\documentclass[12pt, a4paper, openright]{book}
\usepackage{geometry}                % See geometry.pdf to learn the layout options. There are lots.

\usepackage{comment}
\usepackage{color}
\usepackage{fancybox}
\usepackage{mathrsfs}
\usepackage{graphicx}
\usepackage{ascmac}
\usepackage{latexsym}
\usepackage{amsmath,amssymb}
\usepackage{amsthm}
\usepackage{amsfonts}
\usepackage{bm}
\usepackage{ulem}
\usepackage{float}
\usepackage{slashbox}
\usepackage{listings}
\usepackage{cases}
\newcommand{\bra}[1]{\langle {#1} |}
\newcommand{\ket}[1]{| {#1} \rangle}

\newtheorem{lemma}{Lemma}
\newtheorem{definition}{Definition}
\newtheorem{theorem}{Theorem}
\newtheorem{proposition}{Proposition}
\newtheorem{corollary}{Corollary}

\begin{document}
\title{{\Large Doctorate Dissertation}\footnote{This title page is translated from Japanese into English.}\\[1cm]
{\LARGE Entanglement and Causal Relation \\
in distributed quantum computation}\\[6.2cm]
{\large A Dissertation Submitted for Degree of Doctor of Science}\\
{\normalsize December 2015}\\[0.5cm]
{\large Department of Physics, Graduate School of Science,\\The University of Tokyo}\\}
\author{{\large Seiseki Akibue}}
\date{\hfill}                                          % Activate to display a given date or no date

\maketitle

\newpage
\thispagestyle{empty}
\ 
\newpage

\pagenumbering{roman}
\begin{center}
\section*{Abstract}%Why, What and How do we do?
\end{center}
Distributed quantum computation (DQC) is information processing performed over multiple quantum systems connected by a quantum network.
DQC is one of the most promising candidates for realizing a scalable quantum computer.
Quantum communication over the quantum network is indispensable for implementing joint quantum operations over several systems, which is necessary for performing efficient quantum computation in DQC.
Since quantum communication is equivalent to entanglement and classical communication as a resource in DQC, we investigate two different aspects of entanglement and classical communication in DQC.

In the first part of this thesis, we study how to improve the performance of quantum computation over a given quantum network resource by analyzing entanglement resources represented by quantum networks.
To date, quantum networks have been used mainly for quantum communication, i.e. transmitting quantum states between different nodes of the quantum network.
For this purpose, {\it quantum network coding} aiming to improve the performance of quantum communication over a given quantum network has been recently developed.
In contrast, we analyze what kinds of computation can be implemented over a given quantum network resource by introducing a new concept, quantum network coding for quantum computation.
This is because computation can be regarded as a general operations including communication as its special case and it is expected to reduce communication resources in DQC by computing and communicating simultaneously.

We consider a setting of networks where quantum communication for each edge of a network is restricted to sending just one-qubit, but classical communication is unrestricted.
Specifically, we analyze which $k$-qubit unitary operations are implementable over a certain class of networks described by two-dimensional lattices, {\it cluster networks}, by investigating transformations of a given cluster network into quantum circuits. We also analyze which $k$-qubit unitary operations are {\it not} implementable over the cluster networks by using a property of a class of joint quantum operations called {\it separable operations} (SEP). We show that any two-qubit unitary operation is implementable over the {\it butterfly network} and the {\it grail network},  which are fundamental primitive networks for classical network coding. Finally, we analyze probabilistic implementations of unitary operations over cluster networks and obtained necessary and sufficient conditions for implementability.

In the second part of this thesis, we study the role of quantum communication in DQC in terms of entanglement and causal relation in classical communication. 
We investigate resources substituting entanglement and classical communication consumed in entanglement assisted {\it local operations and classical communication} (LOCC).
We start with analyzing the amount of the entanglement resource required for a specific DQC task known as local state discrimination.
The task is discriminating a state from a given set of orthonormal basis states by LOCC with help of entanglement.
We show that entanglement required for the discrimination task allowing only one-round classical communication can be substituted by less entanglement by increasing the rounds of classical communication. 
%Thus we show there exists a tradeoff relation between the amount of entanglement and the rounds of classical communication.

Then we develop a new framework to describe deterministic joint quantum operations in two-party DQC, by using a causal relation between the outputs and inputs of the local operations without predefined causal order but still within quantum mechanics, called {\it ``classical communication" without predefined causal order} (CC*).
We show that local operations and CC* (LOCC*) is equivalent to SEP, which cannot create entanglement from separable states.
This result indicates that entanglement assisted LOCC implementing SEP can be simulated by LOCC*, where no entanglement is needed.
By considering the correspondence between LOCC* and a probabilistic version of LOCC called stochastic LOCC (SLOCC), we show that LOCC* can be interpreted to enhance the success probability of probabilistic operations in SLOCC. We also investigate the relationship between LOCC* and another formalism for deterministic joint quantum operations without predefined partial order (the quantum process formalism) recently developed by Ognyan Oreshkov et al.
Finally we construct an example of non-LOCC SEP by using LOCC*.

\newpage
\chapter*{Acknowledgments}
First, my deepest appreciation goes to my supervisor Professor Mio Murao for all discussions, encouragement, checking a lot of documents and her continuous support through the past five years. Her thought-provoking, accurate and instructive guidance have permitted me to grow my research and find the directions of my research. 
I deeply appreciate Dr. Go Kato and Dr. Masaki Owari as collaborators of a research in this thesis and their kind supports for my next career.
I also appreciate Dr. Akihito Soeda as a collaborator of a research in this thesis and his suggestive comments.
Their advices made enormous contribution to this thesis.

I am deeply grateful to Dr. Peter Turner for all discussions and checking a lot of documents, especially a proposal for Bourses du gouvernement fran\c{c}ais. I would like to appreciate Dr. Damian Markham for giving me an opportunity to stay his group in T\'el\'ecom ParisTech and do a collaborative research for half a year.
I would like to acknowledge to the examiners of this thesis, Dr. Hosho Katsura, Dr. Yasuhiko Arakawa, Dr. Fran\c{c}ois Le Gall, Dr. Naoki Kawashima and Dr. Masato Koashi.
 I would like to thank Yuki Amano and Yumiko Wada for their administrative supports.
I am also indebted to Dr. Harumichi Nishimura, Dr. Barbara Kraus, Dr. Rod Van Meter, Dr. Hiroyasu Tajima, Dr. Takahiko Satoh, Dr. Kenta Takata, Dr. Giulio Chiribella, Dr. Ognyan Oreshkov, Dr. Fabio Costa, Dr. \v{C}aslav Brukner, Dr. Romain All\'eaume, Dr. Eleni Diamanti, Dr. Tom Lawson, Dr. Marc Kaplan, Dr. Alexei Grinbaum, Ms. Christina Giarmatzi, Mr. Amin Baumeler and Mr. Issam Ibnouhsein for their valuable discussions and kind supports. 
I want to thank my colleagues, Dr. Michal Hajdusek, Dr. Yoshifumi Nakata, Dr. Takanori Sugiyama, Dr. Shojun Nakayama, Dr. Eyuri Wakakuwa, Mr. Kotaro Kato, Mr. Jisho Miyazaki, Mr. Kosuke Nakago, Dr. Fabian Furrer, Mr. Atsushi Shimbo, Mr. Hayata Yamasaki, Mr. Ryosuke Sakai, Mr. Hao Qin, Mr. Leonardo Disilvestro, Mr. Adrien Marie, Mr. Thrasyvoulos Karydis and Mr. Adel Sohbi for our discussions and their advice on my research.

I thank all people who I met in Paris when I did the collaborative research as Bourses du gouvernement fran\c{c}ais. Finally, my especial thanks goes to my family, my parents, Makine and Yoshihiro, my wife, Saori, and my son, Nagisa, for their devoted supports and encouragement.

\newpage
\chapter*{Publications}
\noindent \textbf{Journal Articles}
\begin{itemize}
\item A. Soeda, S. Akibue and M. Murao, Two-party LOCC convertibility of quadripartite states and Kraus-Cirac number of two-qubit unitaries, J. Phys. A: Math. and Theo., \textbf{47}, 424036, (2014).
\item S. Akibue and M. Murao, Network coding for distributed quantum computation over cluster and butterfly networks, arXiv:1503.07740, (2015).
\end{itemize}
\noindent \textbf{Conference Talks}
\begin{itemize}
\item \underline{S. Akibue}, G. Kato, M. Owari and M. Murao,  Globalness of separable maps in terms of  classical temporal correlations and quantum spatial correlations, 14th AQIS, Kyoto, Japan, (2014).
\item S. Akibue and \underline{M. Murao}, Implementability of unitary operators over the cluster network with free classical communication, 15th AQIS, Seoul, Korea, (2015).
\item \underline{S. Akibue}, G. Kato, M. Owari and M. Murao,  Globalness of separable maps in terms of time and space resources, 11th QPL, Kyoto, Japan, (2014).
\item \underline{S. Akibue} and D. Markham,  Multipartite correlations with no causal order, New Horizon in Quantum Information Science, Kyoto, Japan, (2014).
\item \underline{S. Akibue} and M. Murao,  Implementability of two-qubit unitary operators over the ladder network with free classical communication, 1st ParQ, Edinburgh, UK, (2013).

\end{itemize}
\noindent \textbf{Conference Proceedings}
\begin{itemize}
\item S. Akibue and M. Murao, Implementability of two-qubit unitary operations over the butterfly network and the ladder network with free classical communication, AIP conference proceedings, 0094-243X ; \textbf{1633}, pp.141f., (2014).
\end{itemize}

\tableofcontents
\pagenumbering{arabic}
\thispagestyle{empty}
\setcounter{page}{-1}
\part{Introduction}
\chapter{Quantum information science}
Quantum information science is an emerging interdisciplinary field of science intersecting quantum physics, information theory and computer science.
In this chapter, we briefly review a historical overview of quantum information science.
A concept of {\it distributed quantum computation} is introduced, and its potential contribution to future information technology and foundations of quantum information science is presented.

\section{Overview of quantum information science}
Quantum mechanics is one of the most significant discovery in science in the twentieth century.
In the early twentieth century, many physicists explored a new theory of physics to capture phenomena that cannot be explained by classical physics such as Newtonian mechanics and electromagnetism. Erwin Schr$\ddot{\rm o}$dinger and Werner Karl Heisenberg have led early developments in formulating quantum theory in mid-1920s. In 1930s, a mathematically rigorous and pragmatic framework of quantum mechanics was established by John von Neumann and Paul Dirac, respectively \cite{vonNeumann, Dirac}.
Quantum mechanics has improved the precision of predictions of empirical results. Moreover, it has substantially changed our understanding of nature since its axioms and consequences are very different from classical physics.
Nowadays, many subfields of physics such as condensed matter physics, optical physics and particle physics, are based on quantum mechanics.
Quantum mechanics has changed not only our understanding of nature but also that of more abstract concepts of {\it information} and {\it computation}.

In information theory, the main interest is to understand how much information we can transmit through a given {\it communication channel}, or just referred as a {\it channel}, physically implemented by an optical fiber, a LAN cable and so on.
A foundation of information theory is built by a paper written by Claude Shannon in 1948 \cite{Shannon}.
He has developed a formalism to describe the amount of information irrespective of the meaning it conveys and the physical systems carrying information.
And he has shown that the amount of information coincides with the optimal compression rate of information under certain setting.
He also defined {\it channel capacity} as the amount of the optimal information transmission rate in a single time by using a channel.
An extension of information theory considering quantum mechanical effects, quantum Shannon theory, has been developed \cite{Schumacher1, Schumacher2} and is still extensively developing.

In computer science, especially in computational complexity theory, understanding properties of computationally difficult problems in principle is the main interest.
A ``computationally difficult problem" is a problem that can be solved by following the right procedure, but takes an extremely long time or a large memory space to solve. For example, as far as we know, factoring a given 1000-bit number is such a typical problem.
The computational difficulty of a problem is evaluated by the optimal time length (time complexity) or the optimal amount of the memory size (space complexity) to solve the problem by using the {\it Turing machine}, a calculation model invented by Alan Turing in 1936 \cite{Turing}.
The easy problem is defined as the problem that can be solved efficiently by Turing machine, i.e. the optimal time length and the optimal amount of the memory size for solving the problem is a polynomial in the size of the input of the problem.
%・その後いくつかの計算機モデル(celluar automatoなど)が提唱されたが、どれも計算能力は等価であることが示された。

Since the computational power of currently widely used silicon-based computers can be regarded as same as Turing machine, computationally difficult problems are intractable by the silicon-based computers.
Then is it really hard to solve a computationally difficult problem no matter how we contrive to solve it? If the nature obeys classical physics, the answer is yes.
Because classical physics can be simulated by the Turing machine efficiently\footnote{It takes hours proportional to $Vt$ for Turing machine to simulate the time evolution of a physical system governed by classical physics, where $V$ is the volume of the physical system and $t$ is the time of the time evolution.}.
However, the nature is governed by quantum mechanics in a microscopic scale. It also seems difficult for Turing machine to simulate quantum mechanics.
In contrast, there  is a possibility that quantum mechanics is efficiently simulatable by using a quantum system.
This implies that {\it quantum computer}, which uses the power of quantum mechanical effects, can be faster than classical computers such as silicon-based computers and Turing machine \cite{Feynman}.

Indeed, quantum algorithms that run faster than all the known classical algorithms are proposed by \cite{Shor, Deutsch3, Grover}.
A significant difference between quantum computation and classical computation is the basic unit of information.
The basic unit of information of quantum computation is a quantum bit, called a {\it qubit}, and that of classical computation is a bit. A bit is a two-valued quantity, and a qubit is a two-level quantum system e.g. a spin of an electron, polarization of a photon.
While a bit can take only two states, 0 or 1, the state representing a qubit can be any superposition of 0 and 1 due to quantum mechanics.
%・BQPは真にP(やPP)よりも大きいのかどうかは不明

We have seen that quantum mechanics provides a new paradigm to information theory and computer science.
Aside from them, quantum mechanics also provides new cryptographic systems \cite{BB84, Crypto}, of which security is guaranteed by the law of physics in contrast to the commonly used cryptographic systems based on computational complexity.
Quantum information science is a field of science to study information processing that uses quantum mechanical effects.
A variety of new information processing schemes have been discovered and rapid progresses in technologies are realizing the schemes.
Quantum information science also provides a new operational perspective of quantum mechanics by using frameworks developed in information theory and computer science.

\section{Distributed Quantum Computation (DQC)}
Distributed computation is computation over a networked computation system in which spatially separated computers are connected by communication channels in order to jointly perform a common task. Cluster computation is an example of distributed computation.
In this thesis, we use ``distributed computation" in a broader sense where each separated computer is not necessary to jointly perform a common task but they can perform their own task by communicating with each other. In this sense, telecommunication and internet are also examples of distributed computation.

{\it Distributed quantum computation} (DQC) is an information processing over multiple separated quantum systems connected by mediating quantum systems (quantum channels), which aims not only quantum computation but also more general distributed tasks, e.g. running a quantum cryptographic protocol.
There are two reasons to assert that DQC will be an infrastructure of the future information society.

First, a practical quantum computer will be based on DQC.
Many different kinds of physical systems have been studied in order to figure out their suitability for implementing quantum computation, such as ion traps \cite{iontrap}, nuclear magnetic resonance \cite{NMR}, quantum dots \cite{quantumdot} and linear optics \cite{linearoptics}. Small quantum computers consisting of several qubits have been already implemented in some systems \cite{Vandersypen, Politi}, however none of these have achieved computation on a scale large enough for practical applications.
Under such circumstances, it is said that DQC is one of the most promising candidates for a scalable quantum computer \cite{Lloyd}.
Moreover, once a practical quantum computer is constructed, it might be cloud computing between a server and end users since quantum computers are highly expensive and the size of the computer will be large.
A secure protocol using quantum communication between the server and the end users is proposed by \cite{BQC}, which can be regarded as DQC over the server and the end users.
Second, quantum cryptography using quantum communication \cite{BB84, Crypto}, which can be regarded as DQC in our definition, is a promising technology for a secure society.
In fact, some quantum cryptography systems have been already commercialized.

In some DQC tasks such as quantum cryptography, quantum communication between the spatially separated quantum systems is an indispensable subroutine. Quantum communication is also an essential resource for DQC for computation to obtain an advantage of quantum mechanical effects since DQC without quantum communication, namely, DQC consisting of a constant size of separated quantum systems connected by just classical communication can be efficiently simulated by classical computers\footnote{It is shown that there exists an advantage of quantum computation using quantum states with marginal quantum entanglement \cite{DQC1-0, DQC1-1, DQC1-2}. However, their protocol allows performing a global unitary operation during computation. In contrast, performing global unitary operations are impossible in DQC without using a quantum channel.}. 
If we need to perform quantum communication by transmitting a quantum state through a quantum channel with small capacity, DQC has to be suspended until all the necessary communications are done, which causes a delay called a {\it bottleneck}.
Quantum communication in DQC is not only transmitting a quantum state from one sender to one receiver via a quantum channel but also transmitting quantum states from many senders to many receivers via a quantum network consisting of quantum channels.
When the scale and complexity of a quantum network grow, the collision of communication pathways between the multiple separated quantum systems causes a serious bottleneck problem, limiting the total performance of DQC.

Some bottleneck problems can be resolved when we can optimize quantum communication so that we reduce the frequency in using quantum channels, however, some bottleneck problems cannot be resolved no matter how we challenge to optimize the communication owing to the restriction originating from the laws of physics.
For a given quantum network, we can define the performance of quantum communication as the set of possible quantum communication within the laws of physics.
The performance of quantum communication may limit the total performance of DQC.
In the first part of this thesis, we are going to investigate the first question:
\begin{itemize}
\item How does the topology of a quantum network consisting of quantum channels affect the performance of quantum communication?
\end{itemize}

The performance of quantum communication through a quantum channel from one sender to one receiver has been extensively studied in quantum Shannon theory \cite{QShannon}.
Recently, the performance of quantum communication through quantum channels from many senders to one receiver, called a multiple-access quantum channel, has been also studied in \cite{MQShannon1, MQShannon2}.
They concentrate on how the performance of quantum communication changes when the capacity of quantum channels is changed while the topology of a quantum network consisting of quantum channels is simple and fixed.

How about the performance of quantum communication over a quantum network consisting of quantum channels between many senders, many receivers and intermediate nodes in addition to senders and receivers?
A crude idea to tackle this problem is just transmitting packets of compressed quantum states and routing the packets at the intermediate nodes like a mail delivery.
However, in \cite{Ahlswede}, it has been shown that processing the packets at the intermediate nodes called {\it network coding} improves communication performance comparing to routing the packets.
In \cite{NCcomp, Soeda}, it was shown that the idea of network coding can be used for computation as well as communication, and communication can be regarded as a special case of computation.

By using the technique of network coding, we analyze implementability of a unitary operation over a given quantum network.
In this thesis, we concentrate on how the performance of quantum communication (and generally, quantum computation) changes when the topology of a quantum network consisting of quantum bipartite channels is changed while the capacity of the quantum channels is fixed.
A unitary operation is an elementary operation in quantum computation and the special class of unitary operations called permutation operations corresponds to transmitting quantum states.
For simplicity, we consider quantum channels are noiseless and have 1-qubit capacity.
We consider a one-shot scenario, i.e. we are allowed to use a given network only once, and concentrate on a {\it cluster network}, which is a certain class of the network consisting of intermediate nodes and the same number of senders and receivers.
The cluster network is a subclass of {\it $k$-pair network}, which has been an actively studied network in both classical network coding and quantum network coding \cite{NCsum, Kobayashi0, Kobayashi1, Kobayashi2}.
Analyzing implementability of a unitary operation over the cluster network can reveal a potential of {\it measurement based quantum computation} (MBQC), an extensively studied model of quantum computation \cite{Raussendorf}.

In the second part of this thesis, Role of entanglement and causal relation in DQC, we are going to investigate the second question:
\begin{itemize}
\item What is the role of quantum communication in DQC?
\end{itemize}
As we mentioned before, quantum communication is indispensable for some DQC tasks and necessary for obtaining a quantum advantage in DQC for computation.
However, how quantum communication enhances DQC has not been fully understood yet\footnote{For example, it is not obvious whether quantum communication is sufficient for obtaining a quantum advantage in DQC for computation. This consideration is rooted in an open problem, whether quantum computation is strictly faster than classical computation.}.
Interpreting the role of quantum communication in a variety of perspectives enriches our understanding DQC and provides a guidance for constructing a new protocol of quantum information processing outperforming the classical counterpart.

Quantum communication can be implemented by quantum teleportation \cite{teleportation} by using {\it quantum entanglement} and classical communication. An entangled state can be shared by quantum communication and classical communication can be performed by quantum communication. Therefore, a pair of entanglement and classical communication is equivalent to quantum communication.
Thus, we investigate the role of quantum entanglement and classical communication in DQC by comparing to another resource substituting them.
%This approach also improves our understanding of the power of quantum entanglement and classical communication in DQC.
%Note that we do not care about the capacity of quantum channels anymore but only assume that they have finite capacity since we would like to investigate the gap between a globally quantum resource and a globally classical resource.
Replacing one resource by another resource in a specific task is a commonly used method in quantum information science in order to understand a role of the resource as presented in the next section.
We summarize what we do in the second part of this thesis in Fig.~\ref{fig:result}.

\begin{figure}[htbp]
 \begin{center}
  \includegraphics[width=95mm]{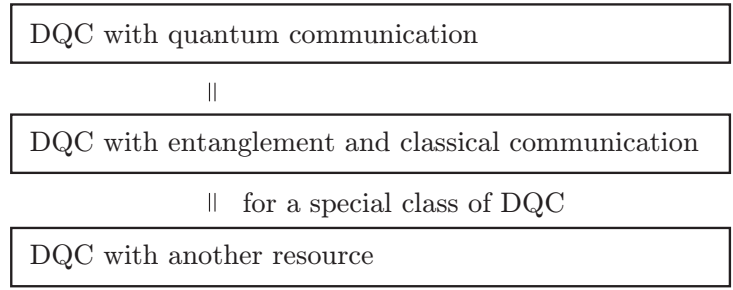}
 \end{center}
 \caption{{\bf An outline of the second part.} In the second part of this thesis, we investigate another resource substituting entanglement and classical communication (or quantum communication) for a task in a special class of DQC in order to understand a role of entanglement and classical communication in DQC.}
 \label{fig:result}
\end{figure}

\section{Quantum entanglement}
Quantum mechanics exhibits many counter-intuitive phenomena that cannot be described in classical physics.  One of such phenomena is the existence of  {\it nonlocal correlations} formulated by Bell and CHSH \cite{Bell}.   They have shown that a quantum mechanical state called an {\it entangled state} shared between spatially separated two parties can produce strong correlations that can never be achieved by any laws of physics based on local realism, e.g.~classical mechanics.   That is, entanglement has a power to enhance correlation in space.   However, entanglement shared between spacelike separated parties cannot be used for communication between the parties.   Indeed, Popescu and Rohrlich \cite{PRbox} have shown that the nonlocal correlations in quantum mechanics is strictly weaker than correlations imposed by the no-signaling condition based on the law of causality of special relativity.   For communication, two parties have to be  within a distance where light can travel, namely, {\it timelike} separated in both quantum and classical cases.

A power of entanglement concerning time also arises when it is accompanied by classical communication.   Quantum teleportation \cite{teleportation} is one of the examples.  Quantum teleportation is a protocol to transmit quantum information represented by an arbitrary and unknown quantum state from a party (sender) to another timelike separated party (receiver) by using shared entanglement and classical communication from the sender to the receiver.  As for transmitting quantum information, quantum teleportation achieves the same goal of direct quantum communication of quantum information, for example, directly sending a photon encoding quantum information through an optical fiber.   But there is an interesting extra property concerning a time-line of events in quantum teleportation.   Quantum communication between the parties (or quantum communication from a mediating third party to both two parties) is necessary in quantum teleportation to share a fixed entangled state between the sender and receiver,  but this event can be done ahead of time before the event that the sender decides what quantum information to send.  The time-limit of the event of the decision is determined  by the timing of classical communication.   In contrast in direct quantum communication, the event of decision should be before the event of quantum communication.  Thus we can slightly ``dodge'' the time-line of the event of quantum communication in transmitting quantum information by using entanglement and classical communication. 
This observation is summarized in Fig.~\ref{fig:timeline}.

\begin{figure}[htbp]
 \begin{center}
  \includegraphics[width=150mm]{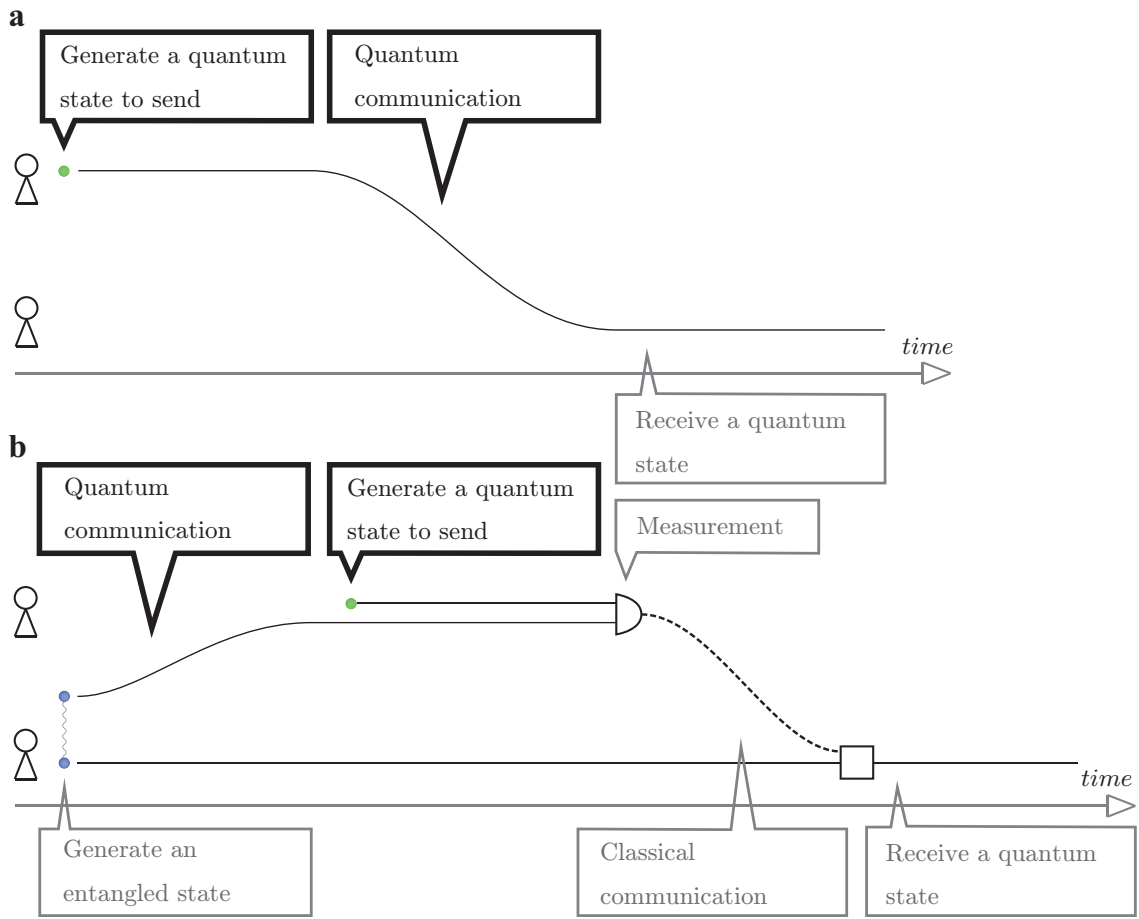}
 \end{center}
 \caption{ {\bf A time-line of events.} 
 ({\bf a}) A time-line of events in direct quantum communication. The event of decision of what quantum information to send should be before the event of quantum communication.
 ({\bf b}) A time-line of events in quantum teleportation. The event of decision can be after the event of quantum communication but should be before the event of classical communication. The detail protocol is shown in Fig.~\ref{fig:Tcircuit}.
}
 \label{fig:timeline}
\end{figure}

These two examples reveal aspects of the power of entanglement concerning space and time.
The power of entanglement can be also understood by analyzing how much entanglement is consumed in a specific information processing task.
For example, entanglement is indispensable for some information processing tasks, such as a specific type of quantum cryptography \cite{Crypto}, quantum teleportation \cite{teleportation} and super dense coding \cite{superdense}. It is considered to be necessary for giving quantum advantage in computation \cite{Vidal, Jozsa} and enhances performances in several information processing tasks, such as classical communication \cite{EAcomm} and communication complexity tasks \cite{CommCmplxty1, CommCmplxty2}.
Analyzing the cost for replacing entanglement by another resource in a specific information processing task \cite{EAcomm, CommCmplxty2} is another way to understand the power of entanglement.
%As we mentioned before, we concentrate on a resource substituting for entanglement and classical communication used in DQC. We can interpret that entanglement assisted classical communication simulates the alternative resource in DQC, which gives DQC with entanglement a potential advantage over DQC without entanglement.
%Since the alternative resource we introduce is a resource concerning space and time, we proceed to the next section.

\section{Space and time}
Spacetime is considered as the background of our description of nature in classical physics and standard quantum field theory, and is considered as dynamically interacting with energy in general relativity.
The difference in the treatment of spacetime is one of the reasons why the unification of quantum mechanics and general relativity into quantum gravity is difficult \cite{Kiefer}.
Space and time are also important notions in computer science. Time complexity (necessary time to compute a problem), space complexity (necessary amount of the memory size to compute a problem) and a tradeoff between them have been extensively studied \cite{Cobham, Nisan}.

DQC can be considered as a joint quantum operation of several parties implemented by connecting each party's quantum operation on a separated system well localized in a spacetime coordinate by using given resources, such as quantum communication, entanglement and classical communication.
For simplicity, we consider DQC between two parties, however, a generalization into several parties is straightforward.
We regard each local operation at a spacetime coordinate belonging to one of the two parties.  
Quantum communication and classical communication can connect two local operations at timelike separated spacetime coordinates linking the output and input of the two operations.
Special relativistic causality introduces a partial order between all the local operations by their spacetime coordinates.
We consider local operations are totally ordered since the partial order can be always extended to a total order.
Entanglement can be shared between any two spacetime coordinates if we assume that entanglement was generated at a spacetime coordinate far in the past, and it has been distributed from that spacetime coordinate.
Since quantum communication can be replaced by quantum teleportation using entanglement assisted classical communication, local operations can be connected by using entanglement and classical communication.
Some of entanglement is shared between different parties and some of entanglement is shared within each party.
We consider a resource substituting entanglement shared between different parties and classical communication between them to understand the role of {\it global} resources.
Examples of a joint quantum operation in DQC is given in Fig.~\ref{fig:DQC}.

\begin{figure}[htbp]
 \begin{center}
  \includegraphics[width=150mm]{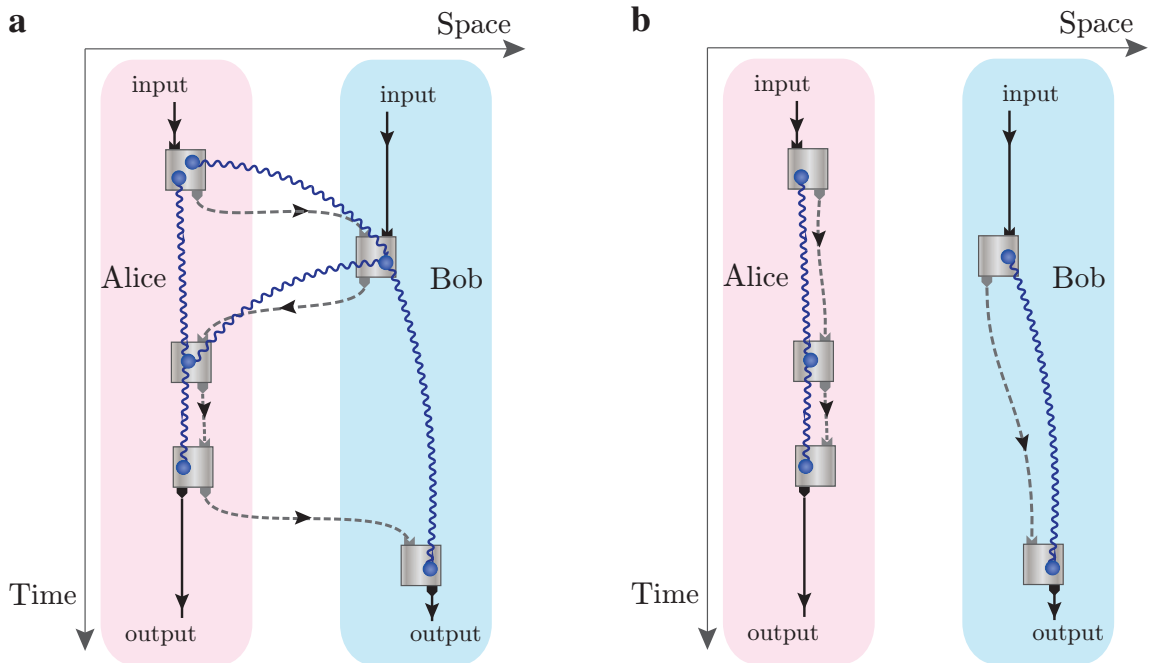}
 \end{center}
 \caption{{\bf Joint quantum operations in DQC.} Dotted arrows represent classical communication between local operations and circles connected by wavy lines represent entangled states between local operations. A box represents a local operation. The first party and the second party is named Alice and Bob, respectively. Entanglement shared between different parties and classical communication between them can be regarded as global resources.
({\bf a}) A joint quantum operation in DQC with global resources, . 
({\bf b}) A joint quantum operation in DQC without global resources. 
 We investigate a resource substituting global resources consumed in a joint quantum operation.}
 \label{fig:DQC}
\end{figure}

The first idea for the alternative resource substituting entanglement and classical communication in DQC is {\it less} entanglement but {\it more} classical communication.
If we are allowed to use more classical communication, it might be possible to reduce the entanglement usage.
We analyze the cost for substituting entanglement by classical communication in a certain DQC task, {\it local state discrimination}, which has been extensively studied in quantum information science \cite{Walgate, 9state, DiVincezo, JNiset}.
Entanglement can be considered as spatial resource since it can generate spatial quantum correlation. On the other hand, the rounds of classical communication between two parties can be considered as temporal resource.
We show that increasing the temporal resource enables decreasing the spatial resource, equivalently, increasing the spatial resource enables decreasing the temporal resource, which can be interpreted as a power of entanglement for parallelizing information processing. 
Analyzing the cost for substituting entanglement by classical communication in other tasks are also investigated in \cite{EAcomm, CommCmplxty2}.

The set of quantum operations in DQC implementable without entanglement between the parties but with classical communication between them is equivalent to a class of quantum operations called {\it local operations and classical communication} (LOCC) \cite{VVedral, MBPlenio, Horodecki, CLMOW12}, which is widely used in quantum information for investigating entanglement and nonlocal properties. 
Thus, if an operation implementable in DQC is an element of LOCC, entanglement consumed in DQC can be replaced by classical communication by definition. That is, we can substitute a global {\it quantum} resource, entanglement, by a completely {\it classical} resource, classical communication.
Can we substitute entanglement by a completely `classical' resource in DQC when we want to implement a quantum operation which is not an element of LOCC?
To study this question, we should clarify implicit assumptions made for classical communication and find out a way for generalization.

Classical communication can connect timelike separated local operations, which implies that special relativistic causality restricts the performance of classical communication, and the total performance of DQC.
However, the performance of classical communication and the total performance of DQC might be changed in the general relativistic spacetime, which allows {\it closed timelike curves} (CTC) as a solution of Einstein's field equations of gravitation, where we cannot define a partial order between spacetime coordinates \cite{CTC1}.
In 1957, Richard Feynman gave the argument as follows: how can we analyze the gravitational field provided by a ball whose position is put into a quantum mechanical superposition? \cite{Feynman2}
Such an argument even implies the possibility of realizing a quantum mechanical superposition of spacetime structures if we crudely combine a consequence of quantum mechanics and general relativity.

Since the existence of the partial order of local operations in the spacetime may not be a fundamental requirement of nature, alternative representations of communication that are not based on the assumption of (the existence of) the partial order have been proposed \cite{OFC, Hardy}.
And it has been shown that there is a possibility of performing joint quantum operations without fixing the partial order of local operations \cite{OFC, Chiribella1, Chiribella2} in the framework of quantum mechanics.
The partial order of the spacetime coordinates are referred to as {\it causal order} in \cite{OFC, Chiribella1, Chiribella2}.
%The partial order of the spacetime coordinates are referred to as {\it partial order} in \cite{OFC, Chiribella1, Chiribella2}.   
Note that they have not tried to analyze quantum mechanics in a curved spacetime \cite{QFTcurve} in this framework, but tried to construct an purely operational formalism without an assumption of the causal order introduced by special relativity but still consistent within standard quantum mechanics. 
Such a challenge would reveal a potential of quantum mechanics and deepen our conceptual understanding of quantum mechanics and causality.
Furthermore, such {\it causally neutral} frameworks give a new insight into existing quantum information processing which have been described by the operational formalism implicitly respecting special relativistic causality.
%And they have shown that standard quantum mechanics is consistent with a new type of causality, called {\it quantum causality}, different from special relativistic causality.

In this thesis, we extend classical communication in such a way that we generalize classical communication into a causal relation between the classical outputs and classical inputs of the local operations, which we call {\it ``classical communication" without predefined causal order} denoted by CC*.
%We focus on {\it classical causal relation} denoted by CC*, which is a correlation between classical systems \footnote{A physical system that can be described by classical physics.} since we would like to substitute quantum entanglement in a `classical' way.
We name a new class of deterministic quantum operations in DQC with CC* but still within quantum mechanics by LOCC*.    We show that LOCC* is equivalent to a certain class of deterministic quantum operations in DQC with entanglement and classical communication known as {\it separable operations} denoted by SEP \cite{VGheorghiu}, which has been introduced for mathematical simplicity to analyze nonlocal quantum tasks in place of LOCC.
Note that there exist elements in SEP that are {\it not} implementable by LOCC \cite{9state, JNiset, DiVincezo, EChitambar}.
That is, if a quantum operation in DQC is an element of SEP but not an element of LOCC, the alternative classical resource is described by CC*. 
 Conventionally two assumptions are put on local operations: (a) they are partially ordered and (b) the choice of a local operation does not depend on resources connecting the local operation. However, when we substitute entanglement consumed in LOCC implementing an element in SEP by a completely `classical' resource, CC* is needed and the two assumptions of local operations have to be relaxed.
Our perspective of understanding the power of entanglement also gives a new characterization of the gap between SEP and LOCC, which has not been well understood \cite{CLMOW12}.

By considering the correspondence between LOCC* and a probabilistic version of LOCC called stochastic LOCC (SLOCC), we analyze the power of LOCC* in terms of enhancing the success probability of probabilistic operations in SLOCC.  We also investigate the relationship between LOCC* and the quantum process formalism for joint quantum operations without partial order developed in \cite{OFC}.
Furthermore, we give an example of the quantum operation which is an element of SEP but not an element of LOCC. Entanglement and classical communication within special relativistic spacetime are necessary to perform such an operations, but by using CC*, entanglement is not necessary.

\section{Organization of this thesis}
This thesis is composed of four parts.

In Part I, we review fundamental formulations of quantum mechanics and quantum information theory.
In Section \ref{sec:QI}, we review the most general physical process described by a  quantum instrument and two mathematical ways to represent the quantum instrument: the Kraus representation and the Choi-Jamiolkowski representation, which are extensively used in this thesis.
In Section \ref{sec:QCM}, we review a model of quantum computation called the circuit model, and its graphical representation.
In Section \ref{sec:LOCC}, we introduce a class of joint quantum operations called local operations and classical communication (LOCC) and related classes of LOCC called stochastic LOCC (SLOCC) and  separable operations (SEP).

In Part II, we study how to improve the performance of quantum computation over a given quantum network resource by analyzing entanglement resources represented by quantum networks.
In Chapter \ref{chap:pre}, we review network coding theory, a model of quantum computation called measurement based quantum computation (MBQC) and classifications of unitary operations in terms of the Kraus-Cirac decomposition and the operator Schmidt decomposition.
In Chapter \ref{chap:QCN}, we give a definition of a $(k,N)$-cluster network and analyze implementability of a $k$-qubit unitary operation over the cluster network in both a deterministic scenario and a probabilistic scenario. We also analyze implementability of a two-qubit unitary operation over the butterfly, grail and square networks.

In Part III, we study the role of quantum communication in DQC in terms of entanglement and causal relation in classical communication.
We investigate resources substituting entanglement and classical communication consumed in entanglement assisted LOCC.
In Chapter \ref{chap:spacetime}, we analyze the amount of the entanglement resource required for a specific DQC task known as local state discrimination.
In Chapter \ref{chap:LOCC*}, we develop a new framework to describe deterministic joint quantum operations in two-party DQC by using ``classical communication" without predefined causal order denoted by CC*.
We show that LOCC* is equivalent to SEP.
We also investigate the relationship between LOCC*, SLOCC and quantum processes.
By using LOCC*, we give an element that resides in the gap between SEP and LOCC.

In Part IV, we summarize the results and outlooks.

\chapter{Preliminaries}
\section{Notation}
The following notation will be used throughout this thesis.
\begin{table}[htb]
\begin{tabular}{ll}
$\overline{a}$&
the complex conjugate of $a$.\\
$a^T$&
the transpose of $a$. The transpose is basis dependent.\\
$a^{\dag}$&
the conjugate transpose of $a$.\\
$\mathcal{H}$&
a finite dimensional Hilbert space.\\
$\mathbf{L}(\mathcal{H})$&
the set of linear operators. \\
$\mathbb{I}_A$&
the identity operator on $\mathcal{H}_A$. \\
$\mathbf{U}(\mathcal{H})$&
the set of unitary operators. $\mathbf{U}(\mathcal{H})=\{M\in\mathbf{L}(\mathcal{H})|M^{\dag}M=\mathbb{I}\}.$\\
$\mathbf{U}_c$&
the set of unitary operators locally unitarily equivalent to\\& a two-qubit controlled phase operator.\\
$\mathbf{Pos}(\mathcal{H})$&
the set of positive semi-definite operators. $\mathbf{Pos}(\mathcal{H})=\{M\in\mathbf{L}(\mathcal{H})|M\geq 0\}.$\\
$\rm{tr}$&
the trace of a linear operator.\\
$\rm{det}$&
the determinant of a linear operator.\\
$\mathbf{D}(\mathcal{H})$&
the set of density operators. $\mathbf{D}(\mathcal{H})=\{M\in\mathbf{Pos}(\mathcal{H})|{\rm tr}[M]=1\}.$\\
$\mathbf{L}(\mathcal{H}_A:\mathcal{H}_B)$&
the set of linear operators. $\mathbf{L}:\mathcal{H}_A\rightarrow\mathcal{H}_B$.\\
$\mathbf{C}(\mathcal{H}_A:\mathcal{H}_B)$&
the set of linear CPTP maps. $\mathbf{C}:\mathbf{L}(\mathcal{H}_A)\rightarrow\mathbf{L}(\mathcal{H}_B)$\\
$\mathbf{U}(\mathcal{H}_A:\mathcal{H}_B)$&
the set of isometry operators. $\mathbf{U}(\mathcal{H}_A:\mathcal{H}_B)=\{M\in\mathbf{L}(\mathcal{H}_A:\mathcal{H}_B)|M^{\dag}M=\mathbb{I}_A\}.$\\
${\rm Sch}\#_B^A(\ket{\psi})$&
the Schmidt rank of $\ket{\psi}\in\mathcal{H}_A\otimes\mathcal{H}_B$.\\
${\rm Op}\#_B^A(M)$&
the operator Schmidt rank of $M\in\mathbf{L}(\mathcal{H}_A\otimes\mathcal{H}_B)$.\\
${\rm KC}\#(U)$&
the Kraus-Cirac number of a two qubit unitary operator $U\in\mathbf{U}(\mathbb{C}^2\otimes\mathbb{C}^2)$.\\

\end{tabular}
\end{table}

\section{Quantum information theory}
In this section, we review fundamental formulations of quantum mechanics and quantum information theory.

\subsection{Quantum mechanics}
\begin{screen}
\textbf{Postulate 1}\,\,\,
A state of a physical system can be described by a vector in a {\it Hilbert space}.
\end{screen}
A Hilbert space is a complex inner product space and is also a complete metric space with respect to the distance function induced by the inner product. An example of Hilbert space is $\mathbb{C}^n$ with respect to the inner product function $(x,y)=\sum_{i=1}^n \overline{x_i}y_i$, where $\overline{a}$ represents the complex conjugate of $a$.
In this thesis, we only consider the cases with a finite $d$-dimensional Hilbert space denoted by $\mathbb{C}^d$ for the reason quoted by Giulio Chiribella et al. \cite{Chiribella3} as follows.
\begin{quotation}
{\it "Another contribution of quantum information has been to shift the emphasis to finite dimensional systems, which allow for a simpler treatment but still possess all the remarkable quantum features. In a sense, the study of finite dimensional systems allows one to decouple the conceptual difficulties in our understanding of quantum theory from the technical difficulties of infinite dimensional systems."}
\end{quotation}
Consider a system $A$ described by the Hilbert space $\mathcal{H}_A$.
We denote a vector in the Hilbert space $\mathcal{H}_A$ by $\ket{\psi}$. We often denote the state by
\begin{equation}
\ket{\psi}_{\mathcal{H}_A}\,\,{\rm or}\,\, \ket{\psi}_A
\end{equation}
in order to specify the Hilbert space or the system that the state belongs to.
We call a two-dimensional Hilbert space a {\it qubit}, which is analogous to a classical bit.
Although the degree of freedom for a qubit is lager than that of a bit, the amount of encodable classical information in a qubit is as same as that to a bit  \cite{Holevo}.

\begin{screen}
\textbf{Postulate 2}\,\,\,
The Hilbert space of a composite physical system consisting of distinct physical systems is given by the tensor product of the Hilbert spaces of the component physical systems.
\end{screen}
Consider two systems $A$ and $B$ with respective their Hilbert spaces $\mathcal{H}_A$ and $\mathcal{H}_B$. If the system $A$ is in state $\ket{\psi}_A$ and the system $B$ in state $\ket{\phi}_B$, the state of total system is given by
\begin{equation}
\label{eq:productstate}
\ket{\psi}_A\otimes\ket{\phi}_B.
\end{equation}
We often use the abbreviated notation $\ket{\psi}_A\ket{\phi}_B$, $\ket{\psi \phi}_{AB}$ or $\ket{\psi,\phi}_{AB}$.
If the state of a composite system can be described in the form of Eq.\eqref{eq:productstate}, the state is called a {\it product state} or a {\it separable state}.
If for any $\ket{\psi}_A\in\mathcal{H}_A$ and $\ket{\phi}_B\in\mathcal{H}_B$ we have
\begin{equation}
\ket{\Psi}_{AB}\neq\ket{\psi}_A\ket{\phi}_B,
\end{equation}
then the state $\ket{\Psi}_{AB}\in\mathcal{H}_A\otimes\mathcal{H}_B$ is called an {\it entangled state}.
For example, an entangled state,
\begin{equation}
\ket{\Phi^+}_{AB}=\frac{1}{\sqrt{2}}(\ket{00}_{AB}+\ket{11}_{AB})
\end{equation}
is called an \textit{EPR state} named after Einstein, Podolsky and Rosen \cite{EPR}, where $\{\ket{0},\ket{1}\}$ is the computational basis.

\begin{screen}
\textbf{Postulate 3}\,\,\,
Any time evolution of a state of a closed quantum system is described by a {\it unitary operator} that corresponds to the time evolution operator of the Hamiltonian of the system.
\end{screen}
The state $\ket{\psi}\in\mathcal{H}$ of the system at time $t_1$ is related to the state $\ket{\psi'}\in\mathcal{H}$ of the system at time $t_2$ by a unitary operator $U\in\mathbf{U}(\mathcal{H})$,
\begin{equation}
\ket{\psi'}=U\ket{\psi}.
\end{equation}
When we control the time evolution of a quantum system that is described by a unitary operator $U$, we regard that a unitary operation $U$ is performed or implemented.

\begin{screen}
\textbf{Postulate 4}\,\,\,
Quantum measurements on a system with Hilbert space $\mathcal{H}$ are described by {\it Positive Operator-Valued Measure} (POVM), a set of positive semi-definite operators $\{M_m\in\mathbf{Pos}(\mathcal{H})\}_{m\in\Omega}$ satisfying the {\it completeness equation},
\begin{equation}
\sum_{m\in\Omega} M_m=\mathbb{I},
\label{eq:measurementcomplete}
\end{equation}
where $m$ corresponds to an index of the possible measurement outcomes. When we measure a state $\ket{\psi}\in\mathcal{H}$ by $\{M_m\}_{m\in\Omega}$, we obtain outcome $m$ with probability
\begin{equation}
p(m)=\bra{\psi}M_m\ket{\psi}.
\end{equation}
\end{screen}
When we perform a quantum measurement described by a POVM whose elements are projection operators, i.e. $M_m^2=M_m$ for all $m$, the measurement is said to be a {\it projective measurement} $\{M_m\}_{m\in\Omega}$.

If we do not have complete knowledge about a state of a quantum system but know a set of possible states $\{\ket{\psi_i}\}_i$ and their probabilities $\{p(i)\}_i$, or an {\it ensemble of states} $\{p(i), \ket{\psi_i}\}_i$,  such a state is described by using a {\it density operator} given by
\begin{equation}
\rho=\sum_i p(i)\ket{\psi_i}\bra{\psi_i}.
\end{equation}
A density operator is an operator on a Hilbert space that is non-negative and has trace equal to one. 
Ensembles of states and density operators are not in one-to-one correspondence. We cannot distinguish states representing different ensembles which are described by the same density operator. So the density operator determines the state of a physical system.
The state is called to be {\it pure} if and only if the state can be represented by a vector in the Hilbert space, that is, the rank of the density operator is one, otherwise the state is said to be {\it mixed}.
The state of a subsystem of a composite quantum system  is provided by the {\it reduced density operator}. Suppose we have physical systems $A$ and $B$, whose composite state is described by a density operator $\rho_{AB}$. The reduced density operator for system $A$ is given by
\begin{equation}
\rho_A={\rm tr}_B[\rho_{AB}]:=\sum_i \bra{i}_B\rho_{AB}\ket{i}_B,
\end{equation}
where ${\rm tr}_B$ is a map between operators known as the partial trace over system $B$, and $\{\ket{i}_B\}_i$ is an orthonormal basis of the Hilbert space of system B. Note that the trace operation does not depend on the orthonormal basis.
The definition of the entangled state is generalized for mixed states as follows.
If the state of a composite system $\rho_{AB}$ is described by
\begin{equation}
\rho_{AB}=\sum_i p(i)\rho_A^{(i)}\otimes\rho_B^{(i)},
\end{equation}
where $\rho_A^{(i)}\in \mathbf{D}(\mathcal{H}_A)$, $\rho_B^{(i)}\in \mathbf{D}(\mathcal{H}_B)$ and $p(i)$ is a probability distribution, the state is called a separable state. If not, the state is called an entangled state.

It is possible to reformulate the postulates of quantum mechanics in terms of density operators instead of state vectors. 

\begin{screen}
\textbf{Postulate 1'}\,\,\,
A state of a physical system can be completely described by a density operator $\rho\in\textbf{D}(\mathcal{H})$.
\end{screen}
\begin{screen}
\textbf{Postulate 2'}\,\,\,
The state of the composite physical system is given by $\rho\in\textbf{D}(\mathcal{H}_1\otimes\mathcal{H}_2\otimes\cdots\otimes\mathcal{H}_n)$, where $\{\mathcal{H}_i\}_i$ are Hilbert spaces of the component physical systems.
\end{screen}
\begin{screen}
\textbf{Postulate 3'}\,\,\,
The time evolution of a closed quantum system is given by $U\rho U^{\dag}$, where $U\in\textbf{U}(\mathcal{H})$.
\end{screen}
\begin{screen}
\textbf{Postulate 4'}\,\,\,
Quantum measurements are described by a POVM $\{M_m\in\mathbf{Pos}(\mathcal{H})\}_{m\in\Omega}$. An outcome $m$ is obtained with probability 
\begin{equation}
p(m)={\rm tr}[M_m\rho].
\end{equation}
When we perform a quantum measurement on a subsystem $\mathcal{H}_A$ of a composite system $\mathcal{H}_A\otimes\mathcal{H}_B$ described by a POVM $\{M_m\in\mathbf{Pos}(\mathcal{H}_A)\}_{m\in\Omega}$, the probability obtaining outcome $m$ is given by
\begin{equation}
p(m)={\rm tr}[(M_m\otimes \mathbb{I}_B)\rho_{AB}],
\end{equation}

and the state of the unmeasured system after the measurement is given by
\begin{equation}
\frac{1}{p(m)}{\rm tr}_{A}[(M_m\otimes \mathbb{I}_B)\rho_{AB}].
\end{equation}
\end{screen}

\subsection{Schmidt decomposition}
Suppose $\ket{\psi}_{AB}$ is a vector in a Hilbert space $\mathcal{H}_A\otimes\mathcal{H}_B$. There exists a set of orthonormal vectors $\{\ket{i}_{A}\in\mathcal{H}_A\}_i$ and a set of orthonormal vectors $\{\ket{i}_{B}\in\mathcal{H}_B\}_i$ such that
\begin{equation}
\ket{\psi}_{AB}=\sum_i \lambda_i \ket{i}_{A}\ket{i}_{B},
\end{equation}
where $\{\lambda_i\}_i$ are non-negative real numbers satisfying $\sum_i \lambda_i^2=1$, known as {\it Schmidt co-efficients}. The number of non-zero coefficients $|\{\lambda_i> 0\}|$ is called as the {\it Schmidt rank} and denoted by ${\rm Sch}\#_B^A(\ket{\psi}_{AB})=|\{\lambda_i> 0\}|$.

\section{Quantum operations}
\label{sec:QI}
A quantum operation is the most general physical process, consisting of unitary evolutions, measurements, discarding subsystems and attaching other subsystems, called {\it ancilla systems}. It represents any realizable physical process that a quantum system can undergo.

Mathematically, a quantum operation can be represented by a {\it quantum instrument}, which is described by a set of linear maps $ \{ \mathcal{M}_o:\mathbf{L}(\mathcal{H}_{in})\rightarrow\mathbf{L}(\mathcal{H}_{out}) \}_o$ that transforms a quantum input state $ \rho\in\mathbf{D}(\mathcal{H}_{in})$ to a quantum output state  given by 
\begin{equation}
\frac{1}{p(o)}\mathcal{M}_o ( \rho)
\end{equation}
 associated with a classical output $o$ ($o=1,2,\cdots, n$) with a probability  distribution 
 \begin{equation}
 p(o)={\rm tr}[\mathcal{M}_o ( \rho) ].
 \end{equation}
If we discard the classical output, a quantum output state is given by
\begin{equation}
\sum_o\mathcal{M}_o ( \rho).
\end{equation}
Each element of instrument $\mathcal{M}_o$ has to be a linear {\it completely positive} (CP) map and a sum of elements $\sum_o \mathcal{M}_o$ has to be a {\it trace preserving} (TP) map to describe quantum operations allowed in quantum mechanics.   A quantum operation conditioned by a classical input $i$ is denoted by $\{ \mathcal{M}_{o|i} \}_o$, a quantum operation without classical output is denoted by $\mathcal{M}_{|i}$ and a quantum operation without classical input and output is denoted by $\mathcal{M}$. $\mathcal{M}_{|i}$ and $\mathcal{M}$ are called deterministic quantum operations, which are described by linear CPTP maps. We define a set of linear CPTP maps from $\mathbf{L}(\mathcal{H}_{in})$ to $\mathbf{L}(\mathcal{H}_{out})$ as $\mathbf{C}(\mathcal{H}_{in}:\mathcal{H}_{out})$. Graphical representations of quantum operations are  given in Fig~\ref{fig:QI}.

\begin{figure}[htbp]
 \begin{center}
  \includegraphics[width=80mm]{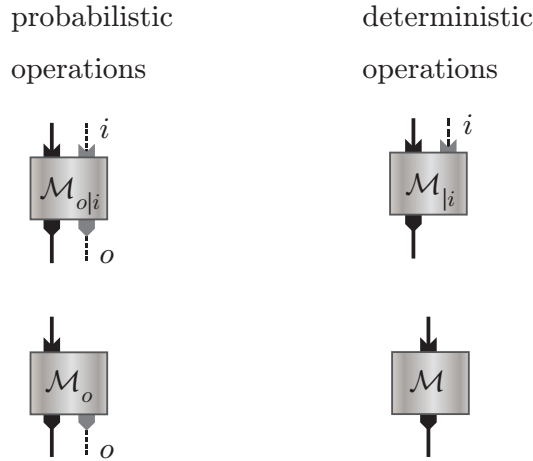}
 \end{center}
 \caption{{\bf Quantum operations}. Probabilistic operations have a classical output corresponding to an outcome of the measurement.  Deterministic operations does not give any classical output.}
 \label{fig:QI}
\end{figure}

There are several mathematical ways to represent the quantum instrument \cite{Choi, Krausrep, Stinespring}. 
We introduce the Kraus representation and the Choi-Jamiolkowski (CJ) representation, which are mainly used in this thesis.
We will show that any quantum instrument is physically realizable, i.e. it can be represented by a sequence of procedures consisting of attaching an ancilla system, applying a unitary time evolution and performing measurement.

\subsubsection{Kraus representation}
Any quantum instrument $\{\mathcal{M}_o:\mathbf{L}(\mathcal{H}_{in})\rightarrow\textbf{L}(\mathcal{H}_{out})\}_o$ can be represented by
\begin{equation}
\mathcal{M}_o(\rho)=\sum_{k} E_{k,o} \rho E_{k,o}^{\dag},
\label{CPTD}
\end{equation}
where $E_{k,o}\in\textbf{L}(\mathcal{H}_{in}:\mathcal{H}_{out})$ satisfying $\sum_{k,o} E_{k,o}^{\dag}E_{k,o}= \mathbb{I}_{in}$ are called \textit{Kraus operators}.
A deterministic quantum operation, a quantum instrument without classical output, $\{\mathcal{M}:\mathbf{L}(\mathcal{H}_{in})\rightarrow\textbf{L}(\mathcal{H}_{out})\}$ is a CPTP map,
\begin{equation}
\mathcal{M}(\rho)=\sum_{k} E_{k} \rho E_{k}^{\dag},
\end{equation}
where $E_{k}\in\textbf{L}(\mathcal{H}_{in}:\mathcal{H}_{out})$ satisfies $\sum_{k} E_{k}^{\dag}E_{k}= \mathbb{I}_{in}$.
We denote the set of all such CPTP maps by $\textbf{C}(\mathcal{H}_{in}:\mathcal{H}_{out})$.
If the quantum operation is deterministic and the number of Kraus operators is one, that is 
\begin{equation}
\mathcal{M}(\rho)=E \rho E^{\dag},
\end{equation}
where $E\in\textbf{L}(\mathcal{H}_{in}:\mathcal{H}_{out})$ satisfies $E^{\dag}E=I_{in}$. Such $E$ are called an \textit{isometry operator}, and the set of all such isometry operators are denoted by $\textbf{U}(\mathcal{H}_{in}:\mathcal{H}_{out})$.
When $\dim(\mathcal{H}_{in})=\dim(\mathcal{H}_{out})$, an isometry operator is equivalent to a unitary operator.

\subsubsection{Choi-Jamiolkowski (CJ) representation}
In Part III, we extensively use the CJ representation to represent quantum operations given by quantum instruments.    For a map representing an element of a quantum instrument $\mathcal{M}_{o|i}:\mathbf{L}(\mathcal{H}_{in})\rightarrow \mathbf{L}(\mathcal{H}_{out})$ where $i$ is an index of the classical input and $o$ is an index of the classical output,  the corresponding CJ operator $M_{o|i}\in \mathbf{L}(\mathcal{H}_{in}\otimes\mathcal{H}_{out})$  is given by
\begin{equation}
M_{o|i}=\sum_{k,l}\ket{k}\bra{l}\otimes\mathcal{M}_{o|i}(\ket{k}\bra{l}),
\end{equation}
where $\{\ket{k}\}_k$ is the computational basis on $\mathcal{H}_{in}$.
 The state of a quantum output for a quantum input $\rho_{in}$  by a linear map $\mathcal{M}_{o|i}$ is obtained by using the CJ operator ${M}_{o|i}$ as 
\begin{equation}
\mathcal{M}_{o|i}(\rho_{in})=\mathrm{tr}_{in}[M_{o|i}(\rho_{in}^T\otimes\mathbb{I}_{out})]
\end{equation}
where $\mathbb{I}_{out}$ is the identity operator on $\mathcal{H}_{out}$ and $\rho_{in}^T$ is the transposition of $\rho_{in}$ with respect to the computational basis.
Note that the output state does not depend on the choice of the computational basis.
$\mathcal{M}_{o|i}$ is completely positive (CP) if and only if $M_{o|i}$ is a positive semi-definite operator. 
$\sum_o\mathcal{M}_{o|i}$ is trace preserving (TP) if and only if $\mathrm{tr}_{out}[\sum_o M_{o|i}]=\mathbb{I}_{in}$.
The CJ representation is unique, i.e. a quantum instrument and a CJ operator are in one-to-one correspondence, while the Kraus representation is not unique.

\subsubsection{Implementation of a quantum instrument}
It is known that there exist physical implementations for any quantum instrument.
\begin{proposition}
A quantum instrument $\{\mathcal{M}_o: \mathbf{L}(\mathcal{H}_{in})\rightarrow\textbf{L}(\mathcal{H}_{out})\}_o$ is physically implementable. 
\end{proposition}
\begin{proof}
 We construct a sequence of procedures realizing a physical process represented by the quantum instrument.
The sequence consists of attaching an ancilla system, applying unitary time evolution and performing measurement.
Let $\{E_{k,o}\in\textbf{L}(\mathcal{H}_{in}:\mathcal{H}_{out})\}_k$ be a set of Kraus operators of $\mathcal{M}_o$.
First, we prepare an initial state $\rho=\sum_i p_i\ket{x_i}\bra{x_i}\in\textbf{D}(\mathcal{H}_{in})$ and an ancilla system $\ket{0}_R\in\mathcal{H}_R$. 
Next, we apply a unitary operator $U\in\mathbf{U}(\mathcal{H}_{in}\otimes\mathcal{H}_R)$ such that
\begin{equation}
U\ket{\psi}_{in}\ket{0}_R=\sum_{k,o}(E_{k,o}\ket{\psi}_{in})\ket{k,o}_M
\end{equation}
for all $\ket{\psi}_{in}\in\mathcal{H}_{in}$, where $\{\ket{k,o}_M\}_{k,o}$ is a set of orthonormal vectors in $\mathcal{H}_M$.
 Note that the dimension of the ancilla system would be changed since the dimension of $\mathcal{H}_{in}$ and that of $\mathcal{H}_{out}$ are different in general. Thus, we denote the ancilla system after performing the unitary operator by $\mathcal{H}_M$, satisfying
\begin{equation}
\dim(\mathcal{H}_{in}\otimes\mathcal{H}_R)=\dim(\mathcal{H}_{out}\otimes\mathcal{H}_{M}).
\end{equation}
We can always construct such a unitary operator $U$ since for any $\ket{\phi}_{in},\ket{\psi}_{in}\in\mathcal{H}_{in}$,
\begin{eqnarray}
\bra{\phi}_{in}\bra{0}_RU^{\dag}U\ket{\psi}_{in}\ket{0}_R&=&\sum_{k,o}\bra{\phi}_{in}E_{k,o}^{\dag}E_{k,o}\ket{\psi}_{in}\\
&=&\bra{\phi}\psi\rangle,
\end{eqnarray}
holds.
Third, we perform a projective measurement $\{\sum_k\ket{k,o}_M\bra{k,o}_M\}_o$ on system $\mathcal{H}_M$. A measurement outcome $o$ is obtained with probability
\begin{eqnarray}
p(o)&=&{\rm tr}\left[(\mathbb{I}_{out}\otimes\sum_k\ket{k,o}_M\bra{k,o}_M) \sum_{a,b,c,d,i}p_iE_{a,b}\ket{x_i}_{in}\bra{x_i}_{in}E_{c,d}\otimes\ket{a,b}_M\bra{c,d}_M\right]\nonumber\\
&=&{\rm tr}\left[\mathcal{M}_o(\rho) \right].
\end{eqnarray}
The state in the system $\mathcal{H}_{out}$ after the measurement is given by
\begin{eqnarray}
\frac{1}{p(o)}{\rm tr}_{M}\left[(\mathbb{I}_{out}\otimes\sum_k\ket{k,o}_M\bra{k,o}_M) \sum_{a,b,c,d,i}p_iE_{a,b}\ket{x_i}_{in}\bra{x_i}_{in}E_{c,d}\otimes\ket{a,b}_M\bra{c,d}_M\right]\nonumber\\
=\frac{1}{p(o)}\mathcal{M}_o(\rho).\,\,\,\,\,\,\,\,\,
\end{eqnarray}
\end{proof}

\section{Quantum computation models}
\label{sec:QCM}
Quantum computation is a process to apply a quantum operation on input state and obtain a classical outputs. There are several models to describe the process. We introduce the circuit model in this section and another model, measurement based quantum computation (MBQC) in Section \ref{sec:MBQC}.
\subsection{Circuit model}
The circuit model of quantum computation describes a quantum operation given by a unitary operation as a combination of \textit{elementary quantum gates}, which corresponds to the elementary logic gates of classical electronic circuits in classical computation \cite{Deutsch2}.
Each wire of a quantum circuit represents a qubit and a sequence of elementary quantum gates are performed on the qubits.
In the first stage, we initialize qubits in a particular input state. In the second stage, we apply a unitary operation corresponding to the computation algorithm to the entire input. This differs from classical computation in that unitary ensures the computation is always reversible. In classical computation, irreversible gates such as the AND gate are used, whose inputs cannot be recovered from the output, meaning that some information about the initial state was lost. Another difference with classical computation is the impossibility of perfectly cloning or copying an unknown quantum state \cite{Noclone}. These facts make the construction of quantum algorithms different from that of classical algorithms. In final, output quantum states are measured.
We use the circuit notation of quantum operations given in Table \ref{table:qcnotation} in this thesis.

\begin{table}[H]
\begin{center}
\begin{tabular}{|c|c|}
\hline
\begin{minipage}{31mm}
\begin{center}
\setlength\unitlength{1truecm}
 \begin{picture}(3,0.5)(-0.45,-0.2)
  \put(0,0.1){\line(1,0){2}}
  \end{picture}
 \end{center}
\end{minipage}
& 
A wire carrying a single qubit (time goes left to right)
\\
\hline
\begin{minipage}{31mm}
\begin{center}
\setlength\unitlength{1truecm}
\begin{picture}(3,0.5)(-0.45,-0.2)
  \put(0,0.1){\line(1,0){2}}
  \put(0,0.15){\line(1,0){2}}
  \end{picture}
\end{center}
\end{minipage}
& 
A wire carrying a single classical bit
\\
\hline
\begin{minipage}{31mm}
\begin{center}
\scalebox{0.1}{\includegraphics{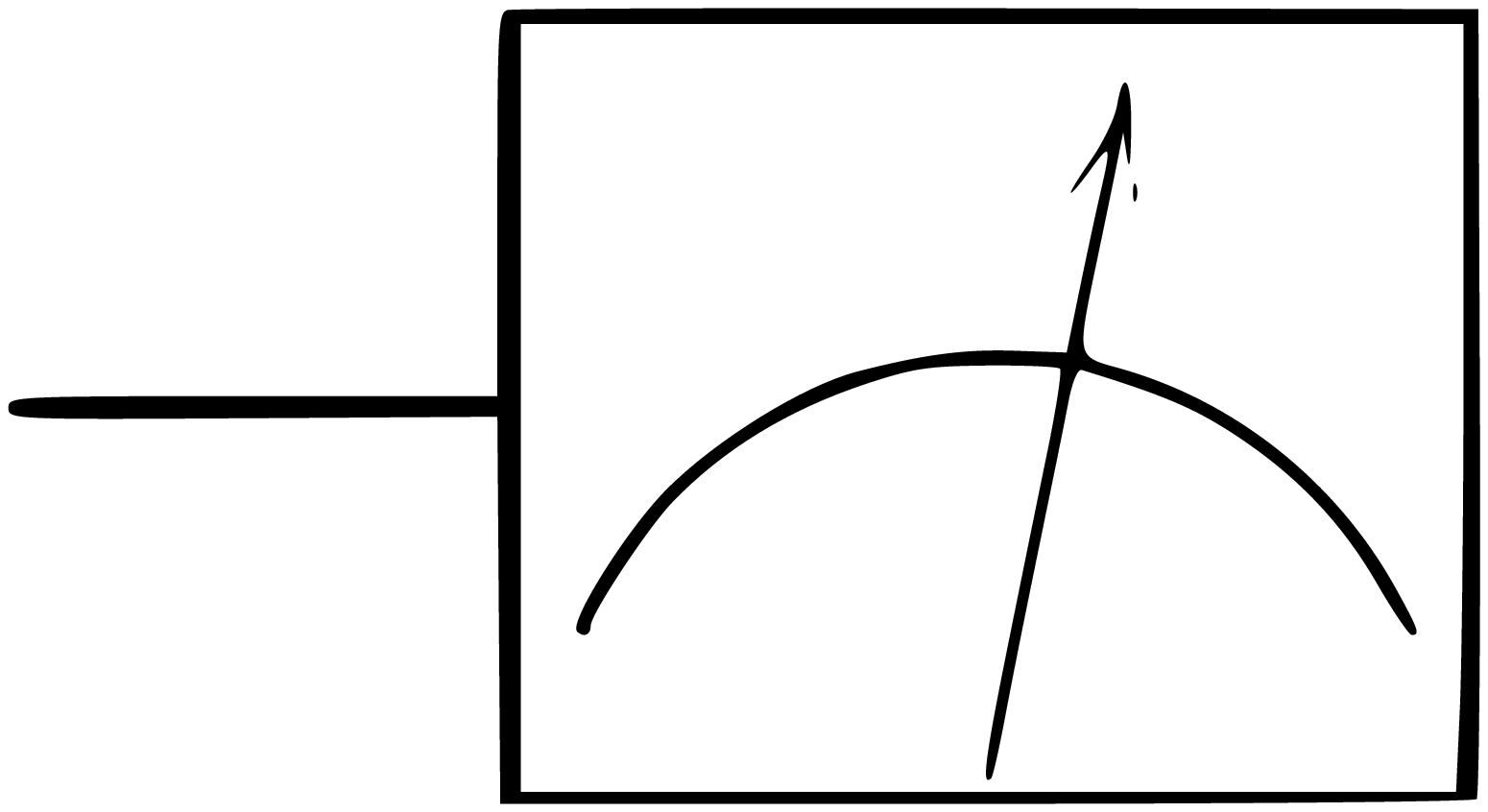}}
\end{center}
\end{minipage}
&
A projective measurement in the \textit{computational basis} $\{ \ket{i}\bra{i} \}_i$.
\\
\hline
\begin{minipage}{31mm}
\begin{center}
\scalebox{1}{\includegraphics{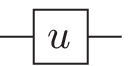}}
\end{center}
\end{minipage}
&
A single unitary gate $u$.
\\
\hline
\begin{minipage}{31mm}
\begin{center}
\scalebox{1}{\includegraphics{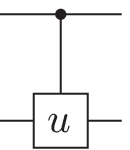}}
\end{center}
\end{minipage}
&
A controlled unitary gate $U=\ket{0}\bra{0}\otimes\mathbb{I}+\ket{1}\bra{1}\otimes u$.
\\
\hline
\begin{minipage}{31mm}
\begin{center}
\scalebox{0.06}{\includegraphics{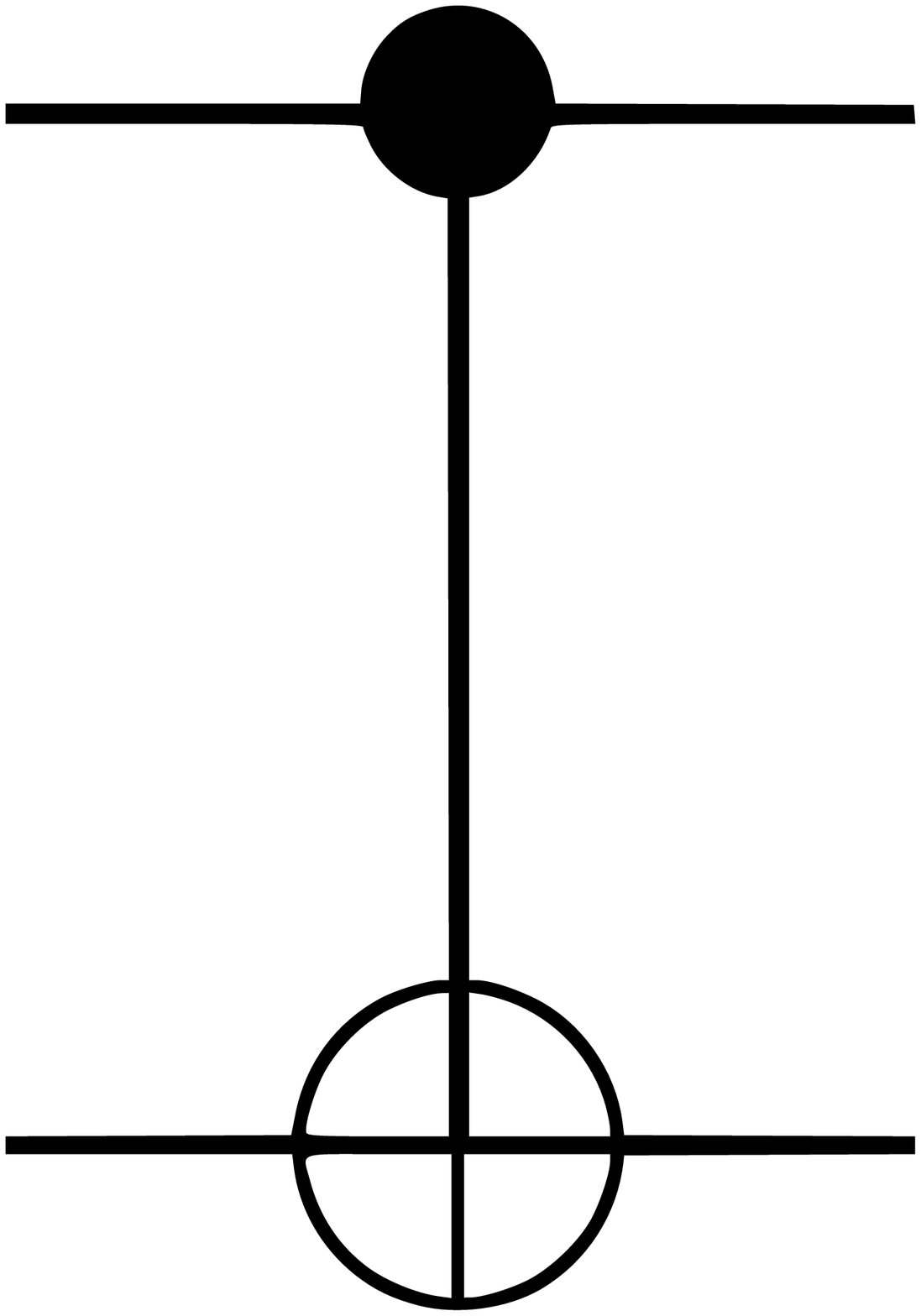}}
\end{center}
\end{minipage}
&
A CNOT gate $U^{CNOT} := \ket{00}\bra{00}+\ket{01}\bra{01}+\ket{10}\bra{11}+\ket{11}\bra{10}$.
\\
\hline
\begin{minipage}{31mm}
\begin{center}
\scalebox{0.1}{\includegraphics{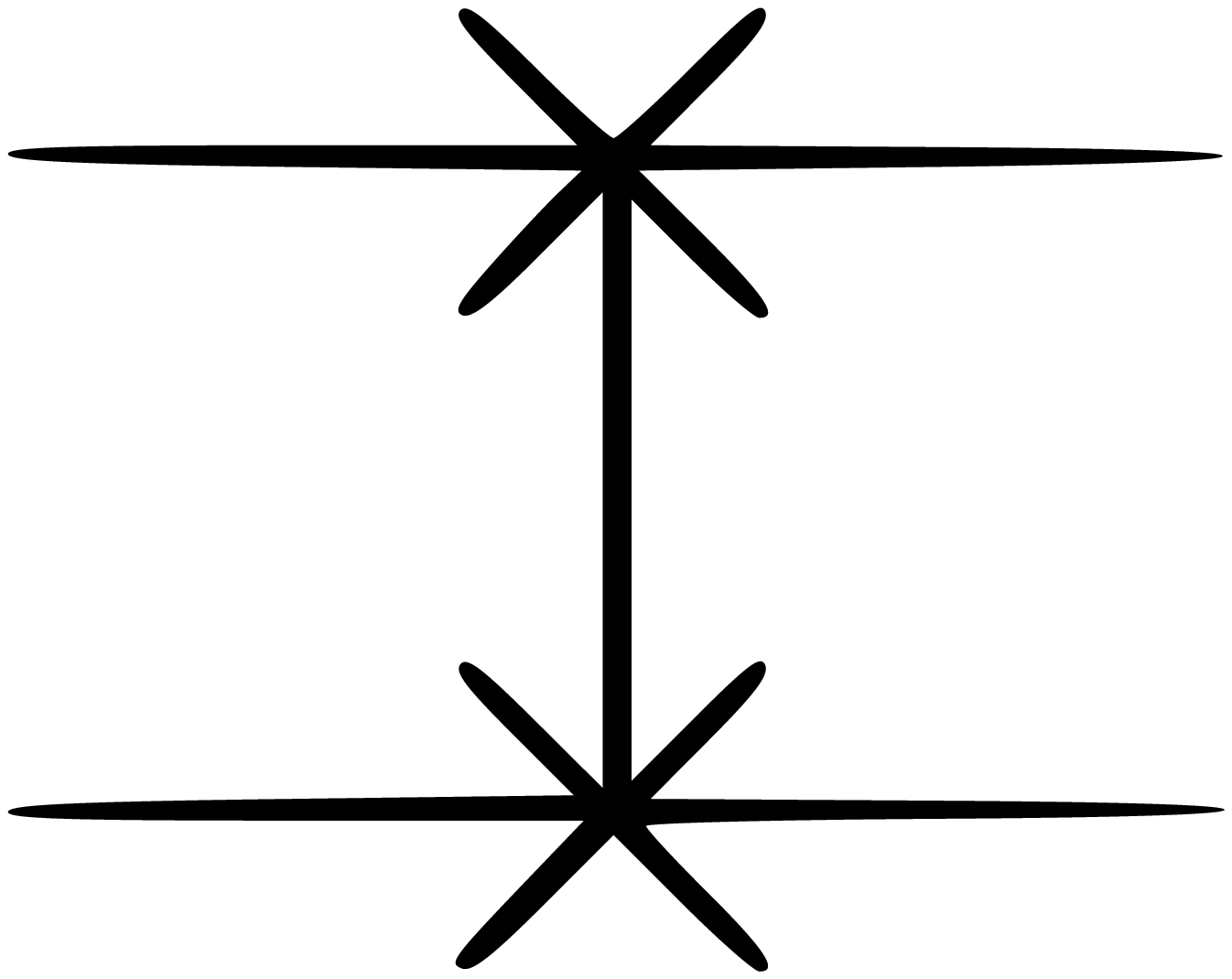}}
\end{center}
\end{minipage}
&
A SWAP operation $U^{SWAP} := \ket{00}\bra{00}+\ket{01}\bra{10}+\ket{10}\bra{01}+\ket{11}\bra{11} $.
\\
\hline
\end{tabular}
\end{center}
\caption{Notation of operations used in quantum circuits.}
\label{table:qcnotation}
\end{table}
Single qubit operators are described by $2\times 2$ unitary matrices in the computational basis, of which the Pauli operators
\begin{equation}
X:=\left(\begin{array}{cc}0 & 1 \\1 & 0\end{array}\right),\,\,\,\,\,\,\,\,Y:=\left(\begin{array}{cc}0 & -i \\i & 0\end{array}\right),\,\,\,\,\,\,\,\,Z:=\left(\begin{array}{cc}1 & 0 \\0 & -1\end{array}\right)
\end{equation}
are often used in this thesis. Any single qubit unitary operator $u$ can be decomposed as
\begin{equation}
u=R_z(\alpha)R_y(\beta)R_z(\gamma),
\label{eq:Eulerdec}
\end{equation}
where $R_z(\theta)=\exp(-i\frac{\theta}{2} Z)$ and $R_y(\theta)=\exp(-i\frac{\theta}{2} Y)$, which is called the Euler decomposition.

We also introduce the Hadamard gate and the phase gate as
\begin{equation}
H:=\frac{1}{\sqrt{2}}\left(\begin{array}{cc}1 & 1 \\1 & -1\end{array}\right),\,\,\,\,\,\,\,\,S:=\left(\begin{array}{cc}1 & 0 \\0 & i\end{array}\right).
\end{equation}
Useful properties of single qubit operators are often used, such as
\begin{equation}
HXH=Z,\,\,\,\,\,\,\,\,SXS^{\dag}=Y,\,\,\,\,\,\,\,\,(SH)Z(SH)^{\dag}=Y.
\end{equation}
For example, these properties are used to easily check,
\begin{equation}
(I\otimes H)U^{CNOT}(I\otimes H)=U^{CZ},
\end{equation}
where $U^{CNOT}$ is the controlled NOT gate (controlled X gate) given by,
\begin{equation}
U^{CNOT}=\left(\begin{array}{cccc}1 & 0 & 0 & 0 \\0 & 1 & 0 & 0 \\0 & 0 & 0 & 1 \\0 & 0 & 1 & 0\end{array}\right),
\end{equation}
and $U^{CZ}$ is the controlled Z gate given by,
\begin{equation}
U^{CZ}=\left(\begin{array}{cccc}1 & 0 & 0 & 0 \\0 & 1 & 0 & 0 \\0 & 0 & 1 & 0 \\0 & 0 & 0 & -1\end{array}\right).
\end{equation}
The controlled NOT gate and the controlled Z gate are special cases of the {\it controlled unitary} operators. The two-qubit controlled unitary operator $U$ is described by
\begin{equation}
U=\ket{0}\bra{0}_C\otimes\mathbb{I}_T+\ket{1}\bra{1}_C\otimes u_T=\left(\begin{array}{cccc}1 & 0 & 0 & 0 \\0 & 1 & 0 & 0 \\0 & 0 & u_{0,0} & u_{0,1} \\0 & 0 & u_{1,0} & u_{1,1}\end{array}\right),
\end{equation}
where $u=\left(\begin{array}{cc} u_{0,0} & u_{0,1} \\ u_{1,0} & u_{1,1}\end{array}\right)$ is a single qubit unitary operator. 
The first system $\mathcal{H}_C$ is called a {\it controlled qubit} and the second system $\mathcal{H}_T$ is called a {\it target qubit}.
Properties of the controlled unitary operators are shown in Part II.

A unitary operation representing an algorithm in the circuit model can be decomposed into a set of smaller unitary `gates' acting on subsets of qubits. A \textit{universal gate set} is a set of unitary gates into which any unitary operation can be decomposed. There exists a variety of universal gate sets. For example, the two-qubit CNOT gate along with arbitrary single qubit unitary gates form a universal gate set.

\section{Local Operations and Classical Communication (LOCC)}
\label{sec:LOCC}
In this section, we introduce a class of joint quantum operations between distant parties where they cannot directly apply any joint quantum operations, but can transmit classical information to each other depending on the outcomes of local quantum operations on their respective subsystems. Such a class of joint quantum operations are referred to as {\it local operations and classical communication}, which is abbreviated to LOCC. 
We also introduce two classes of joint quantum operations related to LOCC.
At the end of this section, we review a more general class of operations implementable by LOCC assisted by pre-shared entanglement.

\subsection{LOCC}
LOCC is a set of deterministic joint quantum operations consisting of a sequence of local quantum operations and classical communication. Probabilistic quantum operations can be also applied in the sequence, however, to make the joint operation to be deterministic, we need to discard (sum up) all the measurement outcomes after all the operations in the sequence are completed. LOCC is widely used in quantum information science for describing practical experimental settings, where global quantum communication is much harder to implement than classical communication. For example, quantifying quantum entanglement \cite{VVedral, MBPlenio} and quantifying the globalness of unitary operations \cite{Soedaglobal} are investigated in terms of LOCC.
We first introduce the simplest class, bipartite one-way LOCC. Then we introduce a more general class, bipartite two-way LOCC, and present the most general class, multipartite two-way LOCC.

\subsubsection{Bipartite one-way LOCC}
We consider the scenario that Alice first performs a local operation represented by $\{ \mathcal{A}_o \}_o$, then sends a classical output $o$ to Bob and Bob performs a local operation represented by $\mathcal{B}_{|o}$ conditioned by the classical input $o$ .   Since we assume that Alice and Bob are acting on different quantum systems at different spacetime coordinates,  the joint quantum operation is described by a tensor product of two local operations. By taking averages over $o$,  we obtain a deterministic joint quantum operation given by 
\begin{equation}
\mathcal{M}=\sum_{o} \mathcal{A}_o\otimes\mathcal{B}_{|o}.
\label{eq:oneLOCC}
\end{equation}
One-way LOCC from Alice to Bob is the set of quantum operations in the form of Eq.\eqref{eq:oneLOCC}. One-way LOCC from Bob to Alice can be defined in a similar way. 

\subsubsection{Bipartite two-way LOCC}
A deterministic joint operation represented by local operations and more general two-way (finite-round) classical communication between two parties is defined by a sequence of Alice's local operations given by $\{ \mathcal{A}_{o_N|i_N}^{(N)}\circ\cdots\circ \mathcal{A}_{o_1}^{(1)} \}_{o_N,\cdots,o_1}$ and another sequence of Bob's local operations  given by $\{ \mathcal{B}_{|i_N'}^{(N)}\circ\cdots\circ \mathcal{B}_{o_1'|i_1'}^{(1)} \}_{o_N',\cdots,o_1'}$.   Here $\circ$ denotes a  connection between two local operations linking a quantum output of a local operation and a quantum input of the next local operation.  The indices $i_k$ and $i'_k$ are classical inputs of the $k$-th operations and $o_k$ and $o'_k$ are classical outputs of the $k$-th operations of Alice and Bob, respectively.   We assume that Alice sends $o_1$ to Bob and Bob receives the classical message as $i'_1$, i.e. $i'_1=o_1$, and similarly, Bob sends $o'_1$ to Alice and so on.
Thus,  bipartite two-way LOCC between the two parties is defined by a set of joint quantum operations represented by

\begin{eqnarray}
\mathcal{M}
&=& \sum_{i_2,\cdots,i_N,o_1,\cdots,o_N} \mathcal{A}_{o_N|i_N}^{(N)}\circ\cdots\circ \mathcal{A}_{o_1}^{(1)}\otimes\mathcal{B}_{|o_N}^{(N)}\circ\cdots\circ\mathcal{B}_{i_2|o_1}^{(1)}.
\end{eqnarray}
An example for $N=3$ is shown in Fig.\ref{fig:twoCC}.  We refer to bipartite two-way LOCC as just bipartite LOCC.
 
 \begin{figure}
 \centering
  \includegraphics[height=.45\textheight]{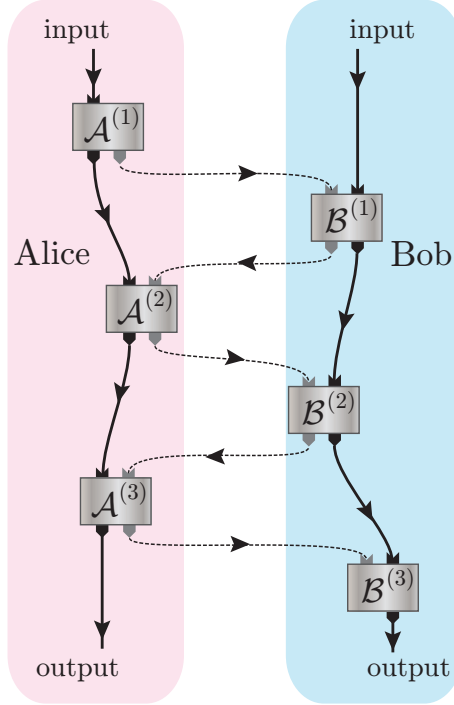}
  \caption{{\bf A bipartite LOCC protocol.} Dotted arrows represent classical communication between local operations and solid arrows represent quantum communication between local operations. A box represents a local operation.}
\label{fig:twoCC}
\end{figure} 

\subsubsection{Multipartite two-way LOCC}
We consider a situation in which multiple (more than two) parties collectively perform local operations, depending on the outcomes of the previous operations.
A set of such deterministic joint quantum operations is the most general LOCC, called multipartite two-way LOCC.
We denote multipartite two-way LOCC as $\mathbf{LOCC}(\mathcal{H}_{I_1},\cdots,\mathcal{H}_{I_N} : \mathcal{H}_{O_1},\cdots,\mathcal{H}_{O_N} )$, where there are $N$ parties and a sequence of local operations performed by the $i$-th party denoted by a quantum operation from $\mathbf{L}(\mathcal{H}_{I_i})$ to $\mathbf{L}(\mathcal{H}_{O_i})$.
Since an element of LOCC is a deterministic quantum operation by definition, we obtain
\begin{equation}
\mathbf{LOCC}(\mathcal{H}_{I_1},\cdots,\mathcal{H}_{I_N} : \mathcal{H}_{O_1}, \cdots,\mathcal{H}_{O_N})\subsetneq\mathbf{C}(\mathcal{H}_{I_1}\otimes \cdots\otimes \mathcal{H}_{I_N} :\mathcal{H}_{O_1}\otimes \cdots\otimes\mathcal{H}_{O_N}).
\end{equation}
In this thesis, we refer to multipartite two-way LOCC as just LOCC.
Note that local operations in LOCC are totally ordered.

\subsection{Separable operation}
We introduce a \textit{separable operation}, which does not describe a physical process but a mathematically conceptual process \cite{VGheorghiu}.
In spite of its clear operational meaning, analysis of quantum information processing tasks under the restriction of LOCC is hard in general since the mathematical structure of LOCC is highly complicated \cite{CLMOW12}.
The separable operation is known to be used to analyze nonlocal quantum tasks in place of LOCC \cite{Akibue, Anthony, DistinguishbySEP} since it has a simpler mathematical structure than LOCC and the class of separable operations is a slightly larger class of deterministic quantum operations than LOCC.
However, only limited number of examples of deterministic quantum operations that are separable operations but {\it not} included in LOCC are known so far \cite{9state, JNiset, DiVincezo, EChitambar}.

\begin{definition}
A CPTP map $\Phi\in {\mathbf C}(\mathcal{H}_{I_1}\otimes\cdots\otimes\mathcal{H}_{I_N} :\mathcal{H}_{O_1}\otimes\cdots\otimes\mathcal{H}_{O_N})$, is said to be a separable operation, if and only if there exists linear operators $\{E_k^{(1)}\in \textbf{L}(\mathcal{H}_{I_1},\mathcal{H}_{O_1})\},\cdots,\{E_k^{(N)}\in \textbf{L}(\mathcal{H}_{I_N},\mathcal{H}_{O_N})\}$ such that
\begin{equation}
\Phi(\rho)=\sum_{k} (E_k^{(1)}\otimes\cdots\otimes E_k^{(N)})\rho (E_k^{(1)}\otimes\cdots\otimes E_k^{(N)})^{\dag}
\label{eq:sCPTP}
\end{equation}
for all $\rho\in \textbf{D}(\mathcal{H}_{I_1}\otimes\cdots\otimes\mathcal{H}_{I_N})$. 
We refer to the set of separable operations as SEP and denote as $\mathbf{SEP}(\mathcal{H}_{I_1},\cdots,\mathcal{H}_{I_N}: \mathcal{H}_{O_1},\cdots,\mathcal{H}_{O_N})$.
\end{definition}
Following relations are obtained by definition of each class.
\begin{eqnarray}
\mathbf{LOCC}(\mathcal{H}_{I_1},\cdots,\mathcal{H}_{I_N} : \mathcal{H}_{O_1}, \cdots,\mathcal{H}_{O_N})&\subsetneq&\mathbf{SEP}(\mathcal{H}_{I_1},\cdots,\mathcal{H}_{I_N} : \mathcal{H}_{O_1}, \cdots,\mathcal{H}_{O_N}),\nonumber\\\\
\mathbf{SEP}(\mathcal{H}_{I_1},\cdots,\mathcal{H}_{I_N} : \mathcal{H}_{O_1}, \cdots,\mathcal{H}_{O_N})&\subsetneq&\mathbf{C}(\mathcal{H}_{I_1}\otimes \cdots\otimes \mathcal{H}_{I_N} :\mathcal{H}_{O_1}\otimes \cdots\otimes\mathcal{H}_{O_N}).\nonumber\\
\end{eqnarray}

\subsection{Stochastic LOCC}
{\it Stochastic LOCC} (SLOCC) is a set of probabilistic joint quantum operations consisting of a sequence of local quantum operations and classical communication. It is not necessary to average over all the measurement outcomes in SLOCC in contrast to LOCC. Bipartite SLOCC is defined by a set of linear CP maps represented by
\begin{eqnarray}
\mathcal{M}
&=& \sum_{(i_1,\cdots,i_N,o_1,\cdots,o_N)\in\mathbb{M}} \mathcal{A}_{o_N|i_N}^{(N)}\circ\cdots\circ \mathcal{A}_{o_1}^{(1)}\otimes\mathcal{B}_{|o_N}^{(N)}\circ\cdots\circ\mathcal{B}_{i_2|o_1}^{(1)},
\end{eqnarray}
where $\mathbb{M}$ is a subset of the set of all the measurement outcomes $\overline{\mathbb{M}}$. Note that an element of SLOCC is not necessary to be TP.
For an input state $\rho$, the probability of obtaining a map in SLOCC $\mathcal{M}$ is given by
\begin{equation}
{\rm tr}[\mathcal{M}(\rho)],
\end{equation}
and the final state after post-selection is given by
\begin{equation}
\frac{\mathcal{M}(\rho)}{{\rm tr}[\mathcal{M}(\rho)]}.
\end{equation}
We denote multipartite SLOCC as $\mathbf{SLOCC}(\mathcal{H}_{I_1},\cdots,\mathcal{H}_{I_N}: \mathcal{H}_{O_1},\cdots,\mathcal{H}_{O_N})$.
By definition, SLOCC is strictly larger than LOCC, i.e.
\begin{eqnarray}
\mathbf{LOCC}(\mathcal{H}_{I_1},\cdots,\mathcal{H}_{I_N} : \mathcal{H}_{O_1}, \cdots,\mathcal{H}_{O_N})&\subsetneq&\mathbf{SLOCC}(\mathcal{H}_{I_1},\cdots,\mathcal{H}_{I_N} : \mathcal{H}_{O_1}, \cdots,\mathcal{H}_{O_N}).\nonumber\\
\end{eqnarray}
In a practical experimental setting, it is not difficult to perform post-selection in a LOCC protocol, which gives a SLOCC protocol, and such a probabilistic operation has a stronger power than a deterministic operation, such as increasing entanglement, discriminating non-orthonormal states perfectly, and enhancing the computational power \cite{Aaronson}.

We can show that any TP elements of SLOCC are included in LOCC as follows.
\begin{proof}
Let $\{M_m\in\mathbf{Pos}(\mathcal{H}_{I_1}\otimes\cdots\otimes\mathcal{H}_{O_N})\}_{m\in\mathbb{M}}$ be the CJ operator of an element of SLOCC, where $\mathbb{M}$ is the set of all the measurement outcomes in a sequence of SLOCC. Then there exists a set of measurement outcomes $\overline{\mathbb{M}}$ such that $\mathbb{M}\subseteq\overline{\mathbb{M}}$ and an element of LOCC whose CJ operator is given by $\{M_m\}_{m\in\overline{\mathbb{M}}}$, i.e.
\begin{equation}
{\rm tr}_{O_1,\cdots,O_N}\left[\sum_{m\in\overline{\mathbb{M}}}M_m\right]=\mathbb{I}_{I_1,\cdots,I_N}.
\end{equation}
If the element of SLOCC is TP, then we obtain
\begin{equation}
{\rm tr}_{O_1,\cdots,O_N}\left[\sum_{m\in\mathbb{M}}M_m\right]=\mathbb{I}_{I_1,\cdots,I_N}.
\end{equation}
These two equations imply
\begin{equation}
{\rm tr}_{O_1,\cdots,O_N}\left[\sum_{m\in\overline{\mathbb{M}}\setminus\mathbb{M}}M_m\right]=0
\end{equation}
and
\begin{equation}
\forall m\in\overline{\mathbb{M}}\setminus\mathbb{M},\,\,\, M_m=0
\end{equation}
since $M_m$ is a positive semi-definite operator. Therefore, the LOCC map $\{M_m\}_{m\in\overline{\mathbb{M}}}$ represents the same map as the SLOCC map $\{M_m\}_{m\in\mathbb{M}}$.

\end{proof}
That is,
\begin{eqnarray}
\mathbf{LOCC}(\mathcal{H}_{I_1},\cdots,\mathcal{H}_{I_N} : \mathcal{H}_{O_1}, \cdots,\mathcal{H}_{O_N})=\,\,\,\,\,\,\,\,\,\,\,\,\,\,\,\,\,\,\,\,\,\,\,\,\,\,\,\,\,\,\,\,\,\,\,\,\,\,\,\,\,\,\,\,\,\,\,\,\,\,\,\,\,\,\,\,\,\,\,\,\,\,\,\,\,\,\,\,\,\,\,\,\,\,\,\,\,\,\,\,\,\,\,\,\,\,\,\,\,\,\,\,\,\,\,\,\,\,\,\,\,\,\,\,\,\nonumber\\
\mathbf{SLOCC}(\mathcal{H}_{I_1},\cdots,\mathcal{H}_{I_N} : \mathcal{H}_{O_1}, \cdots,\mathcal{H}_{O_N})\cap{\mathbf C}(\mathcal{H}_{I_1}\otimes\cdots\otimes\mathcal{H}_{I_N} :\mathcal{H}_{O_1}\otimes\cdots\otimes\mathcal{H}_{O_N}).\nonumber\\
\end{eqnarray}

\subsection{Entanglement assisted LOCC}
When LOCC accompanies with pre-shared entanglement, called {\it entanglement assisted LOCC}, it can implement general global quantum operations acting on the systems of the multiple parties.
We review two types of global quantum operations, {\it quantum teleportation} and a controlled unitary operation, which appear in this thesis.
\subsubsection{Quantum teleportation}
Entanglement shared between spacelike separated parties cannot be used for communication between parties and a finite amount of classical communication cannot be used for exact transmission of an arbitrary quantum state.
However, if entanglement is accompanied by classical communication, exact transmission of an arbitrary quantum state can be achieved.
Such a process is called {\it quantum teleportation} \cite{teleportation}.
In a quantum teleportation protocol, there are two parties, a sender and a receiver.
The sender has a quantum state to be transmitted and an entangled state is pre-shared between the sender and the receiver.
First, the sender performs an appropriate measurement on her system.
Second, she sends the measurement outcome to the receiver.
Third, the receiver performs an appropriate unitary operation on his system corresponding to the measurement outcome.
After all the procedures are done, the state of receiver's system is transformed into a quantum state that was initially possessed by the sender and the state of sender's system is transformed into a particular state corresponding to the measurement outcome not depending on the quantum state initially possessed by her.

We give a quantum circuit representation of the quantum teleportation protocol $\mathcal{T}$ to transmit a one-qubit state by using an EPR state in Fig.~\ref{fig:Tcircuit}.
\begin{figure}
\begin{center}
  \includegraphics[height=.15\textheight]{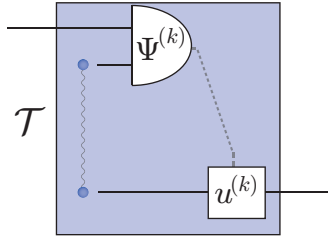}
  \end{center}
  \caption{{\bf A quantum circuit representation of the quantum teleportation protocol $\mathcal{T}$.} $\Psi^{(k)}$ represents the projective measurement $\{\ket{\Psi^{(k)}}\bra{\Psi^{(k)}}\}_k$, $\{\ket{\Psi^{(k)}}\}_k=\{ (u^{(k)}\otimes \mathbb{I})\ket{\Phi^+} \}$ is the Bell basis, $\ket{\Phi^+}$ is an EPR state and $\{u^{(k)}\}_k=\{I,Z,X,ZX\}$ is a set of operations to be applied conditional on the measurement outcome specified by $k$.}
\label{fig:Tcircuit}
\end{figure}

Note that the quantum teleportation protocol is achieved by entanglement assisted bipartite one-way LOCC.
Performing quantum teleportation in both directions, it is possible to implement any global quantum operations between spatially separated parties.
Entanglement consumed in entanglement assisted LOCC is used for quantifying the globalness of a quantum operation \cite{SoedaEPR, DanEPR}.

\subsubsection{Controlled unitary operation}
As we see, it is possible to implement any two-qubit unitary operations between two parties by bi-directional quantum teleportation using two EPR states and bipartite two-way LOCC.
Only one EPR state is sufficient for implementing a two-qubit controlled unitary operation whereas one EPR state is not sufficient for implementing the SWAP operation.
We give a quantum circuit representation of performing a two-qubit controlled unitary operation $U=\ket{0}\bra{0}\otimes\mathbb{I}+\ket{1}\bra{1}\otimes u$ in Fig.~\ref{fig:CUcircuit}.

\begin{figure}
\begin{center}
  \includegraphics[height=.3\textheight]{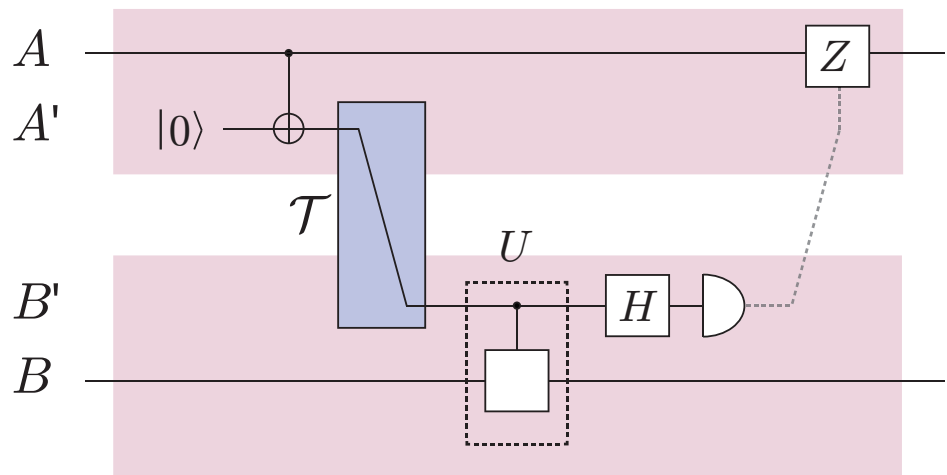}
  \end{center}
  \caption{{\bf A quantum circuit representation of the LOCC protocol implementing a two-qubit controlled unitary operation $U$.} Qubits in the first shaded region are possessed by the first party, those in the second shaded region are possessed by the second party. The protocol consists of introducing an ancillary qubits $\mathcal{H}_{A'}$ at the first party, teleporting the ancillary qubit state from the first party to the second party  represented by qubit $\mathcal{H}_{B'}$, applying $U$ on a controlled qubit $\mathcal{H}_{B'}$  and a target qubit $\mathcal{H}_{B}$ at the second party, performing Hadamard operations and measurements in the computational basis on $\mathcal{H}_{B'}$ and finally applying conditional $Z$ operations depending on the measurement outcome.}
\label{fig:CUcircuit}
\end{figure}

One EPR state and 3-bits of classical communication are consumed in the protocol presented in Fig.~\ref{fig:CUcircuit}. However, one EPR state and 2-bits of classical communication are shown to be sufficient for implementing a two-qubit controlled unitary operation \cite{Eisert}.
The optimal amount of entanglement required for implementing a two-qubit controlled unitary operation is shown in \cite{SoedaEPR, DanEPR}.
Note that one EPR state and 2-bits of classical communication are also sufficient for implementing a controlled unitary operation $U=\ket{0}\bra{0}_C\otimes\mathbb{I}_T+\ket{1}\bra{1}_C\otimes u_T$, where the dimension of the controlled system $\mathcal{H}_C$ is $2$ while the dimension of the target system $\mathcal{H}_T$ can be higher than $2$.

\part{Quantum computation over quantum networks}
\chapter{Preliminaries of Part II}
\label{chap:pre}
In this part, we analyze the first question:
\begin{itemize}
\item How does the topology of a quantum network consisting of quantum channels affect the performance of quantum communication?
\end{itemize}
More specifically, we concentrate on how the performance of quantum communication (generally quantum computation) changes when the topology of a quantum network consisting of quantum channels is changed while the capacity of quantum channels is fixed.
For simplicity, we consider quantum channels are noiseless and have 1-qubit capacity.
We consider a one-shot scenario, i.e. we are allowed to use a given network only once, and concentrate on a {\it cluster network}, which is a certain class of the network consisting of intermediate nodes and the same number of senders and receivers.
The cluster network is a subclass of $k$-pair network, which has been an actively studied network in both classical information theory \cite{NCsum} and quantum information theory \cite{Kobayashi0, Kobayashi1, Kobayashi2}.
We use the technique of {\it network coding} introduced in the following.

\section{Network coding theory}
In this section, we briefly review classical and quantum network coding theory, which is a basis of our analysis.

\subsection{Classical network coding}
In classical (network) information theory, \textit{network coding}, which incorporates processing at each network node in addition to routing, provides efficient transmission protocols that can resolve the bottleneck problem \cite{Ahlswede}.
Consider a communication task over the {\it butterfly network} and the {\it grail network} presented in Fig.~\ref{fig:cgrail} that aims to transmit single bits $x$ and $y$ from $i_1$ to $o_2$ and $i_2$ to $o_1$ simultaneously via nodes $n_1$, $n_2$, $n_3$ and $n_4$. The directed edges denote transmission channels with $1$-bit capacity. One of the channels in each network exhibits the bottleneck without network coding shown in Fig.~\ref{fig:cgrail}. Network coding has been already implemented in wireless network protocols and satellite communication protocols.

\begin{figure}
 \begin{center}
  \includegraphics[height=.3\textheight]{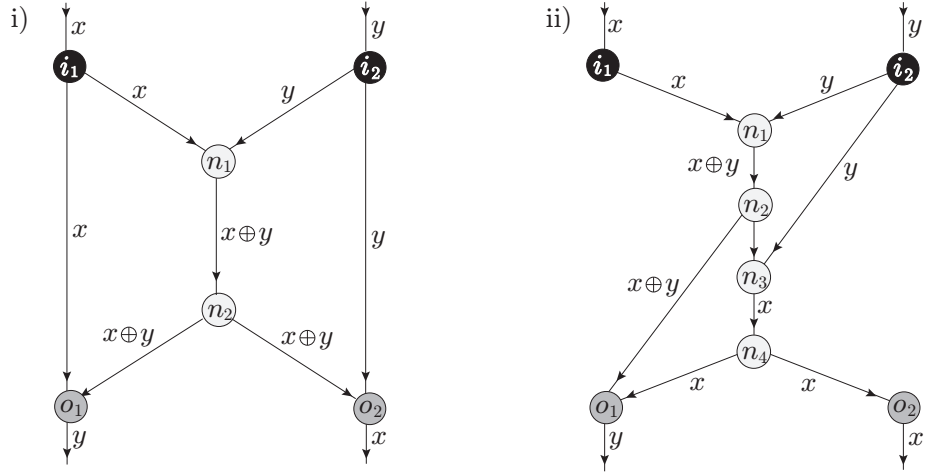}
  \end{center}
  \caption{{\bf Network coding for a classical communication task over i) the butterfly network and ii) the grail network.} Two bits of information $x, y \in \{0, 1\}$ are given at the input nodes $i_1$ and $i_2$, respectively. $x \oplus y$ denotes addition of $x$ and $y$ modulus $2$.}
\label{fig:cgrail}
\end{figure}

Such a communication task over a general two input-output network given by a graph $G$ has been shown to be achievable if and only if $G$ contains an essential substructure \cite{ButterflyGrail}. The butterfly network and the grail network are known to be two such essential substructures.

\subsection{Quantum network coding}
Quantum communication with \textit{quantum network coding} has been studied by analogy to classical network coding \cite{Iwama, Leung, Hayashi,Kobayashi0, Kobayashi1, Kobayashi2}. 
 $k$-pair quantum communication over a network is a unicast communication task to faithfully transmit a $k$-qubit state given at distinct input nodes $\{i_1, i_2,\cdots,i_k\}$ to distinct output nodes $\{o_{1},o_{2},\cdots,o_k\}$ through a given network.
Two examples of $2$-pair quantum communication over a butterfly network and a grail network are shown in Fig.~\ref{fig:qgrail}.
The setting of network coding can be further classified into three cases based on how one treats classical communication, namely, 

\begin{screen}
\begin{enumerate}
\item{The case where each channel can be used for either 1-bit classical communication or 1-qubit quantum communication \cite{Iwama, Leung}.}
\item{The case where each channel can be used for either 2-bit classical communication or 1-qubit quantum communication \cite{Hayashi, Soeda}.}
\item{The case where the channels only restrict the capacity of quantum communication, and classical communication is freely allowed between any two nodes \cite{Kobayashi0, Kobayashi1, Kobayashi2}.}
\end{enumerate}
\end{screen}

\begin{figure}
 \begin{center}
  \includegraphics[height=.33\textheight]{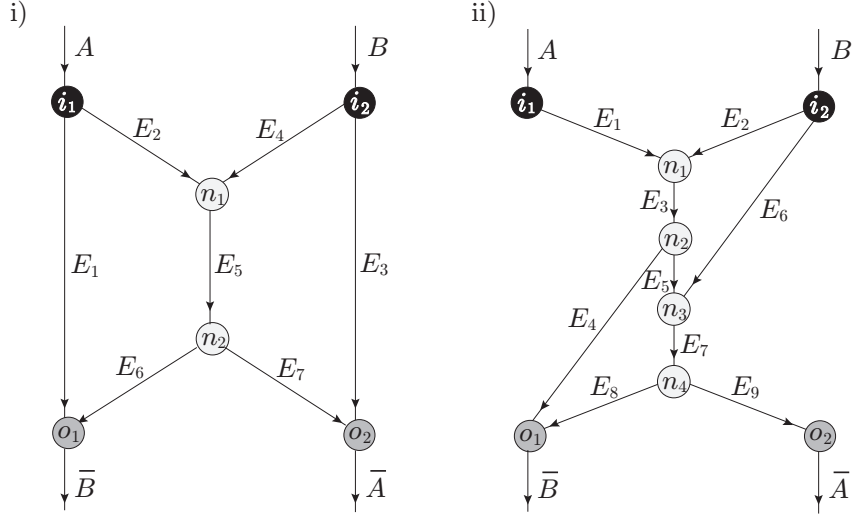}
  \end{center}
  \caption{{\bf The butterfly and the grail quantum network.} i) The butterfly network and ii) the grail network with the input nodes ($i_1$ and $i_2$), output nodes ($o_1$ and $o_2$) and the repeater nodes ($n_1,n_2,n_3$ and $n_4$). The directed edges $A,B,\bar{A},\bar{B}, E_1, E_2,\cdots,E_9$ represent quantum or classical channels. Quantum channels have 1-qubit capacity. Meanwhile there are some settings about classical channel capacity.
Our task is to transmit a given two-qubit state $|{\rm input}\rangle_{i_1,i_2}$ from $i_1$ to $o_2$ and from $i_2$ to $o_1$ simultaneously by using the channels and local quantum operations at each nodes.}
\label{fig:qgrail}
\end{figure}

In quantum network coding, perfect multicast communications are not allowed by the no-cloning theorem \cite{Noclone}. No-cloning theorem also makes it impossible to perform $k$-pair communication over the networks by using a simple extension of classical network coding.
Indeed, in the setting where each edge can be used for either 1-bit classical communication or 1-qubit quantum communication, perfect quantum 2-pair communication over the butterfly network has been shown to be impossible \cite{Iwama, Leung}. However, it has been shown that in the setting where each edge can be used for either 2-bit classical communication or 1-qubit quantum communication, perfect quantum 2-pair communication over the butterfly network is possible, if and only if input nodes share two EPR pairs \cite{Hayashi}. In the setting where each edge has 1-qubit channel capacity and classical communication is freely allowed between any nodes, however, it has been shown that there exists a quantum network coding protocol to achieve the 2-pair quantum communication over the butterfly and grail networks perfectly \cite{Kobayashi0, Kobayashi1, Kobayashi2}.
In this thesis, we focus on the third setting, where classical communication is freely allowed between any two nodes. This setting is justified in practical situations, where classical communication is easier to implement experimentally than quantum communication.
Before proceeding to our research, we review the technique used in quantum network coding in the third setting and note two important points related to our research.

In \cite{Kobayashi2}, Kobayashi {\it et al.} have shown that if classical $k$-pair communication, which is a unicast communication task to faithfully transmit $k$-bit given at distinct input nodes to distinct output nodes simultaneously, is possible, quantum $k$-pair communication over the same network is possible in the quantum setting where free classical communication  is allowed \footnote{Note that a network is a graph consisting nodes (vertices) and channels (edges) as we define later. Thus, if the graph is the same, we say the network is the same although the capacity of channels is different in the quantum setting and the classical setting.}.   Their implementation protocol is composed of two stages. In the first stage, the encoding stage, a unitary operation $U_{encoding}$ described by a gate sequence consisting only of CNOT gates is performed to create a large entangled resource state over the qubits at all nodes of the graph.  The CNOT gates correspond to the operation of copying and exclusive disjunction used for implementing $x \rightarrow x, x$ and $x, y \rightarrow x \oplus y$ in the classical setting of network coding.  In the second stage, the disentangling stage, two kinds of disentangling operations are performed, denoted by maps $\Gamma_{d2}$ and $\Gamma_{d3}$ consisting of the Hadamard gate, measurements in the computational basis, and conditional $Z$ operators depending on the measurement outcomes. 
Using the notation for quantum circuits presented in Table \ref{table:qcnotation} and the notation for nodes and edges presented in Fig. \ref{fig:qgrail}, the gate sequence of $U_{encoding}$ for 2-pair communication over the butterfly network is presented in Fig. \ref{fig:encodingstage}.  Quantum circuits corresponding to $\Gamma_{d2}$ and $\Gamma_{d3}$ are shown in Figures \ref{fig:two} and \ref{fig:three}, respectively.

\begin{figure}[H]
 \begin{center}
  \includegraphics[width=60mm]{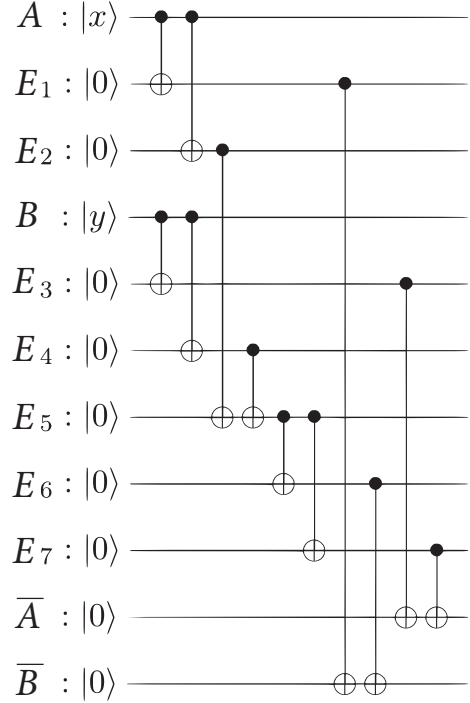}
 \end{center}
 \caption{{\bf The quantum circuit for $U_{encoding}$ in the encoding stage of the protocol presented by the Kobayashi {\it et al.} protocol for $2$-pair communication over the butterfly network.}  The indices of the Hilbert spaces of the qubits transmitted along the corresponding edges are also denoted by $A$, $B$, $\bar{A}$, $\bar{B}$, $E_1, \cdots, E_7$.  Operations on the qubits in the same node are considered local operations. An input state of the qubits $A$ and $B$ is given by $\sum_{x,y =0,1} \lambda_{x,y}\ket{x}_A\ket{y}_B$, where $\{ \ket{x}_A \}$ and $\{ \ket{y}_B \}$ denote the computational bases of qubit $A$ and $B$, respectively, and $\sum_{x,y} |\lambda_{x,y}|^2 =1$.}
 \label{fig:encodingstage}
\end{figure}

\begin{figure}[H]
% \begin{minipage}{0.4\hsize}
  \begin{center}
   \includegraphics[width=45mm]{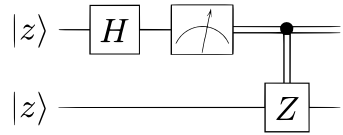}
  \end{center}
  \caption{{\bf A quantum circuit for the map $\Gamma_{d2}$.} It disentangles the first qubit of a two-qubit  state $\sum_{z =0,1} \alpha_z \ket{z}_1 \ket{z}_2 $ to obtain $\sum_{z}\alpha_z \ket{z}_2$ at the second qubit  for all $\alpha_z$ satisfying $\sum_{z} | \alpha_z|^2  =1$, where $\{ \ket{z} \}$ is the computational basis.}
  \label{fig:two}
 \end{figure}
% \end{minipage}
% \begin{minipage}{0.6\hsize}
 \begin{figure}[H]
  \begin{center}
   \includegraphics[width=60mm]{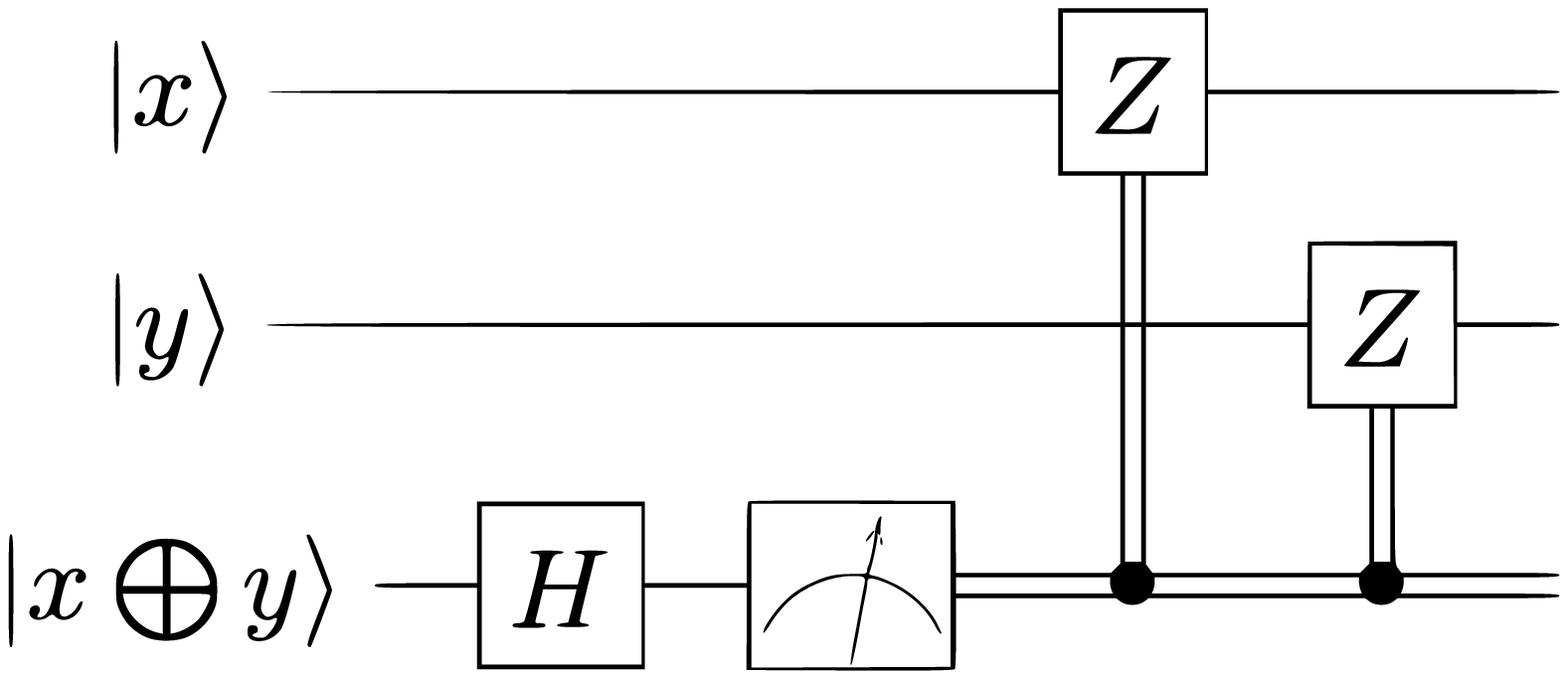}
  \end{center}
  \caption{{\bf A quantum circuit for the map $\Gamma_{d3}$.} It disentangles the third qubit from a three-qubit state $\sum_{x,y =0,1} \lambda_{x,y}\ket{x}_1\ket{y}_2\ket{x \bigoplus y}_3$ to obtain $\sum_{x,y}\lambda_{x,y}\ket{x}_1\ket{y}_2$} on the first two qubits for all $\lambda_{x,y}$ satisfying $\sum_{x,y} | \lambda_{x,y}|^2  =1$, where $\{ \ket{x} \}$ and $\{ \ket{y} \}$ denote the computational basis. 
  \label{fig:three}
% \end{minipage}
\end{figure}

We give an overview of the Kobayashi et al. protocol for 2-pair communication over the butterfly network. We write the initial state of two qubits $A$ and $B$ given at the input nodes $i_1$ and $i_2$, respectively, by
\begin{equation}
\ket{input (\lambda)}_{AB}^{i_1i_2}=\sum_{x,y=0,1} \lambda_{x,y}\ket{x}_{A}^{i_1}\ket{y}_{B}^{i_2}
\label{eqn:initial}
\end{equation}
where $\lambda$ specifies a set of coefficients $\lambda = \{ \lambda_{x,y} \}_{x,y=0,1}$ satisfying $\sum_{x,y} |\lambda_{x,y}|^2 =1$, superscripts $i_1$ and $i_2$ specify the nodes where qubits belong, and $\{ \ket{x} \}$ and $\{ \ket{y} \}$ are the computational basis.  All the other qubits $E_1, \cdots, E_7, \bar{A}$ and $\bar{B}$ are in the fixed state $\ket{0}$.   In general, the initial state given by $\ket{input (\lambda)}$ represents an entangled state of qubits $A$ and $B$.  If necessary, the input state can be restricted to a product state by imposing $\lambda_{x,y} = \mu_x \nu_y$ for $x,y=0,1$ satisfying $\sum_{x} |\mu_x |^2 =1$ and $\sum_{y} |\nu_y |^2 =1$. 

By performing $U_{encoding}$ in the encoding stage, the initial state is transformed to an entangled state of 11 qubits, $A$, $B$, $E_1, \cdots, E_7$, $\bar{A}$ and $\bar{B}$, given by
\begin{eqnarray}
\sum_{x,y=0,1}\lambda_{x,y}\ket{x}_{A}^{i_1}\ket{x}_{E_1}^{o_1}\ket{x}_{E_2}^{n_1}\ket{y}_{B}^{i_2}\ket{y}_{E_3}^{o_2}\ket{y}_{E_4}^{n_1}\ket{x\oplus y}_{E_5}^{n_2}\ket{x\oplus y}_{E_6}^{o_1}\ket{x\oplus y}_{E_7}^{o_2}\ket{x}_{\bar{A}}^{o_2}\ket{y}_{\bar{B}}^{o_1}.
\label{eq2}
\end{eqnarray}
In the disentangling stage, the disentangling operation $\Gamma_{d2}$ on qubits $A$ and $E_1$ is performed to disentangle the qubit $A$.  The resulting state is given by
\begin{eqnarray}
\sum_{x,y=0,1}\lambda_{x,y}\ket{x}_{E_1}^{o_1}\ket{x}_{E_2}^{n_1}\ket{y}_{B}^{i_2}\ket{y}_{E_3}^{o_2}\ket{y}_{E_4}^{n_1}\ket{x\oplus y}_{E_5}^{n_2}\ket{x\oplus y}_{E_6}^{o_1}\ket{x\oplus y}_{E_7}^{o_2}\ket{x}_{\bar{A}}^{o_2}\ket{y}_{\bar{B}}^{o_1}.
\end{eqnarray}
This disentangling operation can be performed when both of the two qubits $A$ and $E_1$ are at the node $i_1$, but it can be also performed after the qubit $E_1$ is transmitted from $i_1$ to $o_1$ using the edge $E_1$.  In this case, 1-bit classical communication is required from the node $i_1$ to the node $o_1$ to perform the conditional $Z$ operation on the qubit $E_1$.   Similarly, the disentangling operation $\Gamma_{d2}$ is performed on the pairs of qubits $\{B, E_3 \}$, $\{ E_2, E_1 \}$, $\{ E_4, E_3 \}$ and $\{ E_5, E_6 \}$ to disentangle qubits $B$, $E_2$, $E_4$, and $E_5$.  Then another disentangling operation $\Gamma_{d3}$ is performed on the sets of qubits $\{ E_1, \bar{B}, E_6 \}$ and $\{ \bar{A}, E_3, E_7 \}$  to disentangle the qubits $E_6$ and $E_7$.  Then we obtain a state of four qubits $E_3$, $\bar{A}$ at the node $o_2$ and $E_1$, $\bar{B}$ at the node $o_1$ given by 
\begin{eqnarray}
\ket{\phi (\lambda)}^{o_2o_1}= \sum_{x,y=0,1}\lambda_{x,y}\ket{x}_{E_1}^{o_1}\ket{y}_{E_3}^{o_2}\ket{x}_{\bar{A}}^{o_2}\ket{y}_{\bar{B}}^{o_1}.
\label{eq4}
\end{eqnarray}

At last, the disentangling operation $\Gamma_{d2}$ is performed on the pairs $\{ E_1, \bar{A}\}$ and $\{ E_3, \bar{B}\}$ to disentangle $E_1$ and $E_3$. Then we obtain the output state
\begin{eqnarray}
\ket{ output (\lambda)}^{o_2o_1} =\sum_{x,y=0,1}\lambda_{x,y}\ket{x}_{\bar{A}}^{o_2}\ket{y}_{\bar{B}}^{o_1}.
\end{eqnarray}
Thus, we achieve $\ket { output ( \lambda )}_{\bar{A}\bar{B}} = \ket{input (\lambda)}_{A, B}$. Note that 7-bits of classical communication is required in the disentangling stage.

At the end of this section, we note two important points of quantum network coding with free classical communication which is related to our research.
\begin{itemize}
\item Applying CNOT gates instead of the operation of copying and exclusive disjunction used for implementing $x \rightarrow x, x$ and $x, y \rightarrow x \oplus y$ in the classical setting of network coding enables to perform quantum communication. However, it is possible to perform a more general controlled unitary operations instead of CNOT by using quantum network, which is a key idea for implementing a unitary operation over a quantum network.
\item In \cite{Kobayashi2}, they have shown that if classical $k$-pair communication is possible, then quantum $k$-pair communication is possible. However, we do not know the converse. Particularly, we do not know whether a quantum communication is possible or not over the {\it square network} presented in Fig.~\ref{fig:square}, where a classical communication is impossible.
\end{itemize}

\begin{figure}
 \begin{center}
  \includegraphics[height=.25\textheight]{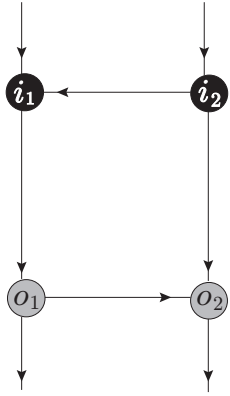}
  \end{center}
  \caption{{\bf The square network with the input nodes ($i_1$ and $i_2$) and output nodes ($o_1$ and $o_2$).} We focus on the setting where quantum channels described by directed edges have 1-qubit capacity and classical communication is freely allowed. The task is to transmit a given two-qubit state $|{\rm input}\rangle_{i_1,i_2}$ from $i_1$ to $o_2$ and from $i_2$ to $o_1$ simultaneously. It has been an open problem whether the task is achievable or not. We show that the task is impossible by using a tool we develop in this Part.}
\label{fig:square}
\end{figure}

\section{Measurement Based Quantum Computation (MBQC)}
\label{sec:MBQC}
Quantum computation over quantum networks is related to a model of quantum computation, {\it measurement based quantum computation} (MBQC).
MBQC was first proposed in 2001 by Raussendorf and Briegel in a much acclaimed paper \cite{Raussendorf}. They showed that a particular highly entangled multi-qubit state called a {\it cluster state} combined with single qubit measurements and classical communication are sufficient for performing universal quantum computation. In the first stage, a cluster state is prepared by initializing all qubits in the $\ket{+}=\frac{1}{\sqrt{2}}(\ket{0}+\ket{1})$ state and then applying controlled Z gates between pairs of neighboring qubits on a square lattice, more generally, one can prepare graph states similarly for neighbors corresponding to the edges of graph.
A definition of a graph state $\ket{\phi}$ is given as follows.
Note that the order in the product has no bearing on the definition since all the controlled Z gates are commutative.
\begin{definition}
For a given graph $G=\{\mathcal{V},\mathcal{E}\}$, the graph state $\ket{\phi}\in\mathbb{C}^{2^{|\mathcal{V}|}}$ is a $|\mathcal{V}|$-qubit state such that
\begin{equation}
\ket{\phi}=\prod_{(u,v)\in\mathcal{E}}U^{CZ}_{u,v} \left(\otimes_{i=1}^{|\mathcal{V}|}\ket{+}_i\right),
\label{eq:clusterstate}
\end{equation}
where $U^{CZ}_{u,v}$ is the controlled Z gate acting only on the $u$-th qubit and the $v$-th qubit.
\end{definition}

In the second stage, we perform a projective measurement $\{\ket{+_{\alpha}}\bra{+_{\alpha}}, \ket{-_{\alpha}}\bra{-_{\alpha}}\}$, where
\begin{equation}
\ket{\pm_{\alpha}}=\frac{1}{\sqrt{2}}(\ket{0}\pm e^{i\alpha}\ket{1}),
\label{eq:projMBQC}
\end{equation}
on all single qubits in a certain sequence, sending the classical outcome bits to the measurement in the next step. The entanglement between measured and unmeasured qubits ensures that the quantum state of the remaining qubits is transformed according to the algorithm given by the initial state and the measurement pattern. The final qubits to be measured define the output of the computation.

Since we use only local measurements, all of the entanglement that we need for the computation must be prepared in the first stage. This model highlights the role of entanglement in quantum computation.

\section{Classifications of unitary operators}
In this section, we first review {\it controlled unitary operation}, since it plays a central role in quantum computation over a quantum network.
In order to characterize the class of implementable unitary operations over a quantum network, we introduce two ways to classify unitary operators, the Kraus-Cirac decomposition and the operator Schmidt decomposition.
Note that the operator Schmidt decomposition is applicable to unitary operators on any non-prime dimensional Hilbert space while the Kraus-Cirac decomposition is applicable only to two-qubit unitary operators.

\subsection{Controlled unitary operation}
A fully controlled unitary operator $U\in\mathbf{U}(\mathcal{H}_c\otimes\mathcal{H}_t)$ is a unitary operator that acts on a control system $\mathcal{H}_c$ and a target system $\mathcal{H}_t$ such that
\begin{equation}
U=\sum_{i=1}^{C}\ket{i}\bra{i}_c\otimes u^{(i)},
\end{equation}
where $\{\ket{i}_c\}_i$ is an orthonormal basis of $\mathcal{H}_c$ and $u^{(i)}\in\mathbf{U}(\mathcal{H}_t)$ is a unitary operator.
$U^{CNOT}$ and $U^{CZ}$ are special cases of two-qubit controlled unitary operators.
A controlled unitary operator is similar to a classical switch where the gate acting on the target system changes depending on the state of the control system. The $i$-th unitary operator $u^{(i)}$ of the set $\{u^{(i)}\}$ is performed on the target system when the state of the control system is prepared to be $\ket{i}_c$. However, unlike the classical switch, we can prepare the state of the control system to be any superposition of states, which creates entanglement between the control system and the target system.

Any two-qubit controlled unitary operators $U$ defined by
\begin{eqnarray}
U:=|0\rangle\langle0|\otimes \mathbb{I}+|1\rangle\langle1|\otimes u
\label{controlledu}
\end{eqnarray}
where $u \in \textbf{U}(\mathbb{C}^2)$, is locally unitarily equivalent to a two-qubit controlled phase operator $U^{(\theta)}$.   $U^{(\theta)}$ can be written by
\[
U^{(\theta)}:=\ket{00}\bra{00}+\ket{01}\bra{01}+\ket{10}\bra{10}+e^{i \theta}\ket{11}\bra{11}.
\]
That is, for all two-qubit controlled unitary operators $U$, there exists $u_1, u_2, u_3, u_4 \in \textbf{U}(\mathbb{C}^2)$ and $\theta\in\mathbb{R}$ such that $U = (u_3 \otimes u_4)  U^{(\theta)}(u_1 \otimes u_2)$. We give a proof of this statement in the following.
\begin{proof}
$u_1, u_2, u_3, u_4 \in \textbf{U}(\mathbb{C}^2)$ and $\theta$ are given by the following.   
\begin{enumerate}
\item{Diagonalize $u$, and obtain $u=e^{i\alpha}(|\psi\rangle\langle\psi|+e^{i\beta}|\psi'\rangle\langle\psi'|)$, where $\{e^{i\alpha},e^{i(\alpha+\beta)}\}$ are the eigenvalues of $u$ and $\{\ket{\psi},\ket{\psi'}\}$ the eigenvectors.}
\item{Set $u_1=\mathbb{I}$,$u_2=|0\rangle\langle\psi|+|1\rangle\langle\psi'|$,
$u_3=\begin{pmatrix}
1&0\\
0&e^{i\alpha}
\end{pmatrix}$,$u_4=|\psi\rangle\langle0|+|\psi'\rangle\langle1|$
and $\theta=\beta$.}
\end{enumerate}
\end{proof}
We summarize the property of the local unitary equivalence of two-qubit controlled unitary operators in Fig.~\ref{fig:LUcontrolphase}.
The set of two-qubit unitary operations that is locally unitarily equivalent to a controlled phase operator is denoted by $\mathbf{U}_c$.

\begin{figure}[H]
% \begin{minipage}{0.4\hsize}
  \begin{center}
   \includegraphics[width=140mm]{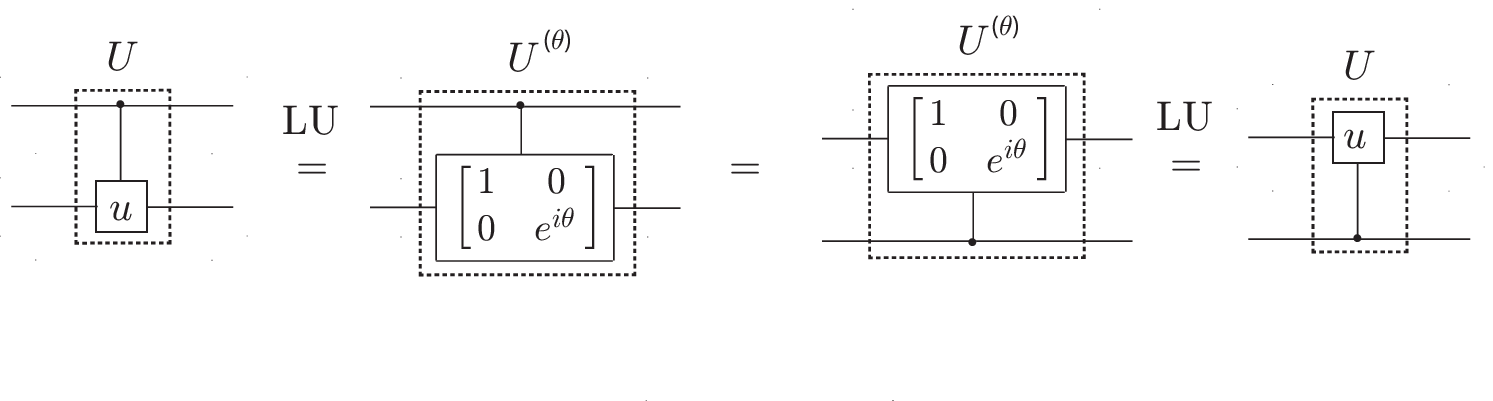}
  \end{center}
  \caption{{\bf Quantum circuits representing the property of local unitary equivalence of a two-qubit controlled unitary operator and a controlled phase operator.}  $\overset{\mathrm{LU}}{=}$ represents local unitary equivalence. Note that the controlled qubit and the target qubit can be flipped in a local unitary equivalence class.}
  \label{fig:LUcontrolphase}
 \end{figure}

\subsection{Kraus-Cirac decomposition}
A general two-qubit unitary operator $U \in \mathbf{U}(\mathcal{H}_{\mathcal{I}_Q} : \mathcal{H}_{\mathcal{O}_Q})$ where $\mathcal{H}_{\mathcal{I}_Q} = \mathcal{H}_{\mathcal{O}_Q}= \mathbb{C}^2 \otimes \mathbb{C}^2$  can be decomposed into a canonical form called the Kraus-Cirac decomposition introduced in \cite{KCdec1,KCdec2, KCdec3} given by
\begin{equation}
U=(u \otimes u^\prime) e^{ i(x X \otimes X+y Y \otimes Y+z Z \otimes Z)} (w \otimes w^\prime).
\label{KCD}
\end{equation}
where $u$, $u^\prime$, $w$ and $w^\prime$ are single qubit unitary operators and $X$, $Y$ and $Z$ are the Pauli operators on $\mathbb{C}^2$ and $x,y,z \in \mathbb{R}$.  Since $u$, $u^\prime$, $w$ and $w^\prime$ represent local unitary equivalence, we focus on analyzing implementability of the two-qubit global unitary part \begin{eqnarray}
U_{global} (x,y,z) := e^{ i(x X \otimes X+y Y \otimes Y+z Z \otimes Z)}
\label{gloablpartofKCD}
\end{eqnarray}
In Eq.~(\ref{gloablpartofKCD}), the parameters $x,y,z$ in $0 \leq x < \pi/2$ (or $0 \leq x \leq \pi/4$ if $z=0$), $0 \leq y \leq \min\{x, \pi/2 - x \}$ and $0 \leq z \leq y $ cover all two-qubit global unitary operators up to the local unitarily equivalence (the Weyl chamber \cite{KCdec2}).
We define the Kraus-Cirac number of a two-qubit unitary operator $U$ as the number of non-zero parameters $x,y,z$ in $U_{global}(x,y,z)$ and denote by ${\rm KC}\#(U)$.  ${\rm KC}\#(U)$ characterizes nonlocal properties (globalness) of $U$ \cite{SoedaAkibue}.   

\subsection{Operator Schmidt decomposition}
The operator Schmidt decomposition can be applied to any unitary operation acting on a finite (non-prime number) dimensional Hilbert space.
The set of linear operators $\mathbf{L}(\mathcal{H}_A)$ forms a Hilbert space with respect to the inner product $(M,N)=\frac{1}{\dim(\mathcal{H}_A)}{\rm tr}(M^{\dag}N)$. Thus we can apply the Schmidt decomposition to operators, such that for any linear operator $M\in\mathbf{L}(\mathcal{H}_A\otimes \mathcal{H}_B)$, there exists a set of orthonormal operators $\{P_i\in\mathbf{L}(\mathcal{H}_A)\}_i$ and  $\{Q_i\in\mathbf{L}(\mathcal{H}_B)\}_i$ satisfying
\begin{equation}
M=\sum_i \lambda_i P_i\otimes Q_i,
\end{equation}
where $\{\lambda_i\}$ are non-negative real numbers known as the {\it operator Schmidt coefficients} \cite{SchmidtDecomp}. The number of non-zero coefficients $|\{\lambda_i> 0\}|$ is known as the {\it operator Schmidt rank}. We denote the operator Schmidt rank of a linear operator $M$ by ${\rm Op}\#_B^A(M)$.

Elements of two-qubit unitary operators $\mathbf{U}(\mathcal{H}_A\otimes\mathcal{H}_B)$ can be classified by the Kraus-Cirac number and the operator Schmidt rank as follows. We summarize the classifications in Table \ref{table:KCnumber}.
\begin{itemize}
\item $U$ with ${\rm KC}\#(U)=0$ is a product of local unitary operators and satisfies ${\rm Op}\#_B^A(U)=1$.
\item $U$ with  ${\rm KC}\#(U)=1$ is locally unitarily equivalent to a controlled unitary operator and satisfies ${\rm Op}\#_B^A(U)=2$.
\item $U$ with  ${\rm KC}\#(U)=2$ is locally unitarily equivalent to a special class of two-qubit unitary operators called a matchgate $U^{match}$ defined by
\begin{equation}
U^{match}=\left(\begin{array}{cccc}u^{(1)}_{0,0} & 0 & 0 & u^{(1)}_{0,1} \\0 & u^{(2)}_{0,0} & u^{(2)}_{0,1} & 0 \\0 & u^{(2)}_{1,0} & u^{(2)}_{1,1} & 0 \\u^{(1)}_{1,0}& 0 & 0 & u^{(1)}_{1,1}\end{array}\right),
\end{equation}
where $u^{(i)}=\left(\begin{array}{cc}u^{(i)}_{0,0} & u^{(i)}_{0,1} \\u^{(i)}_{1,0}& u^{(i)}_{1,1}\end{array}\right)$ is a single-qubit unitary operator whose determinant is equal to 1. A sequence of the matchgates acting only on the nearest neighbor of one-dimensionally alligned qubits corresponds to a physical model of noniteracting fermions \cite{TerhalDiVincenzo} and is efficiently simulatable on a classical computer \cite{Valiant, JozsaMiyake}. $U^{match}$ satisfies ${\rm Op}\#_B^A(U)=4$.
\item The rest of two-qubit unitary operators including the SWAP operator have ${\rm KC}\#(U)=3$ and satisfy ${\rm Op}\#_B^A(U)=4$.
\end{itemize}
\begin{table}[H]
\begin{center}
\begin{tabular}{|c|c|c|}
\hline
${\rm KC}\#(U)$
& 
${\rm Op}\#_B^A(U)$
&
class
\\
\hline
0
& 
1
&
LU
\\
\hline
1
& 
2
&
LU C-phase
\\
\hline
2
& 
4
&
LU matchgate
\\
\hline
3
& 
4
&
$SU(4)$
\\
\hline
\end{tabular}
\end{center}
\caption{{\bf Classification of two-qubit unitary operators $\mathbf{U}(\mathcal{H}_A\otimes\mathcal{H}_B)$ by the Kraus-Cirac number and the operator Schmidt rank.} We refer the class with Kraus-Cirac number 0 as local unitary operators. The class with Kraus-Cirac number 1 is locally unitarily equivalent to controlled phase gates. The class with Kraus-Cirac number 2 is locally unitarily equivalent to matchgates.}
\label{table:KCnumber}
\end{table}

\chapter{Computation over the cluster network}
\label{chap:QCN}
In $k$-pair quantum communication, the output state $\left \vert {\rm output}\right \rangle_{o_{1}\cdots o_k}$ at the output nodes can be regarded as a state obtained by performing a $k$-qubit unitary operation $U$ on the input state $\left \vert {\rm input}\right \rangle_{i_1\cdots i_k}$ given at the input nodes
\begin{equation}
\label{eq:quantumcoding}
\left \vert {\rm output}\right \rangle_{o_1\cdots o_k} = U \left \vert {\rm input}\right \rangle_{i_1\cdots i_k},
\end{equation}
where $U$ is a permutation operator \footnote{A $k$-qubit permutation operator $U$ is defined by $U=\sum_{i_1,\cdots,i_k}\ket{i_{\sigma (1)},\cdots, i_{\sigma (k)}}\bra{i_1,\cdots, i_k}$, where $\sigma$ represents a permutation. An example of a two-qubit permutation operator is a SWAP operator $U^{SWAP}$, defined in Table \ref{table:qcnotation}.}.
We do not need to restrict the $k$-qubit unitary operator $U$ in Eq.\eqref{eq:quantumcoding} to be a permutation operator, it can be a general unitary operator.   This leads to the idea of quantum computation over a quantum network, which aims to perform a unitary operation on a state given at distinct input nodes and to faithfully transmit the resulting state to the distinct output nodes efficiently over the network at the same time \cite{NCcomp, Soeda}. By computing and communicating simultaneously, quantum computation over the network may reduce communication resources in DQC. Since communication can be naturally regarded as a special class of computation, investigating the capability of quantum computation gives us a new insight of quantum network coding, which originally aims at just quantum communication.

We investigate implementability of a unitary operation over a {\it cluster network}, which is a special class of networks with $k$ input nodes and $k$ output nodes, which we call $k$-pair network, as a first step to analyze more general network. The cluster network contains the grail network as its special case and relates to the butterfly network.
We focus on the setting where classical communication is freely allowed between any two nodes.    We present which class of unitary operators is implementable over cluster networks in this setting by investigating transformations of cluster networks into quantum circuits.  The transformation method of cluster networks provides constructions of quantum network coding to implement any two-qubit unitary operations over the grail and butterfly networks, which are fundamental primitive networks for classical network coding.  We also analyze probabilistic implementation of $N$-qubit unitary operations over the cluster network to understand the properties of quantum network coding for quantum computation when the requirement of deterministic implementations are relaxed but that of exact implementations are kept.

We denote the Hilbert space of a set of qubits specified by a set $\mathcal{Q}$  by $\mathcal{H}_{\mathcal{Q}}$  and the Hilbert space corresponding to a qubit $Q_k$ specified by an index $k$ by $\mathcal{H}_{Q_k} =\mathbb{C}^2$.   In our setting where quantum communications are restricted but classical communications are unrestricted, quantum communication of a qubit state between two nodes is replaced by teleportation between two nodes.  Since any direction of classical communications is allowed, quantum communication of a single qubit state can be achieved by sharing a maximally entangled two-qubit state between the nodes and the direction of quantum communication is not limited.  Thus what we can do over a given network in principle is equivalent to perform \textit{local operations} (at each nodes) \textit{and classical communication} (LOCC) assisted by the {\it resource state} that consists of a set of maximally entangled two-qubit states (the EPR states) $\ket{\Phi^{+}}=\frac{1}{\sqrt{2}}(\ket{00}+\ket{11})$ shared between the nodes connected by edges.  

We investigate which unitary operations are implementable by LOCC assisted by the resource state for a given network where nodes are represented by a two-dimensional lattice.  We consider that a node represented by $v_{i,j}$ is on the coordinate of the two-dimensional lattice $(i,j)$ and edges connect nearest neighbor nodes. We call these networks {\it cluster networks}. In Discussion of this part, we examine implementability of a unitary operation over {\it generalized cluster network}, where nodes are represented by a two-dimensional lattice and edges connect nearest neighbor horizontal nodes but edges can connect any pair of vertical nodes. We first give a formal definition of a cluster network.

\begin{definition}
A  network $G=\{\mathcal{V},\mathcal{E},\mathcal{I},\mathcal{O} \}$ is a $(k,N)$-cluster network if and only if, 
\begin{eqnarray}
\mathcal{V}&=&\{v_{i,j};\,1\leq i\leq k,1\leq j\leq N\}\nonumber\\
\mathcal{I}&=&\{v_{i,1};\,1\leq i\leq k\}\nonumber\\
\mathcal{O}&=&\{v_{i,N};\,1\leq i\leq k\}\nonumber\\
\mathcal{E}&=&\mathcal{S} \cup \mathcal{K}
\end{eqnarray}
where
\begin{eqnarray}
\mathcal{S}&=& \{(v_{i,j},v_{i+1,j});\,1\leq i\leq k-1,1\leq j\leq N)\}, \nonumber\\
\mathcal{K}&=& \{(v_{i,j},v_{i,j+1});\,1\leq i\leq k,1\leq j\leq N-1)\}, 
\end{eqnarray}
$k \geq 1$ and  $N \geq 1$.  $\mathcal{V}$ represents the set of all nodes, $\mathcal{I}$ and $\mathcal{O}$ represent $k$ input nodes and $k$ output nodes, respectively.  $\mathcal{E}$ represents the set of all edges and $\mathcal{S}$ and $\mathcal{K}$ represent the sets of vertical and horizontal edges, respectively.
\end{definition}

Next we define the resource state corresponding to the $(k,N)$-cluster network.  We introduce qubits  $S_{i,j}^1$ at node $v_{i,j}$ and $S_{i+1,j}^2$ at node $v_{i+1,j}$ to represent an EPR pair corresponding to an edge $ (v_{i,j},v_{i+1,j}) \in \mathcal{S}$.  Similarly, we introduce qubits specified by $K_{i,j}^1$ at node $v_{i,j}$ and $K_{i,j+1}^2$ at node $v_{i,j+1}$ to represent an EPR pair corresponding to an edge $ (v_{i,j},v_{i,j+1}) \in \mathcal{K}$.   We denote the set of all qubits in the resource state by $\mathcal{R}=\{ S_{i,j}^1, S_{i+1,j}^2|1\leq i \leq k-1, 1\leq j\leq N\}\cup\{K_{i,j}^1, K_{i,j+1}^2| 1\leq i \leq k, 1\leq j\leq N-1 \}$. The resource state $\ket{\Phi}_{\mathcal{R}}$ corresponding to a cluster network is defined by the following.

\begin{definition}
For a given $(k, N)$-cluster network, the resource state $\ket{\Phi}_{\mathcal{R}}\in\mathcal{H}_{\mathcal{R}}$ is defined by
\begin{eqnarray}
\ket{\Phi}_{\mathcal{R}}&=&\otimes_{i=1}^{k-1}\otimes_{j=1}^{N}\ket{\Phi^+}_{S_{i,j}^1,S_{i+1,j}^2}\nonumber\\
&&\otimes_{i=1}^{k}\otimes_{j=1}^{N-1}\ket{\Phi^+}_{K_{i,j}^1,K_{i,j+1}^2}.
\label{resourcestate}
\end{eqnarray}
\end{definition}

For example, the $(3,3)$-cluster network and the corresponding resource state are shown in Fig.~\ref{fig:cluster1}. 
Note that the resource state for a cluster network represented by Eq.~(\ref{resourcestate}) is different from the cluster states used in MBQC, defined by Eq.\eqref{eq:clusterstate}. While we can convert  the resource state for a cluster network into a cluster state by performing a projective measurement on all qubits at each node, a cluster state cannot be converted to the resource state for the corresponding cluster network by LOCC.

\begin{figure}
\begin{center}
  \includegraphics[height=.6\textheight]{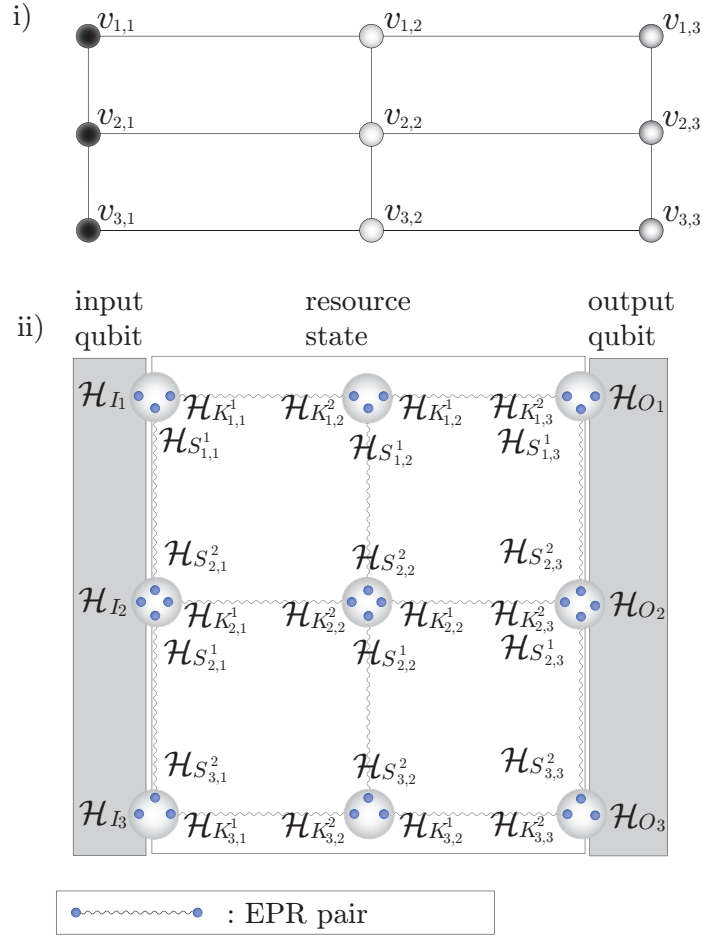}
  \end{center}
\caption{{\bf An example of a cluster network.} i) The $(3,3)$-cluster network with the input nodes $\mathcal{I}=\{ v_{1,1},v_{2,1},v_{3,1} \}$, output nodes $\mathcal{O} = \{ v_{1,3},v_{2,3},v_{3,3} \}$ and $3$ repeater nodes $\{ v_{1,2},v_{2,2},v_{3,2} \}$, and ii) the corresponding resource state. Note that the resource states of the cluster networks are different from the cluster states used in MBQC \cite{Raussendorf}.
}
\label{fig:cluster1}
\end{figure}

Finally we define implementability of a unitary operation over a $k$-pair network. In addition to resource qubits $\mathcal{R}$,
We introduce input qubits $I_i$ at the input node $v_{i,1} \in \mathcal{I}$, output qubits $O_i$ at the output node $v_{i,N} \in \mathcal{O}$, a set of input qubits $\mathcal{I}_Q =\{I_i|1\leq i\leq k\}$ and a set of output qubits $\mathcal{O}_Q =\{O_i|1\leq i\leq k\}$ for a $(k,N)$-cluster network.
 Note that each input and output node has only one input or output qubit since we concentrate on implementability of a unitary operation, and the state of output qubits is initially set to be in $\ket{0}\in\mathcal{H}_{\mathcal{O}_Q}$.

\begin{definition}
For a $(k,N)$-cluster network specified by $G=\{\mathcal{V},\mathcal{E},\mathcal{I},\mathcal{O} \}$, a unitary operation $U \in\mathbf{U}(\mathcal{H}_{\mathcal{I}_Q}:\mathcal{H}_{\mathcal{O}_Q})$ is deterministically implementable over the network if and only if there exists a LOCC map $\Gamma$ such that for any pure state $\ket{\psi}\in \mathcal{H}_{ \mathcal{I}_Q}$,
\begin{equation}
\Gamma(\ket{\psi}\bra{\psi}\otimes\ket{\Phi}\bra{\Phi}_{\mathcal{R}})=U\ket{\psi}\bra{\psi}U^{\dag},
\end{equation}
where LOCC map $\Gamma$ consists of local operations on each node and classical communications and $\mathbf{U}(\mathcal{H}_{\mathcal{I}_Q}:\mathcal{H}_{\mathcal{O}_Q})$ is the set of unitary operations from $\mathcal{H}_{\mathcal{I}_Q}$ to $\mathcal{H}_{\mathcal{O}_Q}$.
\end{definition}
Note that the main difference between this network computation model for implementing a unitary operation over a cluster network and standard MBQC is that any operations inside each node are allowed including adding arbitrary ancilla states in this model whereas only certain projective measurements on the cluster state in each node are allowed in MBQC.  Thus the full set of implementable unitary operations over a ($k,N$)-cluster network is larger than a set of operations implementable by MBQC using the corresponding cluster states converted from the resource state for the ($k,N$)-cluster network by LOCC. In Chapter 5, of this part, we investigate a difference between MBQC and computation over the butterfly network and discuss a potential contribution of our result toward MBQC.

\section{Possible computation}

We propose a method to convert a $(k,N)$-cluster network into quantum circuits representing a class of unitary operators implementable by LOCC assisted by the resource state corresponding to a given cluster network. By using the converted circuit, it is  easier to construct a network coding protocol since a set of implementable unitary operators are represented by a set of parameters of the converted circuit instead of a complicated LOCC protocol. The class of implementable unitary operators represented by the converted circuit is a subset of that over the cluster network  in general since the conversion method does not guarantee to give all possible constructions.  However, in some cases,  the constructions given by the conversion methods cover all possible implementable unitary operations as will be shown in the next section.   

We define a set of vertically aligned nodes $\mathcal{V}^v_j := \{ v_{i,j} \}_{i=1}^{k}$ and a set of vertically aligned edges  $\mathcal{S}_j:=\{ (v_{i,j},v_{i+1,j}) \}_{i=1}^{k-1}$ where $1 \leq j \leq N$.   We also define a set of horizontally aligned nodes $\mathcal{V}^h_i := \{ v_{i,j} \}_{j=1}^N$ and a set of horizontally aligned edges  $\mathcal{K}_i:=\{ (v_{i,j},v_{i,j+1}) \}_{j=1}^{N-1}$ where $1 \leq i \leq k$.  We consider that the EPR pairs given for a set of vertically aligned edges  $\mathcal{S}_j$ are used for implementing global unitary operations between nodes whereas each EPR pair given for a set of horizontal aligned edges  $\mathcal{K}_i$ is used for teleporting a qubit state from node $v_{i,j}$ to node $v_{i,j+1}$.

We investigate which kinds of global operations are implementable between the nodes  in $\mathcal{V}^v_j$ if only one EPR pair is given for each edge and LOCC between the nodes is allowed.   In this case, a two-qubit controlled unitary operation
\begin{eqnarray}
C_{ l;n} (\{ u_n^{(a)} \}_{a=0,1} ):=\sum_{a=0}^1 \ket{a}\bra{a}_l \otimes u_n^{(a)},
\end{eqnarray}
where $l$ represents the vertical coordinate of the node $v_{l,j}$ of a control qubit and $n$ represents the vertical coordinate of the node $v_{n,j}$ of a target qubit, and $u_n^{(a)} (a=0,1)$ are single qubit unitary operations on the target qubit, can be performed by using the method to implement a controlled unitary operation using a EPR pair proposed by \cite{Eisert}. 
If $n \neq l \pm 1$,  all EPR pairs represented by edges between  $l$ and $n$ are consumed to implement the two-qubit control unitary.  When we do not specify the single qubit unitary operations $\{ u_n^{(a)} \}$ on the target qubit we denote a two-qubit controlled unitary operation simply by $C_{ l;n}$.

In addition to the two-qubit control unitary operations, we can perform  three-qubit fully controlled unitary operations defined by 
\begin{eqnarray}
C_{l,m;n} (\{ u_n^{(a b)} \}_{a,b=0,1} ):=\sum_{a=0}^1 \sum_{b=0}^1 \ket{a b}\bra{a b}_{lm} \otimes u_n^{(a b)},
\end{eqnarray}
where $l$ and $m$ represent the vertical coordinates of the nodes $v_{l,j}$ and $v_{m,j}$ of two control qubits, respectively,  and $n$  represents the vertical coordinates of the node $v_{n,j}$ of a target qubit, and $u_n^{(a b)} (a,b=0,1)$
represents single qubit operations on the target qubit.
(See Appendix A.1 for details of the LOCC protocol implementing three-qubit fully controlled unitary operations.)
Note that the indices $l$, $n$ and $m$ should be taken such that $l<n<m$ or $m<n<l$. Similarly to the case of a two-qubit controlled unitary operation, we denote a three-qubit fully controlled unitary operation by $C_{ l,m;n} $ when we do not specify the target single qubit operations.  On the other hand, any four-qubit fully controlled unitary, where three of the four qubits are control qubits and the rest of the qubit is a target qubit, is not implementable on  qubits that are all in different nodes of $\mathcal{V}^v_j$ in a $(k,N)$-cluster network, if a single EPR pair is given for each edge in $\mathcal{S}_j$.  Obviously any single qubit unitary operations  can be implemented on any qubit.

We present a  protocol to convert a given $(k,N)$-cluster network into quantum circuits. First  (step 1 to step 3), we construct quantum circuits of unitary operations that are implementable on qubits in a set of vertically aligned nodes $\mathcal{V}_j^v$ by LOCC assisted by the EPR pairs given for a set of vertically aligned edges $\mathcal{S}_j$ for a certain $j$.   Then (step 4),  we repeat the procedure given by the first part (step 1 to step 3) for different $j$ of $1 \leq j \leq N$.

The conversion protocol: 
\begin{enumerate}
\item Draw $k$ horizontal wire segments where each of the wire segments corresponds to a set of qubits at vertically aligned nodes $\mathcal{V}^v_j$.

\item Symbols representing two-qubit controlled unitary operations $C_{ l;n}$ or three-qubit fully controlled unitary operations $C_{ l,m;n}$ are added on the horizontal wire segments according to the following rules.

\begin{itemize}
\item To represent  $C_{l;n}$, draw a black dot representing a control qubit on the $l$-th wire, draw a vertical segment from the dot to the $n$-th wire segment and draw a box representing a target unitary operation on the $n$-th wire segment at the end of the vertical segment.  Write index $l$ at the side of the vertical segment between the horizontal wire segments. An example is shown in Fig.~\ref{fig:convertedcircuitex1} i).
\item To represent $C_{ l,m;n}$, draw two black dots representing control qubits on the $l$-th and $m$-th wire segments, draw vertical segments from each dot to the $n$-th wire and draw a box representing an arbitrary target unitary operation on the $n$-th wire segment at the end of the vertical segment.  Write indices $l$ and $m$ at the sides of the vertical segment between the horizontal wire segments. An example is shown in Fig.~\ref{fig:convertedcircuitex1} ii)
\item  Multiple gates of $C_{ l;n}$ or $C_{l,m;n}$ can be added as long as there are only one type of index appearing between the horizontal wire segments and no target unitary operation represented by a box is inserted between two black dots on a horizontal wire segment.
 An example of possible circuits generated in this protocol is shown in Fig.~\ref{fig:convertedcircuitex1} iii). We also give an example of circuits that do not follow the rule in Fig.~\ref{fig:convertedcircuitex1} iv).
\end{itemize}

\item Arbitrary single qubit unitary operations represented by boxes are inserted between before and after the sequence of $C_{ l;n}$ and $C_{ l,m;n}$ (but not during the sequence) obtained by step 2.

\item Repeat step 1 to step 3 for each $1 \leq j \leq N$ and connect all the $i$-th horizontal wire segments.
\end{enumerate}

In Appendix A.2, we show that a unitary operation represented by the quantum circuit obtained by step 1 to step 3 of the conversion protocol is implementable in a set of vertically aligned nodes $\mathcal{V}_j^v$, namely, it is implementable by LOCC assisted by $(k-1)$ EPR pairs corresponding to a set of vertically aligned edges $\mathcal{S}_j$.  As examples, quantum circuits converted from the $(2,3)$-cluster and $(3,2)$-cluster networks are shown in Fig.~\ref{fig:convertedcircuit}.

\begin{figure}
\begin{center}
  \includegraphics[height=.15\textheight]{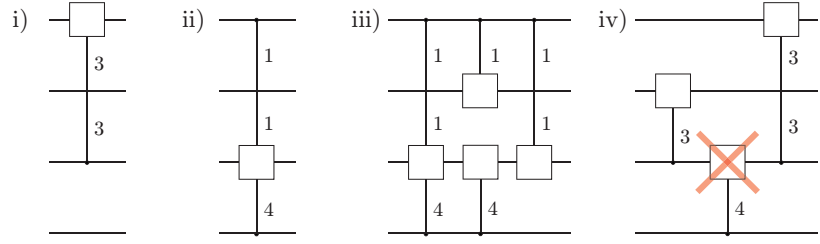}
  \end{center}
 \caption{ i) A symbol representing $C_{3;1}$. ii) A symbol representing $C_{1,4;3}$. iii) An example of circuits generated in step 2 of the conversion protocol. The index in the upper region is 1, that of index in the middle region is 1 and that of index in the lower region is 4.  iv) This conversion is forbidden since there is a target unitary operation inserted between two black dots representing controlled qubits.}
\label{fig:convertedcircuitex1}
\end{figure}

\begin{figure}
\begin{center}
  \includegraphics[height=.6\textheight]{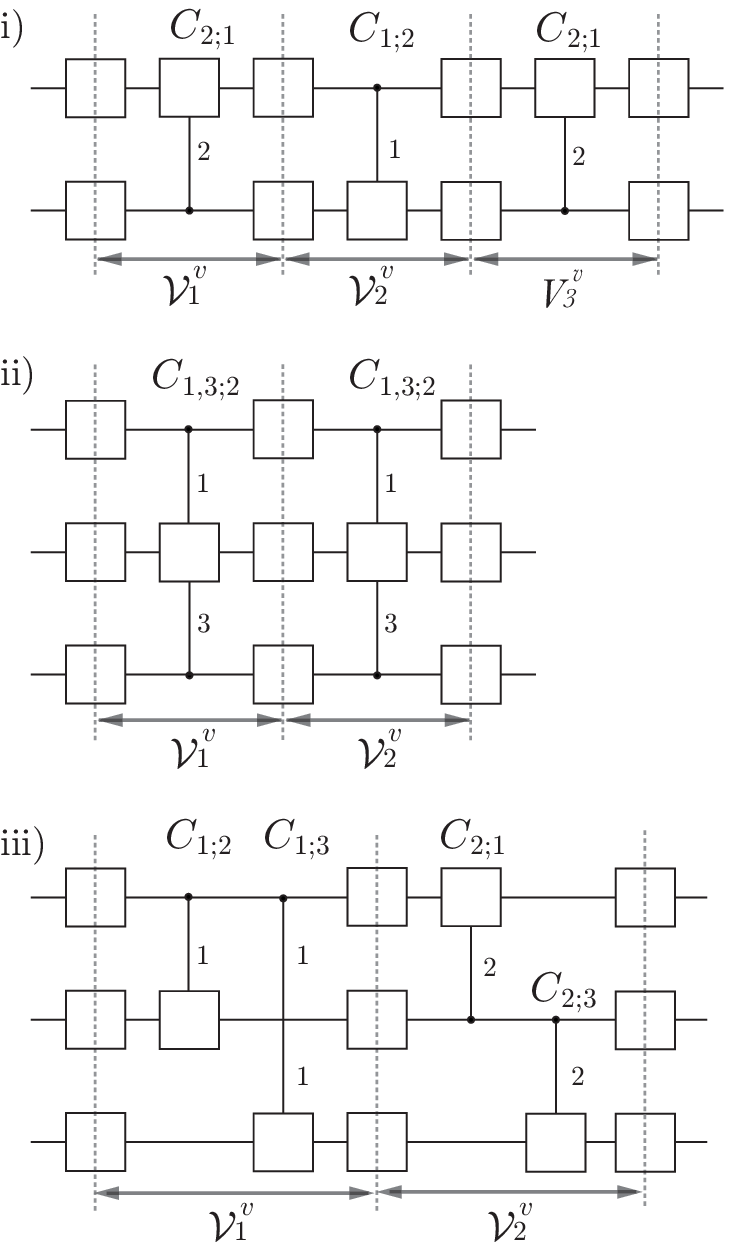}
  \end{center}
  \caption{i) A converted quantum circuit from the $(2,3)$-cluster network.  It is obtained by connecting three segments of circuits generated in step 1 to  step 3 of the protocol corresponding to $\mathcal{V}_1^v$, $\mathcal{V}_2^v$ and $\mathcal{V}_3^v$. It consists of two-qubit controlled unitary operations defined by $C_{ l;n}=\ket{0}\bra{0}_l\otimes u_n^{(0)}+\ket{1}\bra{1}_l\otimes u_n^{(1)}$, where $l$ denotes the wire segment of the control qubit and $u_n^{(i)}$ are arbitrary single qubit unitary operations on the $n$-th qubit, and arbitrary single qubit unitary operations are represented by boxes.
  ii) A converted quantum circuit from the $(3,2)$-cluster network.  It consists of three-qubit fully controlled unitary operations defined by $C_{ l,m;n}=\ket{00}\bra{00}_{l,m}\otimes u_n^{(00)}+\ket{01}\bra{01}_{l,m}\otimes u_n^{(01)}+\ket{10}\bra{10}_{l,m}\otimes u_n^{(10)}+\ket{11}\bra{11}_{l,m}\otimes u_n^{(11)}$, where $l, m$ denotes the wire segments of the control qubits and $u_n^{(ij)}$ are arbitrary single qubit unitary operations on the $n$-th qubit. iii) Another converted quantum circuit obtainable from the $(3,2)$-cluster network.}
\label{fig:convertedcircuit}
\end{figure}

Our conversion method generates infinitely many quantum circuits in general. However for special cluster networks, standard forms of quantum circuits can be obtained.   In Appendix A.3,  we show that  any converted circuit obtained from a $(2,3)$-cluster network can be simulated by the circuit presented in Fig.~\ref{fig:convertedcircuit} i), 
and any converted circuit obtained from a $(3,2)$-cluster network can be simulated by the circuit presented in Fig.\ref{fig:convertedcircuit} ii).

\section{Upper bound of computation}

In this section, we derive the condition for $k$-qubit unitary operations to be implementable over a given cluster network. We show that our conversion method presented in the previous section gives all implementable unitary operations over the $(k,N)$-cluster network for $k=2,3$.
\begin{theorem}
\label{theorem:det}
If i) a $k$-qubit unitary operation $U$ is deterministically implementable over the $(k,N)$-cluster network ($k\geq 2, N\geq 1$), then ii)  the matrix representation of $U$ in terms of the computational basis $U^M$ can be decomposed into
\begin{equation}
U^M=V_1^MV_2^M\cdots V_N^M,
\label{eq:udec}
\end{equation}
where each $V_i^M$ is a $2^k$ by $2^k$ unitary matrix such that
\begin{eqnarray}
V_i^M&=&\sum_{a_1=0}^1\sum_{a_2=0}^1\cdots\sum_{a_{k-1}=0}^1E^{(a_1)}_{1,i}\otimes E^{(a_1,a_2)}_{2,i}\otimes E^{(a_2,a_3)}_{3,i}\nonumber\\
&&\otimes\cdots\otimes E^{(a_{k-2},a_{k-1})}_{k-1,i}\otimes E^{(a_{k-1})}_{k,i},
\label{eq:dec}
\end{eqnarray}
where $E_{i,j}^{(m,n)}$ and $E_{i,j}^{(m)}$ are 2 by 2 complex matrices.
\end{theorem}
 
To prove Theorem \ref{theorem:det}, we use Lemma \ref{lemma:subunitary} represented in Appendix A.4 about a class of bipartite {\it separable maps}  that preserves entanglement.
A bipartite separable map $\Gamma_{sep}$ is a completely positive and trace preserving (CPTP) map as follows:
\begin{equation}
\Gamma_{sep}(\rho_{AB})=\sum_k (A_k\otimes B_k)\rho_{AB}(A_k\otimes B_k)^{\dag},
\end{equation}
where $\sum_k (A_k\otimes B_k)^{\dag}(A_k\otimes B_k)=\mathbb{I}_A\otimes \mathbb{I}_B$.
Since quantum network coding is equivalent to performing LOCC assisted by the resource state in our setting, we have to analyze multipartite LOCC. However, the analysis of multipartite LOCC is extremely difficult. Thus, we analyze multipartite separable maps, which are much easier to analyze than LOCC due to their simple structure. Note that a set of separable maps is exactly larger than that of LOCC \cite {9state, JNiset, DiVincezo, EChitambar}.
 
{\it Proof of Theorem 1.}
Denote by $\mathcal{H}_{\mathcal{I}_Q}=\otimes_{i=1}^k\mathcal{H}_{I_i}$ and $\mathcal{H}_{\mathcal{O}_Q}=\otimes_{i=1}^k\mathcal{H}_{O_i}$ the Hilbert spaces of $k$ input qubits and $k$ output qubits, respectively.    By introducing another  ancillary Hilbert space $\mathcal{H}_{I'_i}$ at the input nodes $v_{i,1}$, denote the Hilbert space of $k$ qubits by $\mathcal{H}_{\mathcal{I}'_Q}=\otimes_{i=1}^k\mathcal{H}_{I'_i}$. A joint state of $k$ copies of the EPR pairs in $\mathcal{H}_{\mathcal{I}_Q} \otimes \mathcal{H}_{\mathcal{I}'_Q}$ is denoted by 
$$\ket{\mathbb{I}}:= \frac{1}{\sqrt{D}} \sum_i\ket{i}_{\mathcal{I}_Q}\ket{i}_{\mathcal{I}'_Q}=\otimes_{i=1}^k\ket{\Phi^+}_{I_i,I'_i}$$ 
where $D=\dim(\mathcal{H}_{\mathcal{I}_Q})=2^k$.    If $U\in\mathbf{U}(\mathcal{H}_{\mathcal{I}_Q}:\mathcal{H}_{\mathcal{O}_Q})$ is deterministically implementable over a $(k,N)$-cluster network for $k \geq 2$ and $N \geq1$, 
and we consider applying $U$ on $\ket{\mathbb{I}}$. That is, there exists a LOCC map $\Gamma$ such that
\begin{equation}
\frac{1}{D} \sum_{i,j}\Gamma(\ket{i}\bra{j}_{\mathcal{I}_Q}\otimes\ket{\Phi}\bra{\Phi}_{\mathcal{R}})\otimes \ket{i}\bra{j}_{\mathcal{I}'_Q}=\ket{U}\bra{U},
\label{loccmapfornecessity}
\end{equation}
where $\ket{\Phi}_{\mathcal{R}}$ is the resource state of the $(k,N)$-cluster network and $\ket{U}$ is defined by
\begin{eqnarray}
\ket{U}:=(U \otimes \mathbb{I}) \ket{\mathbb{I}} \in \mathcal{H}_{\mathcal{O}_Q} \otimes \mathcal{H}_{\mathcal{I}'_Q}.
\label{defUket}
\end{eqnarray}

By defining a map represented by the left hand side of Eq.~(\ref{loccmapfornecessity}) as $\Gamma'(\ket{\Phi}\bra{\Phi}_{\mathcal{R}}):= \frac{1}{D} \sum_{i,j}\Gamma(\ket{i}\bra{j}_{\mathcal{I}_Q}\otimes\ket{\Phi}\bra{\Phi}_{\mathcal{R}})\otimes \ket{i}\bra{j}_{\mathcal{I}'_Q}$, where  $\Gamma^\prime$ is also a LOCC map if we assume two qubits belonging to $\mathcal{H}_{I_i}$ and $\mathcal{H}_{I'_i}$ are in the same input node for all $i$.  Since any LOCC maps are separable maps, there exists a separable map $\Gamma^\prime_{sep}$ satisfying
\begin{equation}
\Gamma^\prime_{sep} (\ket{\Phi}\bra{\Phi}_{\mathcal{R}})=\ket{U}\bra{U},
\label{eq:separable}
\end{equation}
if $U$ is deterministically implementable over a $(k,N)$-cluster network.
Since $\Gamma^\prime_{sep}$ is a map from a pure state to a pure state, the action of $\Gamma^\prime_{sep}$ represented by Eq.\eqref{eq:separable} can be equivelently given by the existence of a set of linear operators (the Kraus operators) $\{A_{i,j}^m\}_m$ for each node $v_{i,j}$ and a probability distribution $\{p_m\}$ 
such that
\begin{eqnarray}
\label{eq:LOCC11}
\forall m;\,\, \otimes_{i=1}^k\otimes_{j=1}^N A_{i,j}^m\ket{\Phi}_{\mathcal{R}}&=&\sqrt{p_m}\ket{U},\\
\sum_m \otimes_{i=1}^k\otimes_{j=1}^N (A_{i,j}^{m\dag}A_{i,j}^m)&=&\mathbb{I},
\end{eqnarray}
where
\begin{eqnarray}
A_{i,1}^m&\in&\mathbf{L}(\mathcal{H}_{v_{i,1}}:\mathcal{H}_{I'_i})\nonumber\\
&&(1\leq i\leq k),\nonumber\\
A_{i,j}^m&\in&\mathbf{L}(\mathcal{H}_{v_{i,j}}:\mathbb{C})\nonumber\\
&&(1\leq i\leq k,2\leq j\leq N-1),\nonumber\\
A_{i,N}^m&\in&\mathbf{L}(\mathcal{H}_{v_{i,N}}:\mathcal{H}_{O_i})\nonumber\\
&&(1\leq i\leq k),
\end{eqnarray}
and $\mathcal{H}_{v_{i,j}}$ is the Hilbert space of qubits of the resource state at node $v_{i,j}$ defined by
\begin{equation}
\mathcal{H}_{v_{i,j}}=\mathcal{H}_{S_{i,j}}\otimes\mathcal{H}_{K_{i,j}}
\end{equation}
\begin{subnumcases}
{\mathcal{H}_{S_{i,j}}=}
\mathcal{H}_{S_{1,j}^1} & ($i=1$) \nonumber\\
\mathcal{H}_{S_{i,j}^1}\otimes\mathcal{H}_{S_{i,j}^2} & ($2\leq i\leq k-1$) \nonumber\\
\mathcal{H}_{S_{k,j}^2} & ($i=k)$\nonumber
\end{subnumcases}
\begin{subnumcases}
{\mathcal{H}_{K_{i,j}}=}
\mathcal{H}_{K_{i,1}^1} & ($j=1$) \nonumber\\
\mathcal{H}_{K_{i,j}^1}\otimes\mathcal{H}_{K_{i,j}^2} & ($2\leq j\leq N-1$) \nonumber\\
\mathcal{H}_{K_{i,N}^2} & ($j=N)$\nonumber
\end{subnumcases}

First, letting $E_m=\otimes_{i=1}^k A_{i,1}^m$, $F_m=\otimes_{i=1}^k \otimes_{j=2}^N A_{i,j}^m$ and applying Lemma \ref{lemma:subunitary} presented in  Appendix A.4, we obtain for all $\{m| p_m\neq 0\}$,
\begin{eqnarray}
\exists \alpha_{1,m}>0,\exists V_{1,m}^M\in\mathbf{U}(\mathbb{C}^{D});\,\,E_m^M=\alpha_{1,m} V_{1,m}^M,
\end{eqnarray}
where $\mathbf{U}(\mathbb{C}^{D})$ is the set of $D$ by $D$ unitary matrices and $E_m^M$ is a $D$ by $D$ matrix satisfying $$(E_m^M)_{a,b}=\bra{a}_{\mathcal{I}'_Q}(\otimes_{i=1}^k A_{i,1}^m)\ket{A_b}_{S_{*,1}^*,K_{*,1}^1}$$ and $$\ket{A_b}_{S_{*,1}^*,K_{*,1}^1}=\otimes_{i=1}^{k-1}\ket{\Phi^+}_{S_{i,1}^1,S_{i+1,1}^2}\otimes \ket{b}_{K_{1,1}^1,\cdots,K_{k,1}^1}.$$
Let
\begin{eqnarray}
A_{1,1}^m&=&\sum_{a_1=0}^1 \bra{a_1}_{S_{1,1}^1}\otimes E_{1,1}^{(a_1),m}\\
A_{i,1}^m&=&\sum_{a_1=0}^1\sum_{a_2=0}^1 \bra{a_1}_{S_{i,1}^1}\bra{a_2}_{S_{i,1}^2}\otimes E_{i,1}^{(a_1,a_2),m}\nonumber\\
&&(2\leq i \leq k-1)\\
A_{k,1}^m&=&\sum_{a_1=0}^1 \bra{a_1}_{S_{k,1}^2}\otimes E_{k,1}^{(a_1),m}
\end{eqnarray}
where $E_{1,1}^{(a_1),m}\in\mathbf{L}(\mathcal{H}_{K_{1,1}^1}:\mathcal{H}_{I'_1})$, $E_{i,1}^{(a_1,a_2),m}\in\mathbf{L}(\mathcal{H}_{K_{i,1}^1}:\mathcal{H}_{I'_i})$ and $E_{k,1}^{(a_1),m}\in\mathbf{L}(\mathcal{H}_{K_{k,1}^1}:\mathcal{H}_{I'_k})$. Thus, $V_{1,m}^M$ can be decomposed into
\begin{eqnarray}
V_{1,m}^M&=&\sum_{a_1,\cdots,a_{k-1}=0}^1 E^{(a_1),m}_{1,1}\otimes E^{(a_1,a_2),m}_{2,1}\otimes\cdots\nonumber\\
&&\otimes E^{(a_{k-2},a_{k-1}),m}_{k-1,1}\otimes E^{(a_{k-1}),m}_{k,1}.
\end{eqnarray}
 Note that we identify a linear operation and its matrix representation in the computational basis, e.g., $E_{1,1}^{(a_1),m}$ is a 2 by 2 complex matrix.

Next, letting $E_m=\otimes_{i=1}^k \otimes_{j=1}^2 A_{i,j}^m$, $F_m=\otimes_{i=1}^k \otimes_{j=3}^N A_{i,j}^m$ and using Lemma \ref{lemma:subunitary} represented in Appendix A.4, we obtain for all $\{m|p_m\neq 0\}$,
\begin{equation}
\exists \alpha_{2,m}>0,\exists V_{2,m}^M\in\mathbf{U}(\mathbb{C}^{D});\,\,E_m^M=\alpha_{2,m} V_{2,m}^M,
\end{equation}
where $E_m^M$ is a $D\times D$ matrix such that 
$$(E_m^M)_{a,b}=\bra{a}_{\mathcal{I}'_Q}(\otimes_{i=1}^k\otimes_{j=1}^2 A_{i,j}^m)\ket{A_b}_{S_{*,1}^*,S_{*,2}^*,K_{*,1}^1,K_{*,2}^*}$$ 
and 
\begin{eqnarray}
\ket{A_b}_{S_{*,1}^*,S_{*,2}^*,K_{*,1}^1,K_{*,2}^*}=\otimes_{i=1}^{k-1}\ket{\Phi^+}_{S_{i,1}^1,S_{i+1,1}^2}  \nonumber \\
\otimes_{i=1}^{k-1}\ket{\Phi^+}_{S_{i,2}^1,S_{i+1,2}^2}
\otimes_{i=1}^{k}\ket{\Phi^+}_{K_{i,1}^1,K_{i,2}^2} \nonumber \\
\otimes \ket{b}_{K_{1,2}^1,\cdots,K_{k,2}^1}. \nonumber
\end{eqnarray}
Let
\begin{eqnarray}
A_{1,2}^m&=&\sum_{a_1=0}^1 \bra{a_1}_{S_{1,2}^1}\otimes E_{1,2}^{(a_1),m}\\
A_{i,2}^m&=&\sum_{a_1=0}^1\sum_{a_2=0}^1 \bra{a_1}_{S_{i,2}^1}\bra{a_2}_{S_{i,2}^2}\otimes E_{i,2}^{(a_1,a_2),m}\nonumber\\
&&(2\leq i \leq k-1)\\
A_{k,2}^m&=&\sum_{a_1=0}^1 \bra{a_1}_{S_{k,2}^2}\otimes E_{k,2}^{(a_1),m},
\end{eqnarray}
where $E_{1,2}^{(a_1),m}\in\mathbf{L}(\mathcal{H}_{K_{1,2}^1}\otimes \mathcal{H}_{K_{1,2}^2}:\mathbb{C})$, $E_{i,2}^{(a_1,a_2),m}\in\mathbf{L}(\mathcal{H}_{K_{i,2}^1}\otimes\mathcal{H}_{K_{i,2}^2}:\mathbb{C})$ and $E_{k,2}^{(a_1),m}\in\mathbf{L}(\mathcal{H}_{K_{k,2}^1}\otimes\mathcal{H}_{K_{k,2}^2}:\mathbb{C})$. By a straightforward calculation, $V_{2,m}^M$ are shown to be decomposed into
\begin{eqnarray}
V_{2,m}^M&=&V_{1,m}^M\sum_{a_1,\cdots,a_{k-1}=0}^1 E'^{(a_1),m}_{1,2}\otimes E'^{(a_1,a_2),m}_{2,2}\otimes\cdots\nonumber\\
&&\otimes E'^{(a_{k-2},a_{k-1}),m}_{k-1,2}\otimes E'^{(a_{k-1}),m}_{k,2},
\end{eqnarray}
where $E'^{(a_1),m}_{1,2},E'^{(a_1,a_2),m}_{i,2}$, and $E'^{(a_1),m}_{k,2}$ are $2\times 2$ complex matrices.

Iterating this procedure, we obtain for all $\{m|p_m\neq 0\}$,
\begin{equation}
\exists \alpha>0,\exists W^M\in\mathbf{U}(\mathbb{C}^{D});\,\,E_m^M=\alpha W^M,F_m^M=\frac{\sqrt{p_m}}{\alpha}\overline{W^M},
\end{equation}
where $W^M$ and $\overline{W^M}$ can be decomposed into
\begin{eqnarray}
W^M=V_1^MV_2^M\cdots V_{N-1}^M\\
\overline{W^M}=U^{M\dag}V_N^M,
\end{eqnarray}
and $V_i^M=\sum_{a_1,\cdots,a_{k-1}=0}^1E^{(a_1)}_{1,i}\otimes E^{(a_1,a_2)}_{2,i}\otimes\cdots\otimes E^{(a_{k-2},a_{k-1})}_{k-1,i}\otimes E^{(a_{k-1})}_{k,i}\in\mathbf{U}(\mathbb{C}^{D})$.
$U^M$ can be decomposed into the form of Eq.\eqref{eq:udec} since $\overline{V_i^M}$ and $V_i^{M\dag}$ can be decomposed into the form of Eq.\eqref{eq:dec}.
\begin{flushright}
$\Box$
\end{flushright}

In the case of the $(2,N)$-cluster networks, which we call $N$-bridge {\it ladder networks}, $V_i$ is locally unitarily equivalent to the two-qubit controlled unitary operation since its operator Schmidt rank is $2$ \cite{SchmidtDecomp}. Thus,  statements i) and ii) in Theorem 1 are equivalent since a sequence of $N$ two-qubit controlled unitary operations is implementable by the converted circuit presented in Fig.~4.3 i). Then we obtain the following theorem for the ladder networks.

\begin{theorem}
A unitary operation $U$ is deterministically implementable over the $N$-bridge ladder network if and only if
$\textsc{KC\#}(U)\leq N$, where $\textsc{KC\#} (U)$ is the Kraus-Cirac number of a two-qubit unitary operator $U$, which is the number of non-zero parameters $x,y,z$ in Eq.~\eqref{gloablpartofKCD} characterizing the global part of $U$.
\label{theorem:main}
\end{theorem}
This theorem is proven by using the following lemma relating the Kraus-Cirac number of a two-qubit unitary operation and the decomposition of the unitary operation into controlled unitary operations  shown in \cite{SoedaAkibue}.

\begin{lemma}
\label{lemma:KC1}
Consider a set of two-qubit unitary operators $\textbf{U}_c$ that is locally unitarily equivalent to a controlled unitary operator.
The decomposition of a two-qubit unitary operator $U$ into a shortest sequence of two-qubit unitary operators in $\textbf{U}_c$ depends on the Kraus-Cirac number $\textsc{KC\#}(U)$ of $U$ as
\begin{eqnarray}
\{U\in SU(4)|\textsc{KC\#}(U)\leq1\}&=&\{U|U\in\textbf{U}_c\}\nonumber\\
\{U\in SU(4)|\textsc{KC\#}(U)\leq2\}&=&\{UV|U,V\in\textbf{U}_c\}\nonumber\\
\{U\in SU(4)|\textsc{KC\#}(U)\leq3\}&=&\{UVW|U,V,W\in\textbf{U}_c\}. \nonumber
\end{eqnarray}
\end{lemma}

{\it Proof of Theorem 2.}
Since \textsc{KC\#}$(U)$ is less than or equal to $N$ if and only if $U$ can be decomposed into $N$ two-qubit controlled unitary operations as shown in Lemma \ref{lemma:KC1}, and $N$ two-qubit controlled unitary operations are deterministically implementable over $N$-bridge ladder network, Theorem 3 is straigtforwardly shown.
\begin{flushright}
$\Box$
\end{flushright}

 We also show that statements  i) and ii) are equivalent in the case of  the $(3,N)$-cluster networks in Appendix A.5.

\section{Butterfly, grail and square networks}

For classical network coding, it has been shown that there exists a network coding protocol over a $2$-pair network, which has two input nodes and two output nodes, if and only if the network has at least one of the butterfly, grail and identity substructures \cite{ButterflyGrail}.
Thus any classical network coding protocol over a $2$-pair network can be reduced into a combination of protocols over the butterfly, grail or identity networks, and these networks are fundamental primitive networks for classical network coding.   As  a first step to investigate implementability of quantum computation over general $2$-pair quantum networks, we investigate  implementability of two-qubit unitary operations over the butterfly and grail networks in this section by using the method converting a $(k, N)$-cluster network into quantum networks introduced in Section 4.1.

By constructing a protocol for implementing $U_{global}  (x,y,z)$ defined in Eq.\eqref{gloablpartofKCD} for arbitrary $x,y,z$, we obtain the following theorem.

\begin{figure}
\begin{center}
  \includegraphics[height=.23\textheight]{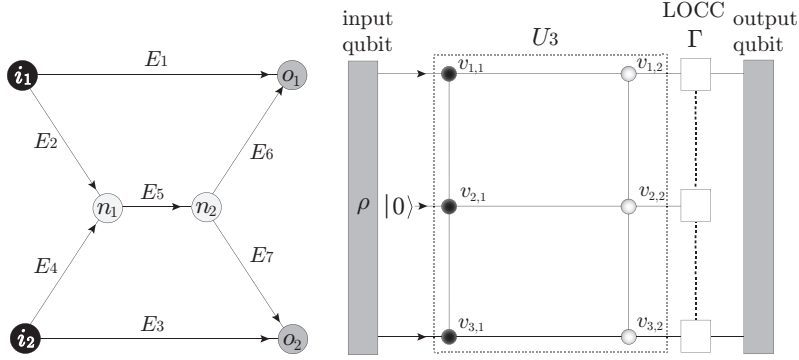}
  \end{center}
  \caption{The nodes $i_1$, $i_2$, $o_1$, $o_2$, $n_1$ and $n_2$ of the butterfly network correspond to the nodes $v_{1,1}$, $v_{3,1}$, $v_{1,2}$, $v_{3,2}$, $v_{2,1}$ and $v_{2,2}$ of a $(3,2)$-cluster network each other. The {\it two}-qubit unitary operation $U_{global} (x,y,z) = e^{ i(x X \otimes X+y Y \otimes Y+z Z \otimes Z)}$ is implementable over a $(3,2)$-cluster network by fixing an input state at the node $v_{2,1}$ at $\ket{0}$, performing an appropriate a {\it three-qubit} unitary operation $U_3$ and performing an appropriate LOCC map $\Gamma$ consisting of a measurement on the qubit at the output node $v_{2,2}$ and the conditional operations on the other output nodes $v_{1,2}$ and $v_{3,2}$ depending on the measurement outcome.
 }
\label{fig:butterflycluster}
\end{figure}

\begin{theorem}
Any two-qubit unitary operation is deterministically implementable over the butterfly network.
\end{theorem}

\begin{proof}
For  implementability of $U_{global} (x,y,z)$ over the butterfly network represented by the left hand side of Fig.~\ref{fig:butterflycluster}, we consider a  $(3,2)$-cluster network represented by the right hand side of Fig.~\ref{fig:butterflycluster} by assigning the nodes $\{ i_1, n_1, i_2, o_1, n_2, o_2 \}$ of the butterfly network  to the nodes $\{ v_{1,1}, v_{2,1}, v_{3,1}, v_{1,2}, v_{2,2}, v_{3,2} \}$ of the $(3,2)$-cluster network, respectively.   In this assignment, the correspondence of the edges of the butterfly network %represented by the notation given by Figure \ref{fig:qgrail} i) 
and the horizontal and vertical sets of edges $\mathcal{K}_1, \mathcal{S}_1, \mathcal{S}_2$ of the $(3,2)$-cluster network is given by 
\begin{eqnarray}
\{ E_1, E_5, E_3\} &\leftrightarrow& \mathcal{K}_1, \nonumber \\
\{ E_2, E_4\} &\leftrightarrow& \mathcal{S}_1, \nonumber \\
\{ E_6, E_7\}, &\leftrightarrow& \mathcal{S}_2.
\end{eqnarray}
Thus any two-qubit unitary operation is deterministically implementable over the butterfly network if any $U_{global} (x,y,z) $ in the form of Eq.~(\ref{gloablpartofKCD}) is deterministically implementable over the $(3,2)$-cluster network where input states are given at nodes $v_{1,1}$ and $v_{3,1}$ and output states are obtained at nodes $v_{1,2}$ and $v_{3,2}$, since the topology of the butterfly network is the same as that of the $(3,2)$-cluster network.  

We construct a protocol implementing two-qubit unitary $U_{global}  (x,y,z)$ by setting a fixed input state at node $v_{2,1}$ and arbitrary two-qubit input state at nodes $v_{1,1}$ and $v_{3,1}$ as a three-qubit input state at input nodes $\mathcal{I}=\{v_{1,1}, v_{2,1}, v_{3,1} \}$, and implementing a three-qubit unitary operation denoted by $U_3$ over the $(3,2)$-cluster network followed by an LOCC map denoted by $\Gamma$ performed at output nodes $\mathcal{O}=\{v_{1,2}, v_{2,2}, v_{3,2} \}$. 
Recall that a unitary operation of represented by the quantum circuit shown in Fig.~\ref{fig:convertedcircuit} ii) is implementable over the $(3,2)$-cluster network.    That is, two three-qubit fully controlled unitary operations $C_{1,3;2}$ are implementable, one at nodes $\mathcal{I}$ and another at nodes $\mathcal{O}$.   The following protocol shows that by choosing appropriate parameters for one of the three-qubit fully controlled unitary operations and one of single-qubit local unitary operations in $U_3$, we can implement $U_{global} (x,y,z)$ with arbitrary $x, y, z$.

The protocol for implementing $U_{global} (x,y,z)$: 
\begin{enumerate}
\item An arbitrary two-qubit input state $\rho$ is given for qubits at input nodes $v_{1,1}$ and $v_{3,1}$ and a fixed input state $\ket{0}$ is prepared for the qubit at node $v_{2,1}$.  
\item Implement $U_3$ of which the quantum circuit representation is given by the left shaded part of Fig.\ref{fig:butterflycircuit} over the $(3,2)$-cluster network.   
\begin{enumerate}
\item All single-qubit unitary operations appearing in the circuit representation of $U_3$
are trivially performed at each node.  
\item The first fully controlled unitary operation implemented at input nodes $\mathcal{I}$ using the EPR pairs represented by  vertical edges $\mathcal{S}_1$ is given by  $C_{1,3;2} (\{ u_n^{(a b)} \}_{a,b=0,1} )$ where $u_n^{(00)}=u_n^{(11)}=\mathbb{I}$ and $u_n^{(01)}=u_n^{(10)}=Z$.
\item To transmit a qubit state from input nodes $v_{i,1}$ to output node $v_{i,2}$ for $i=1,2,3$, quantum teleportation is performed for each $i$ by using the  EPR pair represented by a horizontal edge in $\mathcal{K}_1$.
\item  The second fully controlled unitary operation implemented at output nodes $\mathcal{O}$ contains parameters $y$ and $z$ and  is given by
$C'_{1,3;2} (\{ w_n^{(a b)} \}_{a,b=0,1} )$ where
\begin{eqnarray}
w_n^{(00)}=w_n^{(11)}=e^{i(z-y)}\ket{0}\bra{0}-ie^{i(z+y)}\ket{1}\bra{1}, \nonumber \\
w_n^{(01)}=w_n^{(10)}=e^{-i(z-y)}\ket{0}\bra{0}-ie^{-i(z+y)}\ket{1}\bra{1}. \nonumber
\end{eqnarray}
\item After implementing $C'_{1,3;2} (\{ w_n^{(a b)} \}_{a,b=0,1} )$, a single-qubit unitary operation parameterized by $x$ given by
\begin{equation}
u(x) = \frac{1}{\sqrt{2}} \left(\begin{array}{cc} e^{i x} & -i e^{- i x} \\e^{i x} & i e^{- i x}\end{array}\right)
\end{equation}
is performed at node $v_{2,2} \in \mathcal{O}$. 
\end{enumerate}
\item Perform an LOCC map $\Gamma$ at output nodes $\mathcal{O}$ of which the quantum circuit representation is given by the right shaded part of Fig.\ref{fig:butterflycircuit}.  The map  $\Gamma$ consists of the following three steps.
\begin{enumerate}
\item Perform a projective measurement on the qubit at node $v_{2,2}$  in the computational basis $\{ \ket{0}\bra{0}, \ket{1}\bra{1} \}$.
\item Classically communicate the measurement outcome $k \in \{ 0,1 \}$ from node $v_{2,2}$ to $v_{1,2}$ and also to $v_{3,2}$.
\item If $k=1$, perform a conditional operation $X$ on output qubits at nodes $v_{1,2}$ and $v_{3,2}$, otherwise do nothing.
\end{enumerate}
\end{enumerate}

This protocol maps any input state $\rho$ given at input nodes $v_{1,1}$ and $v_{3,1}$ to   
\begin{eqnarray}
U_{global}  (x,y,z)  \rho U_{global}^{\dag} (x,y,z) =\Gamma(U_3(\rho\otimes\ket{0}\bra{0})U_3^{\dag})
\end{eqnarray}
at output nodes $v_{1,2}$ and $v_{3,2}$ where $\ket{0}$ represents the fixed input state at node $v_{2,1}$.  See Appendix A.6 for details of calculations.   It is straightforward to translate the protocol over the $(3,2)$-cluster network to a protocol to implement $U_{global} (x,y,z) $ over the butterfly network by using the correspondence of vertices and edges.  Thus, $U_{global} (x,y,z) $ is deterministically implementable over the butterfly network.
\end{proof}

\begin{figure}
\begin{center}
  \includegraphics[height=.2\textheight]{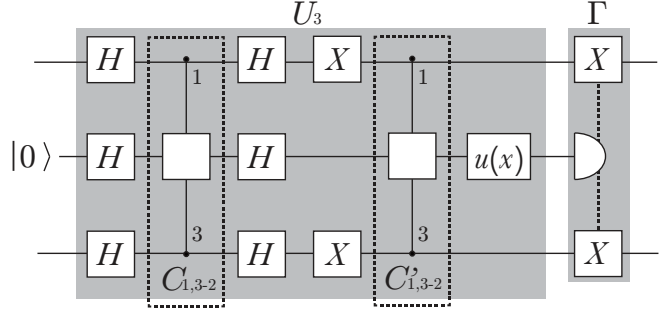}
  \end{center}
  \caption{ A quantum circuit representation of a three-qubit unitary operation $U_3$ (the left shaded part) and an LOCC map $\Gamma$ (the right shaded part) used in a protocol for implementing a two-qubit unitary operation $U_{global} (x,y,z) = e^{ i(x X \otimes X+y Y \otimes Y+z Z \otimes Z)}$ on the first and third qubits.  The input state of the second qubit is fixed in $\ket{0}$.
The single qubit unitary operations represented by boxes are given by $H=(\ket{0}\bra{0}+\ket{0}\bra{1} +\ket{1}\bra{0} - \ket{1}\bra{1})/\sqrt{2}$, $u(x)=H(e^{ix}\ket{0}\bra{0}-ie^{-ix}\ket{1}\bra{1}$) and $X= \ket{0}\bra{1} +\ket{1}\bra{0}$.  The target single-qubit unitary operations of the first three-qubit fully controlled unitary operation $C_{ 1,3;2} (\{ u_n^{(a b)} \}_{a,b=0,1} )$  are given by $u_n^{(00)}=u_n^{(11)}=\mathbb{I}$ and $u_n^{(01)}=u_n^{(10)}=Z$.  The target single-qubit unitary operations of the second three-qubit fully controlled unitary operation $C'_{1,3;2} (\{ w_n^{(a b)} \}_{a,b=0,1} )$  are  given by $w_n^{(00)}=w_n^{(11)}=e^{i(z-y)}\ket{0}\bra{0}-ie^{i(z+y)}\ket{1}\bra{1}$ and $w_n^{(01)}=w_n^{(10)}=e^{-i(z-y)}\ket{0}\bra{0}-ie^{-i(z+y)}\ket{1}\bra{1}$.  The half circle symbol represents a projective measurement in the computational basis $\{ \ket{k} \bra{k} \}_{k=0,1}$.  The single qubit operations (boxes) connected to the measurement symbol by dotted lines represent conditional unitary operations performed only if the measurement result is $k=1$ and do nothing (or perform $\mathbb{I}$) if $k=0$.}
\label{fig:butterflycircuit}
\end{figure}

For implementability of $U_{global} (x,y,z)$ over the grail network, we consider a $(2,3)$-cluster network by a assigning the nodes $\{n_1,n_2,o_1,i_2,n_3,n_4 \}$ of the grail network  to the nodes $\{ v_{1,1}, v_{1,2}, v_{1,3}, v_{2,1}, v_{2,2}, v_{2,3} \}$ of the $(2,3)$-cluster network, respectively (Fig.~\ref{fig:grailcluster}).  The $(2,3)$-cluster network can be converted to a quantum circuit containing three controlled NOT gates and arbitrary single unitary operations that are inserted between the controlled NOT gates.   It is shown that any two-qubit unitary operations $U_{global}  (x,y,z)$ can be decomposed into three controlled NOT gates and single unitary operations inserted between the controlled NOT gates \cite{3CNOT}. Thus any two-qubit controlled unitary operation is deterministically implementable over the grail network.

\begin{figure}
\begin{center}
  \includegraphics[height=.25\textheight]{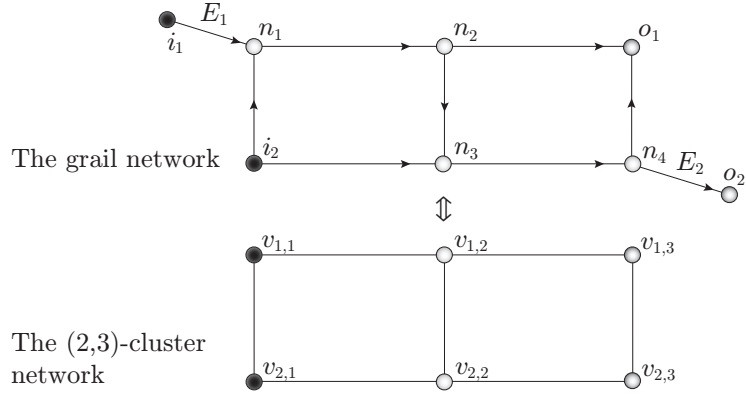}
  \end{center}
  \caption{ The nodes $n_1$, $n_2$, $o_1$, $i_2$, $n_3$ and $n_4$ of the grail network correspond to the nodes $v_{1,1}$, $v_{1,2}$, $v_{1,3}$, $v_{2,1}$, $v_{2,2}$ and $v_{2,3}$ of a $(2,3)$-cluster network, respectively. The set of all unitary operations implementable over the $(2,3)$-cluster network is also implementable over the grail network, since we can use the edges $E_1$ and $E_2$ for just teleporting qubits and the rest of the network forms the $(2,3)$-cluster network, with which any two-qubit unitary operation is implementable.
}  
\label{fig:grailcluster}
\end{figure}

\section{Probabilistic computation}

In this section, we investigate the probabilistic implementation of unitary operations.
There is no classical network coding protocol to  send single bits from $v_{1,1}$ to $v_{2,2}$ and from $v_{2,1}$ to $v_{1,2}$ over  a $(2,2)$-cluster network since there  is no butterfly, grail or identity substructure. 
This task corresponds to  implementing a SWAP operator in quantum network coding.
 It is interesting to  know whether there exists a task that is not achievable by classical network coding but a corresponding task is achievable in a quantum setting or not.  We give a negative result in this section.
 Using Theorem 2, we see that a SWAP operator is not deterministically implementable over  a $(2,2)$-cluster network, which is a 2-bridge ladder network, since the Kraus-Cirac number of the SWAP operator is $3$.
Furthermore, we show that the SWAP operator is not implementable even {\it probabilistically} in this section.

\begin{theorem}
\label{theorem:prob}
A $k$-qubit unitary operation $U$ is probabilistically implementable over the $(k,N)$-cluster network ($k\geq 2, N\geq 1$) if and only if the matrix representation of $U$ in terms of the computational basis $U^M$ can be decomposed into
\begin{equation}
U^M=F_1^MF_2^M\cdots F_{N}^M,
\label{eq:Sunitary}
\end{equation}
where each $F_i^M$ is a $2^k$ by $2^k$ complex matrix that can be decomposed in the same way as Eq.~(\ref{eq:dec})
\end{theorem}

\begin{proof}
Similar to the case of deterministic implementation, we consider applying $U \in\mathbf{U}(\mathcal{H}_{\mathcal{I}_Q}:\mathcal{H}_{\mathcal{O}_Q})$ on a part of $k$ maximally entangled states $\ket{\mathbb{I}} \in \mathcal{H}_{\mathcal{I}_Q} \otimes \mathcal{H}_{\mathcal{I}'_Q}$.  Then $U$ is probabilistically implementable over the $(k,N)$-cluster network ($k\geq 2, N\geq 1$) if and only if there exists a stochastic LOCC (SLOCC) map $\Gamma''$ and non-zero probability $p>0$ such that
\begin{equation}
\Gamma''(\ket{\Phi}\bra{\Phi}_{\mathcal{R}})=p\ket{U}\bra{U},
\label{eq:SLOCC}
\end{equation}
where $\ket{\Phi}_{\mathcal{R}}$ is the resource state of the $(k,N)$-cluster network and $\ket{U} \in \mathcal{H}_{\mathcal{O}_Q} \otimes \mathcal{H}_{\mathcal{I}'_Q}$ is defined by Eq.~(\ref{defUket}).
Eq.~\eqref{eq:SLOCC} is equivalent to the statement that there exists a set of  linear operators $\{A_{i,j}\}$ and non-zero probability $p>0$ such that
\begin{equation}
\label{eq:SLOCC1}
\otimes_{i=1}^k\otimes_{j=1}^N A_{i,j}\ket{\Phi}_{\mathcal{R}}=\sqrt{p}\ket{U}.
\end{equation}
The conditions of $\{A_{i,j}\}$ given by Eq.~\eqref{eq:SLOCC1} is similar to the conditions of Kraus operators $\{A_{i,j}^m\}_m$ given by Eq.~\eqref{eq:LOCC11} presented in the proof of Theorem 1. The index $m$ is dropped in Eq.~\eqref{eq:SLOCC1} since the map we consider is SLOCC instead of LOCC considered in Theorem 1. By taking the correspondence between $A_{i,j}$ and $A_{i,j}^m$, we obtain a decomposition of the form presented in Eq.~\eqref{eq:Sunitary}.

\end{proof}

\begin{lemma}
A SWAP operation $U^{SWAP}=|00\rangle\langle 00|+|01\rangle\langle 10|+|10\rangle\langle 01|+|11\rangle\langle 11|$ is not implementable over the $2$-bridge ladder network with non-zero probability.
\end{lemma}
\begin{proof}
By using Theorem 4, the SWAP operation is probabilistically implementable over the $(2,2)$-cluster network ($2$-bridge ladder network) if and only if there exists a linear operation $P,Q\in\mathbf{L}(\mathcal{H}_1\otimes\mathcal{H}_2)$ and $E_{i,j}^{(k)}\in\mathbf{L}(\mathcal{H}_i)$ such that
\begin{eqnarray}
U^{SWAP}&=&PQ,\\
P&=&E_{1,1}^{(0)}\otimes E_{2,1}^{(0)}+E_{1,1}^{(1)}\otimes E_{2,1}^{(1)}\label{eq:opsch1}\\
Q&=&E_{1,2}^{(0)}\otimes E_{2,2}^{(0)}+E_{1,2}^{(1)}\otimes E_{2,2}^{(1)},\label{eq:opsch2}
\end{eqnarray}
where $\mathcal{H}_i=\mathbb{C}^2$.
 For any linear operations $M$, there exists the operator Schmidt decomposition, and we can define the {\it operator Schmidt rank} ${\rm Op}\#_1^2(M)$, which is the number of non-zero coefficients of the operator Schmidt decomposition. Since $P$ and $Q$ can be decomposed into Eq.\eqref{eq:opsch1} and Eq.\eqref{eq:opsch2}, we can derive
\begin{eqnarray}
{\rm Op}\#_1^2(P)&\leq& 2,\\
{\rm Op}\#_1^2(Q)&\leq& 2.
\end{eqnarray}
Since ${\rm Op}\#_1^2(U^{SWAP})=4$, ${\rm Op}\#_1^2(P)={\rm Op}\#_1^2(Q)=2$. In \cite{inverseSchmidtrank}, it is shown that if ${\rm Op}\#_1^2(P)=2$ and $P$ is invertible, ${\rm Op}\#_1^2(P^{-1})=2$. Thus, the SWAP operation is probabilistically implementable if and only if there exists linear operations $P,Q\in\mathbf{L}(\mathcal{H}_1\otimes\mathcal{H}_2)$ such that
\begin{eqnarray}
Q=U^{SWAP}P,\\
{\rm Op}\#_1^2(P)= 2,\,\,{\rm rank}(P)=4\\
{\rm Op}\#_1^2(Q)= 2,\,\,{\rm rank}(Q)=4.
\end{eqnarray}
In general, we can regard $P$ as a matrix representation of a four qubit pure state $\ket{\Phi}_{1,2,3,4}$;
\begin{equation}
P=\sum_{i=1}^4\bra{i}_{1,2}\ket{\Phi}_{1,2,3,4}\bra{i}_{1,2}.
\end{equation}
Then, the following correspondences are obtained,
\begin{eqnarray}
{\rm rank}(P)=4 &\Leftrightarrow&{\rm Sch}\#_{1,2}^{3,4}(\ket{\Phi})=4,
\label{eq:4qubit1}\\
{\rm Op}\#_1^2(P)=2 &\Leftrightarrow&{\rm Sch}\#_{1,3}^{2,4}(\ket{\Phi})=2,
\label{eq:4qubit2}\\
{\rm Op}\#_1^2(U^{SWAP}P)=2 &\Leftrightarrow&{\rm Sch}\#_{1,4}^{2,3}(\ket{\Phi})=2,
\label{eq:4qubit3}
\end{eqnarray}
where ${\rm Sch}\#_{1,2}^{3,4}(\ket{\Phi})$ is a Schmidt number in terms of a  partition between qubit $1,2$ and qubit $3,4$. We show that there is no four qubit state  simultaneously satisfying Eqs. \eqref{eq:4qubit1}, \eqref{eq:4qubit2}, and \eqref{eq:4qubit3} in Appendix A.7.
\end{proof}

Note that we have shown that $U$ can be decomposed into a particular form represented by Eq.\eqref{eq:udec} {\it if} $U$ is deterministically implementable in Theorem 1 and that $U$ can be decomposed into a particular form represented by Eq.\eqref{eq:Sunitary} {\it if and only if} $U$ is probabilistically implementable in Theorem 4.
And each factor $F_i^M$ in Eq.\eqref{eq:Sunitary} can be a non-unitary complex matrix while each factor $V_i^M$ in Eq.\eqref{eq:udec} must be a unitary matrix.
Whether there exists a difference between implementable unitary operations in the deterministic implementation and in the probabilistic implementation or not is an open problem.

\chapter{Summary and Discussions of Part II}
\section{Summary}%summary of results
We have investigated implementability of $k$-qubit unitary operations over the $(k,N)$-cluster networks where inputs and outputs of quantum computation are given in all separated nodes and quantum communication between nodes is restricted to sending just one-qubit while classical communication is freely allowed. We consider a one-shot scenario where we can use a given cluster network only once and exact implementation without error is required.
We have presented a method to obtain quantum circuit representations of unitary operations implementable over a given cluster network.   For the $(k,N)$-cluster networks with $k=2,3$, we have shown that our method provides all implementable unitary operations over the cluster network. As a first step to find the fundamental primitive networks of network coding for quantum settings, we have shown that both of the butterfly and grail networks are sufficient resources for implementing arbitrary two-qubit unitary operations, meanwhile the $(2,2)$-cluster network is not sufficient to implement arbitrary two-qubit unitary operations even probabilistically. 
\section{Discussions}
\begin{itemize}
\item {\bf Asymptotic implementability of unitary operations over the cluster networks}

In addition to the one-shot scenario, which we have focused on in this part, an {\it asymptotic scenario} is often considered in information science.
For example, a capacity of a channel is usually given by the optimal transmission rate via the channel when asymptotically many uses of the channel is allowed.
In a similar way, implementable unitary operations with asymptotically many uses of a given cluster network can be considered.
However, there is a remarkable gap between the one-shot and a multiple-shot scenario including the asymptotic scenario. For example, the set of unitary operations implementable over the square network ($(2,2)$-cluster network) is the set of the unitary operators whose Kraus-Cirac number is smaller than or equal to $2$ in the one-shot scenario. Meanwhile, it is easy to verify that arbitrary two-qubit unitary operations are implementable over the square network if the network can be used twice. The network coding scheme for this case is presented in Fig.~\ref{fig:superactivation}. That is, the SWAP operation, corresponding to quantum communication from $i_1$ and $i_2$ to $o_2$ and $o_1$ respectively, is always implementable with double uses of the square network.

\begin{figure}
\begin{center}
  \includegraphics[height=.25\textheight]{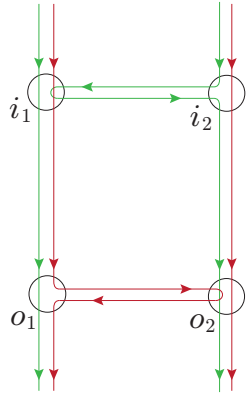}
  \end{center}
  \caption{ {\bf Quantum computation with double uses of the square network.} In this case, any two-qubit unitary operations are implementable over the square network. The green paths and the red paths represent the first pair of qubits and the second pair of qubits, respectively. By performing a given two-qubit unitary operation at nodes $i_1$ on the first pair of qubits and $o_2$ on the second pair of qubits, any two-qubit unitary operation can be applied for two sets of initial states.
}  
\label{fig:superactivation}
\end{figure}

\item {\bf Implementability over a generalized cluster network.}

Our results can be applied to a {\it generalized cluster network} defined as follows.
\begin{definition}
A  network $G=\{\mathcal{V},\mathcal{E},\mathcal{I},\mathcal{O} \}$ is a generalized cluster network if and only if for some $k\geq 1$ and $N\geq 1$, 
\begin{eqnarray}
\mathcal{V}&=&\{v_{i,j};\,1\leq i\leq k,1\leq j\leq N\}\nonumber\\
\mathcal{I}&=&\{v_{i,1};\,1\leq i\leq k\}\nonumber\\
\mathcal{O}&=&\{v_{i,N};\,1\leq i\leq k\}\nonumber\\
\mathcal{E}&=&\mathcal{S}_{sub} \cup \mathcal{K}
\end{eqnarray}
where
\begin{eqnarray}
\mathcal{S}_{sub}&\subseteq&\mathcal{S}_{comp}, \nonumber\\
\mathcal{S}_{comp}&=& \{(v_{m,j},v_{n,j});\,1\leq m<n\leq k,1\leq j\leq N)\}, \nonumber\\
\mathcal{K}&=& \{(v_{i,j},v_{i,j+1});\,1\leq i\leq k,1\leq j\leq N-1)\}.
\end{eqnarray}
\end{definition}
%どこがgeneralizedされたか説明
In this case, if there exists a loop of vertical edges $\mathcal{L}\subseteq \mathcal{S}_{sub}$ such that for some $j$, $L$ and $\{i_m\}_{m=1}^L$,
\begin{equation}
\mathcal{L}=\{e_1=(v_{{i_1},j},v_{{i_2},j}),e_2=(v_{{i_2},j},v_{{i_3},j}),\cdots,e_L=(v_{{i_L},j},v_{{i_1},j})|e_m\neq e_n \,\,{\rm if}\,\, m\neq n\},
\end{equation} 
we can perform a cyclic permutation that transmits a qubit state from $v_{i_1,j}$ to $v_{i_2,j}$, from $v_{i_2,j}$ to $v_{i_3,j}$ and so on (by consuming EPR states corresponding to the looped vertical edges for teleportation) in addition to performing controlled unitary operations presented in Section 4.1.
Thus, quantum computation over a generalized cluster network with a loop of vertical edges has more capability than that without a loop.
Note that the implementable unitary operations over a generalized cluster network are still restricted by Theorem \ref{theorem:det} and \ref{theorem:prob}.

\item {\bf Quantum computation and MBQC over the butterfly network.}

In MBQC, a graph that has a {\it generalized flow} (gflow) is extensively studied since a unitary operation is always implementable over a graph state corresponding to such a graph irrespective of the angle $\alpha$ of each projective measurement defined in Eq.\eqref{eq:projMBQC} with appropriate measurement corrections \cite{gflow1, gflow2}.
\begin{definition}
$(g,\prec)$ is a gflow of a graph $(G,I,O)$, where $g:O^c\rightarrow 2^{I^c}$ and $\prec$ is a strict partial order over $V$, if and only if
\begin{itemize}
\item if $j\in g(i)$ then $i\prec j$
\item if $j\in Odd(g(i))$ then $j=i$ or $i\prec j$
\item $i\in Odd(g(i))$,
\end{itemize}
where $X^c=V\setminus X$,  $Odd(K)=\{u\in V,|N_G(u)\cup K|=1\,\,{\rm mod}\,\, 2\}$ and $N_G(u)$ is the set of vertices neighboring $u$.
\end{definition}
In \cite{gflow3}, it was shown that a graph has a gflow if and only if the graph has a {\it focused gflow} defined as follows.
\begin{definition}
$(g,\prec)$ is a focused gflow of a graph $(G,I,O)$, where $g:O^c\rightarrow 2^{I^c}$ and $\prec$ is a strict partial order over $V$, if and only if
\begin{itemize}
\item if $j\in g(i)$ then $i\prec j$
\item for all $u\in O^c$, $Odd(g(u))\cap O^c=\{u\}$.
\end{itemize}
\end{definition}
The butterfly network does not have a focused gflow as shown below while any two-qubit unitary operations are implementable over the butterfly network in our scenario. 
\begin{proposition}
The butterfly network does not have a focused gflow.
\end{proposition}
\begin{proof}
The butterfly network consists of $I=\{i_1,i_2\}$, $O=\{o_1,o_2\}$, $V=\{n_1,n_2\}\cup I\cup O$ as presented in Fig.~\ref{fig:qgrail}.
$g(n_1)$ does not include $n_1$ since $\prec$ is a strict partial order. If $g(n_1)$ includes $o_1$ or $o_2$, $i_1\in Odd(g(n_1))$ or $i_2\in Odd(g(n_1))$. This contradicts the condition of the focused gflow. Thus, $g(n_1)=\{n_2\}$, i.e. $n_1\prec n_2$.
On the other hand, $\{n_1,o_1\}\subseteq g(n_2)$ or $n_1\notin g(n_2)\wedge o_1\notin g(n_2)$ since $i_1\notin Odd(g(n_2))$. Similarly, $\{n_1,o_2\}\subseteq g(n_2)$ or $n_1\notin g(n_2)\wedge o_2\notin g(n_2)$ since $i_2\notin Odd(g(n_2))$. Thus, $g(n_2)=\{n_1,o_1,o_2\}$ since $n_2\in Odd(g(n_2))$. This is a contradiction to $n_1\prec n_2$.
\end{proof}
This implies that if we restrict the angle of a set of projective measurements in MBQC over the butterfly network, we can implement two-qubit unitary operations. It is unknown how to characterize an implementable unitary operations over a given network in MBQC where the angle of each projective measurement can be restricted. Our results give an upper bound of such an implementable unitary operations over the cluster network.

\end{itemize}

\part{Role of entanglement and causal relation in DQC}

Note that in this part, we mainly consider a bipartite scenario. Thus, we abbreviate bipartite (one-way or two-way) LOCC and bipartite SEP as LOCC and SEP, respectively.
An extension to a multipartite scenario is given in discussions of this part.

\chapter{Resources for state discrimination}
\label{chap:spacetime}
State discrimination is a task to discriminate states chosen from a given set of states $\{\ket{\psi_i}\}_i$ by performing a measurement.
If the set contains non-orthogonal states, it is trivially impossible to discriminate them perfectly.
However, even if the set contains only orthogonal states, it can be impossible to discriminate them perfectly when the state is shared between two parties and the measurement they can perform is somewhat restricted to be local. State discrimination under the restriction of operations to be local is called local state discrimination.
Local state discrimination can be used for characterizing the non-locality of the set of states.
It is shown that any two orthogonal pure states can be distinguished by using LOCC followed by local measurements \cite{Walgate} even if the states are maximally entangled states.
However, several sets of orthogonal product states cannot be distinguished by using LOCC followed by local measurements \cite{9state, DiVincezo, JNiset}. Since any sets of orthogonal product states can be distinguished by using operations in SEP followed by local measurements, such results demonstrate the gap between LOCC and SEP.
We give an example of local state discrimination tasks that cannot be achieved by LOCC but be achieved by SEP as follows.

{\bf Example:} In~\cite{9state}, Bennett et al. have shown that a set of product states of two three-dimensional systems shared between Alice and Bob defined by
\begin{align}
&\bigr|\psi_{1\left(2\right)}\bigr\rangle_{AB} =\bigr|0\bigr\rangle_{A}\bigr|0\pm1\bigr\rangle_{B},\,\,\,\,\,\,\,\,\,
\bigr|\psi_{3\left(4\right)}\bigr\rangle_{AB} =\bigr|0\pm1\bigr\rangle_{A}\bigr|2\bigr\rangle_{B}\nonumber \\
&\bigr|\psi_{5\left(6\right)}\bigr\rangle_{AB} =\bigr|2\bigr\rangle_{A}\bigr|1\pm2\bigr\rangle_{B},\,\,\,\,\,\,\,\,\,
\bigr|\psi_{7\left(8\right)}\bigr\rangle_{AB} =\bigr|1\pm2\bigr\rangle_{A}\bigr|0\bigr\rangle_{B},\,\,\,\,\,\,\,\,\,
\bigr|\psi_{9}\bigr\rangle_{AB} =\bigr|1\bigr\rangle_{A}\bigr|1\bigr\rangle_{B}\label{eq:nine_states},
\end{align}
 where $\{ \ket{j} \}_{j=0}^2$ is the computational basis of a three-level system,  $\ket{a\pm b} :=(\ket{a}\pm\ket{b})/\sqrt{2}$, indices $A$ and $B$ represents Alice's share of the state and Bob's share of the state respectively, is not  deterministically and perfectly distinguishable by any protocol described by a map belonging to LOCC followed by local measurements.

In this chapter, we focus on a set of orthonormal basis states and analyze the amount of entanglement required for discriminating the states by using one-way LOCC and two-way LOCC. We assume classical communication of one-way LOCC is from Alice to Bob. Note that in both one-way LOCC and two-way LOCC, the discrimination task is achievable for any set of orthonormal basis states if sufficient amount of entanglement is shared between Alice and Bob since Alice can teleport her share of the state to Bob and then Bob can perform joint measurement to distinguish the state on the received state and his share of the state.
As a result, we show that the amount of entanglement required for two-way LOCC is less than the lower bound for one-way LOCC in terms of a certain entanglement measure called the Schmidt rank.
That is, entanglement and classical communication required for local state discrimination by using one-way LOCC can be substituted by less entanglement and more rounds of classical communication corresponding to two-way LOCC.
We consider an orthonormal basis on a bipartite system $\mathcal{H}_{A}\otimes\mathcal{H}_{B}$, and entanglement resource $\bigr|\Phi\bigr\rangle$ on the ancilla system $\mathcal{H}_{A'}\otimes\mathcal{H}_{B'}$ in this chapter. We assume that Alice and Bob hold system represented by Hilbert spaces $\mathcal{H}_{AA'}:=\mathcal{H}_{A}\otimes\mathcal{H}_{A'}$ and $\mathcal{H}_{BB'}:=\mathcal{H}_{B}\otimes\mathcal{H}_{B'}$, respectively.

\section{Entanglement resource for one-way LOCC}
In this section, we focus on characterizing the amount of entanglement
resource to discriminate a state from a set of orthonormal basis states by one-way LOCC in terms
of the Schmidt rank. Note that we do not assume the orthonormal basis
as a product basis; thus, all the results in this section are valid
for the cases of all possibly entangled orthonormal basis. Although we assume
that one-way LOCC means one-way LOCC starting from Alice
($\mathcal{H}_{AA'}$) to Bob ($\mathcal{H}_{BB'}$) in this subsection, all the results
can be easily extend to one-way LOCC from Bob to Alice. 

For an orthonormal basis $\left\{ \bigr|\psi_{j}\bigr\rangle_{AB}\right\} _{j=1}^{d_{A}d_{B}}$
on $\mathcal{H}_{A}\otimes\mathcal{H}_{B}$, we define $r_{min}\left(\left\{ \bigr|\psi_{j}\bigr\rangle_{AB}\right\} _{j=1}^{d_{A}d_{B}}\right)$
as the minimum of the Schmidt rank of a state $\bigr|\Phi\bigr\rangle_{A'B'}\in\mathcal{H}_{A'}\otimes\mathcal{H}_{B'}$
such that a set $\left\{ \bigr|\psi_{j}\bigr\rangle_{AB}\otimes\bigr|\Phi\bigr\rangle_{A'B'}\right\} _{j=1}^{d_{A}d_{B}}\subset\mathcal{H}_{AA'}\otimes\mathcal{H}_{BB'}$
can be perfectly discriminated by one-way LOCC, namely,
\begin{align}
 & r_{min}\left(\left\{ \bigr|\psi_{j}\bigr\rangle_{AB}\right\} _{j=1}^{d_{A}d_{B}}\right) %\nonumber \\
:=  \min_{\bigr|\Phi\bigr\rangle}\Big\{ {\rm Sch}\#_{B'}^{A'}\left(\bigr|\Phi\bigr\rangle\right)\ \Big|\ \exists\mathcal{H}_{A'}\otimes\mathcal{H}_{B'},\ \mbox{s.t. } \bigr|\Phi\bigr\rangle\in\mathcal{H}_{A'}\otimes\mathcal{H}_{B'},\nonumber \\
 & \qquad \mbox{and }  \left\{ \bigr|\psi_{j}\bigr\rangle_{AB}\otimes\bigr|\Phi\bigr\rangle_{A'B'}\right\} _{j=1}^{d_{A}d_{B}}\ \mbox{is one-way LOCC distinguishable }\Big\},\label{eq:def_r_min}
\end{align}
where ${\rm Sch}\#_{B'}^{A'}\left(\bigr|\Phi\bigr\rangle\right)$ is the Schmidt
rank of $\bigr|\Phi\bigr\rangle$. Thus, $r_{min}\left(\left\{ \bigr|\psi_{j}\bigr\rangle_{AB}\right\} _{j=1}^{d_{A}d_{B}}\right)$
is the optimal entanglement resource to discriminate $\left\{ \bigr|\psi_{j}\bigr\rangle_{AB}\right\} _{j=1}^{d_{A}d_{B}}$
by one-way LOCC in terms of the Schmidt rank. 

Our main result in this section is stated as the following theorem which completely
characterizes $r_{min}\left(\left\{ \bigr|\psi_{j}\bigr\rangle_{AB}\right\} _{j=1}^{d_{A}d_{B}}\right)$
 and a proof is shown in Appendix B.1.

\begin{theorem}\label{theorem local discrimination} 
For any orthonormal basis $\left\{ \bigr|\psi_{j}\bigr\rangle_{AB}\right\} _{j=1}^{d_{A}d_{B}}\subset\mathcal{H}_{A}\otimes\mathcal{H}_{B}$,
\[
r_{min}\left(\left\{ \bigr|\psi_{j}\bigr\rangle_{AB}\right\} _{j=1}^{d_{A}d_{B}}\right)=d_{min}\left(\left\{ \bigr|\psi_{j}\bigr\rangle_{AB}\right\} _{j=1}^{d_{A}d_{B}}\right),
\]
 where $d_{min}\left(\left\{ \bigr|\psi_{j}\bigr\rangle_{AB}\right\} _{j=1}^{d_{A}d_{B}}\right)$
is defined by
\[
d_{min}\left(\left\{ \bigr|\psi_{j}\bigr\rangle_{AB}\right\} _{j=1}^{d_{A}d_{B}}\right):=\min_{\mathcal{H}_{A}=\bigoplus_{k}\mathcal{M}_{k}}\max_{k}\left\{ \dim\mathcal{M}_{k}\ \Big|\ \forall j,\exists k,\ \mbox{s.t.}\ \bigr|\psi_{j}\bigr\rangle_{AB}\in\mathcal{M}_{k}\otimes\mathcal{H}_{B}\right\} .
\]
\end{theorem}

\section{Entanglement resource for two-way LOCC}
In this section, we focus on the case of orthogonal product bases and study the amount
of entanglement necessary to discriminate one of the basis states by LOCC. In particular
as an example, we consider the basis given by the nine states defined by Eq. \eqref{eq:nine_states}
and their generalization. We study the entanglement
resource necessary to discriminate a state from the nine states by LOCC. 

For the nine states $\left\{ \bigr|\psi_{j}\bigr\rangle_{AB}\right\} _{j=1}^{9}$
defined by Eq. \eqref{eq:nine_states}, we can easily see that there
is no non-trivial subspace $\mathcal{M}\subset\mathcal{H}_{A}$ satisfying
for all $j$, either $\bigl|\psi_{j}\bigr\rangle\in\mathcal{M}\otimes\mathcal{H}_{B}$
or $\bigl|\psi_{j}\bigr\rangle\in\mathcal{M}^{\perp}\otimes\mathcal{H}_{B}$,
where $\mathcal{M}^{\perp}$ is the orthogonal complement of $\mathcal{M}$.
Thus, from Theorem \ref{theorem local discrimination} , the optimal entanglement
resource $\bigl|\Phi\bigr\rangle$ necessary to discriminate the nine
states by one-way LOCC satisfies ${\rm Sch}\#_{B'}^{A'}(\ket{\Phi})=3$.
On the other hand, we construct a two-way LOCC protocol by which
the nine states can be discriminated by consuming an entanglement resource with ${\rm Sch}\#_{B'}^{A'}(\ket{\Phi})=2$. The protocol
can be described as follows: First, we extend the dimension of $\mathcal{H}_{B}$
from $3$ to $d$ by adding additional states $\bigl|i\bigr\rangle_{B}$
for $3\le i\le d$. Then, apply a global unitary $V_{d}$ given by
\begin{eqnarray}
V_{d}&:=&\bigl|0\bigr\rangle\bigl\langle0\bigr|_{A}\otimes\left(\bigl|3\bigr\rangle\bigl\langle1\bigr|+\bigl|1\bigr\rangle\bigl\langle3\bigr|+\bigl|0\bigr\rangle\bigl\langle0\bigr|+\bigl|2\bigr\rangle\bigl\langle2\bigr|+\bigl|4\bigr\rangle\bigl\langle4\bigr|+\cdots+\bigl|d\bigr\rangle\bigl\langle d\bigr|\right)_{B}\nonumber\\
&&+\left(\bigl|1\bigr\rangle\bigl\langle1\bigr|+\bigl|2\bigr\rangle\bigl\langle2\bigr|\right)_{A}\otimes I_{B}.
\label{eq:ent9state}
\end{eqnarray}
As a result, the nine states $\left\{ \bigr|\psi_{j}\bigr\rangle_{AB}\right\} _{j=1}^{9}$
are transformed into %the following states $\left\{ \bigr|\psi'_{j}\bigr\rangle_{AB}\right\} _{j=1}^{9}$:
\begin{align}
\{ &\bigr|\psi'_{1\left(2\right)}\bigr\rangle_{AB} =\bigr|0\bigr\rangle_{A}\bigr|0\pm3\bigr\rangle_{B},\,\,\,\,\,\,\,\,\,
\bigr|\psi'_{3\left(4\right)}\bigr\rangle_{AB}  =\bigr|0\pm1\bigr\rangle_{A}\bigr|2\bigr\rangle_{B}, \nonumber \\
&\bigr|\psi'_{5\left(6\right)}\bigr\rangle_{AB}  =\bigr|2\bigr\rangle_{A}\bigr|1\pm2\bigr\rangle_{B},\,\,\,\,\,\,\,\,\,
\bigr|\psi'_{7\left(8\right)}\bigr\rangle_{AB}  =\bigr|1\pm2\bigr\rangle_{A}\bigr|0\bigr\rangle_{B},\,\,\,\,\,\,\,\,\,
\bigr|\psi'_{9}\bigr\rangle_{AB}  =\bigr|1\bigr\rangle_{A}\bigr|1\bigr\rangle_{B} \}.
\label{eq:nine_states-1}
\end{align}
These states are distinguishable by LOCC. We use a graphical representation of the states in Fig.~\ref{fig:nine_states-1} in order to show a LOCC protocol given in Fig.~\ref{fig:nine_states-d}. Note that we show how the protocol works when $\ket{\psi'_9}=\bigr|1\bigr\rangle_{A}\bigr|1\bigr\rangle_{B}$ is given. When another state is given, the flow of the protocol will be changed since Alice and Bob can change measurements depending on their past measurement outcomes.
\begin{figure}[htbp]
 \begin{center}
  \includegraphics[width=80mm]{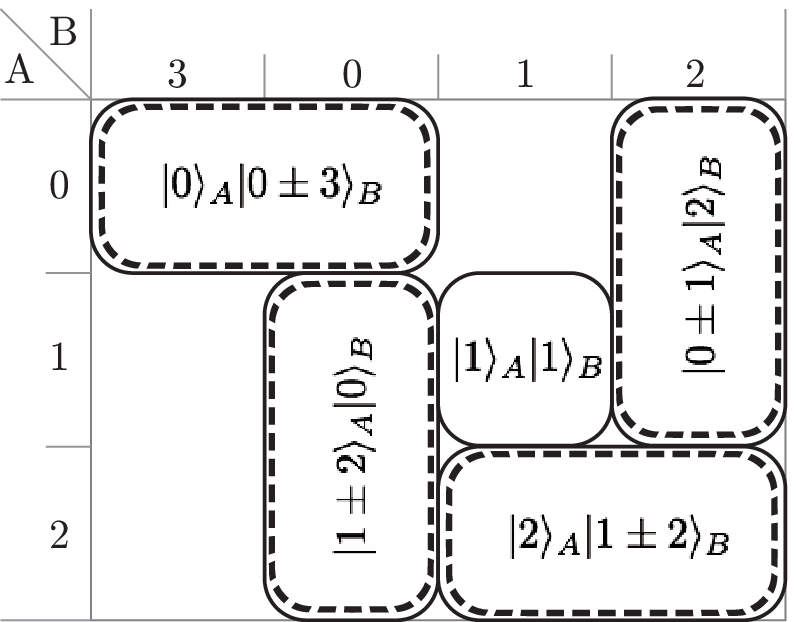}
 \end{center}
 \caption{{\bf A graphical representation of states.} After applying $V_d$, the nine states are transformed into the states represented in this figure.}
 \label{fig:nine_states-1}
\end{figure}
\begin{figure}[htbp]
 \begin{center}
  \includegraphics[width=120mm]{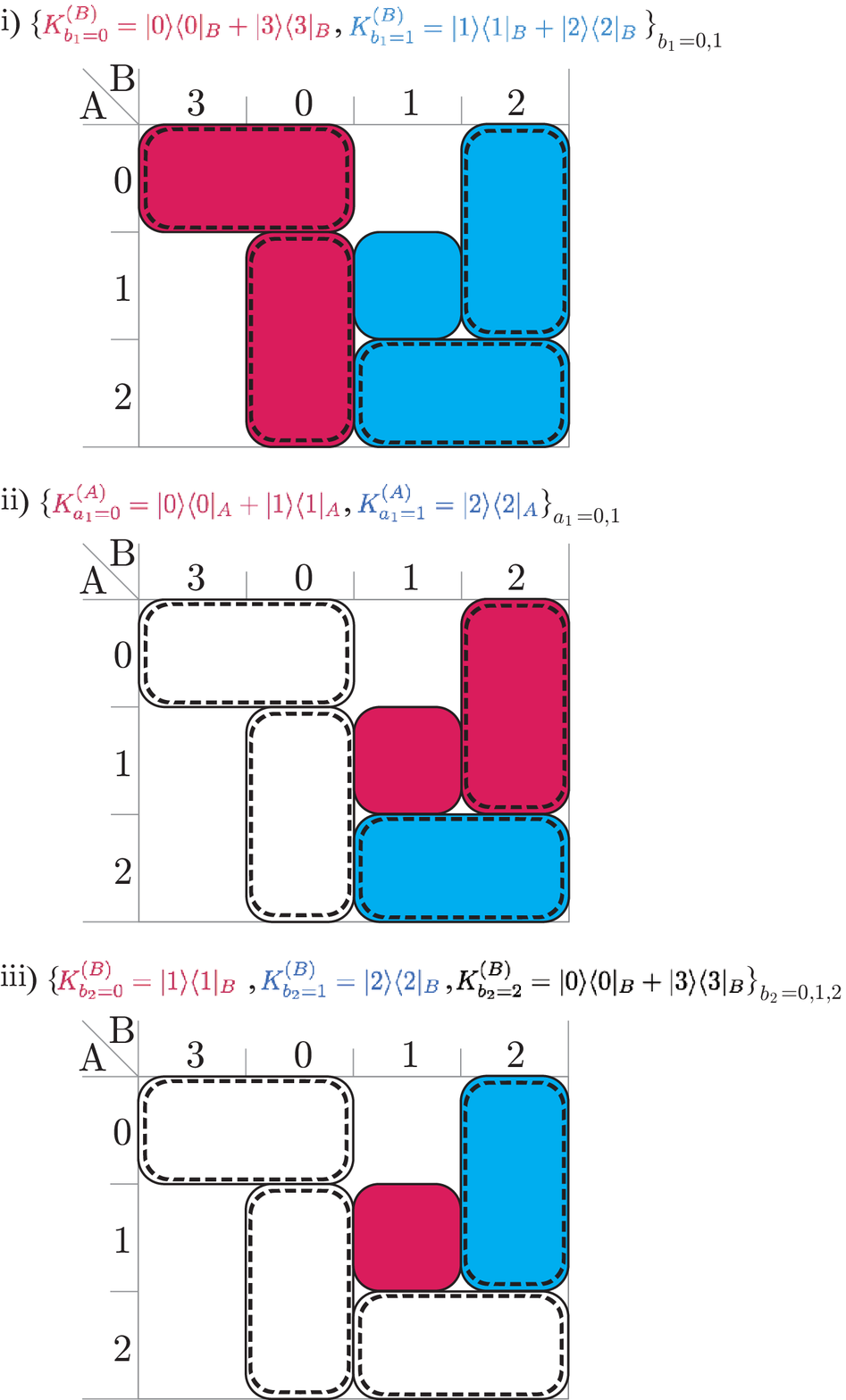}
 \end{center}
 \caption{ {\bf A LOCC protocol for discrimination.} We see how the protocol works when $\ket{\psi'_9}=\bigr|1\bigr\rangle_{A}\bigr|1\bigr\rangle_{B}$ is given. i) First, Bob performs a measurement of which Kraus operators are given by $\{K^{(B)}_{b_1=0}=\ket{0}\bra{0}_B+\ket{3}\bra{3}_B,K^{(B)}_{b_1=1}=\ket{1}\bra{1}_B+\ket{2}\bra{2}_B\}_{b_1=0,1}$ and obtains a measurement outcome $b_1=1$. ii) Second, Alice performs a measurement of which Kraus operators are given by $\{K^{(A)}_{a_1=0}=\ket{0}\bra{0}_A+\ket{1}\bra{1}_A,K^{(A)}_{a_1=1}=\ket{2}\bra{2}_A\}_{a_1=0,1}$ and obtains a measurement outcome $a_1=0$. iii) Third, Bob performs a measurement of which Kraus operators are given by $\{K^{(B)}_{b_2=0}=\ket{1}\bra{1}_B,K^{(B)}_{b_2=1}=\ket{2}\bra{2}_B, K^{(B)}_{b_2=2}=\ket{0}\bra{0}_B+\ket{3}\bra{3}_B\}_{b_2=0,1,2}$, obtains a measurement outcome $b_2=0$ and they know $\ket{\psi'_9}$ is given.}
 \label{fig:nine_states-d}
\end{figure}

$V_{d}$ can be implemented with an EPR state since the operator Schmidt rank of $V_{d}$ is 2 \cite{Cohen}. 

The above discussion can be generalized to a general orthogonal
product basis $\left\{ \bigr|\psi_{j}\bigr\rangle_{AB}\right\} _{j=1}^{d_{A}d_{B}}$ defined by
\begin{align}
\bigl|\psi_{1\left(2\right)}\bigr\rangle & =\frac{1}{\sqrt{2}}\bigl|1\bigr\rangle\otimes(\bigl|1\bigr\rangle\pm\bigl|2\bigr\rangle),\nonumber \\
\bigl|\psi_{3\left(4\right)}\bigr\rangle & =\frac{1}{\sqrt{2}}(\bigl|d_A-1\bigr\rangle\pm\bigl|d_A\bigr\rangle)\otimes\bigl|1\bigr\rangle,\nonumber \\
\bigl|\psi_{m_{1}+5}\bigr\rangle & =\sum_{k=2}^{d_{A}-2}e^{i\frac{2\pi}{d_{A}-3}m_{1}\cdot\left(k-2\right)}\bigl|k\bigr\rangle\otimes\bigl|1\bigr\rangle, \nonumber \\
\bigl|\psi_{d_{A}+m_{2}+2}\bigr\rangle & =\sum_{k=2}^{d_{A}-1}e^{i\frac{2\pi}{d_{A}-2}m_{2}\cdot\left(k-2\right)}\bigl|k\bigr\rangle\otimes\bigl|2\bigr\rangle, \nonumber \\
\bigl|\psi_{2d_{A}+m_{3}}\bigr\rangle & =\bigl|d_A\bigr\rangle\otimes\sum_{k=2}^{d_{B}}e^{i\frac{2\pi}{d_{B}-1}m_{3}\cdot\left(k-2\right)}\bigl|k\bigr\rangle, \nonumber \\
\bigl|\psi_{2d_{A}+d_{B}-1+(d_{A}-1)m_{5}+m_{4}}\bigr\rangle & =\sum_{k=1}^{d_A-1}e^{i\frac{2\pi}{d_A-1}m_{4}\cdot\left(k-1\right)}\bigl|k\bigr\rangle\otimes\bigl|m_{5}+3\bigr\rangle \label{eq generalized nine states}
\end{align}
where $m_i$ for $i=1,2,3,4,5$ satisfies $0\leq m_{1}\le d_{A}-4$, $0\leq m_{2}\le d_{A}-3$, $0\leq m_{3}\le d_{B}-2$,
$0\leq m_{4}\le d_{A}-2$, and $0\leq m_{5}\le d_{B}-3$. 
 A graphical representation of the orthogonal product basis is given in Fig.~\ref{fig:gninestate}.
\begin{figure}[htbp]
 \begin{center}
  \includegraphics[width=90mm]{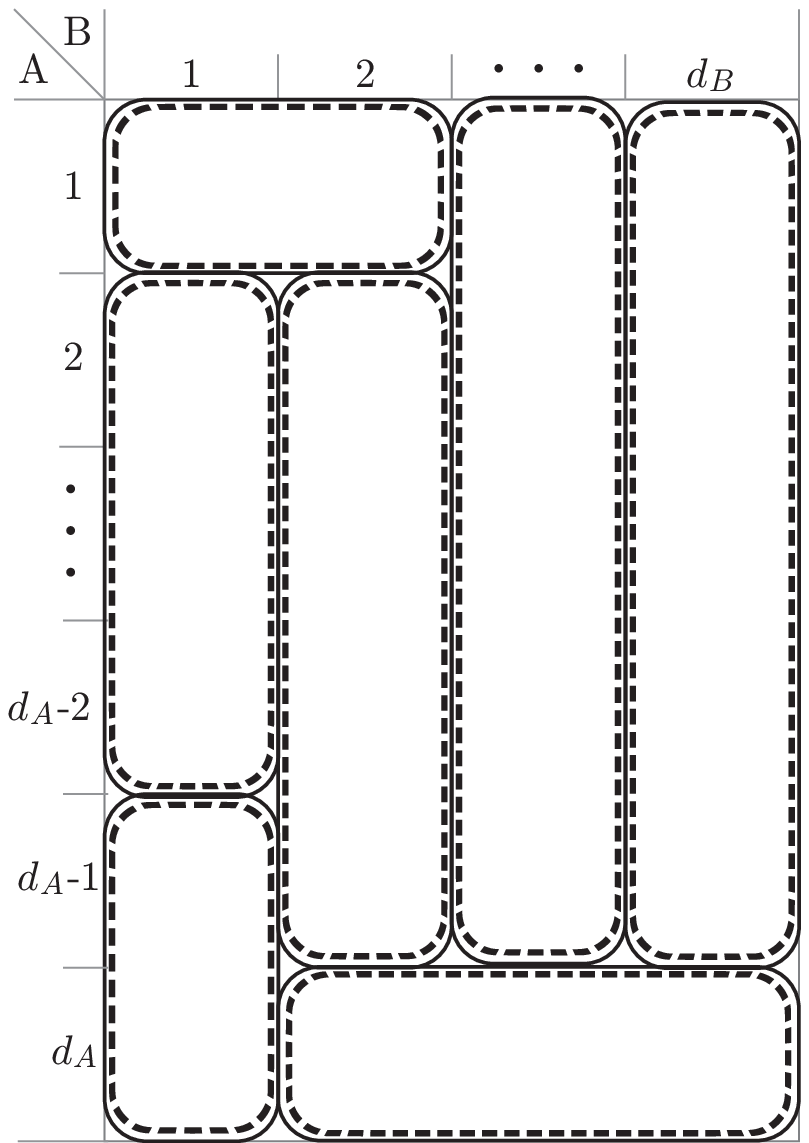}
 \end{center}
 \caption{{\bf A graphical representation of a general orthogonal product basis.}}
 \label{fig:gninestate}
\end{figure}
We consider local state discrimination of $\left\{ \bigr|\psi_{j}\bigr\rangle_{AB}\right\} _{j=1}^{d_{A}d_{B}}$.
Theorem \ref{theorem local discrimination} guarantees that in order
to discriminate a state from this basis states by one-way LOCC starting from Alice to Bob an entanglement resource with the Schmidt rank
$d_{A}$ is necessary. Similarly, to discriminate a state from the same basis states by one-way LOCC
starting from Bob to Alice an entanglement resources
with the Schmidt rank $d_{B}$ is necessary. On the other hand, by
means of the similar discussion as the nine states, it is enough to
use an entanglement resource with the Schmidt rank 2 in case of two-way LOCC. 
Therefore, in the limit of large $d_A$, there is infinite gap between one-way LOCC and two-way LOCC in terms of entanglement resource in terms of the Schmidt rank to distinguish the orthogonal product basis defined by Eq. (\ref{eq generalized nine states}).

\chapter{Resources for SEP}
\label{chap:LOCC*}
In this chapter, we show that a causal relation between the classical outputs and classical inputs of the local operations without predefined partial order, which we call ``classical communication" without predefined causal order (CC*), characterizing a special class of deterministic quantum operations, separable operations.
In Section 7.1, we extend LOCC into a joint quantum operation between two parties implemented without using shared entanglement but with local operations connected by CC*. 
We name a new class of deterministic joint quantum operations obtained by this extension but still within quantum mechanics by LOCC*.
We show that LOCC* with CC* respecting partial order reduces LOCC. 
In Section 7.2, we show that LOCC* can be represented by the $\infty$-shaped loop in Fig.~\ref{fig:loopCC} and is equivalent to SEP \cite{VGheorghiu}, which has been introduced for mathematical simplicity to analyze nonlocal quantum tasks in place of LOCC.  
In Section 7.3, by considering the correspondence between LOCC* and a probabilistic version of LOCC called stochastic LOCC (SLOCC), we analyze the power of CC* in terms of enhancing the success probability of probabilistic operations in SLOCC. 
In Section 7.4, we also investigate the relationship between LOCC* and the quantum process formalism for joint quantum operations without partial order developed in \cite{OFC}.
In Section 7.5, by using our framework, we give an example of the quantum operation which is an element of SEP but not an element of LOCC, i.e., to perform such an operation within special relativistic spacetime, entanglement and classical communication are necessary.

\begin{figure}[htbp]
 \begin{center}
  \includegraphics[width=60mm]{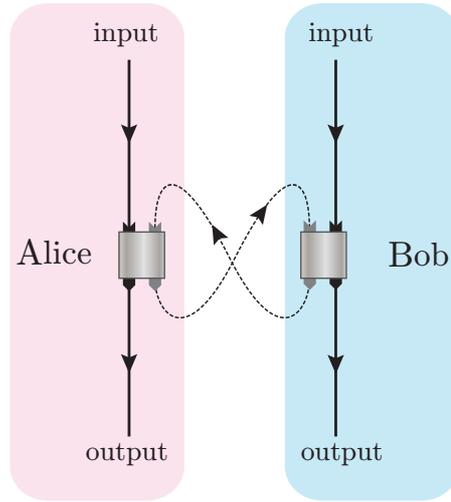}
 \end{center}
 \caption{{\bf A deterministic joint quantum operation consisting of local operations connected by the $\infty$-shaped loop.} Any element of LOCC* can be represented by this joint quantum operation and vice versa as shown in Results.}
 \label{fig:loopCC}
\end{figure}

\section{LOCC*}
We consider a deterministic joint quantum operation (a CPTP map) taking only quantum input and output denoted by $\mathcal{M}$ implemented by two parties, Alice and Bob,  who are connected only by classical communication.  We set Alice performs the first operation.   As a simplest case, we consider to link a classical output  $o$ of Alice's local operation represented by $\{ \mathcal{A}_o \}_o$ and a classical input $i$ of Bob's local operation represented by $\mathcal{B}_{|i}$.   Since Alice and Bob are acting on different quantum systems at different spacetime coordinates,  the joint quantum operation is described by a tensor product of two local operations.  Ability of classical communication indicates that the spacetime coordinate of Alice's local operation and that of Bob's local operation are timelike separated, and Bob's coordinate is in the future cone of Alice's coordinate.   Linking the classical output of Alice $o$ and the classical input of Bob $i$ means that they are perfectly correlated, namely $i=o$ for all $o$.   By taking averages over $o$ and $i$,  we obtain a deterministic joint quantum operation given by 
\begin{equation}
\mathcal{M}=\sum_{o,i}\delta_{i,o} \mathcal{A}_o\otimes\mathcal{B}_{|i}=\sum_{o} \mathcal{A}_o\otimes\mathcal{B}_{|o},
\label{eq:LOCC1}
\end{equation}
where  $\delta_{o,i}$ denotes the Kronecker delta.  This is the simplest case of LOCC called one-way LOCC. 

A deterministic joint operation represented by more general  finite-round LOCC between two parties is defined by  connecting a sequence of Alice's local operations given by $\{ \mathcal{A}^{(N)}_{o_N|i_N}\circ\cdots\circ \mathcal{A}^{(1)}_{o_1} \}_{o_N,\cdots,o_1}$ and another sequence of Bob's local operations  given by $\{\mathcal{B}^{(N)}_{o'_N|i'_N}\circ\cdots\circ \mathcal{B}^{(1)}_{o'_1|i'_1} \}_{o'_N,\cdots,o'_1}$.   Here $\circ$ denotes a connection between two local operations at timelike separated coordinates linking a quantum output of a local operation and a quantum input of the next local operation of the {\it same} party and introduces a total order for each party's local operations.  The indices $i_k$ and $i'_k$ are classical inputs of the $k$-th operations and $o_k$ and $o'_k$ are classical outputs of the $k$-th operations of Alice and Bob, respectively.   In LOCC, not only local operations within the parties are totally ordered, but also the spacetime coordinates of all local operations of Alice and Bob are totally ordered alternately.   Then $o_1$ of Alice's local operation and $i'_1$ of Bob's local operation are linked by classical communication and setting $i'_1=o_1$, and similarly, $o'_1$ of Bob's local operation and $i_{2}$ of Alice's local operation are linked by setting $i_{2}= o'_1$ and so on.
Thus,  finite-round LOCC between the two parties is defined by a set of the joint quantum operation represented by

\begin{eqnarray}
\mathcal{M}&=&\sum_{i_1,\cdots,i'_N,o_1,\cdots, o'_N} 
\delta_{o_N,i'_N} \cdots \delta_{o_1,i'_1} \mathcal{A}^{(N)}_{o_N|i_N}\circ\cdots\circ \mathcal{A}^{(1)}_{o_1|i_1}\otimes \mathcal{B}^{(N)}_{o'_N|i'_N}\circ\cdots\circ \mathcal{B}^{(1)}_{o'_1|i'_1}  \nonumber\\
\label{eq:LOCC2} \\
&=& \sum_{i_1,\cdots,i_N,o_1,\cdots,o_N, o'_N} \mathcal{A}^{(N)}_{o_N|i_N}\circ\cdots\circ \mathcal{A}^{(1)}_{o_1|i_1}\otimes\mathcal{B}^{(N)}_{o_N'|o_N}\circ\cdots\circ \mathcal{B}^{(1)}_{i_2|o_1}
\end{eqnarray}
 where we define $\mathcal{A}^{(1)}_{o_1} := \sum_{i_1} \mathcal{A}^{(1)}_{o_1|i_1}$. An example for $N=3$ is shown in Fig.\ref{fig:twoCC}.  

Now we develop a framework to investigate local operations and classical communication between the parties without the assumption of the existence of  predefined ordering of the spacetime coordinates. 
That is, we keeps the local spacetime coordinates  and the totally ordered structure of local operations within each party but do not assign the {\it global} spacetime coordinate across the parties, and allow connecting any classical inputs and outputs of local operations between different parties as long as the resulting deterministic joint quantum operation is in quantum mechanics. This relaxation allows a generalization of classical communication to a conditional probability distribution $p(i_1,\cdots,i'_N|o_1,\cdots,o'_N)$ linking the classical outputs and the classical inputs of local  operations.
The  corresponding joint quantum operation is represented by
\begin{equation}
\mathcal{M}=\sum_{i_1,\cdots,i'_N,o_1,\cdots,o'_N} p(i_1,\cdots,i'_N|o_1,\cdots,o'_N)\mathcal{A}^{(N)}_{o_N|i_N}\circ\cdots\circ \mathcal{A}^{(1)}_{o_1|i_1}\otimes \mathcal{B}^{(N)}_{o'_N|i'_N}\circ\cdots\circ \mathcal{B}^{(1)}_{o'_1|i'_1} .
\label{eq:LOCC*0}
\end{equation}
This generalization does not guarantee the joint quantum operation represented by Eq.\eqref{eq:LOCC*0} to be TP  whereas its CP property is preserved.   Since we  investigate deterministic joint quantum operations, we require $p(i_1,\cdots,i'_N|o_1,\cdots,o'_N)$ to keep the form of Eq.\eqref{eq:LOCC*0} to represent a CPTP maps.   We call a set of CPTP maps in the form of  Eq.\eqref{eq:LOCC*0}  with $p(i_1,\cdots,i'_N|o_1,\cdots,o'_N)$  as LOCC*.  

For one-way LOCC,  such a generalization corresponds to replacing the delta function $\delta_{o,i}$ in Eq.\eqref{eq:LOCC1} by a conditional probability distribution $ p(i|o) $. This is equivalent to replacing a perfect classical channel by a general noisy classical  channel. However, it is not the case for multiple-round LOCC.  A deterministic joint quantum operation is LOCC if and only if it can be decomposed in the form of Eq.(\ref{eq:LOCC*0}) with $p(i_1,\cdots,i'_N|o_1,\cdots,o'_N)$ respecting {\it causal order} of the classical inputs and outputs of local operations imposed by special relativistic no-signaling conditions. $p(i_1,\cdots,i'_N|o_1,\cdots,o'_N)$ respecting causal order can be regarded as classical communication respecting causal order.
The definition of the causal order of $p(i_1,\cdots,i'_N|o_1,\cdots,o'_N)$ and relationship between LOCC and causal order are given in Appendix B.2.

It is easy to check that without loss of generality, any CPTP map in LOCC* can be implemented by just one local operation performed by each party connected by a conditional probability distribution, since by letting $o_A:=(o_1,\cdots,o_N)$, $i_A:=(i_1,\cdots,i_N)$, $o_B:=(o_1',\cdots,o_N')$ and $i_B:=(i_1',\cdots,i_N')$, we can regard the sequence of Alice's local operations $\{ \mathcal{A}^{(N)}_{o_N|i_N}\circ\cdots\circ \mathcal{A}^{(1)}_{o_1|i_1} \}_{o_N,\cdots,o_1}$ as one quantum instrument $\{\mathcal{A}_{o_A|i_A}\}_{o_A}$ conditioned by the classical input $i_A$, the same with Bob's local operations and $p(i_A,i_B|o_A,o_B)$ is still a conditional probability distribution.   Thus LOCC* is simply defined by a set of deterministic joint quantum operations (= CPTP maps) $\mathcal{M}$ given in the form of
\begin{eqnarray}
\mathcal{M}=\sum_{i_A,i_B,o_A,o_B}p(i_A,i_B|o_A,o_B)\mathcal{A}_{o_A|i_A}\otimes\mathcal{B}_{o_B|i_B},
\label{eq:LOCC*}
\end{eqnarray}
where $p(i_A,i_B|o_A,o_B)$ is a  conditional probability distribution satisfying
\begin{eqnarray}
p(i_A,i_B|o_A,o_B)\geq 0 \mathrm{~~~~~and~~~~~} \sum_{i_A,i_B}p(i_A,i_B|o_A,o_B)=1,
\end{eqnarray}
where $\{ \mathcal{A}_{o_A|i_A} \}_{o_A}$ is Alice's local operation with a classical input $i_A$ and a classical output $o_A$, and $\{ \mathcal{B}_{o_B|i_B} \}_{o_B}$ is Bob's local operation with a classical input $i_B$ and a classical output $o_B$ .
We call $p(i_A,i_B|o_A,o_B)$ in Eq.~\eqref{eq:LOCC*}  linking the classical outputs and classical inputs as CC*, {\it ``classical communication'' without predefined causal order}, with respect to local operations $\{\{\mathcal{A}_{o_A|i_A}\}_{o_A},\{\mathcal{B}_{o_B|i_B}\}_{o_B}\}$.   Note that we can also ``collapse'' the sequential local operations in  multi-round LOCC to represent a joint quantum operation in the form of Eq.~\eqref{eq:LOCC*}.  However in this case if we regard the combined  operations $\{\mathcal{A}_{o_A|i_A}\}_{o_A}$ and $\{\mathcal{B}_{o_B|i_B}\}_{o_B}$ as operations localized in the spacetime of LOCC, the  corresponding $p(i_A,i_B|o_A,o_B)$ cannot be interpreted as classical communication respecting causal order.

Further, we show that LOCC* can be always represented by a $\infty$-shaped  loop shown in Fig.\ref{fig:loopCC}.   It is easy to see that if $p(i_A,i_B|o_A,o_B)=\delta_{i_A,o_B}\delta_{i_B,o_A}$, LOCC* defined by Eq.(\ref{eq:LOCC*})  reduces to 
\begin{eqnarray}
\mathcal{M}=\sum_{a,b}\mathcal{A}_{a|b}\otimes\mathcal{B}_{b|a}.
\label{eq:LOSC}
\end{eqnarray}
In Appendix B.4, we show  the converse, all elements of LOCC* can be decomposed into this form, is also true.   From this form of LOCC*,  it is possible to interpret that CC* in LOCC* can be {\it looped}, namely, Alice's classical input is Bob's classical output and Bob's classical input is Alice's classical output.

\begin{table}
\begin{center}
i) $K_{a|b}^{(A)}$:~~
\renewcommand
\arraystretch{1.1}
\begin{tabular}{|c|p{22mm}|p{22mm}|p{22mm}|}\hline
\backslashbox[1cm]{a\rule[-0mm]{0pt}{0mm}}{b} & 1 & 2 & 3 \\ \hline
1 & $\ket{1}_{A'}\bra{0}_A$ & $\ket{2}_{A'}\bra{0}_A$ & $\ket{3}_{A'}\bra{0+1}_A$ \\ \hline
2 & $\ket{8}_{A'}\bra{1-2}_A$ & $\ket{9}_{A'}\bra{1}_A$ & $\ket{4}_{A'}\bra{0-1}_A$ \\ \hline
3 & $\ket{7}_{A'}\bra{1+2}_A$ & $\ket{6}_{A'}\bra{2}_A$ & $\ket{5}_{A'}\bra{2}_A$ \\ \hline
\end{tabular}
\vspace{0.5cm}
\\ii)$ K_{b|a}^{(B)}$:~~
\begin{tabular}{|c|p{22mm}|p{22mm}|p{22mm}|}\hline
\backslashbox{a}{b} & 1 & 2 & 3 \\ \hline
1 & $\ket{1}_{B'}\bra{0+1}_B$ & $\ket{2}_{B'}\bra{0-1}_B$ & $\ket{3}_{B'}\bra{2}_B$ \\ \hline
2 & $\ket{8}_{B'}\bra{0}_B$ & $\ket{9}_{B'}\bra{1}_B$ & $\ket{4}_{B'}\bra{2}_B$ \\ \hline
3 & $\ket{7}_{B'}\bra{0}_B$ & $\ket{6}_{B'}\bra{1-2}_B$ & $\ket{5}_{B'}\bra{1+2}_B$ \\ \hline
\end{tabular}
\caption{ Tables of the Kraus operators $ K_{a|b}^{(A)}$ of Alice's local operations $\{ \mathcal{A}_{a|b} \}_a$ in i) and $K_{b|a}^{(B)}$ of Bob's local operation $\{\mathcal{B}_{b|a} \}_b$ in  ii).   In the Kraus operator representation, the deterministic joint quantum operation $\mathcal{M}=\sum_{a,b}\mathcal{A}_{a|b}\otimes\mathcal{B}_{b|a}$ transforms any quantum input $\rho_{AB}$ on systems $A$ and $B$ into a quantum output on systems $A'$ and $B'$ as ${\rho'} _{A' B'}=\mathcal{M}(\rho_{AB}) =   \sum_{a,b} K_{a|b}^{(A)} \otimes K_{b|a}^{(B)} \rho_{AB} (K_{a|b}^{(A)}  \otimes K_{b|a}^{(B)})^\dagger$ where $\sum_a (K_{a|b}^{(A)})^\dagger K_{a|b}^{(A)} = \mathbb{I}_A$ for any $b$ and $\sum_b (K_{b|a}^{(B)})^\dagger K_{b|a}^{(B)} = \mathbb{I}_B$ for $a$ with identity operators $\mathbb{I}_A$ and $\mathbb{I}_B$ on system $A$ and $B$, respectively.}
\label{table:9state} 
\end{center}
\end{table}

\section{LOCC* and SEP}

We show that LOCC* provides a new characterization of a set of deterministic joint quantum operations corresponding to SEP by proving that LOCC* is equivalent to SEP.  We start with investigating inclusion relations between LOCC and LOCC*.   It is is easy to verify that LOCC is a subset of LOCC*  by definition of LOCC*.   We show that LOCC* is strictly larger than LOCC by constructing the following  example based on the nine-state discrimination \cite{9state}.

  We show that the the nine-state  {\it can} be deterministically  and perfectly distinguishable by a map in LOCC* followed by local measurements by presenting constructions of Alice's local operation $\{\mathcal{A}_{a|b} \}_a$  and Bob's local operation $\{ \mathcal{B}_{b|a} \}_b$  in the form of Eq.(\ref{eq:LOSC}).   The constructions of $\{\mathcal{A}_{a|b} \}_a$ and  $\{ \mathcal{B}_{b|a} \}_b$ in the Kraus operator representations are given in Table~\ref{table:9state}  i) and ii), respectively.   It is easy to check that for any $\ket{\psi_k }\in \{  \ket{\psi_i} \}_{i=1}^9$, the corresponding $\mathcal{M}$  in LOCC* with the constructions of the local operations transforms
\begin{eqnarray}
\ket{\psi_k}\bra{\psi_k}  \rightarrow 
\sum_{a,b}\mathcal{A}_{a|b}\otimes\mathcal{B}_{b|a} (\ket{\psi_k}\bra{\psi_k})
= \ket{k}_{A'}\bra{k} \otimes \ket{k}_{B'} \bra{k}
\end{eqnarray}
where indices $A'$ and $B'$ denote nine-dimensional output systems for Alice and Bob.    Once Alice and Bob obtain the output state $\ket{k}_{A'}\bra{k} \otimes \ket{k}_{B'} \bra{k}$, they can find out the classical output $k$ by individually performing projective measurements in the basis given by $\{ \ket{j} \}_{j=1}^{9}$. Note that the nine states can be probabilistically distinguished without error by a stochastic LOCC (SLOCC) protocol, which indicates LOCC* is closely related to SLOCC as it will be shown later.

We can further show that LOCC* is equivalent to a well known class of CPTP maps called {\it separable map} (SEP) as summarized in Fig.~\ref{fig:classes}.  SEP is a set of maps representing deterministic joint quantum operations $\mathcal{M}$ that can be written by
\begin{eqnarray}
 \mathcal{M}=\sum_{k}\mathcal{E}^A_{k}\otimes\mathcal{E}^B_{k},
\label{eq:SEP}
\end{eqnarray}
where all the elements of local operations without classical inputs  $\mathcal{E}^A_{k}$ and $\mathcal{E}^B_{k}$ are completely positive.   The set of nine states can be deterministically and perfectly distinguished by using a map in SEP of which Kraus operator representation is given by $\{ \ket{k}_{A'}\otimes \ket{k}_{B'}  \bra{\psi_k}_{AB} \}$ followed by local projective measurements in the basis given by $\{ \ket{k} \}_{k=1}^{9}$.  SEP  is equivalent to a set of CPTP maps that cannot transform any separable states into entangled states \cite{VVedral, MBPlenio}. Therefore, SEP does not have a power to create entanglement between two parties if the quantum input is not entangled. This can be interpreted as a characterization of operational effects of SEP since the characterization is based on the effect of SEP.
 The class SEP includes the class LOCC \cite{CLMOW12}.      Due to the mathematical simplicity of its structure, the class SEP is often used for proving a quantum task to be not implementable by LOCC protocols by showing that the task is not implementable even by using a stronger class of operations, SEP.  However, the gap between LOCC and SEP has not been clear. 

It is easy to see that LOCC* is a subset of SEP from the form of Eq.\eqref{eq:LOCC*},  since $\mathcal{A}_{o_A|i_A}$ and $\mathcal{B}_{o_B|i_B}$ are also completely positive and  $p(i_A,i_B|o_A,o_B)$  is not negative, so we can always transform a map in LOCC* into the form of \eqref{eq:SEP}.   To prove that SEP is a subset of LOCC*,  we use the fact that an element of SEP is implementable by SLOCC with a {\it constant}  success probability \cite{CLMOW12}.  In Appendix B.5, we show that it is possible to reduce the failure probability to be zero in a LOCC* protocol and there always exists a map in LOCC* corresponding to a map in SEP.  Note that the LOCC* map obtained by Appendix B.5 is different from the simpler LOCC* map given in Table~\ref{table:9state} in the case of nine-state discrimination. 
 LOCC* gives a new characterization of SEP in terms of operational resources in the sense that SEP can be interpreted as a set of operations consisting local operations and CC*. Note that in LOCC*, two assumptions of local operations, (a) they are partially ordered and (b) the choice of a local operation does not depend on resources connecting the local operation, are relaxed and the local operations are connected by CC*.

%Since there exist elements in SEP that are {\it not} implementable by LOCC \cite{9state, JNiset, DiVincezo, EChitambar}, shared entanglement is necessary to implement operations in SEP but not included in LOCC with local operations and (normal) classical communication {\it respecting} partial order.  On the other hand, operations in SEP cannot create entanglement between two parties sharing no entanglement similarly to LOCC.  In spite of its clear operational meaning, analysis of quantum information processing tasks under the restriction of LOCC is hard in general since the mathematical structure of LOCC is highly complicated \cite{CLMOW12}. Thus, it is known small amount of examples of DQC in the gap between SEP and LOCC. Owing to the difficulty of LOCC, SEP is used to investigate LOCC in \cite{Akibue, Anthony, DistinguishbySEP}. Giving a good characterization of the gap would improve the bounds. Recently, the gap between SEP and LOCC has been shown even in the case where infinitely many rounds of classical communications are allowed \cite{CLMOW12}, implying that entanglement plays somewhat an essential role to separate SEP from LOCC.   Our result presents an answer for this role involving properties of the spacetime in quantum mechanics.

So far, we have considered the implementation of a deterministic joint quantum operation when no entanglement is shared across the parties.   Now we consider the effect of entanglement shared between the parties.  Since any deterministic joint quantum operations can be implementable by entanglement assisted LOCC, LOCC* can be implementable by entanglement assisted (standard) classical communication respecting causal order.  The results shown in this part suggest that entanglement assisted LOCC implementing SEP can be simulated by LOCC*, where no entanglement is needed.

\begin{figure}
 \centering
  \includegraphics[height=.30\textheight]{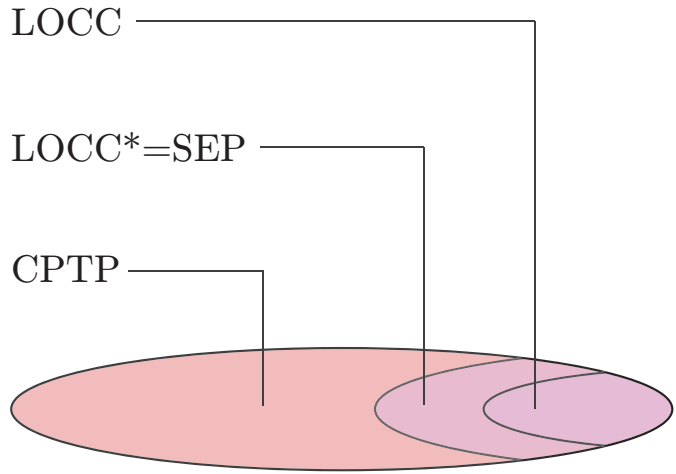}
  \caption{ The inclusion relation between the classes of  deterministic joint quantum operations CPTP, SEP, LOCC*, and LOCC.   LOCC* is equivalent to the set of separable maps (SEP). LOCC is strictly smaller than SEP. LOCC* is strictly smaller than the set of CPTP.}
\label{fig:classes}
\end{figure} 

\section{LOCC* and local post selection}
In this part, we  present the relationship between LOCC* and SLOCC.  SLOCC is a set of (linear) CP maps consisting of local operations and standard classical communication .   In contrast to LOCC, SLOCC contains the CP maps representing the cases where particular measurement outcomes are post-selected.   Since we use a linear map to represent SLOCC, a SLOCC element is not TP but trace decreasing (TD) in general.
We define a class of linear CP maps called SLOCC* that can be simulated by SLOCC.
That is, SLOCC* is the set of linear CP maps from two input systems $X$ and $Y$ to two output system $A$ and $B$ such that 
\begin{equation}
\mathcal{M}=c\Gamma,
\end{equation}
where $\Gamma$ is an element of SLOCC and a positive constant  $c > 0$. Note that SLOCC* contains not only TD maps but also trace increasing maps.  It is easy to see that LOCC* (or SEP) is a subset of SLOCC* since any element of LOCC* (or SEP) is implementable by SLOCC with a constant  success probability independent of inputs.   In the following, we show that a superset of LOCC* where the TP condition is removed from LOCC* is equivalent to SLOCC*.  That is, SLOCC* is equivalent to the set of linear CP maps that can be decomposed into the form of Eq.~\eqref{eq:LOCC*}.
By definition, SLOCC* is equivalent to the set of linear CP maps that can be decomposed into the form of Eq.~\eqref{eq:SEP}, where $k$ corresponds to a measurement outcome.  Without loss of generality, we can restrict $\mathcal{E}_k^A$ and $\mathcal{E}_k^B$ in Eq.~\eqref{eq:SEP} to be CPTD maps since for all linear CP maps $\mathcal{E}^A_k$, there exists natural number $M$ such that  $\frac{1}{M}\mathcal{E}^A_k$ is a CPTD map and 
$\mathcal{M}$ given by Eq.~\eqref{eq:SEP} can be represented by $\sum_{i=1}^M\sum_{k}\tilde{\mathcal{E}^A_{i,k}}\otimes\tilde{\mathcal{E}^B_{i,k}}$, where $\tilde{\mathcal{E}^A_{i,k}}=\frac{1}{M}\mathcal{E}^A_k$ and $\tilde{\mathcal{E}^B_{i,k}}=\mathcal{E}^B_{k}$.
By applying the same technique described in Appendix B.5, any element of SLOCC* can be decomposed into  the form of Eq.~\eqref{eq:LOSC}.   Thus, any SLOCC* element can be decomposed into Eq.~\eqref{eq:LOCC*}.
We summarize the inclusion relation of sets of linear CP maps as shown in Fig.~\ref{fig:classes2}.

Any linear CP maps that can be decomposed into the form of Eq.~\eqref{eq:LOCC*} is simulatable by LOCC with post-selection (SLOCC), and vice versa. A connection between post-selection and the causal structure of the spacetime has been discussed in the context of the probabilistic closed time-like curve (p-CTC) in  \cite{pCTC1, pCTC2, pCTC3}.   Our result indicates the connection between post-selection and CC* without causal order.    In the standard spacetime,  a map in SEP but not in LOCC is implementable by LOCC assisted by entanglement.  Thus we can regard that entanglement accompanied with classical communication is used only for enhancing the success probability of a task achievable without entanglement $p < 1$ to $p=1$ for implementing a map in SEP (=LOCC*,  but not in LOCC).  In this sense,  CC* without causal order can be understood as alternative characterizations of a power of entanglement for enhancing success probability by effectively changing the partial ordering properties of the spacetime.

\begin{figure}
 \centering
  \includegraphics[height=.30\textheight]{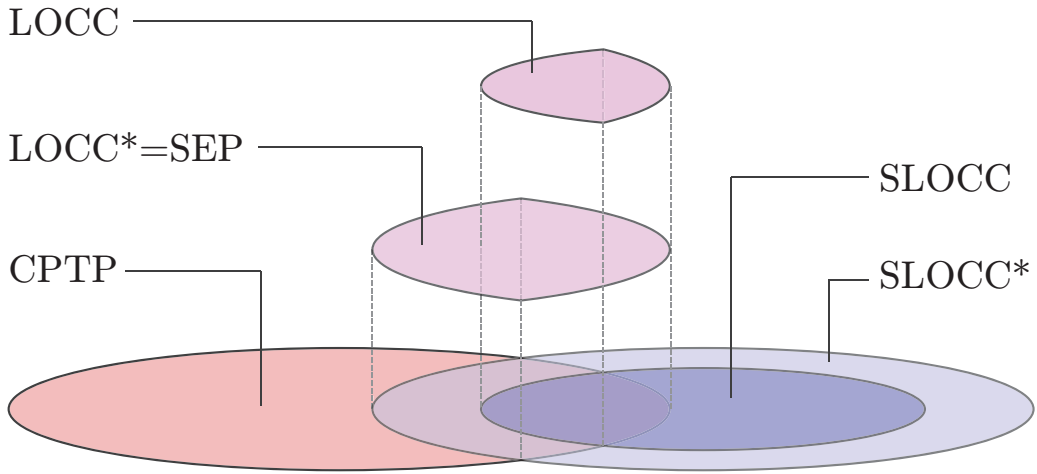}
  \caption{Class of linear CP maps. The intersection of SLOCC and the set of CPTP maps (TP SLOCC elements) is LOCC since any TP SLOCC element can be implemented by a LOCC protocol. The intersection of SLOCC* and the set of CPTP maps (TP SLOCC* elements) is LOCC*.}
\label{fig:classes2}
\end{figure} 

\section{LOCC* and quantum processes}
We discuss correspondences between our formalism of LOCC* and other formalisms of quantum operations without the assumptions of  the partial order of local operations,  {\it higher order formalisms} developed by \cite{OFC, Chiribella2}.   In  the higher order formalism, effects of quantum communication linking local operations of two parties are described as a map that  transforms local operations into a deterministic joint quantum operation.   The map, called  a {\it quantum process},  is a higher order map (supermap) transforming a quantum operation to another quantum operation, whereas a (normal) map  transforms a quantum state to another quantum state.    In \cite{OFC}, the requirements for a quantum process to be consistent with quantum mechanics but without predefined causal order of local operations have been derived.   They have shown that there exists quantum  processes not implementable by {\it quantum} communication linking partially ordered local operations.  Such requirements for a quantum process to be consistent with quantum mechanics can be interpreted as a new kind of causality, which is different from the special relativistic causality but based only on quantum mechanics.  The quantum process shows a new possibility to speed up quantum computers \cite{Qswitchspeed} and implementations of the quantum process are discussed for several settings \cite{Qswitch, Nakago}.

A crucial difference between the higher order formalisms and our formalism of LOCC* is that the local operations are linked by quantum communication in the higher order formalisms,  but we only allow classical communication between the parties.   To compare the two formalisms,  we consider a special type of quantum processes where quantum communication between the parties are restricted to transmitting a probabilistic mixture of ``classical'' states, namely, a set of fixed mutually orthogonal states.  Such a process is called a {\it classical} quantum process (CQP).  In \cite{OFC},  it is shown that a deterministic joint operation described by local operations linked by a classical quantum process does not exhibit the new causality as a (fully) quantum process does. In Appendix B.6, we show that the deterministic joint operations in this case reduce to a probabilistic mixture of two types of operations in one-way LOCC from Alice to Bob and from Bob to Alice.   We denote a set of such deterministic joint quantum operations as LOCQP for comparing to LOCC*.

Since LOCQP is a set of probabilistic mixtures of one-way LOCC, LOCC* is a larger set than LOCQP.  Therefore, CC* used in implementing non-LOCC quantum operations cannot be a classical quantum process linking two local operations.   The gap between LOCC* and LOCQP originates from the conditions on local operations imposed for restricting CC* and a classical quantum process.  CC* in LOCC* is only required to guarantee the joint quantum operation to be deterministic for {\it some} choices of local operations, whereas a classical quantum process in LOCQP is required to guarantee the joint quantum operation to be deterministic for {\it arbitrary} choices of  local operations.  Hence CC* is less restricted than a classical quantum process.   Considering that CC* is simulated by entanglement assisted classical communication, we can conclude that entanglement provides a power to waver restrictions on classical communication linking local operations originated from both the causality in special relativity and the restriction for classical quantum processes when it is accompanied by classical communication.

%{\magenta 集合の式}

\section{LOCQP and SEP}
In bipartite cases, we have shown that LOCQP is equivalent to a set of probabilistic mixtures of one-way LOCC since the classical quantum process is just a probabilistic mixtures of classical communication.
A quantum process is called {\it causally separable} if the quantum process represents a probability mixtures of normal communication channels, otherwise, it is called {\it causally non-separable}.
There exists a classical quantum process $W$ that is causally non-separable in {\it tripartite} cases as shown in Appendix B.8.
This implies that tripartite LOCQP not only contains a simple probability mixture of one-way LOCC but also possibly contains an element of a larger set such as two-way LOCC and even non-LOCC separable operations.
In this section, we give an example of the tripartite non-LOCC separable operations consisting of local operations linked by tripartite LOCQP.
Very few classes of non-LOCC separable operations are known \cite{9state, JNiset, DiVincezo, EChitambar} and the gap between SEP and LOCC has not been clarified so far. LOCQP can be used for a new tool to generate operations in non-LOCC SEP.

In Appendix B.8, we show that CC* defined by
\begin{eqnarray}
p(x,y,z|a,b,c)=\frac{1}{2}
\,\,\,if\,(x,y,z,a,b,c)&=&(0,0,0,0,0,0)\nonumber\\
&=&(0,0,0,1,1,1)\nonumber\\
&=&(0,0,1,0,1,0)\nonumber\\
&=&(0,0,1,1,0,1)\nonumber\\
&=&(0,1,0,0,1,1)\nonumber\\
&=&(0,1,0,1,0,0)\nonumber\\
&=&(0,1,1,0,0,1)\nonumber\\
&=&(0,1,1,1,1,0)\nonumber\\
&=&(1,0,0,0,0,1)\nonumber\\
&=&(1,0,0,1,1,0)\nonumber\\
&=&(1,0,1,0,1,1)\nonumber\\
&=&(1,0,1,1,0,0)\nonumber\\
&=&(1,1,0,0,1,0)\nonumber\\
&=&(1,1,0,1,0,1)\nonumber\\
&=&(1,1,1,0,0,0)\nonumber\\
&=&(1,1,1,1,1,1).\label{eq:nonSEP}
\end{eqnarray}
corresponds to a causally non-separable classical quantum process.
A joint map $\Gamma$ obtained by using this CC* is given by
\begin{equation}
\Gamma=\sum_{x,y,z,a,b,c}p(x,y,z|a,b,c)A_{a|x}\otimes B_{b|y}\otimes C_{c|z},
\end{equation}
where we use the CJ representations of local operations for Alice, Bob and Charlie, for example, and $\{A_{a|x}\}_a$ is the CJ representation of a local quantum operation of Alice in $\mathbf{C}(\mathcal{H}_1:\mathcal{H}_2)$.
We can prove that $\Gamma$ is a separable map that is not in a finite round LOCC with
\begin{eqnarray}
A_{a|x}&=&H^x\ket{a}\bra{a}_1H^x\otimes\ket{xa}\bra{xa}_2\\
B_{b|y}&=&H^y\ket{b}\bra{b}_3H^y\otimes\ket{yb}\bra{yb}_4\\
C_{c|z}&=&H^z\ket{c}\bra{c}_5H^z\otimes\ket{zc}\bra{zc}_6.
\end{eqnarray}
We denote $\ket{000}$, $\ket{101}$, $\ket{+10}$, $\ket{-11}, \ket{100}$, $\ket{001}$, $\ket{-10}$, $\ket{+11}$ as $\ket{\mathbf{0}}$, $\ket{\mathbf{1}}$, $\ket{\mathbf{2}}$, $\ket{\mathbf{3}}$, $\ket{\mathbf{4}}$, $\ket{\mathbf{5}}$, $\ket{\mathbf{6}}$, $\ket{\mathbf{7}}$ and $\ket{x}\bra{x}$ as $[x]$. Then the joint map $\Gamma$ is written as
\begin{eqnarray}
\Gamma&=&\frac{1}{2}\big([\mathbf{000}]+[\mathbf{111}]+[\mathbf{012}]+[\mathbf{103}]+[\mathbf{031}]+[\mathbf{120}]+[\mathbf{023}]+[\mathbf{132}]\nonumber\\
&&+[\mathbf{201}]+[\mathbf{310}]+[\mathbf{213}]+[\mathbf{302}]+[\mathbf{230}]+[\mathbf{321}]+[\mathbf{222}]+[\mathbf{333}]\big).
\end{eqnarray}
We show the existence of a LOCC protocol implies a contradiction.
First, we show the following lemma about a condition of LOCC operators.
\begin{lemma}
If there exist non-zero positive semidefinite operators $A_i\in Pos^*(\mathcal{H}_{1}\otimes\mathcal{H}_{2})$, $B_i\in Pos^*(\mathcal{H}_{3}\otimes\mathcal{H}_{4})$ and $C_i\in Pos^*(\mathcal{H}_{5}\otimes\mathcal{H}_{6})$ such that
\begin{equation}
\sum_iA_i\otimes B_i\otimes C_i =\Gamma,
\end{equation}
where $Pos^*(\mathcal{H})$ is a set of positive operators with at least one non-zero diagonal element, then
\begin{eqnarray}
\forall i,\,\,(A_i, B_i, C_i) \in\{&(\alpha[\mathbf{0}],\beta[\mathbf{0}],\gamma[\mathbf{0}]),&\nonumber\\
&(\alpha[\mathbf{1}],\beta[\mathbf{2}],\gamma[\mathbf{0}]),&\nonumber\\
&(\alpha[\mathbf{2}],\beta[\mathbf{3}],\gamma[\mathbf{0}]),&\nonumber\\
&(\alpha[\mathbf{3}],\beta[\mathbf{1}],\gamma[\mathbf{0}]),&\nonumber\\
&(\alpha[\mathbf{0}],\beta[\mathbf{3}],\gamma[\mathbf{1}]),&\nonumber\\
&(\alpha[\mathbf{1}],\beta[\mathbf{1}],\gamma[\mathbf{1}]),&\nonumber\\
&(\alpha[\mathbf{2}],\beta[\mathbf{0}],\gamma[\mathbf{1}]),&\nonumber\\
&(\alpha[\mathbf{3}],\beta[\mathbf{2}],\gamma[\mathbf{1}]),&\nonumber\\
&(\alpha[\mathbf{0}],\beta[\mathbf{1}],\gamma[\mathbf{2}]),&\nonumber\\
&(\alpha[\mathbf{1}],\beta[\mathbf{3}],\gamma[\mathbf{2}]),&\nonumber\\
&(\alpha[\mathbf{2}],\beta[\mathbf{2}],\gamma[\mathbf{2}]),&\nonumber\\
&(\alpha[\mathbf{3}],\beta[\mathbf{0}],\gamma[\mathbf{2}]),&\nonumber\\
&(\alpha[\mathbf{0}],\beta[\mathbf{2}],\gamma[\mathbf{3}]),&\nonumber\\
&(\alpha[\mathbf{1}],\beta[\mathbf{0}],\gamma[\mathbf{3}]),&\nonumber\\
&(\alpha[\mathbf{2}],\beta[\mathbf{1}],\gamma[\mathbf{3}]),&\nonumber\\
&(\alpha[\mathbf{3}],\beta[\mathbf{3}],\gamma[\mathbf{3}])&|\alpha,\beta,\gamma>0\}.
\label{eq:elements}
\end{eqnarray}
\end{lemma}

\begin{proof}
\begin{equation}
\bra{\mathbf{00}}_{AB}\Gamma\ket{\mathbf{00}}_{AB}=\sum_i\bra{\mathbf{0}}A_i\ket{\mathbf{0}}\bra{\mathbf{0}}B_i\ket{\mathbf{0}}C_i=\frac{1}{2}[\mathbf{0}]_C
\end{equation}
implies
\begin{equation}
\forall i\in\{i|\bra{\mathbf{0}}A_i\ket{\mathbf{0}}\neq0\wedge\bra{\mathbf{0}}B_i\ket{\mathbf{0}}\neq0\},\exists p>0,\,\,C_i=p[\mathbf{0}]_C,
\end{equation}
where $\mathcal{H}_C=\mathcal{H}_5\otimes\mathcal{H}_6$
By a similar calculation, we obtain
\begin{eqnarray}
\forall i\in\{i|\bra{\mathbf{1}}A_i\ket{\mathbf{1}}\neq0\wedge\bra{\mathbf{2}}B_i\ket{\mathbf{2}}\neq0\},\exists p>0,&\,\,C_i=p[\mathbf{0}]_C\nonumber\\
\forall i\in\{i|\bra{\mathbf{2}}A_i\ket{\mathbf{2}}\neq0\wedge\bra{\mathbf{3}}B_i\ket{\mathbf{3}}\neq0\},\exists p>0,&\,\,C_i=p[\mathbf{0}]_C\nonumber\\
\forall i\in\{i|\bra{\mathbf{3}}A_i\ket{\mathbf{3}}\neq0\wedge\bra{\mathbf{1}}B_i\ket{\mathbf{1}}\neq0\},\exists p>0,&\,\,C_i=p[\mathbf{0}]_C\nonumber\\
\forall i\in\{i|\bra{\mathbf{0}}A_i\ket{\mathbf{0}}\neq0\wedge\bra{\mathbf{3}}B_i\ket{\mathbf{3}}\neq0\},\exists p>0,&\,\,C_i=p[\mathbf{1}]_C\nonumber\\
\forall i\in\{i|\bra{\mathbf{1}}A_i\ket{\mathbf{1}}\neq0\wedge\bra{\mathbf{1}}B_i\ket{\mathbf{1}}\neq0\},\exists p>0,&\,\,C_i=p[\mathbf{1}]_C\nonumber\\
\forall i\in\{i|\bra{\mathbf{2}}A_i\ket{\mathbf{2}}\neq0\wedge\bra{\mathbf{0}}B_i\ket{\mathbf{0}}\neq0\},\exists p>0,&\,\,C_i=p[\mathbf{1}]_C\nonumber\\
\forall i\in\{i|\bra{\mathbf{3}}A_i\ket{\mathbf{3}}\neq0\wedge\bra{\mathbf{2}}B_i\ket{\mathbf{2}}\neq0\},\exists p>0,&\,\,C_i=p[\mathbf{1}]_C\nonumber\\
\forall i\in\{i|\bra{\mathbf{0}}A_i\ket{\mathbf{0}}\neq0\wedge\bra{\mathbf{1}}B_i\ket{\mathbf{1}}\neq0\},\exists p>0,&\,\,C_i=p[\mathbf{2}]_C\nonumber\\
\forall i\in\{i|\bra{\mathbf{1}}A_i\ket{\mathbf{1}}\neq0\wedge\bra{\mathbf{3}}B_i\ket{\mathbf{3}}\neq0\},\exists p>0,&\,\,C_i=p[\mathbf{2}]_C\nonumber\\
\forall i\in\{i|\bra{\mathbf{2}}A_i\ket{\mathbf{2}}\neq0\wedge\bra{\mathbf{2}}B_i\ket{\mathbf{2}}\neq0\},\exists p>0,&\,\,C_i=p[\mathbf{2}]_C\nonumber\\
\forall i\in\{i|\bra{\mathbf{3}}A_i\ket{\mathbf{3}}\neq0\wedge\bra{\mathbf{0}}B_i\ket{\mathbf{0}}\neq0\},\exists p>0,&\,\,C_i=p[\mathbf{2}]_C\nonumber\\
\forall i\in\{i|\bra{\mathbf{0}}A_i\ket{\mathbf{0}}\neq0\wedge\bra{\mathbf{2}}B_i\ket{\mathbf{2}}\neq0\},\exists p>0,&\,\,C_i=p[\mathbf{3}]_C\nonumber\\
\forall i\in\{i|\bra{\mathbf{1}}A_i\ket{\mathbf{1}}\neq0\wedge\bra{\mathbf{0}}B_i\ket{\mathbf{0}}\neq0\},\exists p>0,&\,\,C_i=p[\mathbf{3}]_C\nonumber\\
\forall i\in\{i|\bra{\mathbf{2}}A_i\ket{\mathbf{2}}\neq0\wedge\bra{\mathbf{1}}B_i\ket{\mathbf{1}}\neq0\},\exists p>0,&\,\,C_i=p[\mathbf{3}]_C\nonumber\\
\forall i\in\{i|\bra{\mathbf{3}}A_i\ket{\mathbf{3}}\neq0\wedge\bra{\mathbf{3}}B_i\ket{\mathbf{3}}\neq0\},\exists p>0,&\,\,C_i=p[\mathbf{3}]_C.
\end{eqnarray}
Since $A_i$ has at least one non-zero diagonal element and we have
\begin{equation}
\forall i,\,\,\bra{\mathbf{4}}A_i\ket{\mathbf{4}}=\bra{\mathbf{5}}A_i\ket{\mathbf{5}}=\bra{\mathbf{6}}A_i\ket{\mathbf{6}}=\bra{\mathbf{7}}A_i\ket{\mathbf{7}}=0,
\end{equation}
\begin{equation}
\forall i,\,\,\bra{\mathbf{0}}A_i\ket{\mathbf{0}}\neq0\vee\bra{\mathbf{1}}A_i\ket{\mathbf{1}}\neq0\vee\bra{\mathbf{2}}A_i\ket{\mathbf{2}}\neq0\vee\bra{\mathbf{3}}A_i\ket{\mathbf{3}}\neq0.
\end{equation}
The same property also holds for $B_i$.
Thus, 
\begin{equation}
\forall i,\,\,C_i\in\{p[\mathbf{0}]_C,p[\mathbf{1}]_C,p[\mathbf{2}]_C,p[\mathbf{3}]_C|p>0\}.
\end{equation}
The same property also holds for $A_i$ and $B_i$:
\begin{eqnarray}
\forall i,\,\,A_i\in\{p[\mathbf{0}]_A,p[\mathbf{1}]_A,p[\mathbf{2}]_A,p[\mathbf{3}]_A|p>0\}\nonumber\\
\forall i,\,\,B_i\in\{p[\mathbf{0}]_B,p[\mathbf{1}]_B,p[\mathbf{2}]_B,p[\mathbf{3}]_B|p>0\},\\
\end{eqnarray}
where $\mathcal{H}_A=\mathcal{H}_1\otimes\mathcal{H}_2$ and $\mathcal{H}_B=\mathcal{H}_3\otimes\mathcal{H}_4$.
\end{proof}

Next, we prove that the map $\Gamma$ cannot be implemented by an operation in finite-round LOCC.
\begin{proof}
If there exists an element in finite-round LOCC implementing $\Gamma$, there exists {\it accumulated operator} $\{A_{\mathbf{m}}\in \mathbf{Pos}(\mathcal{H}_1\otimes\mathcal{H}_2), B_{\mathbf{m}}\in \mathbf{Pos}(\mathcal{H}_3\otimes\mathcal{H}_4), C_{\mathbf{m}}\in \mathbf{Pos}(\mathcal{H}_5\otimes\mathcal{H}_6)\}$ which represents LOCC, where $\mathbf{m}=(m_1,m_2,\cdots,m_N)$ represents all the measurement outcomes and $m_i$ represents the $i$-th measurement outcome.
By using Lemma 3, for all $\mathbf{m}$ such that $A_{\mathbf{m}}\otimes B_{\mathbf{m}}\otimes C_{\mathbf{m}}\neq \mathbf{0}$, $(A_{\mathbf{m}}, B_{\mathbf{m}},C_{\mathbf{m}})$ is an element of Eq.\eqref{eq:elements}.
Thus, we obtain that a set of all the measurement outcomes $\mathbf{M}$ consists of eleven subsets corresponding to the measurement outcomes,
\begin{eqnarray}
\mathbf{M}&=&\mathbf{M}_{null}\cup \left(\bigcup_{c\in \mathbf{C}}\mathbf{M}_c\right),\\
\mathbf{C}&=&\{000,111,+++,---,01+,1+0,+01,0+-,\nonumber\\&&+-0,-0+,0-1,-10,10-,1-+,-+1,+1-\},
\end{eqnarray}
such that
\begin{eqnarray}
\forall \mathbf{m}\in\mathbf{M}_{null},\,\,{\rm tr}_{2,4,6}[A_{\mathbf{m}}\otimes B_{\mathbf{m}}\otimes C_{\mathbf{m}}]&=&0,\nonumber\\
\forall c\in\mathbf{C}, \forall \mathbf{m}\in\mathbf{M}_{c},\exists \alpha>0\,\,{\rm tr}_{2,4,6}[A_{\mathbf{m}}\otimes B_{\mathbf{m}}\otimes C_{\mathbf{m}}]&=&\alpha[c].
\end{eqnarray}

Since $\Gamma$ is invariant under the permutation of parties, without loss of generality, we can assume that Alice performs the last measurement and obtain the measurement result denoted by $m_N$. Since Bob and Charlie's accumulated operators do not depend on $m_N$, i.e., we have
\begin{eqnarray}
\forall m_1,\cdots,\forall m_N,\forall m'_N,&B_{(m_1,\cdots,m_{N-1},m_N)}=B_{(m_1,\cdots,m_{N-1},m'_N)},\\
\forall m_1,\cdots,\forall m_N,\forall m'_N,&C_{(m_1,\cdots,m_{N-1},m_N)}=C_{(m_1,\cdots,m_{N-1},m'_N)},
\end{eqnarray}
if $(m_1,\cdots,m_{N-1},m_N)\in\mathbf{M}_{c}$, for all $m'_N$
\begin{equation}
(m_1,\cdots,m_{N-1},m'_N)\in\mathbf{M}_{c}\cup\mathbf{M}_{null}.
\end{equation}
This property holds for the other subsets except $\mathbf{M}_{null}$.
Letting
\begin{eqnarray}
\mathbf{m}'&=&(m_1,\cdots,m_{N-1}),\nonumber\\
A_{\mathbf{m}'}&=&\sum_{m_N}A_{\mathbf{m}}\\
B_{\mathbf{m}'}&=&B_{(m_1,\cdots,m_{N-1},0)},\nonumber\\
C_{\mathbf{m}'}&=&C_{(m_1,\cdots,m_{N-1},0)},\\
\end{eqnarray}
we obtain
\begin{eqnarray}
\forall \mathbf{m}'\in\mathbf{M}'_{null},\,\,{\rm tr}_{2,4,6}[A_{\mathbf{m'}}\otimes B_{\mathbf{m'}}\otimes C_{\mathbf{m'}}]&=&0\nonumber\\
\forall c\in\mathbf{C}, \forall \mathbf{m'}\in\mathbf{M'}_{c},\exists \alpha>0\,\,{\rm tr}_{2,4,6}[A_{\mathbf{m'}}\otimes B_{\mathbf{m'}}\otimes C_{\mathbf{m'}}]&=&\alpha[c],
\end{eqnarray}
where $\mathbf{M}'_{null}$ and $\mathbf{M'}_{c}$ are subsets of measurement outcomes $(m_1,\cdots,m_{N-1})$.

Since ${\rm tr}_2[A_{\mathbf{m}'}]={\rm tr}_2[\sum_{m_N}A_{\mathbf{m}}]$ does not depend on $m_{N-1}$ if the $(N-1)$-th measurement is performed by Bob or Charlie, we can repeat this procedure until summing up all the measurement outcomes. Then we obtain
\begin{equation}
{\rm tr}_{2,4,6}\left[\sum_{\mathbf{m}}A_{\mathbf{m}}\otimes B_{\mathbf{m}}\otimes C_{\mathbf{m}}\right]=0.
\end{equation}
This is a contradiction to the property of the accumulated operator of which trace of the sum is identity.

\end{proof}

\chapter{Summary and Discussions of Part III}
\section{Summary}
We have studied the roles of quantum communication connecting local operations in DQC, in particular a class of quantum operations implemented in DQC called SEP. Since quantum communication can be implemented by entanglement assisted classical communication, we analyze entanglement and classical communication required for DQC.

First, we have derived the necessary and sufficient amount of the entanglement resource for a DQC task to perfectly discriminate a state from a given orthonormal basis states by one-way LOCC in terms of an entanglement measure, the Schmidt rank of a required entanglement resource. We have constructed a two-way LOCC protocol which consumes less amount of entanglement resources than the optimal one-way LOCC protocol. 
%Therefore, the entanglement and classical communication required for the discrimination of orthonormal product basis states in one-way LOCC can be substituted by much less entanglement and more rounds of classical communication.

Second, we have developed a framework to describe deterministic joint quantum operations in DQC. In the framework, we have introduced a causal relation between the classical outputs and classical inputs of the local operations without predefined causal order but still within quantum mechanics, called ``classical communication" without predefined causal order (CC*). We have shown that local operations and CC* (LOCC*) with partial order can be simulated by local operations and normal classical communication. On the other hand, CC* without partial order, which can be interpreted as ``classical communication" without causal order, is necessary to implement non-LOCC SEP.
Since we show that LOCC* is equivalent to SEP, entanglement and classical communication required for implementing SEP by entanglement assisted LOCC can be substituted by CC* without using an entanglement resource. 

Third, we have investigated the power of CC* in terms of SLOCC and quantum processes.
We have shown that the super class of LOCC* is simulatable by SLOCC and vice versa, and LOCC* is simulatable by SLOCC with a constant success probability independent of an input state and vice versa.
Then, we have investigated the relationship between LOCC* and local operations connected by a classical quantum process (LOCQP).
We have shown that LOCQP is equivalent to a set of probabilistic mixtures of one-way LOCC in bipartite cases.
On the other hand, we have shown that LOCQP is not included in LOCC in tripartite cases by constructing an example of operations in non-LOCC SEP by using LOCQP.

\section{Discussion}
\begin{itemize}
\item {\bf The amount of entanglement required for local state discrimination}

We have analyzed the difference in the amount of entanglement required for local state discriminiation in one-round LOCC and multi-round LOCC in terms of the Schmidt rank.
However, the minimum amount of required entanglement in terms of entanglement entropy, widely used in quantum information science \cite{entropyent}, is still unknown.
Even the minimum entanglement entropy required for implementing $V_d$ defined in Eq.\eqref{eq:ent9state} is unknown.
In the following, we analyze a lower bound of the entanglement entropy necessary to implement $V_d$.
We define the Schmidt strength $H\left(U\right)$ of a unitary operator $U$ on $\mathcal{H}_{A}\otimes\mathcal{H}_{B}$ as
the Shannon entropy of the square of the operator Schmidt coefficients of $U$.
The Schmidt strength $H\left(V_{d}\right)$ gives a lower bound of the entanglement entropy necessary to implement $V_{d}$ by LOCC \cite{Wakakuwa}. 

Let an operator Schmidt decomposition of $V_d$ be
\begin{equation}
V_d=\sum_k\lambda_k A_k\otimes B_k,
\end{equation}
where $\{A_k\in\mathbf{L}(\mathcal{H}_A)\}_k$ and $\{B_k\in\mathbf{L}(\mathcal{H}_B)\}_k$ are sets of orthonormal operators. The Schmidt strength $H(V_d)$ is defined by
\begin{equation}
H(V_d)=\sum_k -\lambda_k^2 \log_2 (\lambda_k^2).
\end{equation}
In order to apply the operator Schmidt decomposition to $V_d$, we define $V_d'\in\mathbf{L}(\mathcal{H}_B\otimes\mathcal{H}_B:\mathcal{H}_A\otimes\mathcal{H}_A)$ by
\begin{equation}
V_d'=\sum_{i,j}(\bra{i}_B V_d\ket{j}_A)\otimes \ket{j}_A\bra{i}_B,
\end{equation}
where $\{\ket{i}_X\}_i\,\,(X=A,B)$ is the computational basis. By using the operator Schmidt decomposition of $V_d$, $V_d'$ can be decomposed into
\begin{equation}
V_d'=\sum_k\lambda_k \sqrt{\dim(\mathcal{H}_A)\dim(\mathcal{H}_B)}\ket{A_k}_{AA}\bra{B_k}_{BB},
\end{equation}
where $\ket{A_k}=\frac{1}{\sqrt{\dim(\mathcal{H}_A)}}\sum_i(A_k\ket{i}_A)\ket{i}_A$, $\bra{B_k}=\frac{1}{\sqrt{\dim(\mathcal{H}_B)}}\sum_i(\bra{i}_B B_k)\bra{i}_B$. Note that since $\{A_k\}_k$ and $\{B_k\}_k$ are sets of orthonormal operators, i.e. ${\rm tr}(A_k^{\dag}A_l)=\dim(\mathcal{H}_A)\delta_{k,l}$ and ${\rm tr}(B_k^{\dag}B_l)=\dim(\mathcal{H}_B)\delta_{k,l}$, $\{\ket{A_k}\}_k$ and $\{\ket{B_k}\}_k$ are sets of orthonormal states, i.e.
\begin{eqnarray}
\bra{A_k}A_l\rangle&=&\delta_{k,l},\\
\bra{B_k}B_l\rangle&=&\delta_{k,l}.
\end{eqnarray}
The operator Schmidt coefficients $\{\lambda_k\}_k$ is obtained by calculating the eigenvalues of 
\begin{equation}
V_d'V_d^{\prime\dag}=\sum_k\lambda_k^2 \dim(\mathcal{H}_A)\dim(\mathcal{H}_B)\ket{A_k}_{AA}\bra{A_k}_{AA}.
\end{equation}
By straightforward calculation, we obtain
\begin{equation}
V_d'V_d^{\prime\dag}=\begin{pmatrix}d+1&d-1&d-1\\d-1&d+1&d+1\\d-1&d+1&d+1\end{pmatrix},
\end{equation}
where we use the matrix representation in terms of the computational basis. Thus, the operator Schmidt coefficients are given by
\begin{equation}
\{\lambda_k^2\}_{k=0,1}=\left\{\frac{3(d+1)-\sqrt{9d^2-14d+9}}{6(d+1)},\frac{3(d+1)+\sqrt{9d^2-14d+9}}{6(d+1)}\right\}.
\end{equation}
$H\left(V_{d}\right)$ decreases when $d$ increases as shown in Fig.~\ref{fig:SS}.
Since when $d$ goes to infinity, $\lambda_0\rightarrow 0$, $\lambda_1\rightarrow 1$ and $-\lambda_k^2 \log_2 (\lambda_k^2)\rightarrow 0$ for $k=0,1$, $H(V_d)$ decreases to $0$ when $d\rightarrow\infty$.

\begin{figure}
\begin{center}
  \includegraphics[height=.4\textheight]{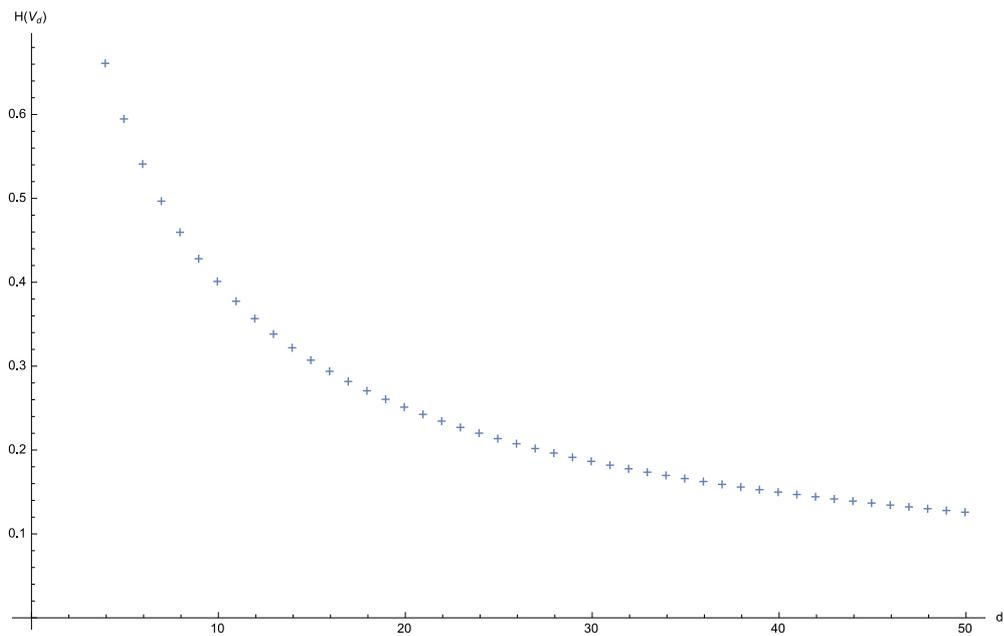}
  \end{center}
  \caption{ {\bf A lower bound of entanglement for implementing $V_d$ by LOCC.} A lower bound is calculated by evaluating $H(V_d)=\sum_i -\lambda_i^2 \log_2 (\lambda_i^2)$, where $\{\lambda_i\}_i$ is the operator Schmidt coefficients of $V_d$.
}  
\label{fig:SS}
\end{figure}

%\item {\bf Negative probability corresponds to no-signaling correlation.}
\item {\bf Multipartite LOCC* and LOCQP}

Multipartite LOCC* is defined by a set of CPTP maps $\mathcal{M}$ given in the form of
\begin{equation}
\mathcal{M}=\sum_{i_1,\cdots,i_N,o_1,\cdots,o_N}p(i_1,\cdots,i_N | o_1,\cdots,o_N)\mathcal{E}_{o_1|i_1}^{(1)}\otimes\cdots\otimes\mathcal{E}_{o_N|i_N}^{(N)},
\label{eq:mLOCC*}
\end{equation}
where $p(i_1,\cdots,i_N | o_1,\cdots,o_N)$ is a conditional probability distribution and $\{\mathcal{E}_{o_k|i_k}^{(k)}:\mathbf{L}(\mathcal{H}_{I_k})\rightarrow\mathbf{L}(\mathcal{H}_{O_k})\}_{o_k}$ is a local operation performed by the $k$-th party with a classical input $i_k$ and a classical output $o_k$. In Appendix B.9, we show that multipartite LOCC* as well as bipartite LOCC* is equivalent to SEP.

On the other hand, multipartite LOCQP is defined by a set of CPTP maps $\mathcal{M}$ given in the form of Eq.\eqref{eq:mLOCC*} where a conditional probability distribution $p(i_1,\cdots,i_N | o_1,\cdots,o_N)$ satisfies
\begin{eqnarray}
\sum_{i_1,\cdots,i_N,o_1,\cdots,o_N}p(i_1,\cdots,i_N | o_1,\cdots,o_N)\mathcal{E}_{o_1|i_1}^{(1)}\otimes\cdots\otimes\mathcal{E}_{o_N|i_N}^{(N)}\nonumber\\
\in \mathbf{C}(\mathcal{H}_{I_1}\otimes\cdots\otimes\mathcal{H}_{I_N}:\mathcal{H}_{O_1}\otimes\cdots\otimes\mathcal{H}_{O_N}),
\end{eqnarray}
for all quantum instruments $\{\mathcal{E}_{o_1|i_1}^{(1)}\}_{o_1}$, $\{\mathcal{E}_{o_2|i_2}^{(2)}\}_{o_2}$ and so on.
It is easy to verify that LOCC* is larger than LOCQP, however we do not know whether multipartite (more than two parties) LOCC* is strictly larger than multipartite LOCQP or not.
In \cite{Amin}, multipartite causally non-separable CQP is proposed.
This implies that we may be able to construct elements of non-LOCC SEP by using such multipartite causally non-separable CQP.

%\item {\bf Quantum processes is not just a folded communication}
%The power of quantum processes has been shown in 
%gameやswitchを見ていると、foldingのようにも思えるが、違うらしい。

\end{itemize}

\part{Conclusion}%What and How do we do? What's the impact?過去形
In this thesis, we have focused on two questions concerning the roles of entanglement in DQC:
\begin{itemize}
\item Which unitary operations are implementable over a given quantum network represented by an entanglement resource?
\item What are resources substituting entanglement and classical communication consumed in entanglement assisted LOCC for implementing SEP?
\end{itemize}

In Part II, we have investigated the first question by analyzing implementability of $k$-qubit unitary operations over the $(k,N)$-cluster networks where inputs and outputs of quantum computation are given in all separated nodes and quantum communication between nodes is restricted to sending just one-qubit while classical communication is freely allowed. We have considered a one-shot scenario where we can use a given cluster network only once and exact implementation without error is required.
We have presented a method to obtain quantum circuit representations of unitary operations implementable over a given cluster network.   For the $(k,N)$-cluster networks with $k=2,3$, we have shown that our method provides all implementable unitary operations over the cluster network. As a first step to find the fundamental primitive networks of network coding for quantum computation, we have shown that both the butterfly and grail networks are sufficient resources for implementing arbitrary two-qubit unitary operations, meanwhile the $(2,2)$-cluster network is not sufficient to implement arbitrary two-qubit unitary operations even probabilistically. 
Our methods are applicable to generalized cluster networks, which may be a first step to investigate quantum computation over a more general network.
It also gives a new insight for implementable unitary operations in MBQC where the angle of each projective measurement can be restricted.

In Part III, we have investigated the second question in terms of entanglement and causal relation in classical communication.
First, we have derived the amount of the entanglement resource that is necessary and sufficient for a DQC task that perfectly discriminates a state from a given orthonormal basis states by one-way LOCC in terms of the Schmidt rank of an entanglement resource. We have constructed a two-way LOCC protocol which consumes less amount of entanglement resources than the optimal one-way LOCC protocol. 
Second, we have developed a framework to describe deterministic joint quantum operations in DQC. In the framework, we have introduced a causal relation between the classical outputs and classical inputs of the local operations without predefined causal order but still within quantum mechanics, called ``classical communication" without predefined causal order (CC*). We have shown that local operations and CC* (LOCC*) with partial order can be simulated by normal LOCC. In contrast, CC* without partial order, which can be interpreted as ``classical communication" without causal order, is necessary to implement an operation in non-LOCC SEP.
Since we show that LOCC* is equivalent to SEP, entanglement and classical communication required for implementing SEP by entanglement assisted LOCC can be substituted by CC*. Note that in LOCC*, two assumptions of local operations, (a) they are partially ordered and (b) the choice of a local operation does not depend on resources connecting the local operation, are relaxed.

We have investigated the power of CC* in terms of SLOCC and quantum processes.
We have shown that the super class of LOCC* is simulatable by SLOCC and vice versa, and LOCC* is simulatable by SLOCC with a constant success probability independent of an input state and vice versa.
Next, we have investigated the relationship between LOCC* and local operations linked by a classical quantum process (LOCQP).
We have shown that LOCQP is equivalent to a set of probabilistic mixtures of one-way LOCC in bipartite cases.
On the other hand, we have shown that LOCQP is not included in LOCC in tripartite cases by constructing an example of operations in non-LOCC SEP by using LOCQP.
In \cite{OFC, Chiribella2}, a few examples of joint quantum operations consisting of two local operations linked by quantum processes that cannot be implemented by using the local operations connected by a causally ordered quantum communication are shown.
However, the joint quantum operations are implementable if we are allowed to use each local operation twice and a causally ordered quantum communication.
In contrast, our example of LOCQP, which is an operation in non-LOCC SEP, cannot be implemented even if finite times use of local operations connected by a causally ordered classical communication are allowed.
Thus, joint quantum operations respecting the restriction for classical quantum processes have a power of not only {\it folding} several local operations localized in a spacetime but more.
LOCC* gives a unified description of LOCC and SEP, and provides a new operational meaning of the gap between LOCC and SEP.
LOCQP gives a new method to construct an element in the gap between LOCC and SEP while very few examples are known.
LOCQP also gives a new insight into the interpretation of the restriction for classical quantum processes.

\appendix
\part{Appendix}
\chapter{Quantum computation over quantum networks}
\section{A LOCC protocol for controlled unitary operations}
We show a construction of a protocol to implement a three-qubit fully controlled unitary operation $C_{l,m;n}$ on qubits located at a set of vertically aligned nodes $\mathcal{V}_j^v$ over the $(k,N)$-cluster networks, where $l$ and $m$ represent two control qubits at nodes $v_{l,j}$ and $v_{m,j}$ respectively, and $n$ represents a target qubit at node $v_{n,j}$.    We present a LOCC protocol to implement $C_{l,m;n}$ assisted by the resource states consisting of the EPR pairs corresponding to the vertical edges $\mathcal{S}_j$ of the $(k,N)$-cluster networks.  

We consider to implement $C_{l,m;n}$ on a state of three qubits indexed by $Q_l$, $Q_m$ and $Q_n$ at node $v_{l,j}$, $v_{m,j}$ and $v_{n,j}$, respectively, and its explicit form is given by
\begin{eqnarray}
C_{l,m;n} (\{ u_n^{(ab)} \}):=\sum_{a=0}^1 \sum_{b=0}^1 \ket{ab}\bra{ab}_{lm} \otimes u_n^{(ab)}
\end{eqnarray}
where $\{ \ket{ab} \}_{a,b=\{ 0,1 \} }$ is the two-qubit computational basis of $\mathcal{H}_{Q_l}  \otimes \mathcal{H}_{Q_m}$ of the two controlled qubits and $u_n^{(ab)}$ acts on $\mathcal{H}_{Q_n}$ of the target qubit.  

To show how our LOCC protocol works, we consider an arbitrary state of the control qubits by 
$
\sum \lambda_{ab}\ket{ab}_{lm} \in \mathcal{H}_{Q_l}  \otimes \mathcal{H}_{Q_m} 
$
where
$\{ \lambda_{ab} \}$ is a set of arbitrary complex coefficients satisfying the normalization condition $\sum_{a,b} | \lambda_{ab} |^2 = 1$ and we represent an arbitrary state of the target qubit by $\ket{\phi} \in \mathcal{H}_{Q_n}$.  In the following, we show that our protocol transforms the joint state of controlled qubits and a target qubit to 
\begin{equation}
C_{l,m;n}  \sum_{a,b} \lambda_{ab}\ket{ab}_{lm} \ket{\phi}_n = \sum_{a,b}\lambda_{ab}\ket{ab}_{lm}u_n^{(ab)}\ket{\phi}_n. \nonumber
\end{equation}

\noindent
The protocol for implementing three qubit fully controlled unitary operations (see Fig.~\ref{fig:3CUcircuit}) :
\begin{enumerate}
\item  Ancillary qubits indexed by $Q_{l^\prime}$, $Q_{m^\prime}$ are introduced at nodes  $v_{l,j}$ and $v_{m,j}$ respectively.  Set both of the ancillary qubits to be in $\ket{0}$.   Each of the two states of control qubits $Q_l$ and $Q_m$ is transformed to a two-qubit state by applying a CNOT gate on the control qubit and the ancillary qubit at the same node, namely applying CNOT on $Q_l$ and $Q_{l^\prime}$ and also $Q_m$ and $Q_{m^\prime}$.  Then the joint state of five qubits $Q_l$, $Q_{l^\prime}$, $Q_m$, $Q_{m^\prime}$ and $Q_n$ is given by 
\begin{equation}
\sum_{a,b}\lambda_{ab}\ket{ab}_{lm}\ket{ab}_{l'm'}\ket{\phi}_n.
\end{equation}
\item By consuming the EPR pairs corresponding to the vertical edges $\mathcal{S}_j$ between $v_{l,j}$ and $v_{n,j}$ and also between  $v_{m,j}$ and $v_{n,j}$, perform quantum teleportation to transmit the states of qubits $Q_{l^\prime}$ and $Q_{m^\prime}$ from nodes $v_{l,j}$ and $v_{m,j}$ to $v_{n,j}$. A circuit representation of the protocol of quantum teleportation represented by $\mathcal{T}$ in Fig. \ref{fig:3CUcircuit} is given by Fig.~\ref{fig:Tcircuit}. Note that in case of $n \neq l\pm1$, we have to repeat the teleportation protocol to transmit a quantum state between the nodes via the neighboring nodes.  Thus all the EPR pairs corresponding to the vertical edges between $l$ and $n$ are consumed for performing teleportation.  We denote indices of two qubits at node $v_{n,j}$  representing the teleported states from $Q_{l^\prime}$ and $Q_{m^\prime}$ by $Q_{l''}$ and $Q_{m''}$, respectively. 
\item At node $v_{n,j}$, perform $C_{ l,m;n}$ on  $\mathcal{H}_{Q_{l''}}  \otimes \mathcal{H}_{Q_{m''}} \otimes \mathcal{H}_{Q_n}$.  Then we obtain the state given by
\begin{equation}
\sum_{a,b}\lambda_{ab}\ket{ab}_{lm}\ket{ab}_{l''m''}u_n^{(ab)}\ket{\phi}_n.
\end{equation}
\item At node $v_{n,j}$, we apply the Hadamard operations and perform projective measurements in the computational basis on both $\mathcal{H}_{Q_{l''}}$ and $\mathcal{H}_{Q_{m''}}$.   The measurement outcomes of qubits $Q_{l''}$  and $Q_{m''}$ are sent to nodes $v_{l,j}$ and $v_{m,j}$, respectively, by classical communication.   At each of nodes  $v_{l,j}$ and $v_{m,j}$, if the measurement outcome is $0$, do nothing, and if the outcome is $1$, perform $Z$ for a correction on qubit $Q_l$ or $Q_m$.   By straightforward calculation, we obtain the state of three qubits $Q_l$, $Q_m$ and $Q_n$ at nodes $v_{l,j}, v_{m,j}$ and $v_{l,n}$, respectively, given by
\begin{equation}
\sum_{a,b}\lambda_{ab}\ket{ab}_{lm}u_n^{(ab)}\ket{\phi}_n.
\end{equation}
\end{enumerate}

Therefore, $C_{l,m;n}$ is successfully applied on the control qubits at nodes $v_{l,j}$ and $v_{m,j}$ and the target qubit at node $v_{n,j}$ by LOCC assisted by the EPR pairs corresponding to the vertical edges $\mathcal{S}_j$ between nodes $v_{l,j}$ and $v_{m,j}$.

\begin{figure}
\begin{center}
  \includegraphics[height=.4\textheight]{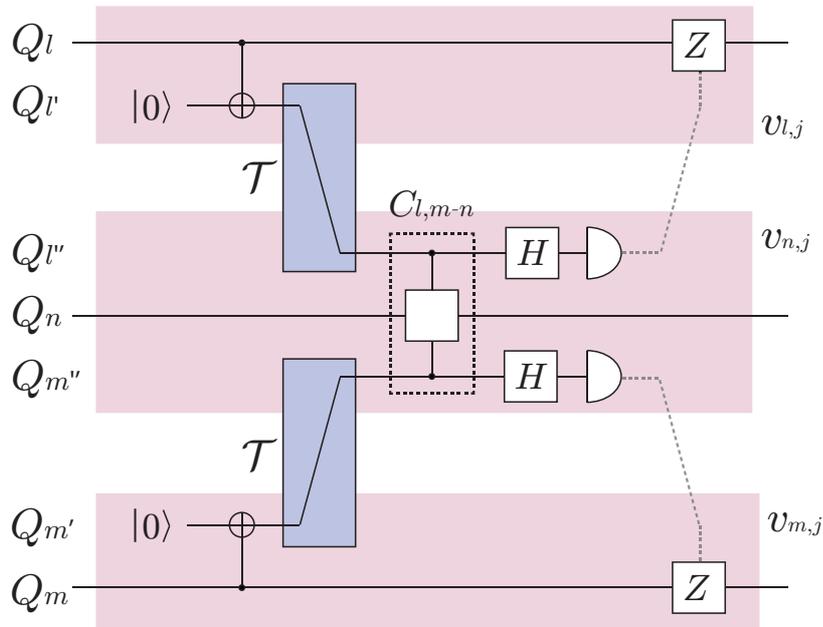}
  \end{center}
  \caption{A quantum circuit representation of the LOCC protocol implementing a three-qubit fully unitary operation $C_{ l,m;n}$, where qubits in the first shaded region are at the node $v_{l,j}$, those in the second shaded region are at the node $v_{n,j}$ and those in the third shaded region are at the node $v_{m,j}$. $\mathcal{T}$ represents teleportation of a qubit state. A circuit representation is given by Fig.~\ref{fig:Tcircuit}. The protocol consists of introducing ancillary qubits $Q_{l'}$  and $Q_{m'}$ at the nodes $v_{l,j}$ and $v_{m,j}$, respectively, teleporting ancillary qubit states from the nodes $v_{l,j}$ and $v_{m,j}$ to the node $v_{n,j}$ represented by qubits $Q_{l''}$  and $Q_{m''}$, applying $C_{ l,m;n}$ on controlled qubits $Q_{l''}$  and $Q_{m''}$ and a target qubit $Q_{n}$ at the node $v_{n,j}$, performing Hadamard operations and measurements in the computational basis on $Q_{l''}$  and $Q_{m''}$ at node $v_{n,j}$ and finally applying conditional $Z$ operations depending on the measurement outcome on two control qubits $Q_l$, $Q_m$ at nodes $v_{l,j}$ and $v_{m,j}$, respectively.}
\label{fig:3CUcircuit}
\end{figure}

\section{LOCC implementation of converted quantum circuits}
In Appendix A.1, we have shown a protocol to implement a three-qubit fully controlled unitary operation in a set of vertically aligned nodes $\mathcal{V}_j^v$. In some cases, we can implement more than one three-qubit or two-qubit controlled unitary operations {\it in parallel} by using the same resource.  We show how a sequence of controlled unitary operations represented by converted circuits can be implemented by LOCC assisted by the resource state given by a collection of $(k-1)$ EPR pairs corresponding to a set of vertically aligned edges $\mathcal{S}_j$ in this appendix.

We introduce a new notation for controlled unitary operations for simplifying and unifying descriptions of two-qubit and three-qubit controlled unitary operations.
We represent a two-qubit controlled unitary operations that controls the $i$-th qubit and targets  the $j$-th qubit as
\begin{equation}
(i,i;j),
\end{equation}
and a three-qubit controlled unitary  gates  that controls  the $i$-th and $j$-th qubit and targets the $k$-th qubit as
\begin{equation}
(i,j;k).
\end{equation}
Note that we represented $(i,i;j)$ as $C_{i;j}$ and $(i,j;k)$ as $C_{i,j;k}$ in the previous sections.
Let $G=\{g_n\}$ be a sequence of controlled unitary operations that is added in step 2 of the conversion protocol.
For example, the converted circuit represented by Fig.~A.2 is described by a sequence
\begin{eqnarray}
g_1&=&(1,1;2)
\label{g1}
\\
g_2&=&(4,4;2)\\
g_3&=&(1,4;2)\\
g_4&=&(4,4;5)\\
g_5&=&(4,4;3)\\
g_6&=&(5,5;6)\\
g_7&=&(4,4;5)
\label{g7}.
\end{eqnarray}

\begin{figure}
\begin{center}
  \includegraphics[height=.24\textheight]{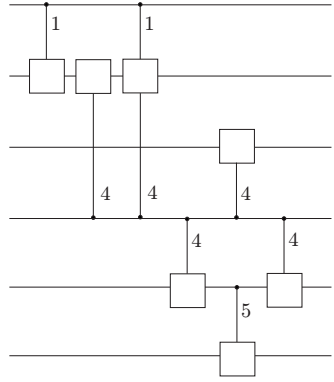}
  \end{center}
  \caption{An example of a converted quantum circuit obtained by step 2 of the conversion protocol.}
\label{fig:circuitsam}
\end{figure}

Let $ \mathcal{C}_i$ be a set of controlled unitary operations that controls the $i$-th qubit:
\begin{equation}
 \mathcal{C}_i=\{(a,b;c)\in G;\, a=i\vee b=i\}.
\end{equation}

For example, for $G$ defined by Eqs.~(\ref{g1})-(\ref{g7}),
\begin{eqnarray}
 \mathcal{C}_1=\{g_1,g_3\}
\label{c1}\\
 \mathcal{C}_4=\{g_2,g_3,g_4,g_5,g_7\}\\
\mathcal{C}_5=\{g_6\}\\
 \mathcal{C}_2=C_3=C_6=\emptyset.
\label{c2}
\end{eqnarray}

Define the {\it range} of $ \mathcal{C}_i\neq \emptyset$ as
\begin{eqnarray}
{\rm range}(\mathcal{C}_i)&=&(\min\{i,\min_{c}\{(a,b;c)\in \mathcal{C}_i\}\},\nonumber\\
&& \max\{i,\max_{c}\{(a,b;c)\in  \mathcal{C}_i\}\}).
\end{eqnarray}

For example,  for $\mathcal{C}_i$ defined by Eqs.~(\ref{c1})-(\ref{c2}),
\begin{eqnarray}
{\rm range}(\mathcal{C}_1)&=&(1,2)\\
{\rm range}(\mathcal{C}_4)&=&(2,5)\\
{\rm range}(\mathcal{C}_5)&=&(5,6).
\end{eqnarray}
 
 All the controlled unitary operations in $G$ are implementable by using the following subroutines extending the one presented in Appendix A.1.
\\\\
\noindent
Subroutines for implementing a sequence of controlled unitary operation in $G$:

{\bf Subroutines}
\begin{enumerate}
\item For applying gates in $\mathcal{C}_i$, we add an ancillary qubit state entangled to the $i$-th qubit state by preparing an ancillary qubit in $\ket{0}$ and applying CNOT where the ancillary qubit is the target qubit of CNOT.  Then the ancillary qubit state is sent from the $i$-th node $v_{i,j} \in \mathcal{V}_j^v$ to the target node by using teleportation.  If several different target qubits are included in $\mathcal{C}_i$, add another ancillary qubit by the same method at a target node, keep one of the ancillary qubits at the target node and send the other to the next target node.   We consume $n_i$ EPR pairs to teleport ancillary qubit states to the target nodes, where $n_i=b-a$ and ${\rm range}(\mathcal{C}_i)=(a,b)$. Since there is no overlap between ranges of $\mathcal{C}_i$ and there is no target unitary operation inserted  between control qubits, we can teleport all the ancillary qubit states entangled to the control states to all the target nodes by just consuming $(k-1)$ EPR pairs. 
\item
We apply all the controlled unitary operations in $G$ in the target nodes by using the teleported ancillary qubit states entangled to the control qubit states as the control qubits.   
\item
We decouple the ancillary qubit states by performing the projective measurements on the ancillary qubits similarly to the case of Appendix A.1 in the target nodes and apply correction unitary operations in the control nodes depending on the measurement outcomes.
\end{enumerate}

\begin{comment}
{\bf Protocol}
\begin{lstlisting}[basicstyle=\ttfamily\footnotesize, frame=single]
n=1;
for(i=1;i<=k;i++) d(i)="not distributed";
while(n<=|G|){
	find i such that g_n\in C_i;
	if(d(i)=="not distributed"){
		do subroutine 1 for C_i;
		d(i)="distributed";
	}
	do subroutine 2 for g_n;
	if(g_j\notin C_i for all j>n){
		do subroutine 3 for C_i;
	}
	n++;
}
 \end{lstlisting}
\end{comment}

\section{Converted circuit of $(2,N)$ and $(3,N)$-cluster network}

First, we prove that any converted circuits of a $(2,N)$-cluster network can be simulated by a circuit consisting of a sequence of $N$  two-qubit unitary operations and local unitary operations.  In this case, 
only two-qubit unitary operations $(1,1;2)$ or $(2,2;1)$ can be added in step 2 of the conversion protocol. Since applying the gate $(1,1;2)$  sequentially for $k\in\mathbb{N}$ times can be simulated by just one use of gate $(1,1;2)$ and gate $(2,2;1)$ can be simulated by one use of gate $(1,1;2)$ and additional local unitary operations, any circuits generated in step 2 of the conversion protocol can be simulated by one use of $(1,1;2)$ and local unitary operations.

Next, we prove that any converted circuits of  a $(3,N)$-cluster network can be simulated by the circuit of a sequence of $N$  three-qubit fully controlled unitary operations given in the form of
\begin{eqnarray}
\ket{00}\bra{00}_{1,3}\otimes u^{(00)}_2+\ket{01}\bra{01}_{1,3}\otimes u^{(01)}_2\nonumber\\
+\ket{10}\bra{10}_{1,3}\otimes u^{(10)}_2+\ket{11}\bra{11}_{1,3}\otimes u^{(11)}_2
\end{eqnarray}
and local unitary operations.
In step 2 of the conversion protocol, every converted circuits can be simulated by six classes of circuits shown in Fig.~10.

\begin{figure}
\begin{center}
  \includegraphics[height=.33\textheight]{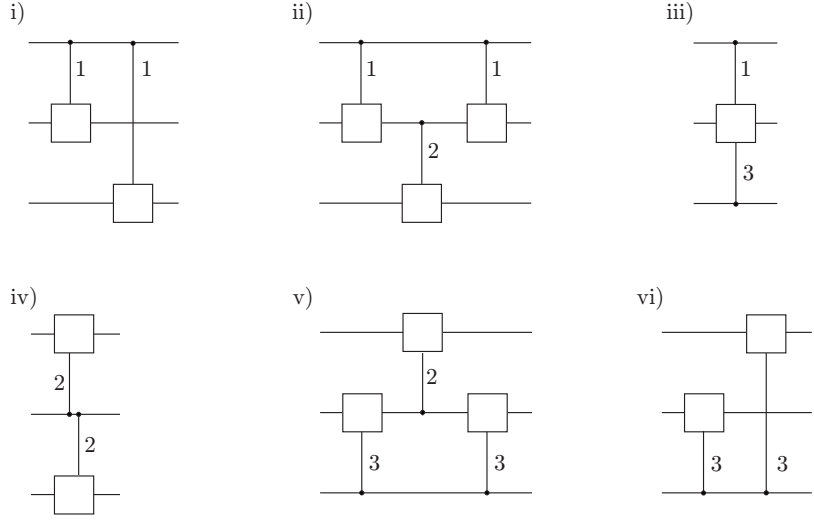}
  \end{center}
  \caption{The six classes of converted quantum circuits obtained by step 2 of the conversion protocol of a $(3,N)$-cluster network.}
\label{fig:convertedcircuit(3,N)}
\end{figure}

In the following, we show that all of these six classes (from class i) to class vi) represented in Fig.~10) can be simulated by a three-qubit fully controlled unitary operation and local unitary operations by investigating each class.
\begin{itemize}
\item[i)] A unitary operation obtained by circuit i) is given by
\begin{eqnarray}
\ket{0}\bra{0}_1\otimes u_2^{(0)}\otimes u_3^{(0)}+\ket{1}\bra{1}_1\otimes u_2^{(1)}\otimes u_3^{(1)}\nonumber\\
\overset{\mathrm{LU}}{=} \ket{0}\bra{0}_1\otimes \mathbb{I}_2\otimes \mathbb{I}_3+\ket{1}\bra{1}_1\otimes u_2^{(1)}u_2^{(0)\dag}\otimes u_3^{(1)}u_3^{(0)\dag}\label{eq:LU1}
\end{eqnarray}
where $u_j^{(i)}$ is a one-qubit unitary operation and $\overset{\mathrm{LU}}{=} $ represents local unitary equivalence. Diagonalize $u_2^{(1)}u_2^{(0)\dag}$ and $u_3^{(1)}u_3^{(0)\dag}$ as
\begin{eqnarray}
u_2^{(1)}u_2^{(0)\dag}&=&v_2 \begin{pmatrix}e^{i\theta_1}&0\\0&e^{i\theta_2}\end{pmatrix}v_2^{\dag}\\
u_3^{(1)}u_3^{(0)\dag}&=&v_3 \begin{pmatrix}e^{i\theta_3}&0\\0&e^{i\theta_4}\end{pmatrix}v_3^{\dag}.
\end{eqnarray}
Since the right handside of Eq.~\eqref{eq:LU1} is locally unitarily equivalent to a diagonal unitary operation in  the computational basis, this circuit can be simulated by a three-qubit fully controlled unitary operation and local unitary operations.

\item[ii)]  In circuit ii), the two-qubit controlled unitary operation $(2,2;3)$ can be decomposed into
\begin{eqnarray}
&\ket{0}\bra{0}_2\otimes u_3^{(0)}+\ket{1}\bra{1}_2\otimes u_3^{(1)}\nonumber\\
=&v_3\left(\ket{0}\bra{0}_2\otimes \mathbb{I}_3+\ket{1}\bra{1}_2\otimes \begin{pmatrix}e^{i\theta_1}&0\\0&e^{i\theta_2}\end{pmatrix}\right)v_3^{\dag}u_3^{(0)}\nonumber\\
=&(\mathbb{I}_2\otimes v_3)\nonumber\\
&\left(\begin{pmatrix}1&0\\0&e^{i\theta_1}\end{pmatrix}\otimes\ket{0}\bra{0}_3+\begin{pmatrix}1&0\\0&e^{i\theta_2}\end{pmatrix}\otimes\ket{1}\bra{1}_3\right)\nonumber\\
&(\mathbb{I}_2\otimes v_3^{\dag}u_3^{(0)}),
\end{eqnarray}
where $v_3$ is a unitary operation that diagonalizes $u_3^{(1)}u_3^{(0)\dag}$.
Thus, this circuit is locally unitarily equivalent to a three-qubit fully controlled unitary operation.

\item[iii)] Circuit iii) consists of just a three-qubit fully controlled unitary operation.

\item[iv)] Circuit iv) can be simulated by a three-qubit fully controlled unitary operation and local unitary operations since we can diagonalize a unitary operation obtained by the circuit in the same way as circuit i).

\item[v)] In the same way as circuit ii),  circuit v) is locally unitarily equivalent to a three-qubit fully controlled unitary operation.

\item[vi)] In the same way as circuit i), circuit vi)  is locally unitarily equivalent to a three-qubit fully controlled unitary operation.
\end{itemize}

\section{Maximally entangled state conversion by SEP}

We prove a lemma about convertibility of maximally entangled states in this appendix.
We analyze a property of bipartite separable maps which preserves entanglement of maximally entangled states.

Let $\ket{\Psi_{in}}=\frac{1}{\sqrt{d}}\sum_i^d \ket{A_i}\ket{B_i}$ and $\ket{\Psi_{out}}=\frac{1}{\sqrt{d}}\sum_i^d \ket{a_i}\ket{b_i}$, where $\{\ket{A_i}\in\mathcal{H}_A\}$ and $\{\ket{B_i}\in\mathcal{H}_B\}$ are orthonormal sets and $\{\ket{a_i}\in\mathcal{H}_a\}$ and $\{\ket{b_i}\in\mathcal{H}_b\}$ are orthonormal bases. Note that $\dim(\mathcal{H}_a)=\dim(\mathcal{H}_b)=d$ but the dimensions of $\mathcal{H}_A$ and $\mathcal{H}_B$ can be higher than $d$.

\begin{lemma}
Let $\{E_k\in\mathbf{L}(\mathcal{H}_A:\mathcal{H}_a)\},\{F_k\in\mathbf{L}(\mathcal{H}_B:\mathcal{H}_b)\}$ be sets of linear operators. If $\{ E_k \otimes F_k \}$ satisfies
\begin{equation}
\sum_k E_k^{\dag}E_k\otimes F_k^{\dag}F_k=\mathbb{I}_{AB}
\end{equation}
and for all $k$,
\begin{equation}
E_k\otimes F_k \ket{\Psi_{in}}=\sqrt{p_k}\ket{\Psi_{out}}
\end{equation}
is satisfied, then for all $\{k|p_k\neq0\}$,
\begin{equation}
\exists \alpha_k>0,\,\exists U_k^M\in\mathbf{U}(\mathbb{C}^d),\: E^M_k=\alpha_k U^M_k, F^M_k=\frac{\sqrt{p_k}}{\alpha_k}\overline{U_k^M},
\end{equation}
where $E^M_k$ and $F^M_k$ are $d$ by $d$ matrices such that $(E^M_k)_{i,j}=\bra{a_i}E_k\ket{A_j}$, $(F^M_k)_{i,j}=\bra{b_i}F_k\ket{B_j}$, $\mathbf{U}(\mathbb{C}^d)$ is the set of $d$ by $d$ unitary matrices and $\overline{U^M}$ is the complex conjugate of $U^M$.
\label{lemma:subunitary}
\end{lemma}

\begin{proof}
By straightforward calculation, we obtain
\begin{eqnarray}
\forall k,\,E^M_k (F^M_k)^{T}&=&\sqrt{p_k}I_d\nonumber\\
\Rightarrow \forall k\in\{k|p_k\neq0\},\,F^M_k&=&\sqrt{p_k}((E^M_k)^{-1})^{T}
\label{eq:transposeM},
\end{eqnarray}
and
\begin{equation}
\sum_{k} (E^M_k)^{\dag}E^M_k\otimes (F^M_k)^{\dag}F^M_k=I_{d^2}.
\label{eq:closureM}
\end{equation}
By using Eq.~(\ref{eq:transposeM}) and Eq.~(\ref{eq:closureM}), we obtain
\begin{eqnarray}
&&{\rm tr}\left(\sum_k E_k^{M\dag}E_k^M\otimes F_k^{M\dag}F_k^M\right)\nonumber\\
&=&{\rm tr}\left(\sum_{\{k|p_k\neq0\}} E_k^{M\dag}E_k^M\otimes F_k^{M\dag}F_k^M\right)+\epsilon\nonumber\\
&=&\sum_{\{k|p_k\neq0\}}p_k {\rm tr}\left(E_k^{M\dag}E_k^M\otimes \overline{(E_k^{M\dag}E_k^M)^{-1}}\right)+\epsilon=d^2\nonumber\\
&\Leftrightarrow& \sum_{\{k|p_k\neq0\}}p_k {\rm tr}\left(E_k^{M\dag}E_k^M\otimes \overline{(E_k^{M\dag}E_k^M)^{-1}}\right)=d^2-\epsilon,\nonumber\\
\label{eq:20150318}
\end{eqnarray}
where $\epsilon={\rm tr}\left(\sum_{\{k|p_k=0\}} E_k^{M\dag}E_k^M\otimes F_k^{M\dag}F_k^M\right)\geq0$.
Since $E_k^M$ is a regular matrix, we can let $P_k=E_k^{M\dag}E_k^M$ be a $d$ by $d$ positive matrix and $\{\lambda_k^i>0|i=0,1,\cdots,d-1\}$ be the set of eigenvalues of $P_k$. Then the eigenvalues of $ \overline{(E_k^{M\dag}E_k^M)^{-1}}=\overline{P_k^{-1}}$ are $\{1/\lambda_k^i|i=0,1,\cdots,d-1\}$ and the condition Eq.~\eqref{eq:20150318} is given by
\begin{equation}
\sum_{\{k|p_k\neq0\}} p_k\sum_i^d \lambda_k^i\sum_j^d \frac{1}{\lambda_k^j}=d^2-\epsilon.
\label{eq:d-e}
\end{equation}
Using the Cauchy-Schwarz inequality, we obtain
\begin{eqnarray}
\sum_i^d \lambda_k^i\sum_j^d \frac{1}{\lambda_k^j}&=&\left(\sum_i^d \sqrt{\lambda_k^i}^2\right)\left(\sum_j^d \sqrt{\frac{1}{\lambda_k^j}}^2\right)\nonumber\\
&\geq& \left(\sum_i^d1\right)^2=d^2.
\label{eq:csinequality}
\end{eqnarray}
The equality holds if and only if $\lambda_k^i=\alpha^2$ for all $i$.
By using  Eqs.~(\ref{eq:d-e})-(\ref{eq:csinequality}) and the fact that $\{p_k|p_k\neq0\}$ is a probability distribution, we obtain
for all $\{k|p_k\neq0\}$,
\begin{eqnarray}
\exists \alpha>0;\:P_k=E_k^{M\dag}E_k^M=\alpha^2\mathbb{I}_d,\\
\epsilon=0.
\end{eqnarray}

\end{proof}

\section{Two conditions in Theorem 1}

We show that two conditions in Theorem 1 are equivalent in the case of the $(3,N)$-cluster networks.
For $k=3$,
the $2^k$ by $2^k$ unitary matrix $V_i^M$ in Theorem 1 is written by
\begin{eqnarray}
V_i^M&=&E_{1,i}^{(0)}\otimes E_{2,i}^{(0,0)}\otimes E_{3,i}^{(0)}\nonumber\\
&&+E_{1,i}^{(0)}\otimes E_{2,i}^{(0,1)}\otimes E_{3,i}^{(1)}\nonumber\\
&&+E_{1,i}^{(1)}\otimes E_{2,i}^{(1,0)}\otimes E_{3,i}^{(0)}\nonumber\\
&&+E_{1,i}^{(1)}\otimes E_{2,i}^{(1,1)}\otimes E_{3,i}^{(1)}.
\end{eqnarray}

By using the result on local unitary equivalence of unitary operations with operator Schmidt rank 2 obtained by Cohen and Yu \cite{Cohen} (Theorem 1 of \cite{Cohen} ), we have
\begin{eqnarray}
V_i^M&\overset{\mathrm{LU}}{=}&\ket{0}\bra{0}_A\otimes W_{BC}^{(0)}+\ket{1}\bra{1}_A\otimes W_{BC}^{(1)}\\
&=&W_{AB}^{(0)}\otimes \ket{0}\bra{0}_C+W_{AB}^{(1)}\otimes\ket{1}\bra{1}_C,
\end{eqnarray}
where $W_{BC}^{(i)}$ and $W_{AB}^{(i)}$ are unitary matrices, $\overset{\mathrm{LU}}{=}$ represents a local unitary equivalence and we identify a three-qubit unitary operation on $\mathcal{H}_A\otimes\mathcal{H}_B\otimes\mathcal{H}_C$ as its matrix representation $V_i^M$.
Thus, it is shown that
\begin{eqnarray}
V_i^M&\overset{\mathrm{LU}}{=}&\ket{00}\bra{00}_{AC}\otimes W_{B}^{(00)}+\ket{01}\bra{01}_{AC}\otimes W_{B}^{(01)}\nonumber\\
&&+\ket{10}\bra{10}_{AC}\otimes W_{B}^{(10)}+\ket{11}\bra{11}_{AC}\otimes W_{B}^{(11)},\nonumber\\
\end{eqnarray}
where $W_B^{(ij)}$ is a $2$ by $2$ unitary matrix.
Statements i) and ii) in Theorem 2 of the main text are equivalent in the case of $(3,N)$-cluster networks since $V_i^M$ is a fully controlled  three-qubit unitary operation and $N$ fully controlled three-qubit unitary operations are implementable by a converted circuit of the $(3,N)$-cluster networks.

\section{Network coding for the butterfly network}

We show that the quantum circuit presented in Fig.~\ref{fig:butterflycircuit}
implements a two-qubit global unitary $U_{global}(x,y,z)$ given by Eq.\eqref{gloablpartofKCD}  for arbitrary parameters $x, y, z \in \mathbb{R}$.
$U_{global}(x,y,z)$ can be decomposed into
\begin{eqnarray}
U_{global}(x,y,z)=\sum_i \lambda_i\ket{\Phi^{(i)}}\bra{\Phi^{(i)}}
\end{eqnarray}
by using its eigenvalues $\{\lambda_i\}_i$ and eigenvectors $\{\ket{\Phi^{(i)}}\}_i$ such that
\begin{eqnarray}
\lambda_0=e^{i(x-y+z)},\lambda_1=e^{i(-x+y+z)},\lambda_2=e^{i(x+y-z)},\lambda_3=e^{i(-x-y-z)}
\end{eqnarray}
\begin{eqnarray}
\ket{\Phi^{(0)}}&=&\frac{1}{\sqrt{2}}(\ket{00}+\ket{11})\\
\ket{\Phi^{(1)}}&=&\frac{1}{\sqrt{2}}(\ket{00}-\ket{11})\\
\ket{\Phi^{(2)}}&=&\frac{1}{\sqrt{2}}(\ket{01}+\ket{10})\\
\ket{\Phi^{(3)}}&=&\frac{1}{\sqrt{2}}(\ket{01}-\ket{10}).
\end{eqnarray}
Thus, in order to show an arbitrary input state $\ket{\phi}$ is transformed into $U_{global}(x,y,z)\ket{\phi}$ through the quantum circuit, it is sufficient to show that the eigenvectors $\{\ket{\Phi^{(i)}}\}_i$ are transformed into $\{\lambda_i\ket{\Phi^{(i)}}\}_i$ and when a measurement $\{M_m\}_{m\in\Omega}$ is performed, the probability of obtaining a measurement outcome must be equal, i.e. for all $m\in\Omega$,
\begin{equation}
\bra{\Phi^{(0)}}M_m\ket{\Phi^{(0)}}=\bra{\Phi^{(1)}}M_m\ket{\Phi^{(1)}}=\bra{\Phi^{(2)}}M_m\ket{\Phi^{(2)}}=\bra{\Phi^{(3)}}M_m\ket{\Phi^{(3)}}
\end{equation}
not to break coherence between the eigenvectors.

We divide the quantum circuit into seven steps from step (i) to step (vii) as shown in Fig.~\ref{fig:butterflycircuit1}.  We show the detail of how the eigenvectors are transformed after each step.

First, we prepare a three-qubit input state
\begin{equation}
\ket{\Phi^{(i)}}_{1,3}\ket{0}_2
\end{equation}
in step (i), where we denote the index of the qubit corresponding to the first horizontal wire as 1 and that of the others likewise. After applying Hadamard gates in step (ii), we obtain
\begin{equation}
H_1 H_3\ket{\Phi^{(i)}}_{1,3}\ket{+}_2,
\end{equation}
where $\ket{\pm}=\frac{1}{\sqrt{2}}(\ket{0}\pm\ket{1})$.
After applying $C_{1,3;2}$ in step (iii), we obtain

\begin{equation}
\frac{1}{\sqrt{2}}\left(H_1 H_3\ket{\Phi^{(i)}}_{1,3}\ket{0}_2+Z_1H_1 Z_3H_3\ket{\Phi^{(i)}}_{1,3}\ket{1}_2\right).
\end{equation}
After applying Hadamard gates and Pauli X operations in step (iv), we obtain
\begin{subnumcases}
{\frac{1}{\sqrt{2}}\left(X_1 X_3\ket{\Phi^{(i)}}_{1,3}\ket{+}_2+\ket{\Phi^{(i)}}_{1,3}\ket{-}_2\right)=}
\ket{\Phi^{(i)}}_{1,3}\ket{0}_2 & ($i=0,2$) \nonumber\\
-\ket{\Phi^{(i)}}_{1,3}\ket{1}_2 & ($i=1,3$).\nonumber
\end{subnumcases}
\begin{eqnarray}
\end{eqnarray}
After applying $C'_{1,3;2}$ in step (v), we obtain
\begin{subnumcases}
{}
e^{i(-y+z)}\ket{\Phi^{(0)}}_{1,3}\ket{0}_2  \nonumber\\
ie^{i(y+z)}\ket{\Phi^{(1)}}_{1,3}\ket{1}_2 \nonumber\\
e^{i(y-z)}\ket{\Phi^{(2)}}_{1,3}\ket{0}_2  \nonumber\\
ie^{i(-y-z)}\ket{\Phi^{(3)}}_{1,3}\ket{1}_2.
\end{subnumcases}
After applying a single qubit unitary operation $u(x)$ given by
\begin{equation}
u(x) = \frac{1}{\sqrt{2}} \left(\begin{array}{cc} e^{i x} & -i e^{- i x} \\e^{i x} & i e^{- i x}\end{array}\right)
\end{equation}
in step (vi), we obtain
\begin{subnumcases}
{}
e^{i(x-y+z)}\ket{\Phi^{(0)}}_{1,3}\ket{+}_2  \nonumber\\
e^{i(-x+y+z)}\ket{\Phi^{(1)}}_{1,3}\ket{-}_2 \nonumber\\
e^{i(x+y-z)}\ket{\Phi^{(2)}}_{1,3}\ket{+}_2  \nonumber\\
e^{i(-x-y-z)}\ket{\Phi^{(3)}}_{1,3}\ket{-}_2.
\end{subnumcases}
After applying the projective measurement in the computational basis and conditional unitary opterations in step (vii), we obtain
\begin{subnumcases}
{}
e^{i(x-y+z)}\ket{\Phi^{(0)}}_{1,3}=\lambda_0 \ket{\Phi^{(0)}}_{1,3} \nonumber\\
e^{i(-x+y+z)}\ket{\Phi^{(1)}}_{1,3}=\lambda_1 \ket{\Phi^{(1)}}_{1,3} \nonumber\\
e^{i(x+y-z)}\ket{\Phi^{(2)}}_{1,3}=\lambda_2 \ket{\Phi^{(2)}}_{1,3}  \nonumber\\
e^{i(-x-y-z)}\ket{\Phi^{(3)}}_{1,3}=\lambda_3 \ket{\Phi^{(3)}}_{1,3}.
\end{subnumcases}
for any measurement outcome. We can verify that the probability of obtaining a measurement outcome is $\frac{1}{2}$ irrespective of eigenvectors.

\begin{figure}
\begin{center}
  \includegraphics[height=.2\textheight]{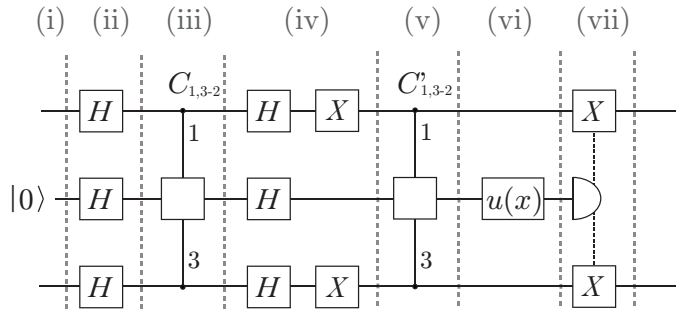}
  \end{center}
  \caption{A protocol to implement a two-qubit unitary operation $U_{global} (x,y,z)$ over the butterfly network. We consider 7 steps  presented in the quantum circuit and denote the steps by Roman numerals, (i) to (vii).  The symbols of gates of the circuit are same as the ones given for Fig.~4.5.}

\label{fig:butterflycircuit1}
\end{figure}

\section{Analysis of four qubit states}

We prove that there is no pure state of four qubits $\ket{\Phi}_{1,2,3,4}$ satisfying
\begin{eqnarray}
\label{eq:sch1}{\rm Sch}\#_{1,2}^{3,4}(\ket{\Phi})&=&4,\\
\label{eq:sch2}{\rm Sch}\#_{2,4}^{1,3}(\ket{\Phi})&=&2,\\
\label{eq:sch3}{\rm Sch}\#_{2,3}^{1,4}(\ket{\Phi})&=&2.
\end{eqnarray}
In \cite{4qubit}, it is shown that any pure states of four qubits can, up to permutations of the qubits, be transformed into one of the following nine families of states by determinant 1 SLOCC:
\begin{eqnarray}
\ket{\Phi_1}&=&\frac{a+d}{2}(\ket{0000}+\ket{1111})+\frac{a-d}{2}(\ket{0011}+\ket{1100})\nonumber\\
&&+\frac{b+c}{2}(\ket{0101}+\ket{1010})+\frac{b-c}{2}(\ket{0110}+\ket{1001})\nonumber\\
\ket{\Phi_2}&=&\frac{a+b}{2}(\ket{0000}+\ket{1111})+\frac{a-b}{2}(\ket{0011}+\ket{1100})\nonumber\\
&&+c(\ket{0101}+\ket{1010})+\ket{0110}\nonumber\\
\ket{\Phi_3}&=&a(\ket{0000}+\ket{1111})+b(\ket{0101}+\ket{1010})\nonumber\\
&&+\ket{0110}+\ket{0011}\nonumber\\
\ket{\Phi_4}&=&a(\ket{0000}+\ket{1111})+\frac{a+b}{2}(\ket{0101}+\ket{1010})\nonumber\\
&&+\frac{a-b}{2}(\ket{0110}+\ket{1001})\nonumber\\
&&+\frac{i}{\sqrt{2}}(\ket{0001}+\ket{0010}+\ket{0111}+\ket{1011})\nonumber\\
\ket{\Phi_5}&=&a(\ket{0000}+\ket{0101}+\ket{1010}+\ket{1111})\nonumber\\
&&+i\ket{0001}+\ket{0110}-i\ket{1011}\nonumber\\
\ket{\Phi_6}&=&a(\ket{0000}+\ket{1111})+\ket{0011}+\ket{0101}+\ket{0110}\nonumber\\
\ket{\Phi_7}&=&\ket{0000}+\ket{0101}+\ket{1000}+\ket{1110}\nonumber\\
\ket{\Phi_8}&=&\ket{0000}+\ket{1011}+\ket{1101}+\ket{1110}\nonumber\\
\ket{\Phi_9}&=&\ket{0000}+\ket{0111},\nonumber
\end{eqnarray}
where $a,b,c,d$ are complex parameters.

Since the Schmidt number of a state cannot be increased under SLOCC and determinant 1 SLOCC is invertible, the Schmidt number of a state is invariant under determinant 1 SLOCC. Thus, we show that no state of the  nine families simultaneously satisfies Eqs.~(\ref{eq:sch1})-(\ref{eq:sch3}).
There are three ways to divide four qubits into a pair of two qubits. We denote the set of Schmidt numbers of a four qubit state $\ket{\Phi}$ for all separations as ${\rm Sch}\#(\ket{\Phi})=\{{\rm Sch}\#_{1,2}^{3,4}(\ket{\Phi}),{\rm Sch}\#_{2,4}^{1,3}(\ket{\Phi}),{\rm Sch}\#_{2,3}^{1,4}(\ket{\Phi})\}$.

\begin{theorem}
There is no four qubit state $\ket{\Phi}\in\mathcal{H}_1\otimes\mathcal{H}_2\otimes\mathcal{H}_3\otimes\mathcal{H}_4$ such that
\begin{equation}
{\rm Sch}\#(\ket{\Phi})=\{4,2,2\}.
\label{eq:schcondition}
\end{equation}
\end{theorem}

\begin{proof}
By straightforward calculation, we can easily check that
\begin{eqnarray}
{\rm Sch}\#(\ket{\Phi_6})&=&\{n_6,n_6,n_6\}\\
{\rm Sch}\#(\ket{\Phi_7})&=&\{3,3,3\}\\
{\rm Sch}\#(\ket{\Phi_8})&=&\{3,3,3\}\\
{\rm Sch}\#(\ket{\Phi_9})&=&\{2,2,2\},
\end{eqnarray}
where $n_6=\#\left\{\sqrt{2},\frac{1}{2}\sqrt{1+4|a|^2}+\frac{1}{2},\frac{1}{2}\sqrt{1+4|a|^2}-\frac{1}{2}\right\}$ and $\#\mathcal{S}$ is the number of non-zero elements of set $\mathcal{S}$. Since $n_6=2$ or $n_6=3$, these four states do not satisfy Eq.~\eqref{eq:schcondition}.

An element of ${\rm Sch}\#(\ket{\Phi_5})$ is $\#\left\{1,\sqrt{2},2|a|\right\}$. To satisfy Eq.~\eqref{eq:schcondition}, $a=0$ is required. Then
\begin{equation}
{\rm Sch}\#(\ket{\Phi_5})=\{2,3,3\},
\end{equation}
which does not satisfy Eq.\eqref{eq:schcondition}.

An element of ${\rm Sch}\#(\ket{\Phi_4})$ is $\#\{|b|\}+\#\{x|x^3-(3|a|^2+2)x^2+(3|a|^4+2|a|^2+1)x-|a|^6=0\}$.
To satisfy Eq.~\eqref{eq:schcondition}, the element must be 2 or 4.
If the element is 2, since $\#\{x|x^3-(3|a|^2+2)x^2+(3|a|^4+2|a|^2+1)x-|a|^6=0\}$ is larger than 1 and is 2 if and only if $a=0$, we have
\begin{equation}
a=b=0.
\end{equation}
Then ${\rm Sch}\#(\ket{\Phi_4})=\{2,2,2\}$. Thus, the element must be 4.
Since $\#\{x|x^3-(3|a|^2+2)x^2+(3|a|^4+2|a|^2+1)x-|a|^6=0\}$ is 3 if and only if $a\neq0$,  we have
\begin{equation}
a\neq0,\,\,b\neq0.
\label{eq:schcon1}
\end{equation}
Another element of ${\rm Sch}\#(\ket{\Phi_4})$ is $\#\{|a-b|\}+\#\{x|64x^3+(\cdots)x^2+(\cdots)x-|a-b|^4|3a+b|^2=0\}$, where we abbreviate coefficients of $x^2$ and $x$. Since this element must be 2, it is necessary that
\begin{equation}
a-b=0\,\,or\,\,3a+b=0.
\label{eq:schcon2}
\end{equation}
The other element of ${\rm Sch}\#(\ket{\Phi_4})$ is $\#\{|a+b|\}+\#\{x|64x^3+(\cdots)x^2+(\cdots)x-|a+b|^4|3a-b|^2=0\}$, where we abbreviate coefficients of $x^2$ and $x$. Since this element must be 2, it is necessary that
\begin{equation}
a+b=0\,\,or\,\,3a-b=0.
\label{eq:schcon3}
\end{equation}
We can easily check that it is impossible to simultaneously satisfy Eqs.~(\ref{eq:schcon1})-(\ref{eq:schcon3}).

${\rm Sch}\#(\ket{\Phi_3})$ is $\{n_3,n_3',n_3'\}$, where
\begin{eqnarray}
n_3&=&\#\{\sqrt{2},|a+b|,|a-b|\},\\
n_3'&=&\#\{\sqrt{1+4|a|^2}\pm1,\sqrt{1+4|b|^2}\pm1\}.
\end{eqnarray}
To satisfy Eq.~\eqref{eq:schcondition}, $n_3'$ must be $2$, that is $a=b=0$. Then $n_3=1$, which does not satisfy Eq.~\eqref{eq:schcondition}.

${\rm Sch}\#(\ket{\Phi_2})$ is $\{n_2,n_2',n_2''\}$, where
\begin{eqnarray}
n_2&=&\#\{|a|,|b|,\sqrt{1+4|c|^2}\pm1\},\\
n_2'&=&\#\{|a+b\pm2c|,\sqrt{1+|a-b|^2}\pm1\},\\
n_2''&=&\#\{|a-b\pm2c|,\sqrt{1+|a+b|^2}\pm1\}.
\end{eqnarray}
In the following, we verify that $\{n_2,n_2',n_2''\}$ cannot be $\{4,2,2\}$, $\{2,4,2\}$ or $\{2,2,4\}$.
\begin{enumerate}
\item $\{n_2,n_2',n_2''\}\neq\{4,2,2\}$:

If $n_2=4$, it is necessary that
\begin{equation}
a\neq0,\,\,b\neq0,\,\,c\neq0.
\label{eq:schcon4}
\end{equation}
If $n_2'=2$, it is necessary that
\begin{eqnarray}
a-b=a+b+2c=0,\\
a-b=a+b-2c=0,\\
or\,\,a+b-2c=a+b+2c=0.
\end{eqnarray}
If $n_2''=2$, it is necessary that
\begin{eqnarray}
a+b=a-b+2c=0,\\
a+b=a-b-2c=0,\\
or\,\,a-b-2c=a-b+2c=0.
\label{eq:schcon5}
\end{eqnarray}
We can easily check that it is impossible to simultaneously satisfy  Eqs.~(\ref{eq:schcon4})-(\ref{eq:schcon5}).

\item $\{n_2,n_2',n_2''\}\neq\{2,4,2\}$:

If $n_2=2$, it is necessary that
\begin{eqnarray}
a=b=0,\\
a=c=0,\\
or\,\,b=c=0.
\end{eqnarray}
With the necessary condition for $n_2''=2$, we obtain that
\begin{equation}
a=b=c=0.
\end{equation}
Then, it is impossible to satisfy $n_2'=4$.

\item $\{n_2,n_2',n_2''\}\neq\{2,2,4\}$:

If $n_2=2$, it is necessary that
\begin{eqnarray}
a=b=0,\\
a=c=0,\\
or\,\,b=c=0.
\end{eqnarray}
With the necessary condition for $n_2'=2$, we obtain that
\begin{equation}
a=b=c=0.
\end{equation}
Then, it is impossible to satisfy $n_2''=4$.
\end{enumerate}

Finally, we analyze ${\rm Sch}\#(\ket{\Phi_1})$. ${\rm Sch}\#(\ket{\Phi_1})$ is $\{n_1,n_1',n_1''\}$, where
\begin{eqnarray}
n_1&=&\#\{|a|,|b|,|c|,|d|\},\\
n_1'&=&\#\{|a+b-c-d|,|a-b+c-d|,\nonumber\\
&&|-a+b+c-d|,|a+b+c+d|\},\\
n_1''&=&\#\{|-a+b+c+d|,|a-b+c+d|,\nonumber\\
&&|a+b-c+d|,|a+b+c-d|\}.
\end{eqnarray}
Note that $n_1$, $n_1'$ and $n_1''$ are invariant under permutation of $a$, $b$, $c$ and $d$.
We verify that $\{n_1,n_1',n_1''\}$ cannot be $\{4,2,2\}$, $\{2,4,2\}$ or $\{2,2,4\}$ in the following.
\begin{enumerate}
\item $\{n_1,n_1',n_1''\}\neq\{4,2,2\}$:

If $n_1=4$, it is necessary that
\begin{equation}
a\neq0,\,\,b\neq0,\,\,c\neq0,\,\,d\neq0.
\end{equation}
If $n_1'=2$, it is necessary that in general
\begin{eqnarray}
a+b-c-d=0,\,\,a-b+c-d=0\\
\Leftrightarrow a=d,\,\,b=c.
\end{eqnarray}
Then
\begin{equation}
n_1''=\#\{|2b|, |2a|, |2a|, |2b|\}=4.
\end{equation}

\item $\{n_1,n_1',n_1''\}\neq\{2,4,2\}$ and $\{n_1,n_1',n_1''\}\neq\{2,2,4\}$: 

If $n_1=2$, it is necessary that in general
\begin{equation}
a=0,\,\,b=0,\,\,c\neq0,\,\,d\neq0.
\end{equation}
Then
\begin{equation}
n_1'=n_1''=\#\{|c+d|,|c+d|,|c-d|,|c-d|\}.
\end{equation}

\end{enumerate}

\end{proof}

\chapter{Role of entanglement and causal relation in DQC}
\section{Entanglement for one-way LOCC}
\textbf{Theorem \ref{theorem local discrimination}.} 
{\it 
For any orthonormal basis $\left\{ \bigr|\psi_{j}\bigr\rangle_{AB}\right\} _{j=1}^{d_{A}d_{B}}\subset\mathcal{H}_{A}\otimes\mathcal{H}_{B}$,
\[
r_{min}\left(\left\{ \bigr|\psi_{j}\bigr\rangle_{AB}\right\} _{j=1}^{d_{A}d_{B}}\right)=d_{min}\left(\left\{ \bigr|\psi_{j}\bigr\rangle_{AB}\right\} _{j=1}^{d_{A}d_{B}}\right),
\]
 where $d_{min}\left(\left\{ \bigr|\psi_{j}\bigr\rangle_{AB}\right\} _{j=1}^{d_{A}d_{B}}\right)$
is defined by
\[
d_{min}\left(\left\{ \bigr|\psi_{j}\bigr\rangle_{AB}\right\} _{j=1}^{d_{A}d_{B}}\right):=\min_{\mathcal{H}_{A}=\bigoplus_{k}\mathcal{M}_{k}}\max_{k}\left\{ \dim\mathcal{M}_{k}\ \Big|\ \forall j,\exists k,\ \mbox{s.t.}\ \bigr|\psi_{j}\bigr\rangle_{AB}\in\mathcal{M}_{k}\otimes\mathcal{H}_{B}\right\} .
\]
}

\begin{proof}
First, we prove the inequality

\begin{equation}
r_{min}\left(\left\{ \bigr|\psi_{j}\bigr\rangle_{AB}\right\} _{j=1}^{d_{A}d_{B}}\right)\le d_{min}\left(\left\{ \bigr|\psi_{j}\bigr\rangle_{AB}\right\} _{j=1}^{d_{A}d_{B}}\right).\label{eq:r_min le d_min}
\end{equation}
 Suppose there exists an orthogonal decomposition $\mathcal{H}_{A}=\bigoplus_{k=1}^{K}\mathcal{M}_{k}$
such that for all $j$, there exists $k$ satisfying $\bigr|\psi_{j}\bigr\rangle_{AB}\in\mathcal{M}_{k}\otimes\mathcal{H}_{B}$.
Then, the following protocol discriminates any basis states in $\left\{ \bigr|\psi_{j}\bigr\rangle_{AB}\right\} _{j=1}^{d_{A}d_{B}}$
perfectly: 
\begin{enumerate}
\item Alice applies a measurement $\left\{ P_{k}\right\} _{k=1}^{K}$
on her system ($\mathcal{H}_{AA'}$), and sends outcome $k$ to
Bob who has system $\mathcal{H}_{B}$, where $P_{k}$ is the orthogonal
projector onto the subspace $\mathcal{M}_{k}$. 
\item Alice teleports the subspace $\mathcal{M}_{k}$ to Bob by using a
maximally entangled state $\bigr|\Phi\bigr\rangle_{A'B'}\subset\mathcal{H}_{A'}\otimes\mathcal{H}_{B'}$
whose Schmidt rank is given by ${\rm Sch}\#_{B'}^{A'}\left(\bigr|\Phi\bigr\rangle_{A'B'}\right)=\max_{k}\left\{ \dim\mathcal{M}_{k}\ \Big|\ \forall j,\exists k,\ \mbox{s.t.}\ \bigr|\psi_{j}\bigr\rangle_{AB}\in\mathcal{M}_{k}\otimes\mathcal{H}_{B}\right\} $. 
\item Bob discriminates the states $\bigr|\psi_{j}\bigr\rangle_{B'B}\in\mathcal{M}_{k}\otimes\mathcal{H}_{B}$
which is now on his subspace. 
\end{enumerate}
The existence of this protocol guarantees the inequality given by \eqref{eq:r_min le d_min}.

Second, we prove the inequality

\begin{equation}
r_{min}\left(\left\{ \bigr|\psi_{j}\bigr\rangle_{AB}\right\} _{j=1}^{d_{A}d_{B}}\right)\ge d_{min}\left(\left\{ \bigr|\psi_{j}\bigr\rangle_{AB}\right\} _{j=1}^{d_{A}d_{B}}\right).\label{eq:r_min ge d_min}
\end{equation}
Suppose $\bigr|\Phi\bigr\rangle_{A'B'}$ attains the optimum defined by Eq.~\eqref{eq:def_r_min}.
That is, a set of states $\left\{ \bigr|\psi_{j}\bigr\rangle_{AB}\otimes\bigr|\Phi\bigr\rangle_{A'B'}\right\} _{j=1}^{d_{A}d_{B}}\subset\mathcal{H}_{AA'}\otimes\mathcal{H}_{BB'}$
can be perfectly discriminated by one-way LOCC, and 
\begin{equation}
r_{min}\left(\left\{ \bigr|\psi_{j}\bigr\rangle_{AB}\right\} _{j=1}^{d_{A}d_{B}}\right)={\rm Sch}\#_{B'}^{A'}\left(\bigr|\Phi\bigr\rangle_{A'B'}\right).\label{eq:r_min=00003Dr_Sch}
\end{equation}
 Then, without loss of generality, we assume that all Alice's POVM
elements are rank one. That is, Alice first applies the POVM represented by $\left\{\bigr|\xi_{x}\bigr\rangle\bigl\langle\xi_{x}\bigr|_{AA'}\right\}$
on her system $\mathcal{H}_{AA'}$, where $\bigr|\xi_{x}\bigr\rangle_{AA'}$
is an unnormalized state on $\mathcal{H}_{AA'}$. We define an unnormalized
state $\bigr|\tilde{\psi}_{jx}\bigr\rangle_{AA'BB'}\in\mathcal{H}_{AA'}\otimes\mathcal{H}_{BB'}$
as a unnormalized state on the total system when the initial unknown
state is denoted by $\bigr|\psi_{j}\bigr\rangle_{AB}$ , and Alice applies the
POVM and obtains outcome $x$, namely,
\begin{eqnarray}
\bigr|\tilde{\psi}_{jx}\bigr\rangle_{A,B,A',B'} & := & (\bigr|\xi_{x}\bigr\rangle\bigl\langle\xi_{x}\bigr|_{AA'}\otimes I_{BB'})\bigr|\psi_{j}\bigr\rangle_{A,B}\otimes\bigr|\Phi\bigr\rangle_{A'B'}
\end{eqnarray}
Since Bob can discriminate any state from the set of states, $\bigl\langle\tilde{\psi}_{jx}|\tilde{\psi}_{j^{\prime}x}\bigr\rangle=\mu_{j}\delta_{jj'}$
holds for all $j$ and $j'$, where $\mu_{j}\ge0$ for all $j$. This can be rewritten as 
\begin{eqnarray}
\bigl\langle\psi_{j}\bigr|S_{x}\otimes Id_{B}\bigl|\psi_{j^{\prime}}\bigr\rangle & = & \mu_{j}\delta_{jj'},\label{eq:rel_1}
\end{eqnarray}
by introducing a positive operator $S_{x}$
on $\mathcal{H}_{A}$ defined as

\begin{eqnarray}
S_{x} & := & \sum_{k=1}^{\min\left(d_{A'},d_{B'}\right)}\lambda_{k}^{2}\bigl\langle e_{k}\bigr|_{A'}\left(\bigl|\xi_{x}\bigr\rangle\bigl\langle\xi_{x}\bigr|_{AA'}\right)\bigl|e_{k}\bigr\rangle_{A'},\label{eq:def_S_x}
\end{eqnarray}
where $\left\{ \lambda_{k}\right\} _{k=1}^{\min\left(d_{A'},d_{B'}\right)}$
and $\left\{ \bigr|e_{k}\bigr\rangle_{A'}\otimes\bigr|f_{k'}\bigr\rangle_{B'}\right\} _{k=1,\ k'=1}^{d_{A'},\quad d_{B'}}$
are the Schmidt coefficients and the Schmidt basis of $\bigr|\Phi\bigr\rangle_{A'B'}$,
respectively, thus $\bigr|\Phi\bigr\rangle_{A'B'}=\sum_{k}\lambda_{k}\bigr|e_{k}\bigr\rangle_{A'}\otimes\bigr|f_{k}\bigr\rangle_{B'}$.
Eq.~\eqref{eq:def_S_x} clearly shows
\begin{equation}
{\rm Sch}\#_{B'}^{A'}\left(\bigr|\Phi\bigr\rangle_{A'B'}\right)\ge\max_{x}\ \mbox{rank}\ S_{x}.\label{eq:r_Sch_phi_ge_max_rank_S}
\end{equation}
On the other hand, by using Eq.~\eqref{eq:rel_1}, we derive
\begin{eqnarray}
\sum_{j'}\ket{\psi_{j'}}\bra{\psi_{j'}}S_{x}\otimes I_{B}\bigl|\psi_{j}\bigr\rangle&=&\sum_{j'}\mu_j\delta_{jj'}\ket{\psi_{j'}}\\
\Leftrightarrow S_{x}\otimes I_{B}\bigl|\psi_{j}\bigr\rangle & = & \mu_{j}\bigl|\psi_{j}\bigr\rangle.\label{eq:rel_2}
\end{eqnarray}
This equation implies that $\left\{ \bigl|\psi_{j}\bigr\rangle\right\} _{j=1}^{d_{A}d_{B}}$
is a simultaneous eigenbasis for a set of positive operators $\left\{ S_{x}\otimes I_{B}\right\} _{x}$
on $\mathcal{H}_{A}\otimes\mathcal{H}_{B}$. In other words, a set
$\left\{ S_{x}\otimes I_{B}\right\} _{x}$ is commutative, i.e. $\left[S_{x}\otimes I_{B},S_{x'}\otimes I_{B}\right]=0$.
Thus a set $\left\{ S_{x}\right\} _{x}$ is commutative as well,
$\left[S_{x},S_{x'}\right]=0$. Hence there exists
a set of projectors $\left\{ P_{l}^{\left(x\right)}\right\} _{l,x}$
satisfying 
\begin{equation}
\forall x,\forall x',\forall l,\forall l',\left[P_{l}^{\left(x\right)},P_{l'}^{\left(x'\right)}\right]=0\label{eq:commutative_P_l^x}
\end{equation}
 such that for all $x$,
\[
S_{x}=\sum_{l}\theta_{l}^{\left(x\right)}P_{l}^{\left(x\right)}
\]
 gives a spectral decomposition of $S_{x}$, and $\theta_{l}^{\left(x\right)}\neq\theta_{l'}^{\left(x\right)}$
for $l\neq l'$. That is, for all $x$,
\begin{equation}
S_{x}\otimes I_{B}=\sum_{l}\theta_{l}^{\left(x\right)}P_{l}^{\left(x\right)}\otimes I_{B}\label{eq:spectral_decomp_S_x_I}
\end{equation}
gives a spectral decomposition of $S_{x}$. Eqs.~\eqref{eq:rel_2},
\eqref{eq:commutative_P_l^x}, and \eqref{eq:spectral_decomp_S_x_I}
lead to the existence of a set of projectors $\left\{ Q_{h}\right\} _{h}$
on $\mathcal{H}_{A}$ satisfying $Q_{h}Q_{h'}=\delta_{hh'}Q_{h}$,
$\sum_{h}Q_{h}=I_{A}$, and
\[
\forall h,\forall j,\ Q_{h}\otimes I_{B}\bigl|\psi_{j}\bigr\rangle=0\ {\rm or}\ \bigl|\psi_{j}\bigr\rangle,
\]
 and there exists a set of non-negative numbers $\left\{ \theta{}_{h}^{'\left(x\right)}\right\} _{h,x}$
such that 
\begin{eqnarray}
\forall x,\ S_{x}=\sum_{h}\theta{}_{h}^{'\left(x\right)}Q_{h},\nonumber\\
\forall h,\ \exists x,\ \theta{}_{h}^{'\left(x\right)}>0 \label{eq:spectral_decomp_S_2}
\end{eqnarray}
 Therefore, by defining a subspace $\mathcal{M}_{h}$ as the support
of $Q_{h}$, $\left\{ \mathcal{M}_{h}\right\} _{h}$ satisfies $\mathcal{H}_{A}=\bigoplus_{h}\mathcal{M}_{h}$
and for all $j$, there exists $h$ such that $\bigl|\psi_{j}\bigr\rangle\in\mathcal{M}_{h}\otimes\mathcal{H}_{B}$.
Hence, we derive
\begin{equation}
\max_{x}\mbox{rank}S_{x}\ge\max_{h}\mbox{rank}Q_{h}=\max_{h}\dim\mathcal{M}_{h}\ge d_{min}\left(\left\{ \bigr|\psi_{j}\bigr\rangle_{AB}\right\} _{j=1}^{d_{A}d_{B}}\right),\label{eq:max_rank_S_x_ge_max_rank_Q_h}
\end{equation}
 where Eq.~\eqref{eq:spectral_decomp_S_2} leads to the first inequality,
the definition of $\mathcal{M}_{h}$ implies the equality, and the
definition of $d_{min}\left(\left\{ \bigr|\psi_{j}\bigr\rangle_{AB}\right\} _{j=1}^{d_{A}d_{B}}\right)$
guarantees the last inequality. Finally, Eqs.~\eqref{eq:r_min=00003Dr_Sch},
\eqref{eq:r_Sch_phi_ge_max_rank_S}, and \eqref{eq:max_rank_S_x_ge_max_rank_Q_h}
implies the inequality given by \eqref{eq:r_min ge d_min}.

Now, we have proven inequalities \eqref{eq:r_min le d_min} and \eqref{eq:r_min ge d_min}.
This completes the proof of the theorem. 
\end{proof}

\section{LOCC and causal order}
We analyze the relationship between LOCC and special relativistic causal order of local operations performed in the spacetime by generalizing a LOCC scenario in a special relativistic framework.   In this scenario, suppose Alice and Bob reside in separate laboratories in spacetime obeying special relativity.  They perform local operations within each laboratory for $N$ times.  Each local operation is assumed to be implemented in a very short time so that it is performed at a single spacetime coordinate.    We assume that the order of local operations in each laboratory is fixed.   However, we do not require to fix the total ordering of local operations across the two parties in contrast to the standard LOCC scenario.

We denote the $k$-th local operations of Alice and Bob are represented by $\mathcal{A}_{o_k|i_k}^{(k)}$ and $\mathcal{B}_{o'_k|i'_k}^{(k)}$.   In this generalized scenario, Alice (Bob) receives classical input $i_k$ ($i'_k$) from outside her (his) laboratory and performs $\mathcal{A}_{o_k|i_k}^{(k)}$ ($\mathcal{B}_{o'_k|i'_k}^{(k)}$) and then sends classical output $o_k$ ($o'_k$) to outside her (his) laboratory.  The deterministic joint quantum operation in this case can be written in the form of Eq.(\ref{eq:LOCC*0}) by means of  $p\left(i_1,\cdots , i'_N|o_1, \cdots, o'_N \right)$  linking the classical outputs and the classical inputs of the local operations. 

The assumption of special relativity adds two conditions for linking two local operations.  An ordering denoted by ``$\prec$'' over local operations associated by the spacetime coordinates is introduced to represent these conditions.   We define $a\prec b$ if the spacetime coordinate of operation $a$ is in the past light cone of the spacetime coordinate of operation $b$.   If neither $a \prec b$ nor $b \prec a$ holds then the spacetime coordinates of $a$ and $b$ are {\it spacelike} separated.  Special relativity guarantees that this ordering is a partial ordering.  Thus, there exists a partial ordering ``$\prec$'' over a set of local quantum operations.  Since local operations within each laboratory are totally ordered by the assumption, the partial ordering must satisfy $\mathcal{A}_{o_1|i_1}^{(1)} \prec \cdots \prec \mathcal{A}_{o_N|i_N}^{(N)}$ and $\mathcal{B}_{o'_1|i'_1}^{(1)} \prec \cdots \prec \mathcal{B}_{o'_N|i'_N}^{(N)}$. Further, $p\left(i_1,\cdots , i'_N|o_1, \cdots, o'_N \right)$ must satisfy the no-signaling condition.  That is, for all local operations, an output of a local operation cannot depend on the input of another local operation performed not in the past.   For example, if $\mathcal{B}_{o'_l|i'_l}^{(l)} 
\prec \mathcal{A}_{o_k|i_k}^{(k)} \prec \mathcal{B}_{o'_m|i'_m}^{(m)}$, then 
\begin{eqnarray}
&p\left(i_1,\cdots,i_k,i'_1,\cdots, i'_l|o_1,\cdots,o_N,o'_1,\cdots,o'_N\right)\nonumber\\
 =&p\left(i_1,\cdots,i_k,i'_1,\cdots, i'_l|o_1,\cdots,o_{k-1},o'_1,\cdots,o'_{m-1}\right)
 \end{eqnarray}
  should be satisfied.

We can show that deterministic joint quantum operations is in LOCC if and only if it has a decomposition in Eq.(4), and there exists a partial ordering over local quantum operations and $p\left(i_1,\cdots , i'_N|o_1, \cdots, o'_N \right)$ in Eq.(\ref{eq:LOCC*0}) satisfies the no-signaling condition with respect to the partial ordering.  The proof of this statement is slightly involved and given in Appendix B.3.   This property of deterministic joint quantum operations in the form of Eq.(\ref{eq:LOCC*0}) suggests that LOCC is a set of all possible deterministic joint quantum operations implementable by local operations and {\it classical communication respecting causal order}.

\section{The rigorous proof of B.2}
In Results, we presented a statement that any deterministic joint quantum operation in LOCC can be written in the form of Eq.(\ref{eq:LOCC*0}) but not all quantum operations written as 
Eq.(\ref{eq:LOCC*0}) are LOCC.  In this section, we show a necessary 
and sufficient condition for a deterministic joint quantum operation in the form of Eq.(\ref{eq:LOCC*0}) to be in LOCC.  We show that the condition is related to special relativistic {\it causal order} of local operations.

We regard $p\left(i_1, \cdots,  i'_N|o_1, \cdots, o'_N\right)$ in Eq.(\ref{eq:LOCC*0}) linking the outputs and inputs of local operations as classical channel between $2 N$ inputs $\{o_k\}_{k=1}^N \cup \{o'_k\}_{k=1}^N$  and $2 N$ outputs $\{i_k\}_{k=1}^N \cup \{i'_k\}_{k=1}^N$. Note that an input $i_k$ for a local operation $\mathcal{A}^{(k)}_{o_k|i_k}$ is an output of the classical channel.   Suppose there exists a  strict partial ordering ``$\prec$'' on the set 
of  local operations $\{\mathcal{A}^{(k)}_{o_k|i_k}\}_{k=1}^N \cup 
\{\mathcal{B}^{(k)}_{o_k|i_k}\}_{k=1}^N$, where  ``strict'' means that 
neither $\mathcal{A}^{(k)}_{o_k|i_k}\prec 
\mathcal{A}^{(k)}_{o_k|i_k}$ nor $\mathcal{B}^{(k)}_{o_k|i_k} \prec 
\mathcal{B}^{(k)}_{o_k|i_k}$
holds, and $\{\mathcal{A}^{(k)}_{o_k|i_k}\}_{k=1}^N$ and $\{\mathcal{B}^{(k)}_{o_k|i_k}\}_{k=1}^N$ are local operations performed in the laboratories of Alice and Bob, respectively.  Each local operation is assumed to be implemented in a very short time so that it is performed at a single spacetime coordinate and is associated with the spacetime coordinate.  We assume the order of local operations in each laboratory is fixed and local operations are performed in the increasing order of $k$, which is the assumption of local temporal ordering we have introduced in Appendix B.2.

Suppose that this partial ordering $\prec$ represents causal order of the spacetime coordinates of local operations.  For example, $\mathcal{A}^{(k)}_{o_k|i_k} \prec \mathcal{B}^{(l)}_{o'_l|i'_l}$ 
means that $\mathcal{B}^{(l)}_{o'_l|i'_l}$ is performed after $\mathcal{A}^{(k)}_{o_k|i_k}$ had been performed.  Neither $\mathcal{A}^{(k)}_{o_k|i_k} \prec \mathcal{B}^{(l)}_{o'_l|i'_l}$ nor $ \mathcal{B}^{(l)}_{o'_l|i'_l}  \prec \mathcal{A}^{(k)}_{o_k|i_k} $ holds if the spacetime coordinates of  $\mathcal{A}^{(k)}_{o_k|i_k}$ and $\mathcal{B}^{(l)}_{o'_l|i'_l}$ are spacelike separated.  Due to the assumption of local temporal ordering, $\mathcal{A}^{(k)}_{o_{k+1}|i_{k+1}}$ is performed after $\mathcal{A}^{(k)}_{o_{k}|i_{k}}$, thus ``$\prec$''  satisfies  $\mathcal{A}^{(k)}_{o_{k}|i_{k}} \prec \mathcal{A}^{(k)}_{o_{k+1}|i_{k+1}}$ for all $k=1, \cdots ,N-1$. Similarly, it satisfies $\mathcal{B}^{(k)}_{o'_{k}|i'_{k}} \prec \mathcal{B}^{(k)}_{o'_{k+1}|i'_{k+1}}$ for all $k=1, \cdots ,N-1$. 

We define a set ${past}(\mathfrak{I})$ representing a set of all ``{\it past}'' 
inputs for a set of outputs $\mathfrak{I}$, where  $\mathfrak{I}$ is a subset of 
$\{i_k\}_{k=1}^N \cup \{i'_k\}_{k=1}^N$.   Formally, by introducing a set of local 
operations $Op\left(\mathfrak{I}\right)$ corresponding to $\mathfrak{I}$ 
as 
\begin{equation}
 Op\left(\mathfrak{I}\right):=\left\{ \mathcal{A}^{(k)}_{o_{k}|i_{k}} 
| i_k \in \mathfrak{I}\right\} \cup \left\{ \mathcal{B}^{(k)}_{o'_{k}|i'_{k}} 
| i'_k \in \mathfrak{I}\right\},
\end{equation}
${past}(\mathfrak{I})$ 
is defined as 
\begin{equation}
 past\left(\mathfrak{I}\right):=\{o_k | \exists \chi \in 
 Op\left(\mathfrak{I}\right), \mathcal{A}^{(k)}_{o_{k}|i_{k}} \prec 
 \chi \} \cup \{o'_k | \exists \chi \in 
 Op\left(\mathfrak{I}\right), \mathcal{B}^{(k)}_{o'_{k}|i'_{k}} \prec \chi \},
\end{equation}
where $\chi$ represent any element of $Op\left(\mathfrak{I}\right)$, that is, there exists 
 $l$ such that $\chi=\mathcal{A}^{(l)}_{o_{l}|i_{l}}$ 
or $\chi=\mathcal{B}^{(l)}_{o'_{l}|i'_{l}}$. Note that  $o_k$ is not 
in $past \left(\{i_k\}\right)$ by definition.

Next we define a classical channel respecting the causal order.  A classical channel $p\left(i_1, \cdots,  i'_N|o_1, \cdots, o'_N\right)$ linking the classical outputs and inputs of local operations  $\{\mathcal{A}^{(k)}_{o_k|i_k}\}_{k=1}^N \cup \{\mathcal{B}^{(k)}_{o_k|i_k}\}_{k=1}^N$  is said to be respecting causal order if there exists a partial ordering $\prec$ {  on the set of 
local operations $\{\mathcal{A}^{(k)}_{o_k|i_k}\}_{k=1}^N \cup \{\mathcal{B}^{(k)}_{o_k|i_k}\}_{k=1}^N$} satisfying the following conditions:
\begin{itemize}
 \item For all $k$, {  $\mathcal{A}^{(k)}_{o_{k}|i_{k}} \prec 
 \mathcal{A}^{(k)}_{o_{k+1}|i_{k+1}}$ and  $\mathcal{B}^{(k)}_{o'_{k}|i'_{k}} \prec \mathcal{B}^{(k)}_{o'_{k+1}|i'_{k+1}}$}.
\item For all $k$ and $l$, 
\begin{eqnarray}\label{App Eq p temporal ordering}
 &p\left(i_1, \cdots, i_k, i'_1, \cdots i'_l|o_1, 
\cdots, o_N, o'_1, \cdots, o'_N\right)\nonumber\\
\label{App Eq p temporal ordering}
= &p \left(i_1, \cdots, i_k, i'_1, 
\cdots i'_l| past \left ( \{ i_a \}_{ a=1 }^k \cup \{ i'_b \}_{ b=1 }^l  \right ) \right).
\end{eqnarray}

\end{itemize} 
Since the classical channel is represented by a conditional probability distribution,
\begin{eqnarray}
 &p\left(i_1, \cdots, i_k, i'_1, \cdots i'_l|o_1, 
\cdots, o_N, o'_1, \cdots, o'_N\right)\nonumber\\
\label{App Eq definition of p}
=& \sum _{i_{k+1},\cdots i_N,i'_{l+1} 
\cdots i_N}
p\left(i_1, \cdots, i_N, i'_1, \cdots i'_N|o_1,
\cdots, o_N, o'_1, \cdots, o'_N\right)
\end{eqnarray}
holds. Hense Eq.(\ref{App Eq p temporal ordering}) can be regarded to represent the no-signaling condition among multiple spacetime coordinates, namely, outputs of the classical channel never depend on the inputs that are not in the past of the outputs. In other words, classical communication represented by a classical channel respects causal order if and only if there exists a partial ordering among local operations and the classical channel satisfies the no-signaling condition with respect to the partial ordering.

We present the main proposition on the relationship between LOCC and causal order.
\begin{proposition}
A deterministic joint quantum operation is in LOCC, if and only if it has a decomposition in 
the form of Eq.(\ref{eq:LOCC*0}) with $p\left(i_1,  \cdots,  i'_N|o_1, \cdots, o'_N\right)$ respecting the causal order.
\end{proposition}
From the proposition, we immediately derive the following corollary about LOCC*.
\begin{corollary}
 A deterministic joint quantum operation is in LOCC* and not in LOCC, if and only if 
 it has a decomposition in the form of Eq.(\ref{eq:LOCC*0}) with $p\left(i_1, \cdots,  i'_N|o_1, \cdots, o'_N\right)$ and for all such decompositions,  
$p\left(i_1,  \cdots,  i'_N|o_1,  \cdots, o'_N\right)$ does not respect causal order.
\end{corollary}

Since the proposition is a necessary and sufficient condition for LOCC, it gives a new characterization of LOCC in terms of $p\left(i_1, \cdots,  i'_N|o_1, \cdots,  o'_N\right)$ and similarly, the corollary gives a new characterization of a deterministic joint quantum operation in LOCC* and not in LOCC, which is nothing but a non-LOCC separable quantum operation as we have shown in Results,  in terms of causal order.  

\vspace{0.5cm}
\noindent{{\bf Proof}}

We provide a proof of the main proposition in the remaining part of this section.
Since the conventional definition of LOCC in Eq.(\ref{eq:LOCC2}) using the totally ordered local operations immediately gives a decomposition with $p\left(i_1, \cdots,  i'_N|o_1, \cdots, o'_N\right)$ respecting causal order, the ``only if'' part is trivial.  Hence we concentrate on proving the ``if'' part of the proposition. 

This proof consists of two parts.   In Part A, we show that for given local operations $\{\mathcal{A}^{(k)}_{o_k|i_k}\}_{k=1}^N \cup \{\mathcal{B}^{(k)}_{o_k|i_k}\}_{k=1}^N$ and classical channel $p\left(i_1, \cdots,  i'_N|o_1, \cdots, o'_N\right)$ representing a deterministic joint quantum operation by Eq.~(\ref{eq:LOCC*0}), it is always possible to construct sets of {\it totally} ordered local operations $\left\{\mathcal{C}^{(k)}_{\mathbb{O}_{f(k,0)}|\mathbb{J}_{f(k,0)}}\right\}_{k=1}^N$ and $\left\{\mathcal{D}^{(k)}_{\mathbb{O}_{f(k,1)}|\mathbb{J}_{f(k,1)}}\right\}_{k=1}^N$, where an input $\mathbb{J}_l$  of a local operation depends only on an output $\mathbb{O}_{l-1}$ of the previous local operation, representing the deterministic joint quantum operation by choosing an appropriate classical channel.     In Part B, we show that another pair of sets of totally ordered local operations $\left\{\mathcal{C'}^{(k)}_{\mathbb{O}_{f(k,0)}|\mathbb{J}_{f(k,0)}}\right\}_{k=1}^N$ and $\left\{\mathcal{D'}^{(k)}_{\mathbb{O}_{f(k,1)}|\mathbb{J}_{f(k,1)}}\right\}_{k=1}^N$ from the local operations and the classical channel constructed in Part A.  The standard form of LOCC given by Eq.~(\ref{eq:LOCC2}) is derived by using  this pair of sets of totally ordered local operations.

\subsubsection*{Part A}

Since there always exists a total ordering which preserves the structure of a given partial ordering on a finite set,
we define a map $f$ satisfying 
\begin{itemize}
 \item $f$ is a bijection from $\{1,\cdots, N\} \times \{0,1\}$ 
to $\{1,\cdots, 2N \}$.
\item $f(k,b)<f(k+1,b)$ for all $k$ and $b$. 
\item {  $\mathcal{A}^{(k)}_{o_k|i_k} \prec \mathcal{B}^{(l)}_{o'_l|i'_l} 
\Rightarrow f(k,0) < f(l,1)$, and $\mathcal{B}^{(k)}_{o'_k|i'_k} \prec \mathcal{A}^{(l)}_{o_l|i_l} 
       \Rightarrow f(k,1) < f(l,0)$} for all $k$ and $l$.
\end{itemize}
By means of $f$, we define a set of classical inputs $\{I_l\}_{l=1}^{2N}$ and a set of outputs $\{O_l\}_{l=1}^{2N}$ of which elements are given by
\begin{equation}\label{app eq def I O}
I_{f(k,0)}=i_k, \ I_{f(k,1)}=i'_k, \  O_{f(k,0)}=o_k, \  
 \mbox{and } 
       O_{f(k,1)}=o'_k 
\end{equation}
for all $k$. We further define a conditional probability distribution $q(I_1,\cdots 
I_{2N}|O_1,\cdots O_{2N})$ by 
\begin{align}\label{app eq def q}
 q(I_1,\cdots, 
I_{2N}|O_1,\cdots, O_{2N}):=
 p(i_1,\cdots, i_N, i'_1, \cdots i'_N|o_1,\cdots, o_N, o'_1, \cdots, o'_N),
\end{align}
where $I_l$ and $O_l$ are related to $i_k$, $i'_k$, $o_k$, and $o'_k$ by 
Eq. (\ref{app eq def I O}). Thus $q( \cdots | \cdots)$ is another representation of $p(\cdots |\cdots )$ in terms of
the newly defined classical inputs and outputs $\{I_l \}_{l=1}^{2N}$ and $\{O_l\}_{l=1}^{2N}$. Using $q(I_1,\cdots  I_{2N}|O_1,\cdots O_{2N})$,  the condition for $p\left(i_1, \cdots,  i'_N|o_1, \cdots, o'_N\right)$  to respect causal order given by Eq.(\ref{App Eq p temporal ordering}) for all $k$ and $l$ is expressed by
\begin{equation}\label{app eq q temporal order}
 q(I_1,\cdots, 
I_{k}|O_1,\cdots, O_{2N}) = q(I_1,\cdots, 
I_{k}|O_1,\cdots, O_{k-1})
\end{equation}
for all $k$. 

Next, we define sets of all possible values of $I_k$ and $O_k$ denoted by $\mathcal{I}_k$ and $\mathcal{O}_k$, respectively,  for $k=1, \cdots 2N$.  We represent variables whose variations are over $\prod_{j=1}^{k}\mathcal{I}_j \times \prod_{j=1}^{k-1}\mathcal{O}_j $ and $\prod_{j=1}^{k}\mathcal{I}_j \times \prod_{j=1}^{k}\mathcal{O}_j $ 
by $\mathbb{J}_k$ and $\mathbb{O}_k$, respectively.    
$\{ \mathbb{J}_k\}_{k=1}^{2N}$ and  $\{ \mathbb{O}_k \}_{k=1}^{2N}$ are 
classical inputs and outputs of new local operations which will be introduced latter.  We define a function $Q(\mathbb{J}_k|\mathbb{O}_{k-1})$ representing a classical channel linking between $\mathbb{J}_k$ and $\mathbb{O}_{k-1}$ as
\begin{align}
 Q(\mathbb{J}_k|\mathbb{O}_{k-1}) := &
 \frac{q(I_1,\cdots,I_k|O_1,\cdots,O_{k-1})}{q(I_1,\cdots,I_{k-1}|O_1,\cdots, O_{k-2})} 
 \cdot \delta \left(\mathbb{J}_k[1,k-1],  
 \mathbb{O}_{k-1}[1,k-1]\right) \nonumber \\
 &\cdot \delta \left(\mathbb{J}_k[k+1,2k-1], 
 \mathbb{O}_{k-1}[k,2k-2]\right), \quad (2 \le k \le 2N) \label{app eq Q 1}\\
  Q(\mathbb{J}_1) :=& q(I_1), \label{app eq Q 2}
\end{align}
where $\delta(x,y)$ is the Kronecker delta, $\mathbb{J}_k=(I_1,\cdots I_{k}, O_1, \cdots, O_{k-1})$ for $2 
\le k \le 2N$, $\mathbb{J}_1=I_1$, and $\mathbb{J}_k[l,m]$ is a vector 
consisting of the partial entries of $\mathbb{J}_k$, from the $l$-th entry through the $m$-th entry for $l<m$.

It is easy to check that $Q(\mathbb{J}_k|\mathbb{O}_{k-1})$ satisfies
\begin{align}
&\sum _{\mathbb{J}_1,\cdots, \mathbb{J}_{2N}} Q(\mathbb{J}_1) \cdot \delta(\mathbb{J}_1,\mathbb{O}_1[1])\prod 
 _{k=2}^{2N}Q(\mathbb{J}_k|\mathbb{O}_{k-1})\cdot \delta 
 (\mathbb{J}_k,\mathbb{O}_k[1,2k-1])  \nonumber \\
=& q(I_1^{(2N)},\cdots I_{2N}^{(2N)}|O_1^{(2N)}, \cdots, O_{2N}^{(2N)}) 
 \cdot 
 \nonumber \\
& \qquad  \prod _{k=2}^{2N}\delta \left( (I_1^{(k)},\cdots, I_{k-1}^{(k)}, 
 O_1^{(k)},\cdots, 
 O_{k-1}^{(k)}),(I_1^{(k-1)},\cdots, I_{k-1}^{(k-1)}, O_1^{(k-1)},\cdots, 
 O_{k-1}^{(k-1)}) \right ),
\label{eqn:Qrelation}
\end{align}
where $I_l^{(k)}$ and $O_l^{(k)}$ are given by $\mathbb{O}_k=(I_1^{(k)},\cdots, 
I_{k}^{(k)}, O_1^{(k)}, \cdots O_{k}^{(k)})$, and $\mathbb{O}_k[l]$ denotes the $l$-th entry of 
$\mathbb{O}_k$.  Eq.~(\ref{eqn:Qrelation}) indicates that $Q(\mathbb{J}_1) \prod 
 _{k=2}^{2N}Q(\mathbb{J}_k|\mathbb{O}_{k-1})$ and $q(I_1^{(2N)},\cdots I_{2N}^{(2N)}|O_1^{(2N)}, \cdots, O_{2N}^{(2N)})$ represent the same classical channel if $\mathbb{J}_k=\mathbb{O}_k[1,k-1]$, $\mathbb{O}_k[1,k-1]=\mathbb{O}_{k-1}[1,k-1]$, and $\mathbb{O}_k[k+1,2k]=\mathbb{O}_{k-1}[k,2k-2]$ are satisfied.

We define a local operation performed in Alice's laboratory for $\mathbb{J}_{f(k,0)}$ and $\mathbb{O}_{f(k,0)}$
denoted by
$\mathcal{C}^{(k)}_{\mathbb{O}_{f(k,0)}|\mathbb{J}_{f(k,0)}}$ as 
\begin{align}
 \mathcal{C}^{(k)}_{\mathbb{O}_{f(k,0)}|\mathbb{J}_{f(k,0)}} :=\delta \left (\mathbb{J}_{f(k,0)}, 
 \mathbb{O}_{f(k,0)}[1,2f(k,0)-1] \right ) \cdot 
 \mathcal{A}^{(k)}_{\mathbb{O}_{f(k,0)}[2f(k,0)]|\mathbb{J}_{f(k,0)}[f(k,0)]}, \label{app eq def mathcal C 1}
\end{align}
where $\mathbb{O}_{1}[1,1]:=\mathbb{O}_{1}[1]$ for $f(k,0)=1$ and  $\mathbb{J}_{k}[l]$ represents the $l$-th entry  of $\mathbb{J}_{k}$.
Similarly, we define a local operation performed in Bob's laboratory denoted by
$\mathcal{D}^{(k)}_{\mathbb{O}_{f(k,1)}|\mathbb{J}_{f(k,1)}}$ as 
\begin{align}
& \mathcal{D}^{(k)}_{\mathbb{O}_{f(k,1)}|\mathbb{J}_{f(k,1)}} :=\delta \left (\mathbb{J}_{f(k,1)}, 
 \mathbb{O}_{f(k,1)}[1,2f(k,1)-1] \right ) \cdot 
 \mathcal{B}^{(k)}_{\mathbb{O}_{f(k,1)}[2f(k,1)]|\mathbb{J}_{f(k,1)}[f(k,1)]}. \label{app eq def mathcal D 1}
\end{align}

By linking classical outputs and inputs of local operations
$\mathcal{C}^{(k)}_{\mathbb{O}_{f(k,0)}|\mathbb{J}_{f(k,0)}}$ and 
$\mathcal{D}^{(k)}_{\mathbb{O}_{f(k,1)}|\mathbb{J}_{f(k,1)}}$ respecting
the total ordering introduced by  the function $f$, we derive a representation of the deterministic joint quantum operation given by Eq.~(\ref{eq:LOCC*0}) as follows:
\begin{align}
& \sum_{\mathbb{J}_1,\cdots ,\mathbb{J}_{2N},\mathbb{O}_1,\cdots , 
 \mathbb{O}_{2N}}Q(\mathbb{J}_1)\cdot \left(\prod 
 _{k=2}^{2N}Q(\mathbb{J}_k|\mathbb{O}_{k-1}) \right) 
 \mathcal{C}^{(N)}_{\mathbb{O}_{f(N,0)}|\mathbb{J}_{f(N,0)}}\circ \cdots \circ
\mathcal{C}^{(1)}_{\mathbb{O}_{f(1,0)}|\mathbb{J}_{f(1,0)}} \nonumber\\
&\otimes \mathcal{D}^{(N)}_{\mathbb{O}_{f(N,1)}|\mathbb{J}_{f(N,1)}}\circ \cdots \circ
\mathcal{D}^{(1)}_{\mathbb{O}_{f(1,1)}|\mathbb{J}_{f(1,1)}} \nonumber \\
=&\sum_{\mathbb{J}_1,\cdots ,\mathbb{J}_{2N},\mathbb{O}_1,\cdots , 
 \mathbb{O}_{2N}}Q(\mathbb{J}_1) \cdot \delta\left( 
 \mathbb{J}_1|\mathbb{O}_1[1]\right)  \cdot \Big(\prod 
 _{k=2}^{2N}Q(\mathbb{J}_k|\mathbb{O}_{k-1})\cdot  \delta\left(\mathbb{J}_{k},\mathbb{O}_{k}[1,2k-1]\right)\Big)  \nonumber \\
& \qquad \mathcal{A}^{(N)}_{\mathbb{O}_{f(N,0)}[2f(N,0)]|\mathbb{J}_{f(N,0)}[f(N,0)]}\circ \cdots \circ
\mathcal{A}^{(1)}_{\mathbb{O}_{f(1,0)}[2f(1,0)]|\mathbb{J}_{f(1,0)}[f(1,0)]} \nonumber\\
&\qquad \otimes \mathcal{B}^{(N)}_{\mathbb{O}_{f(N,1)}[2f(N,1)]|\mathbb{J}_{f(N,1)}[f(N,1)]}\circ \cdots \circ
\mathcal{B}^{(1)}_{\mathbb{O}_{f(1,1)}[2f(1,1)]|\mathbb{J}_{f(1,1)}[f(1,1)]} 
\nonumber \\
=&\sum_{\mathbb{O}_1,\cdots , 
 \mathbb{O}_{2N}}
q(I_1^{(2N)},\cdots I_{2N}^{(2N)}|O_1^{(2N)}, \cdots, O_{2N}^{(2N)}) 
 \cdot 
 \nonumber \\
& \qquad  \prod _{k=2}^{2N}\delta \left( (I_1^{(k)},\cdots, I_{k-1}^{(k)}, 
 O_1^{(k)},\cdots, 
 O_{k-1}^{(k)}),(I_1^{(k-1)},\cdots, I_{k-1}^{(k-1)}, O_1^{(k-1)},\cdots, 
 O_{k-1}^{(k-1)}) \right ), 
 \nonumber \\
& \qquad \mathcal{A}^{(N)}_{\mathbb{O}_{f(N,0)}[2f(N,0)]|\mathbb{O}_{f(N,0)}[f(N,0)]}\circ \cdots \circ
\mathcal{A}^{(1)}_{\mathbb{O}_{f(1,0)}[2f(1,0)]|\mathbb{O}_{f(1,0)}[f(1,0)]} \nonumber\\
&\qquad \otimes \mathcal{B}^{(N)}_{\mathbb{O}_{f(N,1)}[2f(N,1)]|\mathbb{O}_{f(N,1)}[f(N,1)]}\circ \cdots \circ
\mathcal{B}^{(1)}_{\mathbb{O}_{f(1,1)}[2f(1,1)]|\mathbb{O}_{f(1,1)}[f(1,1)]} 
\nonumber\\
=&\sum_{\mathbb{O}_{2N}}
q(I_1^{(2N)},\cdots I_{2N}^{(2N)}|O_1^{(2N)}, \cdots, O_{2N}^{(2N)}) 
 \cdot 
 \nonumber \\
& \qquad \mathcal{A}^{(N)}_{\mathbb{O}_{2N}[N+f(N,0)]|\mathbb{O}_{2N}[f(N,0)]}\circ \cdots \circ
\mathcal{A}^{(1)}_{1,\mathbb{O}_{2N}[N+f(1,0)]|\mathbb{O}_{2N}[f(1,0)]} \nonumber\\
&\qquad \otimes \mathcal{B}^{(N)}_{N,\mathbb{O}_{2N}[N+f(N,1)]|\mathbb{O}_{2N}[f(N,1)]}\circ \cdots \circ
\mathcal{B}^{(1)}_{\mathbb{O}_{2N}[N+f(1,1)]|\mathbb{O}_{2N}[f(1,1)]} 
\nonumber \\
=& \sum_{i_1,\cdots,i_N,i'_1,\cdots, 
 i'_N,o_1,\cdots,o_N,o'_1,\cdots,o'_N}p\left(i_1,\cdots,i_N,i'_1,\cdots, 
 i'_N|o_1,\cdots,o_N,o'_1,\cdots,o'_N\right) \nonumber \\
& \qquad \mathcal{A}^{(N)}_{o_N|i_N}\circ \cdots \circ
\mathcal{A}^{(1)}_{o_1|i_1} \otimes \mathcal{B}^{(N)}_{o'_N|i'_N}\circ \cdots \circ
\mathcal{B}^{(1)}_{o'_1|i'_1}, \label{app eq longest}
\end{align}
 where we used Eqs.~(\ref{app eq def mathcal C 1}) and (\ref{app eq def 
 mathcal D 1}) in the first equality, Eqs.~(\ref{app eq Q 1}) and 
 (\ref{app eq Q 2}) in the second equality, Eq.~(\ref{app eq def q}) in 
 the forth equality, respectively.
Note that in the first line of Eq.~(\ref{app eq longest}), 
classical input $\mathbb{J}_k$ of a local operation only depends on classical 
output $\mathbb{O}_{k-1}$  of the previous local operation in $Q(\mathbb{J}_k|\mathbb{O}_{k-1})$.

\subsection*{Part B}
In this part, we define new local operations performed in Alice's and Bob's laboratories denoted by
$\mathcal{C'}^{(k)}_{\mathbb{O}_{f(k,0)}|\mathbb{O}_{f(k,0)-1}}$ and 
$\mathcal{D'}^{(k)}_{\mathbb{O}_{f(k,1)}|\mathbb{O}_{f(k,1)-1}}$, respectively,
and show that the first line of Eq.~(\ref{app eq longest}) can be transformed to
the standard form of LOCC using these local operations. 
We define
$\mathcal{C'}^{(k)}_{\mathbb{O}_{f(k,0)}|\mathbb{O}_{f(k,0)-1}}$ and 
$\mathcal{D'}^{(k)}_{\mathbb{O}_{f(k,1)}|\mathbb{O}_{f(k,1)-1}}$ 
as 
\begin{align}
 \mathcal{C'}^{(k)}_{\mathbb{O}_{f(k,0)}|\mathbb{O}_{f(k,0)-1}}:= 
 \sum_{\mathbb{J}_{f(k,0)}}Q(\mathbb{J}_{f(k,0)}|\mathbb{O}_{f(k,0)-1})
 \mathcal{C}^{(k)}_{\mathbb{O}_{f(k,0)}|\mathbb{J}_{f(k,0)}}, \nonumber \\
 \mathcal{D'}^{(k)}_{\mathbb{O}_{f(k,1)}|\mathbb{O}_{f(k,1)-1}}:= 
 \sum_{\mathbb{J}_{f(k,1)}}Q(\mathbb{J}_{f(k,1)}|\mathbb{O}_{f(k,1)-1})
 \mathcal{C}^{(k)}_{\mathbb{O}_{f(k,1)}|\mathbb{J}_{f(k,1)}}, \label{app eq def mathcal C' D'}
\end{align}
where we have used $Q(\mathbb{J}_{1}|\mathbb{O}_{0}):=Q(\mathbb{J}_{1})$.  By means of Eqs.~(\ref{app eq Q 1}) and (\ref{app eq Q 2}), we obtain a relation
\begin{align}
 &  \sum_{\mathbb{O}_{f(k,0)}}\mathcal{C'}^{(k)}_{\mathbb{O}_{f(k,0)}|\mathbb{O}_{f(k,0)-1}} = \sum_{I_{f(k,0)},O_{f(k,0)}}\frac{q\left(I_1,\cdots , I_{f(k,0)}|O_1, 
 \cdots, O_{f(k,0)-1}\right)}{q\left(I_1,\cdots , I_{f(k,0)-1}|O_1, 
 \cdots, 
 O_{f(k,0)-2}\right)}\mathcal{A}^{(k)}_{O_{f(k,0)}|\left(I_{f(k,0)}\right)},\label{app  eq tp}
\end{align}
where $I_{f(N,0)}$  and $O_{f(N,0)}$ are the $f(N,0)$-th 
and $2f(N,0)$-th entries of $\mathbb{O}_{f(N,0)}$, respectively, and 
$\{I_l \}_{l=1}^{f(k,0)-1}$ and $\{O_l \}_{l=1}^{f(k,0)-1}$ are given 
by
$$\mathbb{O}_{f(k,0)-1}=\left(I_1, \cdots,  I_{f(k,0)-1}|O_1, \cdots, O_{f(k,0)-1} \right).$$
 Eq.~(\ref{app eq q temporal order}) represents the property of $q(I_1,\cdots, I_{2N}|O_1,\cdots, O_{2N})$ respecting causal order and Eq.~(\ref{app eq tp}) guarantees that  
$\{ 
\mathcal{C'}^{(k)}_{\mathbb{O}_{f(k,0)}|\mathbb{O}_{f(k,0)-1}}\}_{\mathbb{O}_{f(k,0)}}
$ is indeed a quantum instrument for all $k$. Similarly,  $\{ 
\mathcal{D'}^{(k)}_{\mathbb{O}_{f(k,1)}|\mathbb{O}_{f(k,1)-1}}\}_{\mathbb{O}_{f(k,1)}}
$ is also shown to be a quantum instrument for all $k$. 

Now the deterministic joint quantum operation given in the form of Eq.~(\ref{eq:LOCC*0}) can be represented in terms of quantum instruments $\{ 
\mathcal{C'}^{(k)}_{\mathbb{O}_{f(k,0)}|\mathbb{O}_{f(k,0)-1}}\}_{\mathbb{O}_{f(k,0)}}
$ and $\{ 
\mathcal{D'}^{(k)}_{\mathbb{O}_{f(k,1)}|\mathbb{O}_{f(k,1)-1}}\}_{\mathbb{O}_{f(k,1)}}
$ as
\begin{align}
 & \sum_{i_1,\cdots,i_N,i'_1,\cdots, 
 i'_N,o_1,\cdots,o_N,o'_1,\cdots,o'_N}p\left(i_1,\cdots,i_N,i'_1,\cdots, 
 i'_N|o_1,\cdots,o_N,o'_1,\cdots,o'_N\right) \nonumber \\
& \qquad \mathcal{A}^{(N)}_{o_N|i_N}\circ \cdots \circ
\mathcal{A}^{(1)}_{o_1|i_1} \otimes \mathcal{B}^{(N)}_{o'_N|i'_N}\circ \cdots \circ
\mathcal{B}^{(1)}_{o'_1|i'_1}, \nonumber \\
=&
\sum_{\mathbb{O}_1,\cdots , \mathbb{O}_{2N}}
 \mathcal{C'}^{(N)}_{\mathbb{O}_{f(N,0)}|\mathbb{O}_{f(N,0)-1}}\circ \cdots \circ
\mathcal{C'}^{(1)}_{\mathbb{O}_{f(1,0)}|\mathbb{O}_{f(1,0)-1}} \otimes \mathcal{D'}^{(N)}_{\mathbb{O}_{f(N,1)}|\mathbb{O}_{f(N,1)-1}}\circ \cdots \circ
\mathcal{D'}^{(1)}_{\mathbb{O}_{f(1,1)}|\mathbb{O}_{f(1,1)-1}}. 
\label{app eq final}
\end{align}
by using Eq.~(\ref{app eq def mathcal C' D'}).  
The right hand side of Eq.~(\ref{app eq final}) is almost 
in the standard form of LOCC.  The only case of Eq.~(\ref{app eq final}) not fitting in LOCC is the case that Alice (Bob) 
successively performs two local operations.   This case is absorbed in LOCC in the following manner.  Suppose Alice 
successively performs two local operations $\mathcal{C'}^{(k-1)}_{\mathbb{O}_{f(k-1,0)}|\mathbb{O}_{f(k-1,0)-1}}$ and $\mathcal{C'}^{(k)}_{\mathbb{O}_{f(k,0)}|\mathbb {O}_{f(k,0)-1}}$ where $f(k,0)=f(k-1,0)+1$.  Since 
\begin{equation}
\left \{\sum_{\mathbb{O}_{f(k-1,0)}}\mathcal{C'}^{(k)}_{\mathbb{O}_{f(k,0)}|\mathbb
{O}_{f(k,0)-1}}\circ 
\mathcal{C'}^{(k-1)}_{\mathbb{O}_{f(k-1,0)}|\mathbb{O}_{f(k-1,0)-1}} \right 
\}_{\mathbb{O}_{f(k,0)}} \nonumber 
\end{equation}
is a quantum instrument representing a local operation performed in Alice's laboratory,  
we regards these successive local operations as a single local 
operation performed in Alice's laboratory.  Similarly, we combine successive local operations performed in Bob's laboratory as a single local operation.  By repeating this procedure, we can rewrite in the standard form of LOCC, where a sequence of local operations are performed alternatively in Alice's and Bob's laboratories.  Hence, we can conclude that Eq.(\ref{app eq final}) reduces to a standard decomposition of LOCC given in the form of Eq.(\ref{eq:LOCC2}).
Therefore,  the ``{\it if}'' part of the proposition is proven.

\section{Formal mathematical formulation of LOCC*}
We denote the Hilbert spaces  of the systems of Alice's quantum input and Bob's quantum input as $\mathcal{H}_X$ and $\mathcal{H}_Y$, respectively, and those of Alice's  quantum output and Bob's quantum output as $\mathcal{H}_A$ and $\mathcal{H}_B$, respectively.   Since LOCC* is defined as a set of CPTP maps $\mathcal{M}: \mathbf{L}(\mathcal{H}_X \otimes \mathcal{H}_Y) \rightarrow \mathbf{L}(\mathcal{H}_A \otimes \mathcal{H}_B) $ satisfying Eq.\eqref{eq:LOCC*},  LOCC*  is represented by a set of CPTP maps in the CJ representation $M\in \mathbf{L}(\mathcal{H}_X\otimes\mathcal{H}_A\otimes\mathcal{H}_Y\otimes \mathcal{H}_B)$ as 
\begin{eqnarray}
M=\sum_{i_A,i_B,o_A,o_B}p(i_A,i_B|o_A,o_B) A_{o_A|i_A}\otimes B_{o_B|i_B},
\label{eq:LOCC*CJ}
\end{eqnarray}
where $i_A$ and $i_B$ are classical inputs of Alice and Bob, respectively, and $o_A$ and $o_B$ are classical outputs of Alice and Bob respectively, and $p(i_A,i_B|o_A,o_B)$ is a conditional probability distribution, $A_{o_A|i_A}\in \mathbf{L}(\mathcal{H}_X\otimes\mathcal{H}_A)$ is the CJ operator of Alice's local operation and $B_{o_B|i_B}\in \mathbf{L}(\mathcal{H}_Y\otimes\mathcal{H}_B)$ is the CJ operator of the Bob's local operation. Note that these CJ operators must satisfy
\begin{eqnarray}
\forall i_A,\forall o_A, A_{o_A|i_A}&\geq& 0,\\
\forall i_A, \mathrm{tr}_{A} \left[\sum_{o_A}A_{o_A|i_A}\right]&=& \mathbb{I}_{X},\\
\forall i_B,\forall o_B, B_{o_B|i_B}&\geq& 0,\\
\forall i_B, \mathrm{tr}_{B}\left[\sum_{o_B}B_{o_B|i_B}\right]&=& \mathbb{I}_{Y}.
\end{eqnarray}

By taking a special conditional probability distribution given by $p(i_A,i_B|o_A,o_B)=\delta_{i_A,o_B}\delta_{i_B,o_A}$ in Eq.\eqref{eq:LOCC*CJ}, the set of CPTP maps $M\in \mathbf{L}(\mathcal{H}_X\otimes \mathcal{H}_A\otimes\mathcal{H}_Y\otimes\mathcal{H}_B)$ satisfying
\begin{eqnarray}
M=\sum_{a,b} A_{a|b}\otimes B_{b|a}
\label{eq:simpleformLOCC*CJ}
\end{eqnarray}
is easily shown to be in a subset of LOCC*.   We shall show the converse that any element of LOCC* can be decomposed into this form.
For a given conditional probability distribution $p(i_A,i_B|o_A,o_B)$, let local operations of Alice and Bob be 
\begin{eqnarray}
A_{i_B,x|o_A,o_B}&:=&\sum_{i_A} p(i_A,i_B|o_A,o_B)A_{x|i_A}\\
B_{o_A,o_B|i_B,x}&:=&\delta_{o_A}^{(x)}B_{o_B|i_B},
\end{eqnarray}
by introducing a new  index $x$.  Since these newly introduced operators $A_{i_B,x|o_A,o_B}$ and $B_{o_A,o_B|i_B,x}$ satisfy
\begin{eqnarray}
\forall o_A,\forall o_B,\,\,\mathrm{tr}_A\left[\sum_{i_B,x}A_{i_B,x|o_A,o_B}\right]=\mathbb{I}_X,\,\,\,\,\,\,\,\,\,\,\,\,
\forall i_B,\forall x,\,\,\mathrm{tr}_B\left[\sum_{o_A,o_B}B_{o_A,o_B|i_B,x}\right]=\mathbb{I}_Y,\nonumber\\
\end{eqnarray}
they are elements of valid local operations.
By using these local operations, we have
\begin{eqnarray}
\sum_{i_B,x,o_A,o_B}A_{i_B,x|o_A,o_B}\otimes B_{o_A,o_B|i_B,x}&=&\sum_{i_A,i_B,o_A,o_B}p(i_A,i_B|o_A,o_B)A_{o_A|i_A}\otimes B_{o_B|i_B}.\nonumber\\
\end{eqnarray}
Further introducing new classical indices  $a:=(i_B, x)$ and $b:=(o_A, o_B)$, we obtain the form of Eq.\eqref{eq:simpleformLOCC*CJ} where
$A_{a|b}\in \mathbf{L}(\mathcal{H}_X\otimes\mathcal{H}_A)$ and $B_{b|a}\in \mathbf{L}(\mathcal{H}_Y\otimes\mathcal{H}_B)$  satisfy
\begin{eqnarray}
\forall b,\forall a,A_{a|b}&\geq& 0\\
\forall a,\mathrm{tr}_A\left[\sum_a A_{a|b}\right]&=&\mathbb{I}_X\\
\forall a,\forall b, B_{b|a}&\geq& 0\\
\forall a, \mathrm{tr}_B\left[\sum_b B_{b|a}\right]&=&\mathbb{I}_Y.
\end{eqnarray}

\section{Equivalence of LOCC* and SEP}
 For a deterministic joint quantum operation $\mathcal{M}: \mathbf{L}(\mathcal{H}_X \otimes \mathcal{H}_Y) \rightarrow \mathbf{L}(\mathcal{H}_A \otimes \mathcal{H}_B) $ in SEP  defined by Eq.\eqref{eq:SEP}, the corresponding CJ  operator $M\in \mathbf{L}(\mathcal{H}_X\otimes \mathcal{H}_A\otimes\mathcal{H}_Y\otimes\mathcal{H}_B)$ is given by
\begin{equation}
M=\sum_{k=1}^L E_k^A\otimes E_k^B,
\label{eq:defSEP}
\end{equation}
where $E_k^A\in \mathbf{L}(\mathcal{H}_X\otimes \mathcal{H}_A)$ and $E_k^B\in \mathbf{L}(\mathcal{H}_Y\otimes \mathcal{H}_B)$ satisfy
\begin{eqnarray}
E_k^A\geq0,\,\,\,\,\,\,\,\,\,\,\,\, E_k^B \geq0,
\label{eq:SEPop}
\end{eqnarray}
for all $k$.    It is easy to see that LOCC* represented by Eq.\eqref{eq:simpleformLOCC*CJ} is a subset of SEP.   We shall show the converse that it is possible to construct local operators $A_{a|b}\in \mathbf{L}(\mathcal{H}_X\otimes\mathcal{H}_A)$ and $B_{b|a}\in \mathbf{L}(\mathcal{H}_Y\otimes\mathcal{H}_B)$ LOCC* for any given SEP element.

First we show that we can restrict SEP operators to
\begin{eqnarray}
\mathrm{tr}_A\left[E_k^A\right]\leq \mathbb{I}_X,\,\,\,\,\,\,\,\,\,\,\,\, \mathrm{tr}_B\left[E_k^B\right]\leq\mathbb{I}_Y,
\label{eq:2}
\end{eqnarray}
for all $k$ in addition to Eq.\eqref{eq:SEPop} without loss of generality.
Since the CJ operator $M$ defined by Eq. \eqref{eq:defSEP} is TP and SEP operators are positive, we obtain
\begin{equation}
\forall k,\mathrm{tr}_{AB}\left[E_k^A\otimes E_k^B\right]=\mathrm{tr}_A\left[E_k^A\right]\otimes \mathrm{tr}_B\left[E_k^B\right]\leq \mathbb{I}_{XY}.
\label{eq:1}
\end{equation}
Since $\mathrm{tr}_A\left[E_k^A\right]$ and $\mathrm{tr}_B\left[E_k^B\right]$ are positive, they are diagonalizable.   Let the eigenvalues of $\mathrm{tr}_A\left[E_k^A\right]$ and $\mathrm{tr}_B\left[E_k^B\right]$ be $\{\lambda^{A}_i\}$ and $\{\lambda^{B}_j\}$, respectively. Eq.\eqref{eq:1} implies
\begin{equation}
\max_i\{\lambda^A_i\}\max_j\{\lambda^B_j\}\leq 1.
\end{equation}
If $\max_i\{\lambda^A_i\}\leq 1$ and $\max_j\{\lambda^B_j\}\leq 1$, Eq.\eqref{eq:2} is satisfied. We assume that $\max_i\{\lambda^A_i\}\geq 1$ and define
\begin{eqnarray}
\tilde{E_k^A}=\frac{1}{\max_i\{\lambda^A_i\}}E_k^A,\,\,\,\,\,\,\,\,\,\,\,\,
\tilde{E_k^B}=\max_i\{\lambda^A_i\}E_k^B.
\end{eqnarray}
Due to $\tilde{E_k^A}\otimes \tilde{E_k^B}=E_k^A\otimes E_k^B$,  SEP operators $E_k^A$ and $E_k^B$ can be replaced by $\tilde{E_k^A}$ and $\tilde{E_k^B}$. Then the validity of Eq.\eqref{eq:2} is certified.

Now we present a construction of local operators for LOCC* as follows.
\begin{eqnarray}
A_{k|k}&=&E_k^A\,\,\,\,\,\,\mathrm{for}\,\,1\leq k\leq L\\
B_{k|k}&=&E_k^B\,\,\,\,\,\,\mathrm{for}\,\,1\leq k\leq L\\
A_{L+1|k}&=&\frac{1}{\dim(\mathcal{H}_A)}\left(\mathbb{I}_X-\mathrm{tr}_A\left[E_k^A\right]\right)\otimes\mathbb{I}_A\,\,\,\,\,\,\mathrm{for}\,\,1\leq k\leq L\\
B_{L+1|k}&=&\frac{1}{\dim(\mathcal{H}_B)}\left(\mathbb{I}_Y-\mathrm{tr}_B\left[E_k^B\right]\right)\otimes\mathbb{I}_B\,\,\,\,\,\,\mathrm{for}\,\,1\leq k\leq L\\
A_{a|b}&=&\frac{1}{\dim(\mathcal{H}_A)}\mathbb{I}_X\otimes\mathbb{I}_A\,\,\,\,\,\,\mathrm{for}\,\, b>L,L+2\leq a\leq L+3\,\, \mathrm{and}\,\,a+b\equiv 1\,\mod\,2\nonumber\\\\
B_{b|a}&=&\frac{1}{\dim(\mathcal{H}_B)}\mathbb{I}_Y\otimes\mathbb{I}_B\,\,\,\,\,\, \mathrm{for}\,\, a>L,L+2\leq b\leq L+3\,\, \mathrm{and}\,\,a+b\equiv 0\,\mod\,2\nonumber\\\\
A_{a|b}&=&0\,\,\,\,\,\,\,\mathrm{else},\\
B_{b|a}&=&0\,\,\,\,\,\,\, \mathrm{else}.
\end{eqnarray}

\section{CC* and classical quantum processes}
 A quantum process \cite{OFC} is a higher order map transforming quantum operations to another quantum operation.   We analyze the relationship between  CC* in LOCC* and classical quantum processes in the higher order formalism using the CJ operator representation.  We denote the set of all positive semi-definite operators on a Hilbert space $\mathcal{H}$ as $\mathbf{Pos}(\mathcal{H})$ and the set of all CJ operators on $\mathcal{H}_{I}\otimes\mathcal{H}_{O}$ representing CPTP maps from $\mathcal{H}_I$
to $\mathcal{H}_O$ as $CPTP(\mathcal{H}_{I}:\mathcal{H}_{O})=\{ Q \in \mathbf{Pos}(\mathcal{H}_{I}\otimes\mathcal{H}_{O})|\mathrm{tr}_{O} Q =\mathbb{I}_I\}$.

To be consistent with quantum mechanics, it has been proven in \cite{OFC} that a quantum process linking two local operations represented by the CJ operators $M_A\in \mathbf{L}(\mathcal{H}_{I_A}\otimes\mathcal{H}_{O_A})$ and $M_B\in \mathbf{L}(\mathcal{H}_{I_B}\otimes\mathcal{H}_{O_B})$ as inputs can be represented by a positive semi-definite operator $W \in \mathbf{Pos}(\mathcal{H}_{I_A}\otimes\mathcal{H}_{O_A}\otimes\mathcal{H}_{I_B}\otimes\mathcal{H}_{O_B})$ called a {\it process matrix}  satisfying
\begin{equation}
\forall M_A\in CPTP(\mathcal{H}_{I_A}:\mathcal{H}_{O_A}), \forall M_B\in CPTP(\mathcal{H}_{I_B}:\mathcal{H}_{O_B}),\,\, \mathrm{tr}\big[W(M_A^\mathrm{T}\otimes M_B^\mathrm{T})\big]=1,
\label{eq:process}
\end{equation}
where $Q^\mathrm{T}$ represents the transposition of $Q$ with respect to the computational basis used in defining the CJ representation. 

A deterministic joint quantum operation $\mathcal{M}: \mathbf{L}(\mathcal{H}_X\otimes\mathcal{M}_Y)\rightarrow \mathbf{L}(\mathcal{H}_A\otimes\mathcal{H}_B)$ can be regarded to be obtained by transforming two local operations  $\mathcal{A}: \mathbf{L}(\mathcal{H}_{I_A}\otimes\mathcal{H}_X)\rightarrow \mathbf{L}(\mathcal{H}_{O_A}\otimes\mathcal{H}_A)$ and $\mathcal{B}: \mathbf{L}(\mathcal{H}_{I_B}\otimes\mathcal{H}_Y)\rightarrow \mathbf{L}(\mathcal{H}_{O_B}\otimes\mathcal{H}_B)$ by  a quantum process $W \in \mathbf{Pos}(\mathcal{H}_{I_A}\otimes\mathcal{H}_{O_A}\otimes\mathcal{H}_{I_B}\otimes\mathcal{H}_{O_B})$  linking the Hilbert spaces $\mathcal{H}_{I_A}$, $\mathcal{H}_{O_A}$, $\mathcal{H}_{I_B}$ and $\mathcal{H}_{O_B}$.  The CJ operator $M \in \mathbf{L}(\mathcal{H}_X\otimes\mathcal{H}_Y\otimes \mathcal{H}_A\otimes\mathcal{H}_B)$ of $\mathcal{M}$ obtained by transforming the CJ operators of local operations $A \in \mathbf{L}(\mathcal{H}_{I_A}\otimes\mathcal{H}_X \otimes \mathcal{H}_{O_A}\otimes\mathcal{H}_A )$ and $B\in \mathbf{L}(\mathcal{H}_{I_B}\otimes\mathcal{H}_Y \otimes \mathcal{H}_{O_B}\otimes\mathcal{H}_B )$ corresponding to $\mathcal{A}$ and $\mathcal{B}$, respectively, by $W$ satisfying Eq.(\ref{eq:process})  is represented by 
\begin{equation}
M=\mathrm{tr}_{I_A,O_A,I_B,O_B}\big[W(A^{\mathrm{T}_{I_A,O_A}}\otimes B^{\mathrm{T}_{I_B,O_B}})\big],
\label{eq:LOCC**Wmatrix}
\end{equation}
where $Q^{\mathrm{T}_{I_X,O_X}}$ is the partial transposition of $Q$, taking the transposition only in terms of $\mathcal{H}_{I_X}$ and $\mathcal{H}_{O_X}$.

Quantum communication from Alice to Bob or Bob to Alice can be described by $W$ if $W$ is equivalent to the CJ operator  representing a quantum channel  linking two causally ordered local operations belonging to different parties. Moreover, quantum processes can represent ``quantum communication'' without causal order, which is not implementable when the partial order of local operations is fixed but is not be ruled out in the framework of quantum mechanics \cite{OFC}.    The restriction for quantum processes given by Eq.~\eqref{eq:process} can be interpreted to represent a new kind of causality required for linking local operations in quantum mechanics.

 We consider a special class of quantum process where Alice and Bob can communicate only  by ``classical'' states, namely, the states on the Hilbert spaces of $W$ ($\mathcal{H}_{I_A}$, $\mathcal{H}_{O_A}$, $\mathcal{H}_{I_B}$ and $\mathcal{H}_{O_B}$) are restricted to be diagonal with respect to the computational basis.  In such a case,  {\it classical} quantum processes representing classical communication between the parties can be described by a conditional probability distribution $p(i_A,i_B|o_A,o_B)$ satisfying
\begin{equation}
\sum_{i_A,i_B,o_A,o_B} p(i_A,i_B|o_A,o_B)A_{o_A|i_A}\otimes B_{o_B|i_B}\in CPTP(\mathcal{H}_X\otimes\mathcal{H}_Y:\mathcal{H}_A\otimes \mathcal{H}_B)
\label{eq:LOCC**}
\end{equation}
for all quantum instruments $\{A_{o_A|i_A}\}_{o_A}$ and $\{B_{o_B|i_B} \}_{o_B}$, where $\{A_{o_A|i_A}\in \mathbf{Pos}(\mathcal{H}_X\otimes \mathcal{H}_A)\}_{o_A}$ and $\{B_{o_B|i_B}\in \mathbf{Pos}(\mathcal{H}_Y\otimes \mathcal{H}_B)\}_{o_B}$ satisfying $\sum_{o_A}\mathrm{tr}_{A}[A_{o_A|i_A}]=\mathbb{I}_X$ and $\sum_{o_B}\mathrm{tr}_{B}[B_{o_B|i_B}]=\mathbb{I}_Y$.
The proof is given in Appendix B.7.

A  deterministic joint quantum operation $\mathcal{M}$ implementable by  a classical quantum process connecting two local operations is given in the form of
\begin{equation}
\mathcal{M}=\sum_{i_A,i_B,o_A,o_B}p(i_A,i_B|o_A,o_B)\mathcal{A}_{o_A|i_A}\otimes\mathcal{B}_{o_B|i_B},
\end{equation}
where $p(i_A,i_B|o_A,o_B)$ is a conditional probability distribution satisfying Eq.\eqref{eq:LOCC**}. We refer to the set of such joint quantum operations  as LOCQP.  By definition, LOCC* is a larger set than LOCQP  since the condition for $p(i_A,i_B|o_A,o_B)$ given by Eq.\eqref{eq:LOCC**} should be satisfied for {\it all} local operations in LOCQP while it should be satisfied only for {\it some} local operations in LOCC*.   In Appendix B.7, we prove that LOCQP is equivalent to a probability mixture of one-way LOCC in  a bipartite scenario.   Thus,  CC* used in implementing a deterministic joint quantum computation in LOCC* but not in LOCC cannot be represented by a classical quantum process.  This indicates that the required causal property for classical communication liking local operations is weakened for CC* comparing to the special relativistic causality and also to the restriction for classical quantum processes.

\section{CQP and LOCC}
This section consists of two parts. In Part A, we give a rigorous definition of a classical quantum process (CQP) and show a correspondence between a CQP and a conditional probability distribution. In Part B, we show that a set of joint quantum operations consisting local operations and CQP (LOCQP) is equivalent to a set of probabilistic mixture of one-way LOCC in bipartite cases
\subsubsection*{Part A}
A {\it classical quantum process} represents a link of local operations for a situation that Alice and Bob can communicate only by ``classical" states.  In this case, the states on the local input and output Hilbert spaces $\mathcal{H}_{I_A}$, $\mathcal{H}_{O_A}$, $\mathcal{H}_{I_B}$ and $\mathcal{H}_{O_B}$ of a process matrix $W$ are restricted to be diagonal with respect to the computational basis. By denoting the computational basis of $\mathcal{H}_{X}$ as $\{\ket{x}_{X}\}$,  local operations $ M_A\in CPTP(\mathcal{H}_{I_A}:\mathcal{H}_{O_A})$, $ M_B\in CPTP(\mathcal{H}_{I_B}:\mathcal{H}_{O_B})$, and a classical quantum process described by a diagonal process matrix $W\in \mathbf{Pos}(\mathcal{H}_{I_A}\otimes\mathcal{H}_{O_A}\otimes\mathcal{H}_{I_B}\otimes\mathcal{H}_{O_B})$ can be decomposed into
\begin{eqnarray}
 M_A&=&\sum_{i_A,o_A}p_A(o_A|i_A)\ket{i_A,o_A}\bra{i_A,o_A}_{I_A,O_A},\\
 M_B&=&\sum_{i_B,o_B}p_B(o_B|i_B)\ket{i_B,o_B}\bra{i_B,o_B}_{I_B,O_B},\\
W&=&\sum_{i_A,o_A,i_B,o_B}w(i_A,i_B,o_A,o_B)\ket{i_A}\bra{i_A}_{I_A}\otimes\ket{i_B}\bra{i_B}_{I_B}\otimes\ket{o_A}\bra{o_A}_{O_A}\otimes\ket{o_B}\bra{o_B}_{O_B},\nonumber\\ \label{eq:diagonalW}
\end{eqnarray}
 where $w(i_A,i_B,o_A,o_B)$ represents a diagonal element of $W$, and $p_A(o_A|i_A)$ and $p_B(o_B|i_B)$ are conditional probability distributions since $M_A$ and $M_A$ are CPTP maps.  In the following, we show that classical quantum processes correspond to conditional probability distributions, namely, $w(i_A,i_B,o_A,o_B)$ must be a conditional probability distribution so that Eq.~\eqref{eq:process} holds.
The non-negativity of $W$ implies $w(i_A,i_B,o_A,o_B)\geq 0$ and the condition given by Eq.~\eqref{eq:process} is equivalent to 
\begin{equation}
\sum_{i_A,i_B,o_A,o_B} w(i_A,i_B,o_A,o_B)p_A(o_A|i_A)p_B(o_B|i_B)=1
\end{equation}
for all conditional probability distributions $p_A(o_A|i_A)$ and $p_B(o_B|i_B)$.  By choosing $p_A(o_A|i_A)=\delta_{o_A,a}$ and $p_B(o_B|i_B)=\delta_{o_B,b}$, we obtain $\sum_{i_A,i_B} w(i_A,i_B,a,b)=1$ for arbitrary $a$ and $b$.  Thus $w(i_A,i_B,o_A,o_B)$ can be represented by a conditional probability distribution conditioned by $o_A$ and $o_B$ and we define
\begin{equation}
p(i_A,i_B|o_A,o_B):=w(i_A,i_B,o_A,o_B).
\end{equation}

Next, we define a set of deterministic joint quantum operations consisting of local operations linked by a classical quantum process, denoted by LOCC**. Local operations $A\in CPTP(\mathcal{H}_{I_A}\otimes\mathcal{H}_{X}:\mathcal{H}_{O_A}\otimes\mathcal{H}_{A})$ and $B\in CPTP(\mathcal{H}_{I_B}\otimes\mathcal{H}_{Y}:\mathcal{H}_{O_B}\otimes\mathcal{H}_{B})$ linked by a diagonal process matrix $W\in \mathbf{Pos}(\mathcal{H}_{I_A}\otimes\mathcal{H}_{I_B}\otimes\mathcal{H}_{O_A}\otimes\mathcal{H}_{O_B})$ given in the form of Eq.~\eqref{eq:diagonalW} can be decomposed into
\begin{eqnarray}
A&=&\sum_{i_A,o_A}A_{o_A|i_A}\otimes\ket{i_A,o_A}\bra{i_A,o_A}_{I_A,O_A},\\
B&=&\sum_{i_B,o_B}B_{o_B|i_B}\otimes\ket{i_B,o_B}\bra{i_B,o_B}_{I_B,O_B},
\end{eqnarray}
where $\{A_{o_A|i_A}\in \mathbf{Pos}(\mathcal{H}_X\otimes \mathcal{H}_A)\}_{o_A}$ and $\{B_{o_B|i_B}\in \mathbf{Pos}(\mathcal{H}_Y\otimes \mathcal{H}_B)\}_{o_B}$ are the CJ operators of quantum instruments since $A$ and $B$ are CPTP maps.
By straightforward calculation, the CJ operator of a deterministic joint quantum operation is written by
\begin{eqnarray}
M&=&\mathrm{tr}_{I_A,O_A,I_B,O_B}\big[W(A^{\mathrm{T}_{I_A,O_A}}\otimes B^{\mathrm{T}_{I_B,O_B}})\big]
\\
&=&\sum_{i_A,i_B,o_A,o_B}p(i_A,i_B|o_A,o_B)A_{o_A|i_A}\otimes B_{o_B|i_B},
\label{eq:def1_LOCC*}
\end{eqnarray}
 which shows the equivalence of the representations of $M$ given by Eq.~\eqref{eq:LOCC*CJ} and Eq.~\eqref{eq:LOCC**Wmatrix}  when the process is diagonal/classical.  LOCQP is defined by a set of deterministic joint quantum operations of which CJ operators can be represented by Eq.~\eqref{eq:def1_LOCC*}, where $p(i_A,i_B|o_A,o_B)$ satisfies
\begin{equation}
\sum_{i_A,i_B,o_A,o_B} p(i_A,i_B|o_A,o_B)p_A(o_A|i_A)p_B(o_B|i_B)=1
\label{eq:def1_CC**}
\end{equation}
for all conditional probability distributions $p_A(o_A|i_A)$ and $p_B(o_B|i_B)$.   

Further, the condition given by Eq.~\eqref{eq:def1_CC**} is equivalent to that of Eq.~\eqref{eq:LOCC**} if $p(i_A,i_B|o_A,o_B)$ is a conditional probability distribution.  This can be shown as follows.   By letting $\dim(\mathcal{H}_X)=\dim(\mathcal{H}_Y)=\dim(\mathcal{H}_A)=\dim(\mathcal{H}_B)=1$, $A_{o_A|i_A}=p(o_A|i_A)$ and $B_{o_B|i_B}=p(o_B|i_B)$, it is easily checked that Eq.~\eqref{eq:def1_CC**} holds if Eq.~\eqref{eq:LOCC**} holds.  To show the converse, we first check that a map decomposable in the form of Eq.~\eqref{eq:LOCC**} is completely positive.  This is also easily checked since every term in the summation is non-negative.   Then we show that a map decomposable in the form of Eq.~\eqref{eq:LOCC**} is trace preserving when Eq.~\eqref{eq:def1_CC**} holds in the following.
For any operator $\sigma\in \mathbf{L}(\mathcal{H}_X\otimes\mathcal{H}_Y)$,
\begin{eqnarray}
&\mathrm{tr}_{A,B}\left[\mathrm{tr}_{X,Y}\big[\sum_{i_A,i_B,o_A,o_B} p(i_A,i_B|o_A,o_B)(A_{o_A|i_A}\otimes B_{o_B|i_B}) \sigma^\mathrm{T}\big]\right]\\
=&\sum_{k,l}\lambda_{k,l}\sum_{i_A,i_B,o_A,o_B} p(i_A,i_B|o_A,o_B)\mathrm{tr}_{A,X}\big[A_{o_A|i_A}\rho_k^\mathrm{T}\big]\mathrm{tr}_{B,Y}\big[B_{o_B|i_B}\rho_l^\mathrm{T}\big]\\
=&\sum_{k,l}\lambda_{k,l}=\mathrm{tr}[\sigma],
\end{eqnarray}
where $\sigma$ is decomposed as $\sigma=\sum_{k.l}\lambda_{k,l}\rho_k\otimes\rho_l$ by using density operators $\{\rho_k\}_k$ as a basis of the linear space of the operator.  Note that $\mathrm{tr}_{A,X}\big[A_{o_A|i_A}\rho_k^T\big]$ is a conditional probability distribution conditioned by $i_A$ since it satisfies $\sum_{o_A}\mathrm{tr}_{A,X}\big[A_{o_A|i_A}\rho_k^T\big]=\mathrm{tr}_{A,X}\big[\sum_{o_A}A_{o_A|i_A}\rho_k^T\big]=\mathrm{tr}_{X}\big[\rho_k^T\big]=1$.

\subsubsection*{Part B}
We show a proof of the following lemma.

\begin{lemma}
LOCQP is equivalent to a probability mixture of one-way LOCC in bipartite cases.
\end{lemma}

\begin{proof}
In \cite{OFC}, it has been shown that any classical quantum process is {\it causally separable}, i.e. the CJ operator of a classical quantum process $W$ can be decomposed into the form
\begin{equation}
W=qW_{A\rightarrow B}+(1-q)W_{B\rightarrow A},
\end{equation}
where $q\in [0,1]$, $W_{A\rightarrow B}\in \mathbf{Pos}(\mathcal{H}_{I_A}\otimes\mathcal{H}_{O_A}\otimes\mathcal{H}_{I_B}\otimes\mathcal{H}_{O_B})$ is diagonal with respect to the computational basis and satisfies the conditions given by 
\begin{eqnarray}
W_{A\rightarrow B}=\mathbb{I}_{O_B}\otimes W_{I_A,O_A,I_B},\,\,\,\,\,\,\,\,\,
\mathrm{tr}_{I_B}\left[W_{I_A,O_A,I_B}\right]=\mathbb{I}_{O_A}\otimes\rho_{I_A},\,\,\,\,\,\,\,\,\,
\mathrm{tr}_{I_A}\left[\rho_{I_A}\right]=1,
\label{eq:AtoB}
\end{eqnarray}
and the similar conditions are satisfied by $W_{B\rightarrow A}\in \mathbf{Pos}(\mathcal{H}_{I_A}\otimes\mathcal{H}_{O_A}\otimes\mathcal{H}_{I_B}\otimes\mathcal{H}_{O_B})$.  

In \cite{Chiribella1},   it has been proven that an operator $W_{A\rightarrow B}$ satisfying the conditions given by Eq.~\eqref{eq:AtoB}  but not necessarily being diagonal in the computational basis corresponds to a special type of quantum process called {\it quantum comb} where Alice's operation and Bob's operation are linked by quantum communication from Alice to Bob.   Thus a causally separable process can be interpreted as a probabilistic mixture of quantum communication from Alice to Bob and that from Bob to Alice. When a causally separable process is classical (diagonal with respect to the computational basis), the process can be interpreted as a probabilistic mixture of classical communication from Alice to Bob and that from Bob to Alice.

 Let denote the diagonal elements of $W_{A\rightarrow B}$ and $W_{B\rightarrow A}$ with respect to the computational basis  by $p_{A\rightarrow B}(i_A,i_B|o_A,o_B)$ and $p_{B\rightarrow A}(i_A,i_B|o_A,o_B)$, respectively.   It is easy to verify that  $p_{A\rightarrow B}(i_A,i_B|o_A,o_B)$ and $p_{B\rightarrow A}(i_A,i_B|o_A,o_B)$ are conditional probability distributions since $W_{A\rightarrow B}$ and $W_{B\rightarrow A}$ correspond to classical quantum processes.

Then $p(i_A,i_B|o_A,o_B)$ satisfying Eq.~\eqref{eq:LOCC**}
can be decomposed into
\begin{equation}
p(i_A,i_B|o_A,o_B)=qp_{A\rightarrow B}(i_A,i_B|o_A,o_B)+(1-q)p_{B\rightarrow A}(i_A,i_B|o_A,o_B).
\end{equation}
Eq.~\eqref{eq:AtoB} implies $p_{A\rightarrow B}(i_A,i_B|o_A,o_B)$ does not depend on $o_B$.  We define $p_{A\rightarrow B}(i_A,i_B|o_A):=p_{A\rightarrow B}(i_A,i_B|o_A,o_B)$.  The operator $W_{I_A,O_A,I_B}$ in Eq.~\eqref{eq:AtoB} is given by
\begin{equation}
W_{I_A,O_A,I_B}=\sum_{i_A,i_B,o_A}p_{A\rightarrow B}(i_A,i_B|o_A)\ket{i_A}\bra{i_A}_{I_A}\otimes\ket{o_A}\bra{o_A}_{O_A}\otimes\ket{i_B}\bra{i_B}_{I_B}.
\end{equation}
Due to Eq.~\eqref{eq:AtoB}, $\sum_{i_B}p_{A\rightarrow B}(i_A,i_B|o_A)$ does not depend on $o_A$.  We define $p_{A\rightarrow B}(i_A):=\sum_{i_B}p_{A\rightarrow B}(i_A,i_B|o_A)$.
We further define a set $X=\{x|p_{A\rightarrow B}(x)\neq 0\}$ and
\begin{eqnarray}
p'_{A\rightarrow B}(i_B|o_A,i_A)&:=&\frac{p_{A\rightarrow B}(i_A,i_B|o_A)}{p_{A\rightarrow B}(i_A)}\,\,\,\,(i_A\in X)\\
p'_{A\rightarrow B}(i_B|o_A,i_A)&:=&p'_{A\rightarrow B}(i_B)\,\,\,\,(i_A\notin X),
\end{eqnarray}
where $p'_{A\rightarrow B}(y)$ is an arbitrary probability distribution.  Note that $p'_{A\rightarrow B}(i_B|o_A,i_A)$ satisfies all properties required for a conditional probability distribution. Thus, we have
\begin{equation}
p_{A\rightarrow B}(i_A,i_B|o_A,o_B)=p_{A\rightarrow B}(i_A)p'_{A\rightarrow B}(i_B|o_A,i_A),
\end{equation}
and similarly,
\begin{equation}
p_{B\rightarrow A}(i_A,i_B|o_A,o_B)=p_{B\rightarrow A}(i_B)p'_{B\rightarrow A}(i_A|o_B,i_B).
\end{equation}
Combining these results, a deterministic joint quantum operation $\mathcal{M}$ in LOCQP is given by
\begin{eqnarray}
\mathcal{M}&=&q\sum_{i_A,i_B,o_A,o_B}p_{A\rightarrow B}(i_A,i_B|o_A,o_B)\mathcal{A}_{o_A|i_A}\otimes\mathcal{B}_{o_B|i_B}\nonumber\\
&&+(1-q)\sum_{i_A,i_B,o_A,o_B}p_{B\rightarrow A}(i_A,i_B|o_A,o_B)\mathcal{A}_{o_A|i_A}\otimes\mathcal{B}_{o_B|i_B}.
\end{eqnarray}
Since  $\sum_{i_A,i_B,o_A,o_B}p_{A\rightarrow B}(i_A,i_B|o_A,o_B)\mathcal{A}_{o_A|i_A}\otimes\mathcal{B}_{o_B|i_B}=\sum_m\mathcal{A}_m\otimes\mathcal{B}_{|m}$ where 
\begin{equation}
m:=(i_A,o_A),~\mathcal{A}_{i_A,o_A}:=p_{A\rightarrow B}(i_A)\mathcal{A}_{o_A|i_A},~\mathcal{B}_{|i_A,o_A}:=\sum_{i_B,o_B}p'_{A\rightarrow B}(i_B|o_A,i_A)\mathcal{B}_{o_B|i_B}
\end{equation}
represents operations in one-way LOCC,  the deterministic joint quantum operation $\mathcal{M}$ in LOCQP is concluded to be a probability mixture of one-way LOCC.

\end{proof}

\section{Classical causally non-separable process}
We analyze a tripartite case for describing deterministic joint quantum operations.
Consider Alice performs a CPTP map $A\in CPTP(\mathcal{H}_1\otimes\mathcal{H}_7:\mathcal{H}_2\otimes\mathcal{H}_8)$, Bob performs a CPTP map $B\in CPTP(\mathcal{H}_3\otimes\mathcal{H}_9:\mathcal{H}_4\otimes\mathcal{H}_{10})$, Charlie performs a CPTP map $C\in CPTP(\mathcal{H}_5\otimes\mathcal{H}_{11}:\mathcal{H}_6\otimes\mathcal{H}_{12})$ and spaces $\mathcal{H}_i$ $(7\leq i\leq 12)$ are connected by a process matrix $W\in \mathbf{Pos}(\mathcal{H}_7\otimes\mathcal{H}_8\otimes\mathcal{H}_9\otimes\mathcal{H}_{10}\otimes\mathcal{H}_{11}\otimes\mathcal{H}_{12})$, where
$\dim(\mathcal{H}_i)=2$ $(7\leq i\leq 12)$.
\begin{equation}
W=\frac{1}{8}\left[\mathbb{I}_{7,8,9,10,11,12}+Z_7\mathbb{I}_8\mathbb{I}_9Z_{10}Z_{11}Z_{12}+Z_7Z_8Z_9\mathbb{I}_{10}\mathbb{I}_{11}Z_{12}+\mathbb{I}_7Z_8Z_9Z_{10}Z_{11}\mathbb{I}_{12}\right]
\end{equation}
is known to be a classical causally non-separable process matrix \cite{Amin}, where we abbreviate $\otimes$.
This process matrix corresponds to CC* given in Eq.~\eqref{eq:nonSEP}.

\section{Multipartite LOCC*}
 For a deterministic joint quantum operation $\mathcal{M}$ in SEP, the corresponding CJ  operator $M$ is given by
\begin{equation}
M=\sum_{l=1}^L E_l^{(1)}\otimes\cdots\otimes E_l^{(N)},
\label{eq:defSEP}
\end{equation}
where $E_l^{(n)}\in \mathbf{L}(\mathcal{H}_{I_n}\otimes \mathcal{H}_{O_n})$ satisfies
\begin{eqnarray}
E_l^{(n)}\geq0,
\end{eqnarray}
for all $l$ and $n$.    It is easy to see that LOCC* represented by Eq.\eqref{eq:mLOCC*} is a subset of SEP.   We shall show the converse that it is possible to construct a probability distribution $p(i_1,\cdots,i_N|o_1,\cdots,o_N)$ and local operators $\{\{E_{o_n|i_n}^{(n)}\in \mathbf{L}(\mathcal{H}_{I_n}\otimes \mathcal{H}_{O_n})\}_{o_n}\}_{n=1}^N$ in LOCC* for any given SEP element.

Without loss of generality, we can assume that
\begin{eqnarray}
\mathrm{tr}_{O_n}\left[E_l^{(n)}\right]\leq \mathbb{I}_{I_n},
\label{eq:2}
\end{eqnarray}
for all $l$ and $n$.

Now we present a construction of local operators for LOCC* as follows.
\begin{eqnarray}
p(i_1,\cdots,i_N|o_1,\cdots,o_N)&=&\prod_{n=1}^N\delta_{i_n,o_n}\\
E_{l|l}^{(n)}&=&E_l^{(n)}\,\,\,\,\,\,\mathrm{for}\,\,1\leq n\leq N,\,\,1\leq l\leq L \\
E_{L+1|l}^{(n)}&=&\frac{1}{\dim(\mathcal{H}_{O_n})}\left(\mathbb{I}_{I_n}-\mathrm{tr}_{O_n}\left[E_l^{(n)}\right]\right)\otimes\mathbb{I}_{O_n}\nonumber\\
&&\mathrm{for}\,\,1\leq n\leq N,\,\,1\leq l\leq L+1\\
E_{1|L+1}^{(n)}&=&\frac{1}{\dim(\mathcal{H}_{O_n})}\mathbb{I}_{I_n}\otimes\mathbb{I}_{O_n}\,\,\,\,\,\,\mathrm{for}\,\,1\leq n\leq N\\
E_{o_n|i_n}^{(n)}&=&0\,\,\,\,\,\,\,\mathrm{else}.
\end{eqnarray}

\end{document}